\documentclass[a4paper,12pt,twoside]{book}
\usepackage[utf8]{inputenc}
\usepackage[T1]{fontenc} 
\usepackage{color}
\usepackage{amsthm}
\usepackage{amsmath}
\usepackage{graphicx}

\usepackage{amssymb,amsmath, bm}
\usepackage{mathtools}
\usepackage{relsize,exscale}
\usepackage{subcaption}
\usepackage{fancyhdr}   
\usepackage{verbatim}
\usepackage{ mathrsfs }
\usepackage[unicode]{hyperref}
\usepackage{wrapfig}
\usepackage{caption}
\usepackage[left=2.5cm, right=2.5cm, top=2.5cm, bottom=3.55cm]{geometry}

\usepackage{caption}
\usepackage{float}
\usepackage{ulem}

\usepackage{tabularx}
\usepackage{booktabs}
\usepackage{array}

\newcommand{\kolo}[1]{\!\vphantom{#1}\stackrel{\circ}{#1}\!\vphantom{#1}\!}

\newcommand{\one}[1]{\vphantom{#1}\overset{1}{#1}\vphantom{#1}}
\newcommand{\two}[1]{\vphantom{#1}\overset{2}{#1}\vphantom{#1}}

\newcommand{\Kz}{{\bf K}}
\newcommand{\Ko}{\!\one{K}}
\newcommand{\Kt}{\!\two{K}}
\newcommand{\No}{\!\one{N}}
\newcommand{\Nt}{\!\two{N}}
\newcommand{\Ao}{\!\one{A}}

\newcommand{\ho}{\!\one{h}}

\newcommand{\tA}{\widetilde{A}}
\newcommand{\tAo}{\!\one{\tA}}

\newtheorem{theorem}{Theorem}[section]
\newtheorem{lemma}[theorem]{Lemma}
\newcommand{\dd}{\text{d}}
\newcommand{\cP}{{\cal P}}
\newcommand{\cJ}{{\cal J}}
\newcommand{\cO}{{\cal O}}

\newcommand{\cF}{{\cal F}}

\newcommand{\sgn}{\text{sgn}}

\newcommand{\Lag}{\mathcal{L}}
\newcommand{\mnabla}{\kolo \nabla}
\newcommand{\mBox}{\kolo \Box}

\newcommand{\mGamma}{\kolo \Gamma}

\begin{document}
    

    \begin{titlepage}%
    \let\footnotesize\small
    \let\footnoterule\relax
    \begin{center}%
        {\Large\textbf{Theory of Gravity  \\ as \\ Theory of Local Inertial Frames  }\par}
         \vspace{0.5cm}
        {{\small\textbf{Teoria Grawitacji  jako  Teoria Lokalnych Układów Inercjalnych}
        \par }}
        \vspace{0.7cm}
        {\Large\selectfont{Bartłomiej Bąk} \par}
        \fontsize{12pt}{14pt}\selectfont
        \vspace{1.4cm}
        \begin{figure}[h!]
            \centering
            \includegraphics[width=0.5 \textwidth]{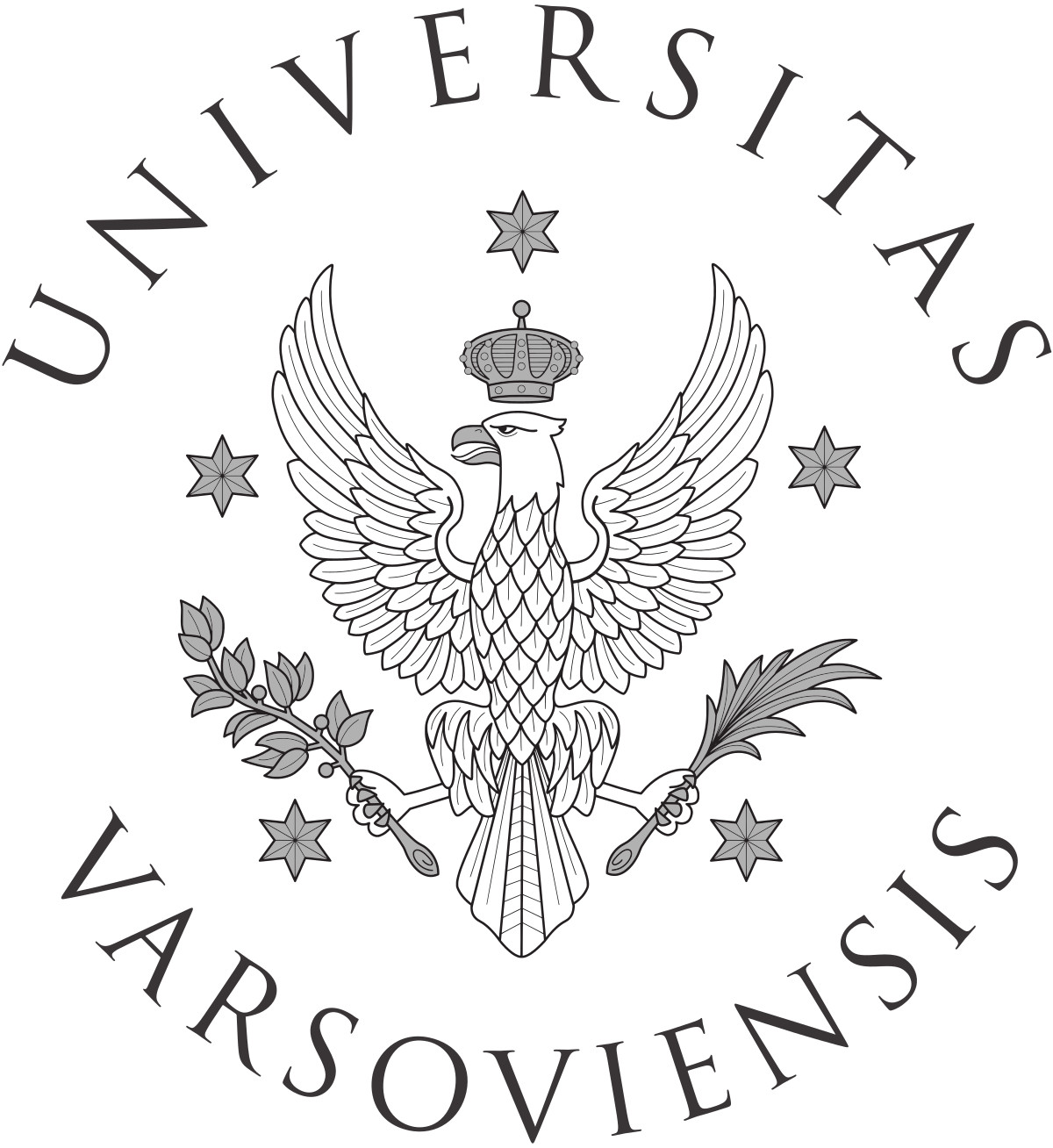}
        \end{figure}
        \vspace{1.4cm}
        \begin{centering}
            
            The thesis written under the supervision of 

             \ 
             
            {\large prof. dr. hab. Jacek Jezierski}\\ and\\ {\large  prof. dr. hab. Jerzy Kijowski}\\
            \vspace{0.7cm}
            Department of Mathematical Methods in Physics \\
            Faculty of Physics \\
            University of Warsaw

        \end{centering}
        \vspace{10mm plus .1fill}
        {\fontsize{14pt}{18pt} Warsaw, September 2025}
    \end{center}
\end{titlepage}%
\newpage
\pagenumbering{gobble} 
\begin{center}
    Gravity, Local Inertial Frames, Matter fields, Non-metricity  \\
\end{center}
\newpage 
\pagenumbering{arabic} 
    \section*{Abstract}

It is proved that the affine theory of the full Riemann tensor constitutes the extended theory of gravity. The fundamental object here is a spacetime symmetric connection. Its physical interpretation is that of a field of local inertial frames. In this approach, gravity arises, in a natural way, as a local version of   Newton's First Law. A variational formulation of the theory is presented, based on the results introduced in the article~\cite{nonmetricity}, co-authored by the present author and one of the supervisors. The general framework is supported by several examples.

A central result of the dissertation is the transition from the affine picture to the metric picture. A simplified version of this procedure — valid for a specific class of Lagrangians — was already analyzed in~\cite{nonmetricity}. It was proved there that the non-metricity in the affine picture can be interpreted as a matter field in the metric picture. This is precisely  Hermann Weyl's interpretation of electromagnetism (cf.~\cite{weyl}). Here, the transformation from the full Riemann tensor theory to the conventional metric theory is examined in full generality for the first time and illustrated using the previously introduced examples.

\section*{Streszczenie}

W pracy dowodzi się, że teoria afiniczna pełnego tensora Riemanna jest rozszerzoną teorią grawitacji. Podstawowym obiektem jest tu symetryczna koneksja na czasoprzestrzeni. Jej interpretacja fizyczna to pole lokalnych układów inercjalnych. W tym ujęciu grawitacja pojawia się, w sposób naturalny, jako lokalna wersja Pierwszej Zasady Dynamiki Newtona. Przedstawiono wariacyjne sformułowanie teorii, oparte na wynikach artykułu~\cite{nonmetricity}, którego współautorem jest autor dysertacji i jeden z promotorów. Ogólny schemat teorii jest poparty przykładami.

Głównym rezultatem rozprawy jest przejście od obrazu afinicznego do obrazu metrycznego. Uproszczona wersja tej procedury –- zawężona do pewnej klasy lagranżjanów –- została już opracowana w pracy~\cite{nonmetricity}. Udowodniono tam, że niemetryczność w obrazie afinicznym można interpretować w obrazie metrycznym jako pole materii. Tak właśnie interpretował pole elektromagnetyczne Hermann Weyl (zob.~\cite{weyl}). W niniejszej pracy po raz pierwszy opisano, w pełnej ogólności, transformację od pełnej teorii tensora Riemanna do konwencjonalnej teorii metrycznej, ilustrując ją wcześniej wprowadzonymi przykładami.

\newpage
    \section*{Acknowledgements}

First and foremost, I am profoundly grateful to my supervisors,  \textbf{prof. dr hab. Jerzy Kijowski} and \textbf{prof. dr hab. Jacek Jezierski},  for their unwavering guidance, support, and patience throughout these years, many enlightening discussions, and advice. You have been the most influential figures in my academic life, and I will never forget that it was you who introduced me to the beauty of mathematical methods in physics and taught me how to write, present, and defend my results. Finally, you gave me the opportunity  to  provide my own, independent work.

\ 

I would also like to thank the entire community of the \textbf{Faculty of Physics}, and in particular the \textbf{Department of Mathematical Methods in Physics}, to which I belong, for their constant help, encouragement, and inspiring atmosphere.  

 \ 
 
My deep gratitude goes to my teachers, especially my physics teachers \textbf{Tomasz Fatyga} and \textbf{Genowefa Gajger}, whose passion and dedication first inspired me to pursue physics.  

\ 

A heartfelt thanks is due to my friends \textbf{Kasia Wardęga}, \textbf{Paulina Michalak}, and \textbf{Robert Grosz}, whom I met during my studies, for their unfailing friendship, encouragement, and support in both academic and everyday matters.  

\ 

I am especially grateful to my office-mates from room 5.68, in particular \textbf{Bartosz Zawora} and \textbf{Norbert Mokrzański}, for countless valuable discussions, shared dinners, and a daily dose of good humour, which made research much more enjoyable.  

\ 

Special thanks are also reserved for my friends from the \textbf{University of Warsaw Judo Section}, led by \textbf{Artur Stepnowski}, for the invisible yet invaluable physical and mental support. The training gave me a place to clear my mind, regain balance, and strengthen myself both physically and mentally.  

\ 

I cannot fully express my gratitude to my family, especially my parents,  who nurtured my growth since childhood and never held me back. You gave me the freedom and strength to embark on this scientific journey,  and I hope I have made good use of it.  

\ 

Last but certainly not least, I wish to express my deepest gratitude to my beloved wife \textbf{Karolina} and my son \textbf{Wiktor}, for whom I never give up and from whom I draw my greatest motivation to continue my work and responsibilities.

\newpage
    
    \tableofcontents
    
    \chapter{Introduction}
\section{Content}

The thesis consists of an introduction, three main chapters, a summary, and an appendix.

The \textbf{Introduction} provides a brief motivation for the research, an overview of the dissertation, and a description of the conventions, notation, and key geometrical objects used throughout the work.

 \textbf{Chapter~2: Preliminaries}  is divided into four sections:
\begin{itemize}
    \item \textbf{Origins} – discusses variational calculus, the relationship between the affine connection, inertial reference frames, and the gravitational field, as well as the variational structure of the metric picture.
    \item \textbf{Variational structure in the affine picture} – presents the variational formula for affine Lagrangians.
    \item \textbf{Construction of affine Lagrangians} – outlines methods for building affine Lagrangians from geometric quantities.
    \item \textbf{The scheme of deriving the approximated affine Lagrangians and field equations} – explains the scheme of derivation the field equations.
\end{itemize}

\textbf{Chapter~3: Affine Lagrangians} demonstrates the application of the above scheme to four explicit affine Lagrangians.

\textbf{Chapter~4: Metric Lagrangians} presents the passage from the affine picture to the metric picture at the variational level and derives the corresponding metric Lagrangians for the previously introduced examples.

The \textbf{Summary} briefly outlines the main results obtained in the thesis and suggests directions for future research.

The \textbf{Appendix} contains supplementary material on classical electrodynamics, Proca theory, and Fierz–Lanczos theory, which are used in the main body of the dissertation.

\section{Motivation}

The origins of gravity trace back to the 17th century, when Sir Isaac Newton formulated the law of universal gravitation and the three laws of dynamics in his renowned work \textit{Philosophiae Naturalis Principia Mathematica}~\cite{newton}. It is no exaggeration to say that Newton was one of the greatest scientists of all time, whose work influenced and inspired generations of physicists and mathematicians. The significance of his contributions can neither be overlooked nor overstated, as many of the theories and ideas presented therein remain relevant to this day. For example, the determination of spacecraft or satellite trajectories is still based on Newtonian gravity.

However, since the 19th century, observations — most notably by the astronomer Urbain Le Verrier, who discovered the anomalous apsidal precession of Mercury~\cite{Mercury} — have led to the conclusion that Newton’s description of gravity is insufficient.

A major breakthrough came in 1915, when Albert Einstein introduced a new framework in which gravity is associated with the curvature of spacetime (see \cite{ein1, ein2}). This concept, known as the \textit{general theory of relativity}, remains one of the most important theories in modern physics. Interestingly, the connection between gravity and geometry was originally proposed by the mathematician William Clifford in 1876~\cite{clifford}, a fact acknowledged by Einstein himself. Unfortunately, Clifford’s contribution was largely forgotten. The geometric development and, in particular, the variational formalism of Einstein’s gravity were further elaborated by David Hilbert in 1915~\cite{hil}. In this dissertation, such an approach is referred to as the \textit{metric picture}.

In 1919, Attilio Palatini~\cite{palatini} proposed a new formulation in which both the metric and an affine connection are treated as independent configuration fields. In this setting, the connection is not assumed \textit{a priori} to be compatible with the metric structure. In the vacuum case, the compatibility condition between the connection and the metric arises as one of the Euler–Lagrange equations, thereby reproducing the Einstein–Hilbert results. This approach, referred to as the \textit{Palatini picture}, represents an intermediate stage between the metric picture and the affine picture presented below — cf.~\cite{APP, nonmetricity}.

The next major development was introduced by Jerzy Kijowski in 1978~\cite{newvariationalprinciple}, following an earlier paper by his colleague Wiktor Szczyrba in 1976~\cite{Szczyrba} concerning the \textit{multisymplectic structure} of gravity theory. Kijowski discovered that Einstein’s equations can be derived from  a purely affine Lagrangian, depending solely on the connection and its first derivatives. The metric tensor emerges here as a momentum canonically conjugate to the Ricci tensor. His initial formulation considered gravity coupled to simple matter models, such as scalar or electromagnetic fields -- cf.~\cite{Kij-Fer}. This simplicity consists in the fact that the Lagrangian of the theory is sensitive to the symmetric Ricci tensor only. A natural extension includes the full Ricci tensor, comprising both its symmetric and skew-symmetric parts. The motivation for studying this more general theory lies in the observation that, when the full Ricci tensor is considered, the affine connection becomes non-metric. In this case, the non-metricity can be interpreted as a matter field in the metric picture — and \textit{vice versa}: matter fields in the metric picture (under suitable conditions) can be reinterpreted as components of a symmetric but non-metric affine connection — cf.~\cite{lic, mag, APP, nonmetricity}. Interestingly, the idea that matter can influence the metricity of the connection was already proposed by Hermann Weyl in 1918~\cite{weyl}.

A further generalisation of the affine framework involves the full Riemann tensor, which \textit{a priori} contains three independent components: the algebraic trace (i.e., the Ricci tensor), which splits into symmetric and skew-symmetric parts, and the remaining traceless part of the Riemann tensor. These geometric structures are assumed to correspond to physical fields or phenomena.

The symmetric part of the Ricci tensor is naturally associated with gravity. The affine theory of standard gravity is realised by the (unique, see \textbf{Chapter~\ref{Lambda vacuum grav}}) affine Lagrangian $\Lag_A = C\, \sqrt{|\det K|}$, where $K$ is the symmetric Ricci tensor, and $C$ is a global constant with units of $[\textbf{cm}^2]$ (in the geometrised unit system~\cite{Gravitation}), which naturally provides room for the cosmological constant $\Lambda$ (with units of $[\textbf{cm}^{-2}]$). In this case, the connection becomes metric due to the field equations.

The skew-symmetric part, being a closed 2-form (as shown later), can be interpreted as the electromagnetic field or, more generally, the Proca field (a massive spin-1 boson). Such fields also appear in modern cosmology in connection with so-called \textit{dark photons}~\cite{thedarkphoton, holdom, rogatko2024}. Interestingly, the only natural candidate for the affine Lagrangian of the full Ricci tensor $R$ is $\Lag_A = C\, \sqrt{|\det R|}$~\cite{kij2024}. Furthermore, to link the skew-symmetric Ricci tensor (which is dimensionless) with the electromagnetic field strength tensor (the Faraday 2-form, with dimensions~$[\textbf{cm}]$), a coupling constant is necessary — and again, the square root of the cosmological constant is the natural choice. The affine theory of the full Ricci tensor was the main topic of the author’s Bachelor Thesis~\cite{lic}. Interestingly, the conjecture that the skew-symmetric Ricci tensor is related to the Faraday 2-form was first proposed by Hermann Weyl in 1918~\cite{weyl}. This idea is a smooth continuation of the earlier observation linking non-metricity to matter, because in this case the affine connection is no longer metric, and the non-metricity is precisely described by the skew-symmetric Ricci tensor.

The traceless part of the Riemann tensor, however, has no clear physical interpretation at present. Nonetheless, it has been suggested that it may serve as a model for dark matter — an elusive component of the universe that has been observed indirectly for decades but remains poorly understood. Currently, dark matter is often described as an additional matter field, sometimes in combination with modifications of standard gravity — cf.~\cite{borowiec}. This conjecture aligns well with the emergence of effective matter fields from a non-metric affine connection. Moreover, upon transition to the metric picture, the traceless Riemann component gives rise to several effective matter fields. One of them is associated with the Weyl tensor, which is known to describe a massless spin-2 field — cf.~\cite{marian, lanczos}. Unfortunately, in this case, there is no obvious “natural” candidate for the affine Lagrangian. Consequently, various proposals have been put forward and investigated.

Of course, many other well-established extensions of Einsteinian gravity exist. These include:
\begin{enumerate}
    \item \textit{Lovelock gravity}~\cite{lovelock}, a metric theory involving higher-order terms of the Riemann tensor;
    \item \textit{Horndeski gravity}~\cite{horndeski}, a scalar–tensor theory linear in the Ricci tensor;
    \item $f(R)$ \textit{gravity}~\cite{buchdahl}, a class of theories based on functions of the Ricci scalar.
\end{enumerate}

Several other modern approaches are inspired by these models — cf.~\cite{  banados, fRgravity,Extended, cosmo, harada2, harada1, horava, Hor-Lif,  vollick}. Some of them are loosely related to the affine theory of the full Riemann tensor presented in this dissertation; however, none fully encompass it. This makes the current study a novel and independent exploration of an extended theory of gravity.

\section{Conventions, notation and useful geometric objects}
\sectionmark{Conventions}

\subsection{List of symbols}

\begin{tabularx}{\textwidth}{
  >{\raggedright\arraybackslash}p{1.9cm}
  >{\raggedright\arraybackslash}X
  >{\centering\arraybackslash}p{1.8cm}
}
\toprule
\textbf{Symbol} & \textbf{Meaning / Description} & \textbf{Equation} \\
\midrule
$g_{\mu\nu}$ & Metric tensor of the four-dimensional Lorentzian manifold \\
$\delta^{\kappa}_{\mu}$ & Four-dimensional Kronecker delta function \\
$\epsilon^{\kappa\lambda\mu\nu}$ &  Four-dimensional Levi-Civita symbol \\
$\Gamma^{\kappa}_{\ \lambda\mu}$ & General symmetric affine connection coefficients & \eqref{aff con} \\
$\mGamma^{\kappa}_{\ \lambda\mu}$ & Metric connection coefficients (Christoffel symbols) & \eqref{mGamma} \\
$\nabla_{\nu}$ & Covariant derivative with respect to~$\Gamma^{\kappa}_{\ \lambda\mu}$ & \eqref{nablaaf} \\
$\mnabla_{\nu}$ & Covariant derivative with respect to~$\mGamma^{\kappa}_{\ \lambda\mu}$ & \eqref{met con} \\
$\mBox$ & D'Alembert operator with respect to~$\mGamma^{\kappa}_{\ \lambda\mu}$ &\eqref{mnabla F}  \\
$N^\kappa_{\ \lambda\mu}$ & Non-metricity tensor & \eqref{decGamma} \\
$A_{\mu}$ & Algebraic trace of the non-metricity tensor $N^\kappa_{\ \lambda\mu}$ &\eqref{slad N} \\
$A^\kappa_{\ \lambda\mu}$ & Algebraically traceless part of the non-metricity tensor $N^\kappa_{\ \lambda\mu}$& \eqref{bezsladowe N} \\
$h^k$ & Metric trace of the tensor $A^{\kappa}_{\ \lambda\mu}$ & \eqref{def h}\\
$\tA^\kappa_{\ \lambda\mu}$ & Totally traceless part of the non-metricity tensor $N^\kappa_{\ \lambda\mu}$& \eqref{def tildeA}\\
$R^{\kappa}_{\ \lambda\mu\nu}$ & Riemann curvature tensor of the connection $\Gamma^\kappa_{\ \lambda\mu}$ &\eqref{def: tensor riemanna}\\
$R_{\mu\nu}$ & Ricci tensor of the  connection $\Gamma^{\kappa}_{\ \lambda\mu}$ &\eqref{def: tensor Ricciego} \\
$K_{\mu\nu}$ & Symmetric part of the Ricci tensor  $R_{\mu\nu}$ &\eqref{def: tensor K} \\
$Z_{\mu\nu}$ & Totally traceless part of the Ricci tensor $R_{\mu\nu}$ & \eqref{def: tensor Z}\\
$R$ & Metric trace of the tensor  $R_{\mu\nu}$ & \eqref{def: tensor Z} \\
$F_{\mu\nu}$ & Skew-symmetric part of the Ricci tensor  $R_{\mu\nu}$ &\eqref{def: tensor F}\\
$W^{\kappa}_{\ \lambda\mu\nu}$ & Algebraically traceless part of the Riemann tensor $R^{\kappa}_{\ \lambda\mu\nu}$ & \eqref{trless W} \\
$W^{\kappa}_{\ \nu}$ & Metric trace of the tensor $W^{\kappa}_{\ \lambda\mu\nu}$& \eqref{trace W}\\
$\widetilde{W}^{\kappa}_{\ \lambda\mu\nu}$ & Totally traceless part of the Riemann tensor $R^{\kappa}_{\ \lambda\mu\nu}$& \eqref{W decomposition}\\ 
$\mathfrak{w}_{\kappa\lambda\mu\nu}$ & Weyl tensor of the  connection $\Gamma^\kappa_{\ \lambda\mu}$ &\eqref{def frakw}\\
$\mathfrak{h}_{\kappa\lambda\mu\nu}$ & $\widetilde{W}_{\kappa \lambda\mu\nu}$ tensor component & \eqref{def frakh}\\
$K^{\kappa}_{\ \lambda\mu\nu}$ & Kijowski curvature tensor of the   connection $\Gamma^{\kappa}_{\ \lambda\mu}$ &\eqref{rel KR}\\
$U^{\kappa}_{\ \lambda\mu\nu}$ & Algebraically traceless part of the Kijowski tensor $K^{\kappa}_{\ \lambda\mu\nu}$ &\eqref{rel UW} \\
$\kolo{R}^{\kappa}_{\ \lambda\mu\nu}$ & Riemann curvature tensor of the metric connection $\mGamma^{\kappa}_{\ \lambda\mu}$ &\eqref{metRimm}\\
$\kolo{R}$ & Ricci scalar of the tensor  $\kolo{R}_{\mu\nu}$& \eqref{Ricci deco} \\
$\kolo{K}_{\mu\nu}$ & Symmetric part of the Ricci tensor  $\kolo{R}_{\mu\nu}$ & \eqref{K} \\
$\kolo{F}_{\mu\nu}$ & Skew-symmetric part of the Ricci tensor  $\kolo{R}_{\mu\nu}$ &\eqref{koloF}\\ 
$\kolo{W}^{\kappa}_{\ \lambda\mu\nu}$ & Algebraically traceless part of the Riemann tensor $\kolo{R}^{\kappa}_{\ \lambda\mu\nu}$ & \eqref{rozklad pelny W}  \\
$\kolo{\mathfrak{w}}_{\kappa\lambda\mu\nu}$ & Weyl tensor of the metric connection $\mGamma^\kappa_{\ \lambda\mu}$ & \eqref{Ricci deco}\\
$\kolo{K}^{\kappa}_{\ \lambda\mu\nu}$ & Kijowski curvature tensor of the metric connection $\mGamma^{\kappa}_{\ \lambda\mu}$ &\eqref{metKij}\\
$\kolo{U}^{\kappa}_{\ \lambda\mu\nu}$ & Algebraically traceless part of the Kijowski tensor $\kolo{K}^{\kappa}_{\ \lambda\mu\nu}$ & \eqref{rozklad pelny U}  \\
$\Lag$ & Lagrangian (the scalar density)& \eqref{action} \\
$\Lag_H$ & Hilbert Lagrangian associated to the metric Ricci tensor $\kolo{K}_{\mu\nu}$ & \eqref{LagH0} 
\end{tabularx}

\begin{tabularx}{\textwidth}{
  >{\raggedright\arraybackslash}p{1.9cm}
  >{\raggedright\arraybackslash}X
  >{\centering\arraybackslash}p{1.8cm}
}
\toprule
\textbf{Symbol} & \textbf{Meaning / Description} & \textbf{Equation} \\
\midrule
$\Lag_{\rm matt}$ & Matter Lagrangian &\eqref{delta lagmatt} \\
$\Lag_g$ & Metric Lagrangian & \eqref{def Lagg} \\
$\Lag_A$& Affine Lagrangian& \eqref{var0} \\
$\phi$ & Matter field & \eqref{action}\\
$p^{\nu}$ & Momentum conjugate to $\phi$ & \eqref{def momentum}\\
$\cP^{\lambda\mu}_{\ \ \kappa}$ & Partial derivative of $\Lag_{\rm matt}$ with respect to the $\mGamma^{\kappa}_{\ \lambda\mu}$ & \eqref{def calP0}\\
${\cal R}^{\mu\nu\kappa}$ & Linear combination of $\cP^{\lambda\mu}_{\ \ \kappa}$ and the metric tensor $g^{\mu\nu}$ & \eqref{def: calR}\\
${\cal Y}^{\mu\nu  \kappa}$ & Linear combination of ${\cal R}^{\lambda\mu  \kappa}$ & \eqref{cal Y} \\
$\pi^{\mu\nu}$ & Metric density and the momentum conjugate to $K_{\mu\nu}$ & \eqref{pi2}\\
$\pi_{\kappa}^{\ \lambda\mu\nu}$ & Linear combination of the momentum $\pi^{\mu\nu}$ and  $\delta^{\kappa}_{\lambda}$ & \eqref{pi40} \\
$\kolo{G}_{\mu\nu}$ & Metric Einstein tensor& \eqref{def ein tensor}\\
$\kolo{\cal G}_{\mu\nu}$ & Metric Einstein tensor density & \eqref{def ein tensor}\\
$\cP_{\kappa}^{\ \lambda\mu\nu}$ & Momentum conjugate to $K^{\kappa}_{\ \lambda\mu\nu}$ & \eqref{def cPK} \\
$\chi^{\mu\nu}$ & Momentum conjugate to $F_{\mu\nu}$ & \eqref{rel chi} \\
$\Omega_{\kappa}^{\ \lambda\mu\nu}$ & Momentum conjugate to $U^{\kappa}_{\ \lambda\mu\nu}$ & \eqref{rel Omega} \\
$\cO_{\kappa}^{\ \nu}$ & Metric trace of the momentum $\Omega_{\kappa}^{\ \lambda\mu\nu}$& \eqref{def tcO} \\
$\widetilde{\Omega}_{\kappa}^{\ \lambda\mu\nu}$ & Totally traceless part of the momentum ${\Omega}_{\kappa}^{\ \lambda\mu\nu}$ & \eqref{Omega decomposition}\\
$\mnabla_{\nu}\mathfrak{O}_{\kappa}^{\ \lambda\mu\nu}$ & Totally traceless part of the divergence $\mnabla_{\nu}\Omega_{\kappa}^{\ \lambda\mu\nu}$ & \eqref{mathfrakO}\\
$\Sigma_{\kappa}^{\ \lambda\mu\nu}$ & Momentum conjugate to $W^{\kappa}_{\ \lambda\mu\nu}$ & \eqref{rel Sigma} \\
$\Sigma_{\kappa}^{\ \nu}$ & Metric trace of the momentum $\Sigma_{\kappa}^{\ \lambda\mu\nu}$ & \eqref{def tSigma}\\
$\widetilde{\Sigma}_{\kappa}^{\ \lambda\mu\nu}$ & Totally traceless part of the momentum ${\Sigma}_{\kappa}^{\ \lambda\mu\nu}$ & \eqref{Sigma decomposition}\\
$\mathcal{J}^\mu$ & Divergence of the momentum $\chi^{\mu\nu}$ & \eqref{cal J}\\
$\Lambda$ & Cosmological constant & \eqref{lagLam} \\
$V_0, V_1, \dots V_6$ & Possible contractions of four Riemann tensors and two Levi-Civita symbols; variants & (\ref{w0}-\ref{w6})\\
$RRRR$ & Symbolical notion of \textbf{scalar densities} of weight “2'' which represents  contractions of four Riemann tensors $R^{\kappa}_{\ \lambda\mu\nu}$  with two Levi-Civita symbols $\epsilon^{\kappa\lambda\mu\nu}$ & \eqref{lagF} \\
{$KKKK$, $KKKF$, $KKKW$, $KKFF$, $KKFW$, $KKWW$} & Symbolical notion of \textbf{scalar densities} of weight “2'' which represents  contractions of symmetric Ricci tensors $K_{\mu\nu}$, skew-symmetric Ricci tensors $F_{\mu\nu}$ and algebraically  traceless part of Riemann tensor $W^{\kappa}_{\ \lambda\mu\nu}$ with two Levi-Civita symbols $\epsilon^{\kappa\lambda\mu\nu}$ & \eqref{RRRR}\\ \\
$KKK$, $KFF$, $KFW$, $KWW$, $KKF$, $KKW$ & Symbolical notion of \textbf{tensor densities} of weight “2'' which represents  contractions of symmetric Ricci tensors $K_{\mu\nu}$, skew-symmetric Ricci tensors $F_{\mu\nu}$ and algebraically  traceless part of Riemann tensor $W^{\kappa}_{\ \lambda\mu\nu}$ with two Levi-Civita symbols $\epsilon^{\kappa\lambda\mu\nu}$; they are related with derivatives of the affine Lagrangian $\Lag_A$ with respect to the  proper tensors& \eqref{derK} \\ \\
$o(F,W)$& Third- and higher-order terms in the tensors $F_{\mu\nu}$ and $W^{\kappa}_{\ \lambda\mu\nu}$ & \eqref{RRRR}
\end{tabularx}

\begin{tabularx}{\textwidth}{
  >{\raggedright\arraybackslash}p{1.9cm}
  >{\raggedright\arraybackslash}X
  >{\centering\arraybackslash}p{1.8cm}
}
\toprule
\textbf{Symbol} & \textbf{Meaning / Description} & \textbf{Equation} \\
\midrule

$\alpha$ & Global constant coefficient chosen to match the specific variant $V_0, V_1, \dots, V_6$ &  \eqref{lagF}\\
$\sigma$ & Sign of the $KKKK$ contraction specified for each variant & \eqref{reldetK} \\
$\gamma$ & Positive numerical coefficient related with  $KKKK$ contraction specified for each variant& \eqref{reldetK}\\
$\sigma_K$ & Sign of  $\det K$ & \eqref{reldetK}\\
$\sigma_g$ & Sign of  $\det g$; due to the Lorentzian signature, $\sigma_g=-1$ & \eqref{def sigg}\\
$\Kz$ & Symbolical notion of the non-perturbed symmetrical Ricci tensor $K_{\mu\nu}$ &\eqref{rachperp}\\
$\Ko$ & Symbolical notion of the first-order perturbation of the symmetrical Ricci tensor $K_{\mu\nu}$&\eqref{rachperp}\\
$\Kt$ & Symbolical notion of the second-order perturbation of the symmetrical Ricci tensor $K_{\mu\nu}$&\eqref{rachperp} \\
$ggg\Ko$, $gggW$, $gg\Ko\Ko$, $ggg\Kt$, $ggFF$, $ggFW$, $ggWW$ & Symbolical notion of \textbf{scalar densities} of weight “2'' which represents  contractions of metric tensors $g_{\mu\nu}$, first- and second-order perturbations of the symmetric Ricci tensors $\Ko$, $\Kt$, skew-symmetric Ricci tensors $F_{\mu\nu}$ and algebraically  traceless part of Riemann tensor $W^{\kappa}_{\ \lambda\mu\nu}$ with two Levi-Civita symbols $\epsilon^{\kappa\lambda\mu\nu}$ & (\ref{fo1}-\ref{fo2})\\
\\ 
$gg\Ko$, $ggW$, $g\Ko\Ko$, $gg\Kt$, $g\Ko W$, $gFF$, $gFW$, $gWW$, $ggF$, $ggW$ & Symbolical notion of \textbf{tensor densities} of weight “2'' which represents  contractions of metric tensors $g_{\mu\nu}$, first- and second-order perturbations of the symmetric Ricci tensors $\Ko$, $\Kt$, skew-symmetric Ricci tensors $F_{\mu\nu}$ and algebraically  traceless part of Riemann tensor $W^{\kappa}_{\ \lambda\mu\nu}$ with two Levi-Civita symbols $\epsilon^{\kappa\lambda\mu\nu}$; they are related with derivatives of the affine Lagrangian $\Lag_A$ with respect to the  proper tensors & (\ref{fo3}-\ref{fo4})\\
\\
$Q_{\mu\nu}$& Difference between $K_{\mu\nu}$ and $\kolo{K}_{\mu\nu}$ & \eqref{def Q}\\
$D^{\kappa}_{\ \lambda\mu\nu}$& Linear part of the difference between  $U^{\kappa}_{\ \lambda\mu\nu}$ and  $\kolo{U}^{\kappa}_{\ \lambda\mu\nu}$& \eqref{def tensor D} \\
$\mathfrak{U}^{\kappa}_{\ \lambda\mu\nu}$& Sum of $\kolo{U}^{\kappa}_{\ \lambda\mu\nu}$ and $D^{\kappa}_{\ \lambda\mu\nu}$ &\eqref{def mathfrakU}\\
$\mathfrak{U}^{\kappa}_{\ \nu}$& Metric trace of $\mathfrak{U}^{\kappa}_{\ \lambda\mu\nu}$ & \eqref{def mathfrakU 2}\\
$C^{\kappa}_{\ \lambda\mu\nu}$& Linear part of the difference between $W^{\kappa}_{\ \lambda\mu\nu}$ and  $\kolo{W}^{\kappa}_{\ \lambda\mu\nu}$& \eqref{def tensor C}\\
$C^{\kappa}_{\ \nu}$& Metric trace of   $C^{\kappa}_{\ \lambda\mu\nu}$ & \eqref{decomp tensor C2} \\
$\Lambda_{\mathrm{eff}}$ & Effective cosmological parameter & \eqref{def lam eff}\\
$C_F$, $C_W$, $I_{FW}$ & Coupling constants from the theory of the full Ricci tensor with a fixed background field & (\ref{KKFF v10}-\ref{KKWW v10})\\
$T^{\mu\nu}$ & Stress-energy tensor& \eqref{SET}\\
$\mathfrak{b}$ & Born-Infeld coupling constant & \eqref{BI const}
\end{tabularx}

\begin{tabularx}{\textwidth}{
  >{\raggedright\arraybackslash}p{1.9cm}
  >{\raggedright\arraybackslash}X
  >{\centering\arraybackslash}p{1.8cm}
}
\toprule
\textbf{Symbol} & \textbf{Meaning / Description} & \textbf{Equation} \\
\midrule
$f_{\mu\nu}$ & Faraday two-form; electromagnetic tensor & \eqref{def: dwuforma faradaya}\\
$a_{\mu}$ & Electromagnetic potential& \eqref{def: dwuforma faradaya}\\
${\cal F}^{\mu\nu}$ & Dual electromagnetic tensor; momentum conjugate to $a_{\mu}$ & \eqref{constrel}\\
${\cal T}^{\mu\nu}$ & Stress-energy tensor density & \eqref{calT ed}\\
$B_{\mu\nu}$ & Proca field & \eqref{Proca field}\\
$b_{\mu}$ & Proca potential & \eqref{Proca field} \\
${\cal B}^{\mu\nu}$ & Dual Proca tensor; momentum conjugate to  $b_{\mu}$& \eqref{const Proca} \\
$m$ & Mass parameter for the Klein-Gordon or Proca equation & \eqref{eq P2} \\
$L_{\kappa\lambda\mu}$ & Lanczos potential & \eqref{def L pot}\\
$\mathfrak{L}_{\kappa\lambda\mu\nu}$& Lanczos field & \eqref{def lanczos field}\\
\bottomrule
\end{tabularx}

\subsection{Metric structure} 

In this dissertation, the metric tensor $g_{\mu\nu}$ (with Greek indices) is always a four-dimensional symmetric tensor with signature $(-+++)$. \textbf{It is the only object that can raise or lower indices.} It naturally defines the metric connection as:
\begin{align}
    \mnabla_{\kappa} g_{\mu\nu} := g_{\mu\nu,\kappa} - \mGamma^{\sigma}_{\ \kappa\mu}\, g_{\sigma\nu} - \mGamma^{\sigma}_{\ \kappa\nu}\, g_{\mu\sigma} =0\, ,
    \label{met con}
\end{align}
which implies the well-known formula for Christoffel symbols $\mGamma^{\kappa}_{\ \lambda\mu}$:
\begin{align}
    \mGamma^{\kappa}_{\ \lambda\mu} = \frac 12\, g^{\kappa\sigma}\left(g_{\sigma\lambda,\mu} + g_{\sigma \mu, \lambda} - g_{\lambda\mu,\sigma}\right)\, .
    \label{mGamma}
\end{align}
The circles above the covariant derivative and Christoffel symbols indicate that these objects are associated with the \textit{metric structure}, as a non-metric connection will also appear later.

\subsection{General affine connection structure}

The general affine connection is not necessarily metric. Therefore:
\begin{equation}
    \nabla_{\kappa} g_{\mu\nu} \neq 0\, ,
    \label{nablaaf}
\end{equation}
however, the connection is symmetric (torsionless):
\begin{equation}
    \Gamma^{\kappa}_{\ \lambda\mu} = \Gamma^{\kappa}_{\ \mu\lambda}, 
    \label{aff con}
\end{equation}
Such an exclusion is motivated by the following observation: the torsion is, by definition, a difference between two connections represented by a skew-symmetric tensor (with respect to the lower indices).  It means that the torsion could be algebraically separated from the connection without any field equations or geometric properties. Therefore, a theory of a non-symmetric connection from the very beginning is equivalent to the theory of a symmetric connection interacting with an extra (skew-symmetric) tensor field (cf. \cite{nonmetricity, Geometria}). Importantly, the affine connection has also an independent physical interpretation as a field of \textit{local inertial frames} -- see \textbf{Chapter~\ref{intro conn}}.

If the theory is  associated with not only symmetric affine connection $\Gamma$, but also with the metric structure, there could be defined a difference between  general connection $\Gamma^{\kappa}_{\ \lambda\mu}$ and metric connection $\mGamma^{\kappa}_{\ \lambda\mu}$  denoted as  \textit{non-metricity tensor} $N^{\kappa}_{\ \lambda\mu}$:
\begin{equation}
    N^{\kappa}_{\ \lambda\mu}:= \Gamma^{\kappa}_{\ \lambda\mu} -   \mGamma^{\kappa}_{\ \lambda\mu}  \, .
     \label{decGamma}
\end{equation}
Obviously,  $N^{\kappa}_{\ \lambda\mu}$ is also symmetric with respect to the lower indices, as $\Gamma^{\kappa}_{\ \lambda\mu} $ and $ \mGamma^{\kappa}_{\ \lambda\mu}$ are. The above definition could also be understood as a decomposition of the connection $\Gamma$ into the metric part $\mGamma$ which satisfies equation~\eqref{met con}, and the remaining non-metric part~$N^{\kappa}_{\ \lambda\mu}$.

\subsection{Non-metricity tensor decomposition}
\label{non deco}
The non-metricity tensor $N^{\kappa}_{\ \lambda\mu}$ can be decomposed into the trace part $A_{\mu}$ and the algebraically traceless part $A^{\kappa}_{\ \lambda\mu}$ as follows:
\begin{align}
N^{\kappa}_{\ \lambda \mu} = A^{\kappa}_{\ \lambda \mu} + \frac 25\,\left(\delta^{\kappa}_{\lambda}\, A_{\mu} + \delta^{\kappa}_{\mu}\, A_{\lambda} \right)\, ,
\label{rozklad tensora N}
\end{align}
where
\begin{align}
 A_{\mu} &= \frac 12 \, N^{\kappa}_{\ \kappa \mu} \, , \label{slad N}\\
 A^{\kappa}_{\ \lambda \kappa} &=0  \label{bezsladowe N} \, .
\end{align}
The special choice of the trace representation is related to the expression for the skew-symmetric part of the Ricci tensor~\eqref{rozklad pelny F} (see below).

\ 

If the metric structure is given, then the “algebraically traceless part” $ A^{\kappa}_{\ \lambda \kappa} $ could be non-trivially contracted:
\begin{align}
    h^{\kappa}:=A^{\kappa}_{\ \lambda \kappa} g^{\lambda\kappa}\, ,
    \label{def h}
\end{align}
and decomposed:
\begin{align}
     A^{\kappa}_{\ \lambda \mu} =\widetilde{ A}^{\kappa}_{\ \lambda \mu} -\frac{1}{18}\left( \delta^{\kappa}_{\lambda}\, h_{\mu} + \delta^{\kappa}_{\mu}\, h_{\lambda} - 5g_{\lambda\mu}\, h^{\kappa}\right)\, ,
     \label{def tildeA}
\end{align}
where $\widetilde{ A}$ is a totally traceless part.

The final decomposition of the non-metricity tensor $N^{\kappa}_{\ \lambda \mu}$ is the following:
\begin{align}
N^{\kappa}_{\ \lambda \mu} = \widetilde{ A}^{\kappa}_{\ \lambda \mu} -\frac{1}{18}\left( \delta^{\kappa}_{\lambda}\, h_{\mu} + \delta^{\kappa}_{\mu}\, h_{\lambda} - 5g_{\lambda\mu}\, h^{\kappa}\right) + \frac 25\,\left(\delta^{\kappa}_{\lambda}\, A_{\mu} + \delta^{\kappa}_{\mu}\, A_{\lambda} \right)\, .
\label{rozklad tensora N tot}
\end{align}
Interestingly,  $\widetilde{ A}_{[\kappa  \lambda] \mu}$ has the same properties as the Lanczos potential -- see \textbf{Appendix~\ref{lanczospot}}.

\subsection{Curvature tensors}

Any connection structure induces the curvature. In our case, it is called \textit{Riemann curvature tensor} $R^{\kappa  }_{\ \lambda \mu \nu}$ and is defined as usual \cite{Geometria}:
\begin{align}
    R^{\kappa  }_{\ \lambda \mu \nu}:=-\Gamma^{\kappa}_{\ \lambda \mu, \nu}+\Gamma^{\kappa}_{\ \lambda \nu, \mu}-\Gamma^{\sigma}_{\ \lambda \mu}\, \Gamma^{\kappa}_{\ \nu \sigma} + \Gamma^{\sigma}_{\ \lambda \nu}\, \Gamma^{\kappa}_{\ \mu \sigma}\, .
    \label{def: tensor riemanna}
\end{align}
This tensor is, by definition, skew-symmetric in the last two lower indices:
\begin{align}
     R^{\kappa  }_{\ \lambda \mu \nu} = -  R^{\kappa  }_{\ \lambda  \nu \mu}\, .
     \label{skew riem}
\end{align}
Moreover, it satisfies \textit{the first Bianchi identity}:
\begin{align}
    R^{\kappa}_{\ [\lambda \mu \nu]} = 0 \ \Longleftrightarrow R^{\kappa}_{\ \lambda \mu \nu} + R^{\kappa}_{\ \nu\lambda \mu} + R^{\kappa}_{\ \mu \nu \lambda}=0\, .
    \label{1bianchirieman}
\end{align}

The Ricci tensor $R_{\lambda\nu}$ is an algebraic trace of the above Riemann tensor:
\begin{align}
    R_{\lambda \nu}:=R^{\kappa  }_{\ \lambda \kappa \nu} = -\Gamma^{\kappa}_{\ \lambda \kappa, \nu}+\Gamma^{\kappa}_{\ \lambda \nu, \kappa}-\Gamma^{\sigma}_{\ \lambda \kappa}\, \Gamma^{\kappa}_{\ \nu \sigma} + \Gamma^{\sigma}_{\ \lambda \nu}\, \Gamma^{\kappa}_{\ \kappa \sigma}\, .
    \label{def: tensor Ricciego}
\end{align}
  In general, the Ricci tensor $R_{\mu\nu}$ does not have any symmetry, so it can be decomposed into a symmetric part $K_{\mu\nu}$ and a skew-symmetric part $F_{\mu\nu}$:
\begin{align}
K_{\lambda \nu } &:= R_{(\lambda \nu)} = \Gamma^{\kappa}_{\ \lambda \nu, \kappa} -\Gamma^{\kappa}_{\ \kappa (\lambda, \nu)} + \Gamma^{\sigma}_{\ \lambda \nu}\, \Gamma^{\kappa}_{\ \kappa \sigma} - \Gamma^{\sigma}_{\ \kappa \lambda}\, \Gamma^{\kappa}_{\ \nu \sigma} \, ,  \label{def: tensor K}\\
F_{\lambda \nu} &:= R_{[\lambda \nu]} = -\Gamma^{\kappa}_{\ \kappa [\lambda , \nu]}  \, .
\label{def: tensor F}
\end{align}
Of course, if the connection is metric and torsionless, the skew-symmetric part $F_{\mu\nu}$ vanishes automatically:
\begin{align}
    \kolo{F}_{\mu\nu} = - \mGamma^{\kappa}_{\ \kappa[\lambda,\nu]} =  \partial_{[\nu|} \left(\frac 12\,g^{\kappa\sigma}\,g_{\kappa\sigma,|\lambda]} \right)  = \partial_{[\nu}\partial_{\lambda]} \left(\log \sqrt{|\det g|} \right) =0\, .
    \label{koloF}
\end{align}
It means that this part will depend only on the non-metricity tensor $N^{\kappa}_{\ \lambda\mu}$.

\ 

The algebraic decomposition of the Riemann tensor $R^{\kappa}_{\ \lambda \mu \nu} $ for its irreducible elements is the following \cite{nonmetricity}:
\begin{align}
R^{\kappa}_{\ \lambda \mu \nu} = \frac 13 \left(\delta^{\kappa}_{\mu}\, K_{\lambda \nu} - \delta^{\kappa}_{\nu}\, K_{\lambda \mu} \right) + \frac 15 \left(2\,\delta^{\kappa}_{\lambda}\, F_{\mu \nu} + \delta^{\kappa}_{\mu}\, F_{\lambda \nu} - \delta^{\kappa}_{\nu}\, F_{\lambda \mu} \right) + W^{\kappa}_{\ \lambda \mu \nu} \, ,
\label{rozklad: tensor riemanna}
\end{align}
where $W^{\kappa}_{\ \lambda \mu \nu} $ denotes the algebraically traceless part of the Riemann tensor:
\begin{align}
    W^{\kappa}_{\ \kappa \mu \nu} = W^{\kappa}_{\  \mu \kappa\nu} = 0\, ,
    \label{trless W}
\end{align}
which also satisfies conditions~\eqref{skew riem} and~\eqref{1bianchirieman}. Importantly, the tensor $W^{\kappa}_{\ \lambda \mu \nu} $ \textbf{is not} the Weyl tensor, because the metric structure is necessary to define the Weyl tensor, whereas for the existence of an abstract symmetric connection $\Gamma$ no metric is needed. The “extraction'' of the Weyl tensor from $W^{\kappa}_{\ \lambda \mu \nu} $ is presented in \textbf{Chapter~\ref{metric dec}}.

\ 

For further purposes, it will be useful to introduce the  Kijowski tensor $K^{\kappa}_{\ \lambda \mu \nu}$ (cf.~\cite{lic, mag, nonmetricity, Geometria, universality, kij2024}),  which is equivalent to the Riemann tensor $R^{\kappa}_{\ \lambda \mu \nu}$:
\begin{align}
K^{\kappa}_{\ \lambda \mu \nu}:=  -\frac 23 \, R^{\kappa}_{\ (\lambda \mu) \nu}\, . 
\label{rel KR}
\end{align}
This tensor is symmetric with respect to the first two lower indices:
\begin{align}
K^{\kappa}_{\ \lambda \mu \nu} = K^{\kappa}_{\ \mu \lambda \nu} \, ,
\label{sym kij}
\end{align}
which makes it much more suitable to describe the relation between derivatives of $\Gamma$ and the corresponding canonical momenta -- see formula~\eqref{def cPK}. It satisfies an analogue of the first Bianchi identity -- cf.~\eqref{1bianchirieman}:
\begin{align}
    K^{\kappa}_{\ (\lambda \mu \nu)} = 0 \ \Longleftrightarrow K^{\kappa}_{\ \lambda \mu \nu} + K^{\kappa}_{\ \nu\lambda \mu} + K^{\kappa}_{\ \mu \nu \lambda }=0\, .
    \label{eq: pierwsza tozsamosc Bianchiego dla K}
\end{align} 
It is easy to prove that the inverse relation between the Riemann tensor and the Kijowski tensor is given by:
\begin{align}
    R^{\kappa}_{\ \lambda \mu \nu} = -2 \, K^{\kappa}_{\ \lambda [\mu \nu]}\, .
    \label{Riem od Kij}
\end{align}
The Kijowski tensor $K^{\kappa}_{\ \lambda \mu \nu}$ can also be expressed in terms of the connection $\Gamma$:
\begin{align}
K^{\kappa}_{\ \lambda \mu \nu}  = \Gamma^{\kappa}_{\ \lambda \mu, \nu} - \Gamma^{\kappa}_{\ (\lambda \mu, \nu)} + \Gamma^{\sigma}_{\ \lambda \mu}\, \Gamma^{\kappa}_{\ \nu \sigma} -  \Gamma^{\sigma}_{\ (\lambda \mu}\, \Gamma^{\kappa}_{\ \nu) \sigma} \, .
\label{def: krzywizna kijowskiego}
\end{align}
The decomposition of the Kijowski tensor $K^{\kappa}_{\ \lambda \mu \nu}$ for irreducible parts is the following:
\begin{align}
K^{\kappa}_{\ \lambda \mu \nu}  =  -\frac 19\left(\delta^{\kappa}_{\lambda}\, K_{\mu \nu} + \delta^{\kappa}_{\mu}\, K_{\lambda \nu} - 2 \delta^{\kappa}_{\nu}\, K_{\lambda \mu} \right) - \frac 15\left(\delta^{\kappa}_{\lambda}\, F_{\mu \nu} + \delta^{\kappa}_{\mu}\, F_{\lambda \nu} \right) + U^{\kappa}_{\ \lambda \mu \nu} \, ,
\label{dec kijowski}
\end{align}
where $K_{\mu\nu}$ and $F_{\mu\nu}$ are components of the Ricci tensor~\eqref{rozklad: tensor riemanna}, whereas $U^{\kappa}_{\ \lambda \mu \nu}$ is the remaining algebraically traceless part, related to the tensor $W^{\kappa}_{\ \lambda \mu \nu}$ in the same manner as the Kijowski and Riemann tensors~\eqref{rel KR}:
\begin{align}
U^{\kappa}_{\ \lambda \mu \nu} &=- \frac 23\, W^{\kappa}_{\ (\lambda \mu) \nu}\, , & W^{\kappa}_{\ \lambda \mu \nu} &= -2 \, U^{\kappa}_{\ \lambda [\mu \nu]}\, .
\label{rel UW}
\end{align}

 \subsection{Decomposition of curvature tensors for metric and non-metric parts}

If the affine connection $\Gamma$ could be decomposed for the metric part $\mGamma$ and the non-metricity tensor $N$~\eqref{decGamma}, then the curvature tensors $R^{\kappa}_{\ \lambda \mu \nu}$~\eqref{def: tensor riemanna} and $K^{\kappa}_{\ \lambda \mu \nu}$~\eqref{def: krzywizna kijowskiego} decompose as follows \cite{mag}:
\begin{align}
R^{\kappa  }_{\ \lambda \mu \nu}&=  \kolo{R}^{\kappa}_{ \ \lambda \mu\nu} + \mnabla_{\mu} N^{\kappa}_{\ \nu \lambda} - \mnabla_{\nu} N^{\kappa}_{\ \mu \lambda} + N^{\sigma}_{\ \lambda \nu}\, N^{\kappa}_{\ \mu \sigma} - N^{\sigma}_{\ \lambda \mu}\, N^{\kappa}_{\ \nu \sigma} \, , \label{metRimm}\\
K^{\kappa  }_{\ \lambda \mu \nu} &=  \kolo{K}^{\kappa}_{ \ \lambda \mu\nu} + \mnabla_{\nu} N^{\kappa}_{\ \lambda \mu } - \mnabla_{(\nu} N^{\kappa}_{\ \lambda \mu) } +   N^{\sigma}_{\ \lambda \mu}\,  N^{\kappa}_{\ \nu  \sigma }- N^{\sigma}_{\ (\lambda \mu}\,  N^{\kappa}_{\ \nu)  \sigma }\, . \label{metKij}
\end{align}
Analogously, the symmetric Ricci tensor $K_{\mu\nu}$~\eqref{def: tensor K} and skew-symmetric Ricci tensor $F_{\mu\nu}$~\eqref{def: tensor F} are decomposed:
\begin{align}
    K_{\mu \nu} &=  \kolo{K}_{\mu \nu} + \mnabla_{\kappa} N^{\kappa}_{\ \mu \nu} - \mnabla_{(\mu}  N^{\kappa}_{\ \nu) \kappa}  + N^{\sigma}_{\ \mu \nu}\, N^{\kappa}_{\ \kappa  \sigma} - N^{\sigma}_{\ \kappa \mu }\, N^{\kappa}_{\ \nu \sigma} 
    \label{K}\, , \\
    F_{\mu \nu} &=   -  N^{\kappa}_{\ \kappa [\mu , \nu]}\, .  
    \label{F}
\end{align}

The above decomposition of $K_{\mu\nu}$~\eqref{K} and $F_{\mu\nu}$~\eqref{F} goes even further, due to the algebraic decomposition of the  non-metricity tensor $N^{\kappa}_{\ \lambda\mu}$~\eqref{rozklad tensora N}:
\begin{align}
K_{\mu \nu}&=\kolo{K}_{\mu \nu} +  \mnabla_{\kappa} A^{\kappa}_{\ \mu \nu}  - \frac 65\,  \mnabla_{(\mu}  A_{\nu)} - A^{\sigma}_{\ \rho \mu}\,A^{\rho}_{\ \nu \sigma} + \frac 65\, A_{\sigma}\, A^{\sigma}_{\ \mu \nu} + \frac {12}{25}\, A_{\mu}\, A_{\nu} \, ,\label{rozklad pelny K} \\
F_{\mu \nu}&=A_{\nu,\mu} - A_{\mu,\nu} =-2{A}_{[\mu,\nu]} = 2 \mnabla_{[\mu} A_{\nu]} \label{rozklad pelny F}\, .
\end{align}
The skew symmetric tensor $F_{\mu\nu}$  depends only on the trace part $A_{\mu}$ of the  non-metricity tensor $N^{\kappa}_{\ \lambda\mu}$, and is precisely a closed 2-form:
\begin{align}
    F=\dd A \ \Longrightarrow \ \dd F=0\ \Longrightarrow \ \partial_{[\alpha} F_{\mu\nu]} = 0\, .
    \label{closedF}
\end{align}
Interestingly, the algebraically traceless tensors $W^{\kappa}_{\ \lambda \mu \nu}$~\eqref{trless W} and $U^{\kappa}_{\ \lambda \mu \nu}$~\eqref{rel UW} depend only on the algebraically traceless part of the  non-metricity tensor $A^{\kappa}_{\ \lambda\mu}$~\eqref{bezsladowe N}:
\begin{align}
W^{\kappa}_{\ \lambda \mu \nu}&=\kolo{W}^{\kappa}_{\ \lambda\mu \nu} + \mnabla_{\mu}{A}^{\kappa}_{\ \nu \lambda } - \mnabla_{\nu}{A}^{\kappa}_{\ \mu \lambda } + \frac 13\,\left( \delta^{\kappa}_{\nu} \mnabla_{\sigma} {A}^{\sigma}_{\ \mu\lambda} - \delta^{\kappa}_{\mu} \mnabla_{\sigma} {A}^{\sigma}_{\ \nu\lambda}\right) + \nonumber \\
&\quad  +{A}^{\sigma}_{\ \lambda \nu}\, A^{\kappa}_{\ \mu \sigma} - {A}^{\sigma}_{\ \lambda \mu}\, A^{\kappa}_{\ \nu \sigma}  + \frac 13 \, \left(\delta^{\kappa}_{\mu}\, {A}^{\rho}_{\ \nu \sigma}\, {A}^{\sigma}_{\ \rho \lambda} - \delta^{\kappa}_{\nu}\, {A}^{\rho}_{\ \mu \sigma}\, {A}^{\sigma}_{\ \rho \lambda}\right)  \, , \label{rozklad pelny W} \\
U^{\kappa}_{\ \lambda \mu \nu}&=\kolo{U}^{\kappa}_{\ \lambda\mu \nu}  + \frac 23\, \left( \mnabla_{\nu}{A}^{\kappa}_{\ \lambda  \mu  } -  \mnabla_{(\lambda }A^{\kappa}_{\ \mu) \nu }\right) - \frac 29\, \left( \delta^{\kappa}_{\nu} \mnabla_{\sigma} {A}^{\sigma}_{\ \lambda\mu} - \delta^{\kappa}_{(\lambda| } \mnabla_{\sigma} {A}^{\sigma}_{\ |\mu)\nu }\right) + \nonumber \\
&\quad +\frac 23\left(  {A}^{\sigma}_{\ \lambda \mu}\, A^{\kappa}_{\ \nu \sigma} - A^{\kappa}_{\ \sigma(\lambda} \, A^{\sigma}_{\ \mu) \nu} \right) - \frac 29 \left(\delta^{\kappa}_{(\lambda}\, {A}^{\sigma}_{\ \mu) \rho  }\, {A}^{\rho}_{\ \nu \sigma} - \delta^{\kappa}_{\nu}\, {A}^{\rho}_{\   \sigma(\lambda}\, {A}^{\sigma}_{\ \mu)\rho }\right)    \, . \label{rozklad pelny U}  
\end{align}

\subsection{Metric decomposition}
\label{metric dec}

The Riemann tensor~\eqref{def: tensor riemanna} of the general symmetric affine connection $\Gamma$ was decomposed with respect to the algebraic traces~\eqref{rozklad: tensor riemanna}. However, the appearance of the metric structure allows further decompositions, due to the possibility of defining  metric traces as contractions with the metric tensor $g_{\mu\nu}$. The symmetric Ricci tensor $K_{\mu\nu}$ decomposes automatically:
\begin{align}
     K_{\mu\nu}&= Z_{\mu\nu} +\frac 14\, R\,  g_{\mu\nu}\, , &  R&:= R_{\mu\nu}\, g^{\mu\nu} = K_{\mu\nu}\,g^{\mu\nu}\, , & Z_{\mu\nu}\,g^{\mu\nu}=0\, ,
    \label{def: tensor Z}
\end{align}
whereas the skew-symmetric part $F_{\mu\nu}$ does not have any metric traces:
\begin{align}
    F_{\mu\nu}\, g^{\mu\nu}=0\,.
\end{align}

For the metric Riemann tensor $\kolo{R}_{\alpha\beta\mu\nu}$, there is a so-called \textit{Ricci decomposition}~\cite{Weinberg}, which in four dimensions takes the following form:
\begin{align}
    \kolo{R\, }_{  \kappa\lambda\mu\nu}&=\frac{\kolo{R}}{6}\left(g_{\kappa\nu}\, g_{\lambda\mu} - g_{\kappa\mu}\, g_{\lambda\nu} \right)  + \frac 12 \left(\kolo{K}_{\kappa\mu}\, g_{\lambda\nu}- \kolo{K}_{\kappa\nu}\, g_{\lambda\mu}  + \kolo{K}_{\lambda\nu}\, g_{\kappa\mu} -\kolo{K}_{\lambda\mu}\, g_{\kappa\nu}  \right) + \kolo{\mathfrak{w} }_{  \kappa\lambda\mu\nu}\, ,
    \label{Ricci deco}
\end{align}
where $\kolo{K}_{\mu\nu}$ is the metric  Ricci tensor (which is always symmetric -- see~\eqref{koloF}), $\kolo{R}$ is the metric Ricci scalar, and $\kolo{\mathfrak{w}}_{  \alpha \beta\mu\nu}$ is the metric Weyl tensor. Moreover, the metric Riemann tensor satisfies additional algebraic conditions:
\begin{align}
    \kolo{R}_{\alpha \beta\mu\nu} &= \kolo{R}_{\mu\nu\alpha \beta}\, , &
    \kolo{R}_{\alpha \beta\mu\nu} &= -\kolo{R}_{\beta\alpha\mu\nu}\, ,
\end{align}
as well as a special differential identity, called \textit{the  second Bianchi identity}:
\begin{align}
    \mnabla_{[\alpha} \kolo{R}_{\kappa\lambda] \mu\nu} = 0 
    \qquad \Longrightarrow \qquad 
    \mnabla_{\alpha} \kolo{R}_{\kappa\lambda\mu\nu} 
    + \mnabla_{\kappa} \kolo{R}_{\lambda\alpha\mu\nu} 
    + \mnabla_{\lambda} \kolo{R}_{\alpha\kappa\mu\nu} = 0\, ,
    \label{2Bian}
\end{align}
which induces, after the  contraction with two metric tensors, \textit{the  contracted second Bianchi identity}:
\begin{align}
    \mnabla_{\nu} \kolo{K}^{\nu}_{\ \mu} = \frac 12 \mnabla_{\mu} \kolo{R}\, ,
    \label{cont 2Bian}
\end{align}
whereas the contraction with only one metric tensor generates:
\begin{align}
    \mnabla^{\kappa}\kolo{\mathfrak{w}}_{\kappa\lambda\mu\nu} = \mnabla_{[\mu}\kolo{K}_{\nu]\lambda} +\frac{1}{6}g_{\lambda[\mu}\mnabla_{\nu]}\kolo{R}\, . 
    \label{cont 2Bian 2}
\end{align}

Due to the  formula~\eqref{Ricci deco}, the relation between the algebraically traceless metric Riemann tensor $\kolo{W}$~\eqref{rozklad: tensor riemanna} (remembering that $\kolo{F}_{\mu\nu}=0$~\eqref{koloF}) and the  totally traceless metric Riemann tensor (Weyl tensor) $\kolo{\mathfrak{w} }_{  \alpha \beta\mu\nu}$~\eqref{Ricci deco} is the following:
\begin{align}
    \kolo{W}_{\kappa\lambda\mu\nu} &= \kolo{\mathfrak{w} }_{\kappa\lambda\mu\nu} +\frac{\kolo{R}}{6}\left(g_{\kappa\nu}\, g_{\lambda\mu} - g_{\kappa\mu}\, g_{\lambda\nu} \right)  + \frac 12 \left(\kolo{K}_{\kappa\mu}\, g_{\lambda\nu} - \kolo{K}_{\kappa\nu}\, g_{\lambda\mu}\right) + \nonumber \\
    &\quad  + \frac{1}{6}\left( \kolo{K}_{\lambda\nu}\, g_{\kappa\mu}  -\kolo{K}_{\lambda\mu}\, g_{\kappa\nu}  \right)\, .
    \label{rel W mathfrakw}
\end{align}

The metric decomposition of the algebraically traceless tensor $W^{\kappa}_{\ \lambda\mu\nu}$ is considerably more involved and was presented in~\cite{marian}. Nevertheless, it is very instructive to include it here.  Firstly, the tensor $W^{\kappa}_{\ \lambda\mu\nu}$ has only one non-trivial metric trace:
\begin{align}
    W^{\kappa}_{\ \nu}:= W^{\kappa}_{\ \lambda\mu\nu}\, g^{\lambda\mu}\, ,
    \label{trace W}
\end{align}
which is, by definition, algebraically traceless:
\begin{align}
     W^{\kappa}_{\ \kappa}=0\, .
\end{align}
For the tensor $\kolo{W}_{\kappa\lambda\mu\nu}$~\eqref{rel W mathfrakw}, the following trace equals:
\begin{align}
    \kolo{W}_{\kappa \nu} = -\frac 43 \left( \kolo{K}_{\kappa\nu} - \frac{\kolo{R}}{4}\, g_{\kappa\nu}\right) = -\frac 43   \kolo{Z}_{\kappa\nu} \, ,
    \label{def metric W2}
\end{align}
what implies that the tensor $\kolo{W}_{\kappa \nu}$ is symmetric and proportional to the traceless metric Ricci tensor $\kolo{Z}_{\kappa \nu}$~\eqref{def: tensor Z}.

\ 

The further decomposition of the tensor $ W^{\kappa}_{\ \lambda\mu\nu}$ is presented in the lemma below:
\begin{lemma}
\label{lemm W dec}
Let $W^{\kappa}_{\ \lambda\mu\nu}$ be a tensor satisfying the following algebraic conditions:
\begin{align}
    W^{\kappa}_{\ [\lambda\mu\nu]} &= 0\, , &
    W^{\kappa}_{\ \lambda\mu\nu} &= -W^{\kappa}_{\ \lambda\nu\mu}\, , &
    W^{\kappa}_{\ \kappa\mu\nu} &= 0\, , &
    W^{\kappa}_{\ \mu\nu\kappa} &= 0\, .
\end{align}

Define:
\begin{align}
    W_{\kappa\lambda\mu\nu} &:= g_{\kappa \sigma}\, W^{\sigma}_{\ \lambda\mu\nu}\, , &
    W_{\kappa\nu} &:= W_{\kappa\lambda\mu\nu}\, g^{\lambda\mu}\, ,
\end{align}
and let $\widetilde{W}_{\kappa\lambda\mu\nu}$ denote the totally traceless part of $W_{\kappa\lambda\mu\nu}$ (i.e., traceless both algebraically and metrically), possessing the same symmetries as $W_{\kappa\lambda\mu\nu}$. Then the following decomposition holds:
\begin{align}
    W_{\kappa\lambda\mu\nu}
    &= \widetilde{W}_{\kappa\lambda\mu\nu}
    -\frac{1}{6}\, g_{\kappa\lambda}\, W_{[\mu\nu]}
    + \frac{1}{8}\left( g_{\kappa\nu}\, W_{(\lambda\mu)} - g_{\kappa\mu}\, W_{(\lambda\nu)} \right)+ \nonumber \\
    &\quad + \frac{1}{12}\left( g_{\kappa\nu}\, W_{[\lambda\mu]} - g_{\kappa\mu}\, W_{[\lambda\nu]} \right)
    + \frac{3}{8}\left( W_{(\kappa\nu)}\, g_{\lambda\mu} - W_{(\kappa\mu)}\, g_{\lambda\nu} \right)+ \nonumber \\
    &\quad + \frac{5}{12}\left( W_{[\kappa\nu]}\, g_{\lambda\mu} - W_{[\kappa\mu]}\, g_{\lambda\nu} \right)\, .
    \label{W decomposition}
\end{align}
\end{lemma}
\begin{proof}
    The proof relies on the verification of all presented conditions.
\end{proof}
However, the task is not yet completed, because the tensor $\widetilde{W}_{\kappa\lambda\mu\nu}$ can also be decomposed into tensors with special algebraic properties.  Indeed, it is skew-symmetric in the last two indices and does not have any specific symmetries in the first two indices. Thus:
\begin{align}
    \widetilde{W}_{\kappa\lambda\mu\nu} = \widetilde{W}_{[\kappa\lambda]\mu\nu} + \widetilde{W}_{(\kappa\lambda)\mu\nu}\, .
\end{align}
Now, there is defined the following tensor:
\begin{align}
    \mathfrak{w}_{\kappa\lambda\mu\nu} = \widetilde{W}_{[\kappa\lambda]\mu\nu} + \widetilde{W}_{[\mu\nu]\kappa\lambda}\, ,
    \label{def frakw}
\end{align}
which is precisely the Weyl tensor, and the following tensor:
\begin{align}
    \mathfrak{h}_{\kappa\lambda\mu\nu} = \widetilde{W}_{[\kappa\lambda]\mu\nu} - \widetilde{W}_{[\mu\nu]\kappa\lambda}\, .
      \label{def frakh}
\end{align}
Therefore:
\begin{align}
    \widetilde{W}_{[\kappa\lambda]\mu\nu} &= \frac 12\, \left( \widetilde{W}_{[\kappa\lambda]\mu\nu} + \widetilde{W}_{[\mu\nu]\kappa\lambda} \right) + \frac 12\, \left( \widetilde{W}_{[\kappa\lambda]\mu\nu} - \widetilde{W}_{[\mu\nu]\kappa\lambda} \right) =\frac 12\, \mathfrak{w}_{\kappa\lambda\mu\nu} +  \frac 12\, \mathfrak{h}_{\kappa\lambda\mu\nu}\, .
    \label{rozklad tildeW}
\end{align}

\subsection{Independent components}

An important topic related to the presented decomposition concerns the independent components (sometimes referred to as \textit{degrees of freedom}) of each element of the Riemann curvature tensor  $R^{\kappa}_{\ \lambda \mu\nu}$ ~\eqref{def: tensor riemanna} of the general symmetric affine connection $\Gamma$. Initially, the Riemann tensor has 80 independent components. This is because it is skew-symmetric in the last two lower indices, resulting in 6 degrees of freedom in four-dimensional spacetime. The first lower index does not exhibit any additional symmetry, leading to $4\cdot 6=24$ components; however, the Bianchi identities eliminate four of them for each dimension. Consequently, all three lower indices together carry 20 degrees of freedom, which are then multiplied by 4 due to the fully independent first upper index.

The Riemann tensor has only one algebraic trace: the Ricci tensor $R_{\mu\nu}$~\eqref{def: tensor Ricciego}, which is just a $4\times 4$ matrix, and carries 16 degrees of freedom, which splits for 10 degrees for the symmetric Ricci tensor $K_{\mu\nu}$~\eqref{def: tensor K} and $6$ for the skew-symmetric $F_{\mu\nu}$~\eqref{def: tensor F}. As a result, the algebraically traceless Riemann tensor $W^{\kappa}_{\ \lambda\mu\nu}$~\eqref{rozklad: tensor riemanna} has $80-16=64$ independent components.

When the metric structure is introduced, further decomposition becomes possible, and the resulting objects need to be defined. This procedure applies trivially to the symmetric Ricci tensor $K_{\mu\nu}$~\eqref{def: tensor Z}, which has the only one scalar trace $R$ and the traceless part $Z_{\mu\nu}$ with  $10-1=9$ independent components. As it was shown, the tensor \( W^{\kappa}_{\ \lambda\mu\nu} \) has a more complicated  structure. Firstly, the metric trace \( W^{\kappa}_{\ \nu} \)~\eqref{trace W} has 15 components, because it is a traceless \( 4 \times 4 \) matrix. Naturally, after lowering the first index, this matrix can be expressed as the sum of a symmetric part \( W_{(\mu\nu)} \), with 9 degrees of freedom, and a skew-symmetric \( W_{[\mu\nu]} \) part, with 6 degrees of freedom. This implies that the totally traceless part \( \widetilde{W}_{\kappa \lambda\mu\nu} \) has \( 64 - 15 = 49 \) independent components.
 
It is known that the Weyl tensor \( \mathfrak{w}_{\kappa\lambda\mu\nu} \) has 10 degrees of freedom (see \cite{marian}, or Weinberg's book \cite{Weinberg}, p. 146). The “twin” of the Weyl tensor, \( \mathfrak{h}_{\kappa\lambda\mu\nu} \)~\eqref{rozklad tildeW}, is also skew-symmetric in the first and second pairs of indices. Moreover, it is skew-symmetric with respect to the exchange of these two pairs. This implies that it can be represented as a \( 6 \times 6 \) skew-symmetric matrix with 15 degrees of freedom. Consequently, the last component \( \widetilde{W}_{(\kappa \lambda)\mu\nu} \) carries \( 49 - 10 - 15 = 24 \) degrees of freedom. All these objects, along with their respective numbers of independent components, are summarized in the table below:
\\

\begin{tabularx}{\textwidth}{>{\raggedright\arraybackslash}m{3cm} X}
\toprule
\textbf{Tensor} & \textbf{Number of independent components } \\
\midrule  
 $R^{\kappa}_{\ \lambda\mu\nu}$ &  80 \\
         $R_{\mu\nu}$ & 16 \\
         $F_{\mu\nu}$ & 6 \\
         $K_{\mu\nu}$ & 10 \\
         $R$ & 1 \\
         $Z_{\mu\nu}$ & 9 \\
         $W^{\kappa}_{\ \lambda\mu\nu}$ &  64 \\
         $W^{\kappa}_{\ \nu}$ & 15 \\
         $\widetilde{W}_{\kappa\lambda\mu\nu}$ & 49 \\
         $\mathfrak{w}_{\kappa\lambda\mu\nu}$ & 10 \\
         $\mathfrak{h}_{\kappa\lambda\mu\nu}$ & 15 \\
         $\widetilde{W}_{[\kappa\lambda]\mu\nu}$ & 25 \\
         $\widetilde{W}_{(\kappa\lambda)\mu\nu}$ & 24 \\
\bottomrule
\end{tabularx}

\

The same analysis for the metric Riemann tensor $\kolo{R}^{\kappa}_{\ \lambda\mu\nu}$ is presented in~\cite{Weinberg}. However, for the sake of completeness, it is also included here:

\ 

\begin{tabularx}{\textwidth}{>{\raggedright\arraybackslash}m{3cm} X}
\toprule
\textbf{Tensor} & \textbf{Number of independent components } \\
\midrule  
 $\kolo{R}^{\kappa}_{\ \lambda\mu\nu}$ & 20 \\ 
         $\kolo{K}_{\mu\nu}$ & 10 \\
         $\kolo{R}$ & 1 \\
         $\kolo{Z}_{\mu\nu}$ & 9 \\ 
         $\kolo{\mathfrak{w}}_{\kappa\lambda\mu\nu}$ & 10 \\ 
\bottomrule
\end{tabularx}
    \chapter{Preliminaries}

\section{Origins}

\subsection{Variational calculus}
\label{varcal}
The calculus of variations was an ingenious idea developed at the end of the 17th century by some of the most influential minds of the time, including Pierre de Fermat, Isaac Newton, Gottfried Leibniz, Jakob and Johann Bernoulli, and Marquis de l'Hôpital. It was initially devised to solve the problem of the brachistochrone: the curve (or trajectory) along which the time of motion in a uniform gravitational field is the shortest. The core innovation of this method lies in treating entire curves (functions) as “variables” and finding the one that minimises the time.

Of course, time is not the only quantity that can be minimised. Another example is the geodesic problem, where the goal is to find the shortest path on a given surface, minimising the curve's length. Similarly, in the problem of the catenary (the curve of a hanging chain), the quantity minimised is energy. In general, the object being minimised is referred to as the \textit{action}, typically defined as an integral. This method revolutionised mathematics and physics and remains in use to this day. Inspired by this approach, Pierre Louis Maupertuis described it as the \textit{principle of least action}, popularly denoted as the \textit{Maupertuis principle}.

The formalisation and further development of the calculus of variations were initiated by Leonhard Euler, a student of Johann Bernoulli, and continued by Giuseppe Luigi Lagrangia (better known as Joseph-Louis Lagrange)\footnote{T his information was taken from \cite{kerner, AKW}.}. The resulting equations that the minimising function must satisfy are known as \textit{the Euler-Lagrange equations}, while the integrand used to compute the action is called the \textit{Lagrangian}. Over time, this method was generalised to describe far more complex systems involving multiple parameters (coordinates), leading to the development of field theory. A classic example of such an application is the problem of finding the shape of a stretched membrane. This generalisation has been successfully applied to many important physical theories, such as electrodynamics and gravity.

However, after years of study, scientists discovered that, in general, there is not any true “minimum''. Instead, the method of variations identifies functions corresponding to critical points (which may be minima, maxima, or saddle points). This observation plays a crucial role in the variational  formulation of general relativity. Therefore, a brief modern overview of these ideas, as applied to field theory, is presented below\footnote{Those examples and comments were already presented in  \cite{nonmetricity}, written by the author of this dissertation and J. Kijowski -- one of the supervisors.}.

\

Consider a scalar field\footnote{The scalar behaviour of the field is assumed to simplify the notation. The geometric character of the field does not affect the final result.} $\phi$ which depends on coordinates $(x^{\mu})$. Suppose that the Lagrangian of the model depends on the field $\phi$ and its derivatives $\phi_{,\nu}$ only. Hence, the action $S$ is defined as an integral (non-oriented) of the Lagrangian $\Lag$ over the region $\Omega$ with measure $\dd \mu$:
\begin{align}
    S:=\int_{\Omega} \Lag(\phi,\phi_{,\nu})\,  \dd \mu \, .
    \label{action}
\end{align}
 To find the field which “optimises'' the action, a one-parameter family of scalar fields $\phi$ is considered:  $\phi=\phi(x^{\mu},\epsilon)$, where $\epsilon$ is a continuous parameter which distinguishes members of this family. Hence, the variational calculus relies on finding an extremum of action with respect to this extra parameter. Like in other typical “optimising'' problems, it is necessary to calculate a derivative and equate it to zero. The derivative with respect to $\epsilon$ deserves a special symbol:
 \begin{align}
     \frac{\partial}{\partial \epsilon}:=\delta\, ,
 \end{align}
 and will be called as a \textit{variation}. This name very precisely describes the idea of changing (varying) the functions among the family. As it was written before, to find the extremum, the variation of action has to vanish:
 \begin{align}
     \delta S=\delta \int_{\Omega} \Lag(\phi,\phi_{,\nu})\,  \dd \mu = 0\, ,
 \end{align}
 for any variation $\delta \phi$.
In general, the variation does not commute with the integral but, for sufficiently smooth fields, the variation of the integral is equal to the integral of the variation\footnote{In modern approach this problem is solved by resigning of “global'' point of view and by looking on this problem locally, where the Lagrangian has an interpretation of the infinitesimal action -- see~\cite{Tulcz}. This idea will be used and described in next parts of this dissertation.}. Hence:
\begin{align}
    \delta S=0 \iff \int_{\Omega}\delta \Lag(\phi,\phi_{,\nu})\,  \dd \mu=0\, .
\end{align}
The $\epsilon$  appears via the field $\phi$ and its partial derivatives (with respect to coordinates) $\phi_{,\nu}$, then:
\begin{align}
    \delta \Lag = \frac{\partial \Lag}{\partial \phi}\, \delta \phi + \frac{\partial \Lag}{\partial \phi_{,\nu}}\, \delta \phi_{,\nu} \, .
    \label{dafafd}
\end{align}
Due to the fact that the parameter $\epsilon$  and coordinates $(x^{\mu})$ are mutually independent objects, the variation $\delta$ and partial derivative $\partial_{\nu}$ trivially commute:
\begin{align}
    \delta \phi_{,\nu} =  \frac{\partial^2\phi}{\partial\epsilon\partial x^{\nu}} =  \frac{\partial^2\phi}{\partial x^{\nu}\partial\epsilon} = \partial_{\nu} \delta \phi\, .
\end{align}
In classical textbooks, this simple conclusion is called “the fundamental lemma of the calculus of variation''. Consider the following quantity:
\begin{align}
    p^{\nu}:=\frac{\partial \Lag}{\partial \phi_{,\nu}}\, ,
    \label{def momentum}
\end{align}
which is called a \textit{momentum canonically conjugate} to the field $\phi$. Now, using the integration by parts in $\delta \Lag$~\eqref{dafafd}, one has:
\begin{align}
    \delta \Lag  = \frac{\partial \Lag}{\partial \phi}\, \delta \phi + \partial_{\nu}\left(p^{\nu}\, \delta \phi\right) - p^{\nu}_{\ ,\nu}\, \delta \phi  \, .
    \label{varlag}
\end{align}
Integration over the region $\Omega$ implies:
\begin{align}
    \int_{\Omega}\delta \Lag \, \dd \mu &= \int_{\Omega}\left(\frac{\partial \Lag}{\partial \phi} - p^{\nu}_{\ ,\nu}\right)\, \delta \phi\, \dd \mu   +
    \int_{\Omega} \partial_{\nu}\left(p^{\nu}\, \delta \phi\right) \, \dd \mu    = \nonumber \\
    &=\int_{\Omega}\left(\frac{\partial \Lag}{\partial \phi} - p^{\nu}_{\ ,\nu}\right)\, \delta \phi\, \dd \mu   +  \int_{\partial \Omega} p^{\nu}\, \delta \phi \, ,
\end{align}
where in the last equality was used the Stokes theorem. In mechanics, the values of all functions $\phi$ were fixed at the boundary $\partial \Omega$, therefore $\delta \phi \big|_{\partial \Omega}=0$. For example, in the brachistochrone problem, the demanded function has fixed starting and finishing points. Hence, posing the boundary conditions guarantees vanishing of the boundary term $\partial_{\nu}\left(p^{\nu}\, \delta \phi\right)$, and then vanishing of the $\int_{\Omega}\delta \Lag \, \dd \mu$ is obtained by the vanishing of the volume (bulk) term, what introduce the famous Euler-Lagrange equations:
\begin{align}
    \frac{\partial \Lag}{\partial \phi} - p^{\nu}_{\ ,\nu} =0\, .
    \label{def EL eq}
\end{align}
In mechanics, or more generally, in statics, everything works perfectly. Especially, due to the fact that obtained this way equations are elliptic (like Laplace equation), for whom the  Dirichlet problem (prescribed boundary conditions) is well-posed. The problem appears in dynamics, where the system is typically described by the hyperbolic equation (like wave equation), where Dirichlet conditions do not work at all. Furthermore, for arbitrarily given boundary conditions, the solution does not exist! The presence of this effect is perfectly visible for the wave equation in two-dimensional spacetime $\mathbb{R}^2=\{(t,x)\}$:
\begin{equation}\label{wave}
    \left( \frac{\partial^2}{\partial x^2} - \frac{\partial^2}{\partial t^2} \right) \varphi = 0 \, ,
\end{equation}
Implying advanced and retarded coordinates $(u,v)=(t-x,t+x)$ and twice integrating it over the rectangle
\begin{align}
{\cal R}=\{(u,v)\in\mathbb{R}^2 : \, u_0\leq u \leq u_0+2\delta,\, v_0\leq v \leq v_0+2\epsilon\}\, ,
\end{align}
is easy to prove that field equation~\eqref{wave} for the function $\varphi(u,v)$ is equivalent to the following identity
\begin{align}
   \varphi(u_0+2\delta,\, v_0+2\epsilon) - \varphi(u_0+2\delta,\, v_0) -\varphi(u_0,\, v_0+2\epsilon)+\varphi(u_0,\, v_0) = 0\, ,
\end{align}
for any choice of four numbers: $\{u_0,v_0,\epsilon,\delta \}$. The above equation  could be simply re-transformed to standard spacetime coordinates $(t,x)=(\frac{v+u}{2}, \frac{v-u}{2})$. Then function $\phi(t,x)$ satisfies:
\begin{align}
    \varphi(t_0+\epsilon +\delta,\, x_0+\epsilon - \delta)- \varphi(t_0+\delta ,\,  x_0-\delta)  - \varphi(t_0+\epsilon,\,  x_0+\epsilon) + \varphi(t_0,\, x_0)  =0 \, .   \label{ident}
\end{align}
Putting $t_0=0$, $x_0=x$, $\delta = x$ and $\epsilon = 1-x$ provides to an identity which must be fulfilled for any $0 \le x \le 1$:
\begin{equation}\label{iden-boundary}
  \varphi(1,1-x)- \varphi(x,0)  - \varphi(1-x,1) +  \varphi(0,x) = 0\, .
\end{equation}
Consider the spacetime volume ${\cal O}$:
\begin{equation}\label{obszar}
        {\cal O} = \left\{ (t,x):\  0 \leq t \le 1 \, ; \ 0\leq x \le 1 \right\} \, .
\end{equation}
Now, the field equation implies a constraint in space of boundary data: the value of the field on the upper wall (i.e.: $\varphi(1,\cdot)$) is uniquely given by its value on the remaining three walls (i.e.: $\varphi(\cdot , 0)$, $\varphi(\cdot , 1)$ and $\varphi(0, \cdot)$). There is no solution of the wave equation if the boundary data do not satisfy the constraint   defined by equation~\eqref{iden-boundary}! Moreover, the field equation~\eqref{wave} is {\em equivalent} to this constraint!

Hence, the “brachistochrone” philosophy relies on believing  that imposed boundary conditions will be satisfied by the obtained equations. Unfortunately, if  the boundary data is chosen randomly, then the probability that there is any solution that satisfies this choice is precisely zero -- like the probability of choosing a natural number from all reals.

Of course, the constraint~\eqref{iden-boundary} is still “relatively manageable” for the simple spacetime rectangle~\eqref{obszar}, whereas for a generic spacetime volume ${\cal O}$ (e.g., a time slice $\{a \le t \le b ; x \in \mathbb{R}\}$) it is a much worse, very singular, non-closed subspace in any reasonable topology of boundary data. The conclusion is very simple: the “brachistochrone” philosophy for  theories/problems described by hyperbolic equations totally breaks down. Although it works perfectly for the elliptic cases.

The situation is hard, but not hopeless. If the boundary term could not be eliminated in general, then there should be taken a different strategy, called “on shell'' philosophy. This procedure relies on allowing only those fields, which satisfy the Euler-Lagrange system~\eqref{def EL eq}. Due to that, the variation of the Lagrangian~\eqref{varlag} is restricted to the boundary term but in agreement with field equations. Therefore, the variation $\delta\Lag$~\eqref{varlag} “on shell'' is equal:
\begin{align}
    \delta \Lag  = \partial_{\nu}\left(p^{\nu}\, \delta \phi\right) = p^{\nu}_{\ ,\nu}\, \delta \phi + p^{\nu}\, \delta \phi_{,\nu}  \, .
    \label{lagonshell}
\end{align}
The corresponding field equations are the following:
\begin{align}
    p^{\nu} &=\frac{\partial \Lag}{\partial\phi_{,\nu}}\, , \\
    p^{\nu}_{\ ,\nu} &=\frac{\partial \Lag}{\partial\phi}\, ,
\end{align}
where the first one is exactly the definition of the momentum~\eqref{def momentum}, whereas the second one is precisely the Euler-Lagrange equation~\eqref{def EL eq}. From a geometrical point of view, at each spacetime point, field equations~\eqref{lagonshell} can be considered as a symplectic relation (i.e. a Lagrangian submanifold) in a symplectic space parameterised by the following “generalised jets” of fields: $(\varphi ,\,  \varphi_{,\lambda} , \, p^{ \lambda},\,  p^{\lambda}_{ \ ,\nu})$. This approach was rigorously defined in \cite{Tulcz, CJK, Kij-Moreno2015}, but its strength consists in the fact that it is very well adapted for practical calculations in both the Lagrangian and Hamiltonian formalism (especially when constraints are present) and avoids the ridiculous procedure of “imposing the spacetime-boundary conditions”. Practically, this concept provides to the so-called \textit{control theory} which relies on splitting the  canonical field variables into two groups: the “control parameters” (those, which appear under the sign “$\delta$” -- in case of~\eqref{lagonshell} these are configuration variables $\varphi$ and their “velocities” $\varphi_{,\lambda}$) -- and the “response parameters” (in case of~\eqref{lagonshell} these are momenta $p^{\nu}$ and \underline{only} their “currents” $j= p ^{ \nu}_{\ ,\nu}$). Field equations are then considered as the “control -- response relation”. It will be very useful to describe and simplify formalism in the sequel. Moreover, this procedure provides the conclusion that the Lagrangian could be treated as a fundamental quantity -- in the opposite to the initial case, where the action was a starting object.

All these techniques were informally present in classical texts, written by Lagrange, Hamilton, Carath{\'e}odory, and other pioneers of the calculus of variations. The example of classical mechanics, formulated as a symplectic relation:
\begin{equation}
    \delta L(q, \dot{q}) = \frac {{\rm d}}{{\rm d}t} \left( p\,  \delta q\right)
     = \dot{p} \, \delta q + p\,  \delta \dot{q}\, ,
\end{equation}
which is equivalent to:
\begin{align}
\dot{p} &=  \frac {\partial L}{\partial q}\, , & p &= \frac {\partial L}{\partial \dot{q}}   \, ,
\end{align}
with respect to the canonical symplectic form:
\begin{align}\label{Omega1}
    \omega = \frac {{\rm d}}{{\rm d}t} \left( \delta p \wedge \delta q\right) =
    \delta \dot{p} \wedge \delta q + \delta p \wedge \delta \dot{q} \, ,
\end{align}
was first formulated by W.M.Tulczyjew (cf.~\cite{Tulcz}).
Legendre transformation, like the transition from the Lagrangian to the Hamiltonian picture, is simply described in this formalism as an exchange between control and response parameters: $p$ {\it versus} $\dot{q}$ in~\eqref{Omega1}. The Hamiltonian description is given by:
\begin{align}
    \delta H(q,p) = \delta(p\dot{q}-L) = -\dot{p} \, \delta q + \dot{q}\,  \delta p\, ,
\end{align}
where:
\begin{align}
    -\dot{p} &= \frac{\partial H}{\partial q}\, ,&  \dot{q} &= \frac{\partial H}{\partial p}\, .
\end{align}
It is worthwhile to notice that the well-known canonical symplectic form $\dd p \wedge \dd q$ has no natural analogue in field theory (derived from multiple integrals), whereas the form $\omega$~\eqref{Omega1} has a unique canonical field-theoretical counterpart $\partial_{\nu}\left(\dd p^{\nu} \wedge \dd \phi \right)$.

 \subsection{Connection, inertial reference frames and gravitational field}
\label{intro conn}

The \textit{affine connection} on the tangent bundle is one of the fundamental constructions in differential geometry. However, this kind of structure is not irreducible because the tangent bundle, unlike a principal bundle\footnote{The connection associated with a principal bundle, which is widely used in Yang-Mills field theory, is not considered in this dissertation.} or an associated bundle, possesses an additional structure called a \textit{soldering form} \cite{kobayashi}, which links the “vertical directions'' with the “horizontal directions''. It is represented by the symmetric part of the connection, whereas the remaining skew-symmetric part is geometrically a tensor, commonly referred to as \textit{torsion}. Due to the intrinsic (canonical) structure of the tangent bundle, a general affine connection decomposes into two independent parts: the symmetric connection and the torsion, which, as a tensor, can a priori be treated as an external matter field. The only irreducible component is the symmetric connection, which will be the main focus of this chapter.

The best-known example of a symmetric connection is the \textit{metric connection} (\textit{Levi-Civita connection}), which naturally arises in Riemannian geometry. However, the concept of a symmetric affine connection can also be applied to describe physical phenomena such as local inertial reference frames (observers) or the gravitational field, which will be shown below.

All ideas and results presented in this subsection are taken from Kijowski's textbook~\cite{Geometria} and article~\cite{kij2024}.

\

The story begins from  Newton's laws of dynamics, precisely, from the second one, which could be written in the following way:
\begin{align}
    \ddot{x}^k =  f^{k}\, ,
\end{align}
where $\ddot{x}^k$ denotes the second derivative with respect to some parameter (e.g., biological proper time $s$ of a pilot of the spacecraft at  the position $x^k$), whereas $f^k$ are components of the force per unit mass. Unfortunately, the above equation is valid only in inertial frames, which are introduced in the first law of dynamics, because after coordinate transformations, appear extra terms represented by, e.g.,  Coriolis or centrifugal forces. Indeed, the transformation from the inertial coordinate system $(x^k)$ to  arbitrary coordinates $(y^k)$, which we use to parameterise spacetime points,  produces:
\begin{align}
    \ddot{x}^k  = \frac{\dd ^2 x^k}{\dd s^2} =  \frac{\dd }{\dd s}\left(\frac{\partial x^k}{\partial y^l}\, \dot{y}^l\right) = \frac{\partial^2 x^k}{\partial y^l\, \partial y^m}\, \dot{y}^l\, \dot{y}^m + \frac{\partial x^k}{\partial y^l}\, \ddot{y}^l\, .
\end{align}
The transformation matrix $ \frac{\partial x^k}{\partial y^l}$ is locally invertible, so the second Newton's law equals:
\begin{align}
    \ddot{y}^l  + \Gamma^l_{\ mn}\, \dot{y}^m\, \dot{y}^n  = \frac{\partial y^l}{\partial x^k}\, f^k\, ,
    \label{2law}
\end{align}
where
\begin{align}
    \Gamma^l_{\ mn} := \frac{\partial y^l}{\partial x^k}\,\frac{\partial^2 x^k}{\partial y^m\, \partial y^n}\, .
    \label{gamsymb}
\end{align}
Elements $\Gamma^l_{\ mn}$ of the above array are called the  connection coefficients, or shortly, the \textit{connection}. If the coordinate system $(x^k)$ is inertial and the force $f^k$ vanishes, then $(y^l)$ is also inertial if and only if all second derivatives  vanish: $\frac{\partial^2 x^k}{\partial y^m\, \partial y^n}=0$. Although the inertial reference frame cannot be interpreted with just one coordinate system, due to the fact that $(y^l)$ is a representative of the whole class of coordinate systems, which differ one by one by linear transformations. Precisely, the mentioned class (which is also known as \textit{the inertial reference frame}) is an \textit{equivalence class} $[(x^k)]$  generated by \textit{the equivalence relation} “$\sim$'':
\begin{align}
    \left\{ (x^k) \sim  (y^k) \right\} \ \ \  \Longleftrightarrow \ \ \  \frac{\partial^2 x^k}{\partial y^m\, \partial y^n}=0\, .
\end{align}
As proved in the textbook \cite{Geometria}, chapter 7.3, the above relation is symmetric, reflexive, and transitive and, whence, is a genuine equivalence relation between coordinate systems. Any of its equivalence classes can be identified with Newton's “inertial reference frame''.

For purposes of the theory of gravity, the  \textit{local} version of this relation is necessary:
\begin{align}
    \left\{(x^k) \sim_{\textbf{m}} (y^l) \right\} \Longleftrightarrow \frac{\partial^2 x^k}{\partial y^m\, \partial y^n} (\textbf{m})=0\, ,
\end{align}
therefore,  any of its equivalence classes at the point \textbf{m} can be called a “local inertial reference system at \textbf{m}''.

The theory of gravity consists, therefore, in replacing the First Newton's Law (``\textit{There is a global inertial frame\dots }'') by its \textit{local} version (``\textit{At each spacetime point there is a local inertial frame\dots}) -- cf.~\cite{kij2024}.  The above  symbols $\Gamma^k_{\ lm}$~\eqref{gamsymb} describe the deviation of the coordinate system (used for calculations) from the inertial frame. However, the system of coordinates, which is inertial at the point $\textbf{m}$, will no longer be inertial in the neighborhood of that point, emphasizing the necessity of discussing everything locally.

If a coordinate system exists in which the connection coefficients vanish everywhere, the frame is globally inertial. This crucial observation was used by Einstein to describe the phenomenon of gravitation. As an example, let us consider an orbiting spacecraft, where the gravitational field is effectively eliminated due to the circular motion counteracting the centrifugal force, resulting in a state of weightlessness inside. This implies that, at every moment, there exists a special inertial frame in which gravity is absent, and the spacecraft's trajectory is locally “as straight as possible''.

But globally, since the spacecraft follows a circular trajectory, there is no global inertial frame in the sense of Newton, and these local inertial frames are different at different points. These two aspects suggest that the gravitational field curves not only trajectories but spacetime itself and is fundamentally described by the field of inertial frames, i.e.  by the connection. Below, more mathematical aspects of the connection will be presented, which will be useful in the sequel.

\

 Let $R (\mathbb{M})$ denote all reference frames on spacetime $\mathbb{M}$, whereas $R_{\textbf{m}} (\mathbb{M})$ refers to those at the point $\textbf{m}$. For the given coordinate system $(y^k)$, all geometrical objects (vectors, tensors, etc.) acquire a coordinate description. The same happens with~$R(\mathbb{M})$. However, the obtained structure is slightly different from the mentioned ones. For the local reference frame $r=[(x^k)]\in R_{\textbf{m}} (\mathbb{M})$, is considered the following table of numbers:
\begin{align}
    \Gamma^l_{\ mn} = \frac{\partial y^l}{\partial x^k}\,\frac{\partial^2 x^k}{\partial y^m\, \partial y^n}\, ,
\end{align}
where $(x^k)$ is a representative of $r$. As it was mentioned, the above table uniquely characterises the equivalence class $r$. It implies that two representatives $(x^k)$ and~$(y^k)$ are related in the following way:
\begin{align}
    x^l= y^l + \frac 12\,   \Gamma^l_{\ mn} \, y^m\, y^n\, .
\end{align}
In the above equality, there is assumed that both systems are “centred'' at the same point~$\textbf{m}$. This way $R(\mathbb{M})$ stays a fiber bundle over $\mathbb{M}$ with coordinates $(y^l, \Gamma^l_{\ mn})$. It will be very educative to show how  $\Gamma^l_{\ mn}$ transforms. Introducing a new coordinate system $z^a$, one has:
\begin{align}
    \Gamma^l_{\ mn} &= \frac{\partial y^l}{\partial x^k}\,\frac{\partial}{\partial y^m} \left(\frac{\partial x^k}{\, \partial y^n} \right) =  \frac{\partial y^l}{\partial z^a}\, \frac{\partial z^a}{\partial x^k}\, \frac{\partial }{\partial y^m}\, \left(\frac{\partial x^k}{\partial z^b}\, \frac{\partial z^b}{\, \partial y^n} \right) = \nonumber \\
    &= \frac{\partial y^l}{\partial z^a}\, \frac{\partial z^a}{\partial x^k}\,  \left[  \left(\frac{\partial }{\partial y^m}\, \frac{\partial x^k}{\partial z^b} \right) \, \frac{\partial z^b}{\, \partial y^n} +  \frac{\partial x^k}{\partial z^b}\, \left( \frac{\partial }{\partial y^m}\,\frac{\partial z^b}{\, \partial y^n} \right)\right] = \nonumber\\
    &= \frac{\partial y^l}{\partial z^a}\, \frac{\partial z^a}{\partial x^k}\,  \left[  \left(\frac{\partial z^c}{\partial y^m}\, \frac{\partial }{\partial z^c}\, \frac{\partial x^k}{\partial z^b} \right) \, \frac{\partial z^b}{\, \partial y^n} +  \frac{\partial x^k}{\partial z^b}\, \left( \frac{\partial }{\partial y^m}\,\frac{\partial z^b}{\, \partial y^n} \right)\right] = \nonumber\\
    &=\frac{\partial y^l}{\partial z^a}\,   \frac{\partial z^c}{\partial y^m} \, \frac{\partial z^b}{\, \partial y^n}  \, \left(  \frac{\partial z^a}{\partial x^k} \,   \frac{\partial^2 x^k}{\partial z^b\, \partial z^c}\right)  + \frac{\partial y^l}{\partial z^a}\,    \frac{\partial^2 z^a}{\, \partial y^n\, \partial y^m} \, .
\end{align}
Denoting by
\begin{align}
    \widetilde{\Gamma}^{a}_{\ bc}:=   \frac{\partial z^a}{\partial x^k} \,   \frac{\partial^2 x^k}{\partial z^b\, \partial z^c} \, ,
\end{align}
the above transformation formula equals:
\begin{align}
    \Gamma^l_{\ mn} =  \frac{\partial y^l}{\partial z^a}\,   \frac{\partial z^c}{\partial y^m} \, \frac{\partial z^b}{\, \partial y^n}  \, \widetilde{\Gamma}^{a}_{\ bc}  + \frac{\partial y^l}{\partial z^a}\,    \frac{\partial^2 z^a}{\, \partial y^n\, \partial y^m} \, .
\end{align}
It shows that the above transformation law is linear (first order in $\Gamma$) but is not homogeneous (an extra additive term appears). Thus, the fiber bundle $R_{\textbf{m}} (\mathbb{M})$ is an \textit{affine bundle}.

\

As it was mentioned before, the connection $\Gamma$ describes the equivalence class of inertial frames and, \textit{ex definitione} depends on the chosen coordinate system. Precisely, the uniform movement in Cartesian coordinates $(x^k)$ is given by
\begin{align}
    \ddot{x}^k =0\, ,
\end{align}
whereas the same movement in the spherical coordinates $(y^k)$ will be given by
\begin{align}
    \ddot{y}^k =- \Gamma^k_{\ lm}\, \dot{y}^l\, \dot{y}^m\, ,
\end{align}
where $\Gamma^k_{\ lm}$ does not vanish for all $k,l,m$. Furthermore, the non-vanishing components of the connection may be associated with the presence of a gravitational field (and consequently, curved spacetime) or with the use of a non-inertial reference frame. This naturally leads to the question: “Is there a criterion that can distinguish 'fictitious' forces from the gravitational field?'' This question is, of course, equivalent to the following: “How can we determine whether the connection is flat?''.

Such a criterion indeed exists, and it is precisely the curvature tensor, which is constructed from the connection $\Gamma$ and its partial derivatives $\partial\Gamma$. In this dissertation, two equivalent curvature tensors are used: the older and more widely known Riemann tensor~\eqref{def: tensor riemanna}, defined as
\begin{align}
    R^{\kappa  }_{\ \lambda \mu \nu}:=-\Gamma^{\kappa}_{\ \lambda \mu, \nu}+\Gamma^{\kappa}_{\ \lambda \nu, \mu}-\Gamma^{\sigma}_{\ \lambda \mu}\, \Gamma^{\kappa}_{\ \nu \sigma} + \Gamma^{\sigma}_{\ \lambda \nu}\, \Gamma^{\kappa}_{\ \mu \sigma}\, ,
\end{align}
and the less commonly used, yet particularly useful (especially in variational calculus), Kijowski tensor~\eqref{def: krzywizna kijowskiego}, given by
\begin{align}
K^{\kappa}_{\ \lambda \mu \nu}  = \Gamma^{\kappa}_{\ \lambda \mu, \nu} - \Gamma^{\kappa}_{\ (\lambda \mu, \nu)} + \Gamma^{\sigma}_{\ \lambda \mu}\, \Gamma^{\kappa}_{\ \nu \sigma} -  \Gamma^{\sigma}_{\ (\lambda \mu}\, \Gamma^{\kappa}_{\ \nu) \sigma} \, .
\end{align}

Among all the symmetric affine connections $\Gamma$, there is a special one which is often discussed in the literature: \textit{the Levi-Civita connection}, also called \textit{the metric connection}. Its connection coefficients $\mGamma^{\kappa}_{\ \nu \sigma}$ are called \textit{Christoffell symbols}. It is defined as the unique symmetric connection, which is compatible with the metric structure - see~\eqref{met con}:
\begin{align}
    \mnabla_{\kappa} g_{\mu\nu} := g_{\mu\nu,\kappa} - \mGamma^{\sigma}_{\ \kappa\mu}\, g_{\sigma\nu} - \mGamma^{\sigma}_{\ \kappa\nu}\, g_{\mu\sigma} =0\, .
\end{align}

\section{Variational structure in the metric picture}
\label{metpic}

The metric picture is the most popular description of gravity. The configuration space is spanned by the second jet of the metric tensor $(g_{\mu\nu}, g_{\mu\nu,\kappa}, g_{\mu\nu,\kappa\lambda})$, which means that  gravity is a “second order'' theory. However, derivatives of the metric cannot appear freely, because partial derivatives of tensors, in general, are not tensors. Precisely, in standard approaches, all derivatives of the  metric tensor appear via the  Ricci tensor~$\kolo{K}_{\mu\nu}$~\eqref{def: tensor K}. The corresponding Lagrangian is called \textit{Hilbert Lagrangian} and has the following form:
\begin{align}
    \Lag_{H} = \frac{\sqrt{|\det g|}}{16 \pi} \kolo{K}_{\mu\nu}\, g^{\mu\nu}\, .
    \label{LagH0}
\end{align}
This Lagrangian is a well-known and deeply studied object; however, here are reminded some useful formulas and properties of it. At first, let us define the following tensor densities:
\begin{align}
{\pi}^{\mu\nu} &:= \frac{\sqrt{|\det g|} }{16 \pi} \,  g^{\mu\nu}
    \, , 
    \label{pi2} \\
    \pi_{\kappa}^{\ \lambda\mu\nu} &:=  \delta_\kappa^\nu \,
{\pi}^{\lambda\mu} -\delta_{\kappa}^{(\lambda}\, \pi^{\mu)\nu}\, .
\label{pi40}
\end{align}
Then, the variational formula $\delta \Lag_H$, which was first fully derived\footnote{A phrase “fully derived'' means that all components of the variational formula were written explicitly. For example, in famous textbook of Wheeler, Misner,Thorn \textit{Gravitation} \cite{Gravitation} (page 520, formula 21.86) was calculated only the volume (bulk) part of the variation of the Hilbert Lagrangian density and neglected the boundary term.}  in \cite{pieszy}, is given by:
 \begin{align}
    \delta \Lag_H &= -\frac{1}{16\pi}\kolo{\cal G}^{\mu\nu}\, \delta g_{\mu\nu}  + \partial_{\nu}\left(\pi_{\kappa}^{\ \lambda\mu\nu}\, \delta \mGamma^{\kappa}_{\ \lambda\mu}  \right) = \\ 
    &=\kolo{K}_{\mu\nu}\, \delta \pi^{\mu\nu}  + \partial_{\nu}\left(\pi_{\kappa}^{\ \lambda\mu\nu}\, \delta \mGamma^{\kappa}_{\ \lambda\mu}  \right) \, , 
    \label{delta Lag H}
    \end{align}
    where $\kolo{\cal G}_{\mu\nu}$ is an Einstein tensor density:
\begin{align}
    \kolo{\cal G}^{\mu\nu} := \sqrt{|\det g|}\, G^{\mu\nu} = \sqrt{|\det g|}\, \left(\kolo{K}^{\mu\nu} - \frac12\, g^{\mu\nu}\, \kolo{K}_{\alpha\beta}\, g^{\alpha\beta}\right)\, ,
    \label{def ein tensor}
\end{align}
The  metric density $\sqrt{|\det g|}\, g^{\mu\nu}$ was first used by V.A.~Fock in his textbook \cite{Fock} to simplify the notation, whereas the incorporation of the gravitational constant (which is in geometrical units equal to “1'') and $\frac{1}{16\pi}$ factor was proposed by J. Kijowski in~\cite{newvariationalprinciple}. Therefore, the  Hilbert Lagrangian~\eqref{LagH0} could be written as:
\begin{align}
    \Lag_{H} = \kolo{K}_{\mu\nu}\,   \pi^{\mu\nu} \, ,
    \label{LagH}
\end{align}
what implies the following variational formula:
\begin{align} 
    \delta \Lag_H =  \delta\left(\pi^{\mu\nu}\, \kolo{K}_{\mu\nu} \right) = \pi^{\mu\nu}\, \delta \kolo{K}_{\mu\nu} + \kolo{K}_{\mu\nu}\, \delta\pi^{\mu\nu}  \, ,
    \label{delta Lag H2}
\end{align}
which, compared with the formula ~\eqref{delta Lag H} gives:
\begin{align}
     \pi^{\mu\nu}\, \delta \kolo{K}_{\mu\nu} = \partial_{\nu}\left(\pi_{\kappa}^{\ \lambda\mu\nu}\, \delta \mGamma^{\kappa}_{\ \lambda\mu}  \right)\, .
     \label{paigoaesg}
\end{align}
The addition of external matter fields slightly modifies the variational formula in the metric picture due to its dependence on the metric tensor \( g_{\mu\nu} \) and, possibly, its derivatives \( g_{\mu\nu,\kappa} \) if covariant derivatives of the field are involved\footnote{For details see \cite{nonmetricity}.}. Thus, the configuration space  has the following form:
\begin{align}
     \left(\phi, \phi_{,\alpha}, g_{\mu\nu}, g_{\mu\nu,\kappa} \right) = \left(\phi, \phi_{,\alpha}, g_{\mu\nu}, \mGamma^{\kappa}_{ \ \lambda \mu}  \right) = \left(\phi, \mnabla_{\alpha}\phi , g_{\mu\nu} \right)\, ,
\end{align}
what induces the variational formula for the matter Lagrangian $\Lag_{\rm matt}$:
\begin{align}
    \delta \Lag_{\rm matt} = \frac{\partial \Lag_{\rm matt}}{\partial g_{\mu\nu}}\, \delta g_{\mu\nu} + \cP^{\lambda\mu}_{\ \ \kappa}\, \delta \mGamma^{\kappa}_{ \ \lambda \mu} + \left(\frac{\partial \Lag_{\rm matt}}{\partial \phi} - p^{\nu}_{\ ,\nu} \right)\, \delta \phi + \partial_{\nu}\left(p^{\nu}\, \delta \phi \right)\, ,
    \label{delta lagmatt}
\end{align}
where
\begin{align}
    \cP^{\lambda\mu}_{\ \ \kappa}:= \frac{\partial\Lag_{\rm matt}}{\partial \mGamma^{\kappa}_{\ \lambda\mu}}\, .
    \label{def calP0}
\end{align}
Now, the metric Lagrangian $\Lag_g$, which describes the interaction between geometry and matter, is simply a sum of the Hilbert and matter Lagrangians:
\begin{align}
    \Lag_g:= \Lag_H + \Lag_{\rm matt }\, ,
    \label{def Lagg}
\end{align}
and the corresponding variational formula is as follows:
\begin{align}
    \delta \Lag_g &= \left( \frac{\partial \Lag_{\rm matt}}{\partial g_{\mu\nu}} -\frac{1}{16\pi}\kolo{\cal G}^{\mu\nu}\right) \, \delta g_{\mu\nu} + \cP^{\lambda\mu}_{\ \ \kappa}\, \delta \mGamma^{\kappa}_{ \ \lambda \mu}  + \left(\frac{\partial \Lag_{\rm matt}}{\partial \phi} - p^{\nu}_{\ ,\nu} \right)\, \delta \phi +\nonumber\\
    &\quad +\partial_{\nu}\left(\pi_{\kappa}^{\ \lambda\mu\nu}\, \delta \mGamma^{\kappa}_{\ \lambda\mu} +p^{\nu}\, \delta \phi \right) \,.
    \label{laggaaaa}
\end{align}
Of course, the term $\cP^{\lambda\mu}_{\ \ \kappa}\, \delta \mGamma^{\kappa}_{ \ \lambda \mu}$ depends on the metric and its derivatives, but it was proven\footnote{The analogous proof is presented in \textbf{Chapter~\ref{chap passage}} in \textbf{Lemma~\ref{lemma cR}}.} in \cite{mag, nonmetricity} that:
\begin{eqnarray*}
    \mathcal{P}_{\ \ \lambda}^{ \mu\nu}\,  \delta \mGamma^{\lambda}_{\  \mu\nu} =  \partial_{\kappa} \left( {\cal R}^{\mu\nu\kappa}\, \delta g_{\mu\nu} \right) - \left(\mnabla_{\kappa} \mathcal{R}^{\mu\nu\kappa}  \right)\,  \delta g_{\mu\nu}      \, ,
\end{eqnarray*}
where:
\begin{align}
    {\cal R}^{\mu\nu\kappa} &:=  \frac 12 \left(  \mathcal{P}^{ \kappa\mu\nu} + \mathcal{P}^{\kappa \nu \mu} - \mathcal{P}^{ \mu\nu\kappa}  \right) \, .
\label{def: calR}
\end{align}
The variational formula for the metric Lagrangian~\eqref{laggaaaa} can be expressed as:
\begin{align}  
    \delta \Lag_g &= \left[ \frac{\partial \Lag_{\rm matt}}{\partial g_{\mu\nu}} -\frac{1}{16\pi}\kolo{\cal G}^{\mu\nu}- \left(\mnabla_{\kappa} \mathcal{R}^{\mu\nu\kappa}  \right)\right] \, \delta g_{\mu\nu}   + \left(\frac{\partial \Lag_{\rm matt}}{\partial \phi} - p^{\nu}_{\ ,\nu} \right)\, \delta \phi +\nonumber\\  
    &\quad +\partial_{\nu}\left({\cal R}^{\kappa\lambda\nu}\, \delta g_{\kappa\lambda}  
 + \pi_{\kappa}^{\ \lambda\mu\nu}\, \delta \mGamma^{\kappa}_{\ \lambda\mu} +p^{\nu}\, \delta \phi \right)\, .  
\end{align}  
The field equations are determined by the bulk (volume) terms:
\begin{align}  
    \frac{\delta \Lag_{g}}{\delta g_{\mu\nu}}&=0&  \Longrightarrow &&\frac{\partial \Lag_{\rm matt}}{\partial g_{\mu\nu}} &= \frac{1}{16\pi}\kolo{\cal G}^{\mu\nu} + \left(\mnabla_{\kappa} \mathcal{R}^{\mu\nu\kappa}  \right)\, , \\
     \frac{\delta \Lag_{g}}{\delta \phi}&=0&\Longrightarrow& &  \frac{\partial \Lag_{\rm matt}}{\partial \phi} &= p^{\nu}_{\ ,\nu} \, .  
\end{align}  
The first equation corresponds to Einstein’s field equation, while the second one represents the Euler-Lagrange equation for the matter field $\phi$.  The remaining boundary term encodes the symplectic relation between control and response parameters.

\section{Variational structure in the affine picture}
\label{chap var str aff}

The affine picture was firstly proposed by Jerzy Kijowski in \cite{newvariationalprinciple}, where  the affine Lagrangian $\Lag_A$ depends on the first jet of the symmetric affine connection $(\Gamma, \partial \Gamma)$. The motivation of such construction was briefly  presented in \textbf{Chapter~\ref{intro conn}}. Thus, the variation $\delta \Lag_A$ \textit{on shell} (cf. \textbf{Chapter~\ref{varcal}}) is given by the boundary term~\eqref{lagonshell}:
\begin{align}
    \delta \Lag_A = \partial_{\nu}\left( \cP_{\kappa}^{\ \lambda\mu\nu}\, \delta \Gamma^{\kappa}_{\ \lambda\mu }\right) = \partial_{\nu}\cP_{\kappa}^{\ \lambda\mu\nu}\, \delta \Gamma^{\kappa}_{\ \lambda\mu } + \cP_{\kappa}^{\ \lambda\mu\nu}\, \delta \Gamma^{\kappa}_{\ \lambda\mu,\nu }\, , 
    \label{var0}
\end{align}
where  $\cP_{\kappa}^{\ \lambda\mu\nu}$ is  a momentum canonically conjugated to the connection $\Gamma^{\kappa}_{\ \lambda\mu}$:
\begin{align}
    \cP_{\kappa}^{\ \lambda\mu\nu}:=\frac{\partial \Lag_A}{\partial \Gamma^{\kappa}_{\ \lambda\mu,\nu}}\, .
    \label{def: cP}
\end{align}
Importantly, the connection $\Gamma^{\kappa}_{\ \lambda\mu}$ and its first derivatives $\Gamma^{\kappa}_{\ \lambda\mu,\nu}$ are not tensors. Therefore, they cannot appear freely in the Lagrangian, which must be a scalar density. The connection naturally appears in two ways: through curvature tensors or via covariant derivatives, which are not considered in this approach\footnote{The affine theory that includes covariant derivatives of additional matter fields was presented~in~\cite{nonmetricity}.}.

This affine theory is assumed to be described by “first-order Lagrangians,” meaning that the affine Lagrangian does not depend on second (or higher) derivatives of the connection $\Gamma$. Consequently, derivatives of the connection $\Gamma^{\kappa}_{\ \lambda\mu,\nu}$ must be arranged in the Riemann (or equivalently, Kijowski) curvature tensor\footnote{The construction of higher-order curvature tensors is also possible — see \cite{senger}.}.

For practical purposes, as will be seen below, the variational calculus will employ the Kijowski tensor $K^{\kappa}_{\ \lambda\mu\nu}$~\eqref{def: krzywizna kijowskiego}. Formally, this means that in the configuration space $(\Gamma,\partial\Gamma)$, a map — a coordinate transformation — is introduced:
\begin{align}
\left(\Gamma^{\kappa}_{\ \lambda\mu}, \Gamma^{\kappa}_{\ \lambda\mu,\nu}\right) \longmapsto
\left(\Gamma^{\kappa}_{\ \lambda\mu}, K^{\kappa}_{\ \lambda\mu\nu}\right)\, .
\label{map}
\end{align}
It implies the following simple theorem:
\begin{theorem}
\label{th cPdG}
    The variation~\eqref{var0} of the affine Lagrangian $\Lag_A\left(\Gamma^{\kappa}_{\ \lambda\mu}, K^{\kappa}_{\ \lambda\mu\nu}\right)$ is given by the following formula:
    \begin{align}
     \delta \Lag_A =\partial_{\nu}\left( \cP_{\kappa}^{\ \lambda\mu\nu}\, \delta \Gamma^{\kappa}_{\ \lambda\mu }\right) = \left( \nabla_{\nu}\cP_{\kappa}^{\ \lambda\mu\nu}\right)\, \delta \Gamma^{\kappa}_{\ \lambda\mu } + \cP_{\kappa}^{\ \lambda\mu\nu}\, \delta K^{\kappa}_{\ \lambda\mu\nu} \, , \label{var1}
\end{align}
where $\cP_{\kappa}^{\ \lambda\mu\nu}$ satisfies:
\begin{align}
    \cP_{\kappa}^{\ \lambda\mu\nu}&=\cP_{\kappa}^{\ \mu\lambda\nu}\,, & \cP_{\kappa}^{\ (\lambda\mu\nu)}&=0\, .
    \label{cond calP}
\end{align}
\end{theorem}
\begin{proof}
The proof strictly relies on tensor calculus. Firstly, if the Lagrangian depends on the derivatives of the connection only through the Kijowski tensor \( K^{\kappa}_{\ \lambda\mu\nu} \)~\eqref{def: krzywizna kijowskiego}, then the momentum \( \mathcal{P}_{\kappa}^{\ \lambda\mu\nu} \) must satisfy condition~\eqref{cond calP} in order to be well-defined as a derivative~\eqref{def: cP}. Furthermore, as a derivative of the Lagrangian (a scalar density) with respect to the Kijowski tensor, it must itself be a tensor density. Condition~\eqref{cond calP} represents a dual symmetry of the Kijowski tensor — see~\eqref{eq: pierwsza tozsamosc Bianchiego dla K}. If the above condition is not imposed, a fictitious gauge will emerge. Therefore:
\begin{align}
    \cP_{\kappa}^{\ \lambda\mu\nu}\, \delta K^{\kappa}_{\ \lambda\mu\nu} &=  \cP_{\kappa}^{\ \lambda\mu\nu}\, \delta \Gamma^{\kappa}_{\ \lambda \mu, \nu} -  \underbrace{\cP_{\kappa}^{\ \lambda\mu\nu}\, \delta \Gamma^{\kappa}_{\ (\lambda \mu, \nu)}}_{=0} +  \nonumber \\
    &\quad +\cP_{\kappa}^{\ \lambda\mu\nu}\, \delta \left(\Gamma^{\sigma}_{\ \lambda \mu}\, \Gamma^{\kappa}_{\ \nu \sigma}\right) -  \underbrace{\cP_{\kappa}^{\ \lambda\mu\nu}\, \delta \left(\Gamma^{\sigma}_{\ (\lambda \mu}\, \Gamma^{\kappa}_{\ \nu) \sigma}\right)}_{=0} =\nonumber \\
    &=\partial_{\nu}\left( \cP_{\kappa}^{\ \lambda\mu\nu}\, \delta \Gamma^{\kappa}_{\ \lambda\mu} \right) - \left(\partial_{\nu}  \cP_{\kappa}^{\ \lambda\mu\nu}\right)\, \delta  \Gamma^{\kappa}_{\ \lambda\mu}+ \nonumber \\
    &\quad +\left(\cP_{\kappa}^{\ \nu\sigma\lambda}\, \Gamma^{\mu}_{\ \nu\sigma} + \cP_{\sigma}^{\ \lambda\mu\nu}\,\Gamma^{\sigma}_{\ \nu\kappa} \right)\, \delta \Gamma^{\kappa}_{\ \lambda\mu}\, .
\end{align}
Fortunately, terms proportional to $\delta \Gamma^{\kappa}_{\ \lambda\mu}$ combine to the covariant derivative. Indeed:
\begin{align}
    \left(\nabla_{\nu}  \cP_{\kappa}^{\ \lambda\mu\nu}\right)\, \delta  \Gamma^{\kappa}_{\ \lambda\mu}&=  \left(\partial_{\nu}  \cP_{\kappa}^{\ \lambda\mu\nu} \underline{- \Gamma^{\sigma}_{\ \nu\sigma}\, \cP_{\kappa}^{\ \lambda\mu\nu}} - \Gamma^{\sigma}_{\ \nu\kappa}\, \cP_{\sigma}^{\ \lambda\mu\nu} + \Gamma^{\lambda}_{\ \nu\sigma}\, \cP_{\kappa}^{\ \sigma\mu\nu} +\right. \nonumber \\
    &\quad \left.+ \Gamma^{\mu}_{\ \nu\sigma}\, \cP_{\kappa}^{\ \lambda\sigma\nu} \underline{+ \Gamma^{\nu}_{\ \nu\sigma}\, \cP_{\kappa}^{\ \lambda\mu\sigma }} \right)\, \delta  \Gamma^{\kappa}_{\ \lambda\mu} = \nonumber\\
    &=\left(\partial_{\nu}  \cP_{\kappa}^{\ \lambda\mu\nu}  - \Gamma^{\sigma}_{\ \nu\kappa}\, \cP_{\sigma}^{\ \lambda\mu\nu} +2 \Gamma^{\mu}_{\ \nu\sigma}\, \cP_{\kappa}^{\ \lambda\sigma\nu}\right)\, \delta  \Gamma^{\kappa}_{\ \lambda\mu}\, . 
\end{align}
The first underlined term appears because $\cP$ is a tensor density, whereas the other terms arise from the definition of the covariant derivative of a tensor. Using the condition~\eqref{cond calP}, the last term equals:
\begin{align}
    \Gamma^{\mu}_{\ \nu\sigma}\, \cP_{\kappa}^{\ \lambda\sigma\nu} \, \delta  \Gamma^{\kappa}_{\ \lambda\mu} = -\Gamma^{\mu}_{\ \nu\sigma}\,\left( \cP_{\kappa}^{\ \sigma\nu\lambda}+ \cP_{\kappa}^{\ \nu\lambda\sigma}\right) \, \delta  \Gamma^{\kappa}_{\ \lambda\mu} \, ,
\end{align}
which implies:
\begin{align}
     2\Gamma^{\mu}_{\ \nu\sigma}\, \cP_{\kappa}^{\ \lambda\sigma\nu} \, \delta  \Gamma^{\kappa}_{\ \lambda\mu} =-\Gamma^{\mu}_{\ \nu\sigma}\,  \cP_{\kappa}^{\ \sigma\nu\lambda}\, \delta  \Gamma^{\kappa}_{\ \lambda\mu} \, .
\end{align}
Finally,
\begin{align}
     \cP_{\kappa}^{\ \lambda\mu\nu}\, \delta K^{\kappa}_{\ \lambda\mu\nu} = \partial_{\nu}\left( \cP_{\kappa}^{\ \lambda\mu\nu}\, \delta \Gamma^{\kappa}_{\ \lambda\mu} \right) - \left(\nabla_{\nu}  \cP_{\kappa}^{\ \lambda\mu\nu}\right)\, \delta  \Gamma^{\kappa}_{\ \lambda\mu}\, ,
\end{align}
which finishes the proof.
\end{proof}
 As mentioned earlier, $K^{\kappa}_{\ \lambda\mu\nu}$ is equivalent to the Riemann tensor $R^{\kappa}_{\ \lambda\mu\nu}$ but has different symmetries (cf. formulae~\eqref{rel KR} and~\eqref{Riem od Kij}). The practical advantage of introducing the Kijowski tensor now becomes evident: it shares the same symmetry in the first two lower indices as the connection $\Gamma^{\kappa}_{\ \lambda\mu}$. Consequently, the momentum $\cP_{\kappa}^{\ \lambda\mu\nu}$ remains a proper tensor density, as it is defined as the derivative of the affine Lagrangian (a scalar density) with respect to the Kijowski tensor~\eqref{var1}:
\begin{align}
    \cP_{\kappa}^{\ \lambda\mu\nu}:= \frac{\partial \Lag_A}{\partial \Gamma^{\kappa}_{\ \lambda\mu,\nu}} =\frac{\partial \Lag_A}{\partial K^{\kappa}_{\ \lambda\mu\nu}}\, .
    \label{def cPK}
\end{align}

The application of the decomposition~\eqref{var1} of the Kijowski tensor $K^{\kappa}_{\ \lambda\mu\nu}$ to the variation $\delta \Lag_A$~\eqref{var1}  induces the “analogue” decomposition of the momentum $ \cP_{\kappa}^{\ \lambda\mu\nu}$:
\begin{lemma}
\label{lemma dec cP}
    The momentum $\cP_{\kappa}^{\ \lambda\mu\nu}$~\eqref{def cPK} that satisfies:
    \begin{align}
        \cP_{\kappa}^{\ \lambda\mu\nu} &= \cP_{\kappa}^{\ \mu\lambda\nu}\, , &  \cP_{\kappa}^{\ (\lambda\mu\nu)} &=0\, ,
    \end{align}
    decomposes as follows:
\begin{align}
    \cP_{\kappa}^{\ \lambda \mu \nu}  =  \pi_{\kappa}^{\ \lambda\mu\nu} -\delta^{(\lambda}_{\kappa}\, \chi^{\mu) \nu }  + \Omega_{\kappa}^{\ \lambda \mu \nu} \, ,
\label{eq: rozklad pedu P} 
\end{align}
where
\begin{align}
\pi_{\kappa}^{\ \lambda\mu\nu} &=  \delta_\kappa^\nu \,
{\pi}^{\lambda\mu} -\delta_{\kappa}^{(\lambda}\, \pi^{\mu)\nu}\, ,
\label{pi4} \\
\pi^{\mu \nu} &=-\frac 23 \cP_{\kappa}^{\  \kappa (\mu \nu)} \, ,  \\
\chi^{\mu \nu} &=-\frac 25  \cP_{\kappa}^{\ \kappa [\mu \nu]} \, ,
\end{align}
and $\Omega_{\kappa}^{\ \lambda\mu\nu}$ is the remaining, algebraically traceless part of $\cP_{\kappa}^{\ \lambda\mu\nu} $.
\end{lemma}
The proof contains a simple verification of given conditions. Finally, the variation~\eqref{var1} takes the following form:
\begin{align}
 \delta \Lag_A= \left(\nabla_{\nu}\cP_{\kappa}^{\ \lambda\mu\nu} \right)\, \delta\Gamma^{\kappa}_{\ \lambda\mu} + \pi^{\mu\nu}\, \delta K_{\mu\nu}+\chi^{\mu\nu}\, \delta F_{\mu\nu}  +  \Omega_{\kappa}^{\ \lambda\mu\nu}\, \delta U^{\kappa}_{\ \lambda\mu\nu}\, .
 \label{varLA0}
\end{align}

The above formula corresponds to the symplectic structure of the theory, where the field equations (equivalent to the Euler-Lagrange equations — see \textbf{Chapter~\ref{varcal}}) are as follows:

\begin{align}  
    \frac{\partial \Lag_A}{\partial \Gamma^{\kappa}_{\ \lambda\mu}} &= \nabla_{\nu}\cP_{\kappa}^{\ \lambda\mu\nu} \label{1fieldeq0} \, , \\  
   \pi^{\mu\nu} &=\frac{\partial \Lag_A}{\partial K_{\mu\nu}}\, , \label{rel pi}\\  
    \chi^{\mu\nu} &=\frac{\partial \Lag_A}{\partial F_{\mu\nu}}\, , \label{rel chi} \\  
    \Omega_{\kappa}^{\ \lambda\mu\nu} &=\frac{\partial \Lag_A}{\partial U^{\kappa}_{\ \lambda\mu\nu}}\, . \label{rel Omega}  
\end{align}  

These results complete the variational description in the affine picture. The symplectic formula~\eqref{varLA0} and the field equations (\ref{1fieldeq0}-\ref{rel Omega}) have been derived. The only non-trivial aspect lies in the identification of the momentum \(\pi^{\mu\nu}\) with the metric tensor \(g_{\mu\nu}\)~\eqref{pi2}. However, this assumption leads to physically acceptable conclusions and corresponds to the simplest case, where only the symmetric Ricci tensor \(K_{\mu\nu}\) is considered.  

Moreover, this approach is consistent with the \textit{Palatini variational principle}, in which the metric tensor and curvature are treated on equal footing as independent configurations. In such a framework, the metricity of the connection arises naturally for a class of theories that do not depend on covariant derivatives, such as scalar field theory and electrodynamics — see  \cite{nonmetricity}. A few examples of affine theories will be presented in the sequel.

\subsection{The first field equation and the non-metricity equation}
\label{first field eq}

The examples of affine theories which will be discussed in this dissertation will depend only on the curvature tensor. It means that the introduced configuration space $(\Gamma^{\kappa}_{\ \lambda\mu}, K^{\kappa}_{\ \lambda\mu\nu})$ -- see~\eqref{map} -- is restricted  \textbf{only} to the Kijowski curvature tensor $K^{\kappa}_{\ \lambda\mu\nu}$~\eqref{def: krzywizna kijowskiego}. Thus, \textit{the first field equation}~\eqref{1fieldeq0} is the  following\footnote{In the paper \cite{nonmetricity} is presented the affine theory of only symmetric Ricci tensor $K_{\mu\nu}$ with external fields and then, the right-hand side of the first field equation could be non-zero, due to dependence of the Lagrangian on covariant derivatives of those external fields.}:
\begin{align}
    \frac{\partial \Lag_A}{\partial \Gamma^{\kappa}_{\ \lambda\mu}} = \nabla_{\nu}\cP_{\kappa}^{\ \lambda\mu\nu} =0\, .
    \label{1fieldeq}
\end{align}
Furthermore, using the decomposition of the momentum $\cP_{\kappa}^{\ \lambda\mu\nu}$~\eqref{eq: rozklad pedu P}, it could be simply proved that the above equation is equivalent to
\begin{align}
\nabla_{\kappa} \pi^{\lambda \mu} =  - \frac 23 \, \delta^{(\lambda}_{\kappa} \mathcal{J}^{\mu)}- \nabla_{\nu} \Omega_{\kappa}^{\ \lambda \mu \nu}\, ,
\label{1row}
\end{align}
where
\begin{align}
\mathcal{J}^{\mu} := \nabla_{\nu}\chi^{\mu\nu} = \partial_{\nu}\chi^{\mu\nu}= \chi^{\mu\nu}_{\ \ ,\nu}\, .
\label{cal J}
\end{align}
The equality of covariant and partial divergences of $\chi^{\mu\nu}$ is proven in the following lemma:
\begin{lemma}
\label{lem div}
The covariant divergence (with respect to any symmetric affine connection $\Gamma$) of a skew-symmetric tensor density $\chi^{\mu\nu}$  is equal to the partial divergence of $\chi^{\mu\nu}$:
\begin{align}
    \nabla_{\nu} \chi^{\mu\nu} = \partial_{\nu} \chi^{\mu\nu}\, .
\end{align}
\end{lemma}
\begin{proof}
The proof is obtained via the explicit calculus:
    \begin{align}
        \nabla_{\nu}\chi^{\mu\nu} - \partial_{\nu} \chi^{\mu\nu} =  - \Gamma^{\kappa}_{\ \kappa\nu}\, \chi^{\mu\nu} + \Gamma^{\mu}_{\ \kappa\nu}\, \chi^{\kappa\nu} + \Gamma^{\nu}_{\ \kappa\nu}\, \chi^{\mu\kappa}=0\, . 
    \end{align}
    The first term appears due to the density character of   $\chi^{\mu\nu}$ and it simply cancels with the last term, whereas the second term vanishes due to the contraction of opposed symmetries between connection $\Gamma$ and tensor density $\chi$.
\end{proof}

According to the discussion about the Hilbert Lagrangian in \textbf{Chapter~\ref{metpic}}, the fundamental relation between the metric tensor \(g_{\mu\nu}\) and the momentum \(\pi^{\mu\nu}\) was presented in~\eqref{pi2}. Therefore, equation~\eqref{1row} describes the deviation from the metricity of the connection, as the vanishing of the right-hand side of this equation is equivalent to the metricity condition:
\begin{align}
\nabla_{\kappa} \pi^{\lambda \mu}=0 \Longrightarrow \nabla_{\kappa} g_{\lambda \mu}=0 \Longrightarrow \Gamma^{\kappa}_{\ \lambda\mu} = \mGamma^{\kappa}_{\ \lambda\mu}\, .  
\label{meteq}
\end{align}

Hence, equation~\eqref{1row} will be referred to as \textit{the non-metricity equation}, and it uniquely induces the decomposition of \(\Gamma\)~\eqref{decGamma} into the metric part \(\mGamma\) and the non-metricity tensor \(N\) at the very first stage. Specifically, the metric connection~\(\!\mGamma\) serves as a \textit{general solution of a homogeneous system of equations}, whereas the non-metricity tensor \(N\) acts as a \textit{particular solution of a non-homogeneous system of equations}. Using this decomposition, equation~\eqref{1row} can be reformulated as an algebraic equation for the non-metricity tensor \(N\), as presented in the theorem below:
\begin{theorem}
\label{th non-metricity}
    The non-metricity equation~\eqref{1row}:
    \begin{align}
\nabla_{\kappa} \pi^{\lambda \mu} =  - \frac 23 \, \delta^{(\lambda}_{\kappa} \mathcal{J}^{\mu)}- \nabla_{\nu} \Omega_{\kappa}^{\ \lambda \mu \nu}\, ,
\label{1row1}
\end{align}
for the  connection $\Gamma$ is equivalent with the below linear equation for the  non-metricity tensor $N:=\Gamma-\mGamma$:
\begin{align}
    \left(\delta^{\alpha}_{\kappa}\, \delta^{\beta}_{\lambda}\, \delta^{\gamma}_{\mu} - \Delta^{\alpha\beta\gamma}_{\kappa\lambda\mu} \right) N_{\alpha\beta\gamma} &=\frac{8\pi}{\sqrt{|\det g|}}\, \left[\frac{2}{3}\, g_{\kappa(\lambda}\, \cJ_{\mu)} - g_{\lambda\mu}\, \cJ_{\kappa} + \right. \nonumber  \\
    &\quad \left. +\mnabla_{\nu} \left(\Omega_{\kappa\lambda\mu}^{\ \ \ \nu} - 2\Omega_{(\lambda\mu)\kappa}^{\ \ \ \nu}+g_{\kappa(\lambda}\, \cO_{\mu)}^{\ \ \nu} - \frac 12 \, g_{\lambda\mu}\, \cO_{\kappa}^{\ \nu} \right)\right]\, ,   \label{1row2}\\
  \frac{\sqrt{|\det g|}}{8\pi}\,   \Delta^{\alpha\beta\gamma}_{\kappa\lambda\mu}  &:= 2\delta_{(\lambda}^{\gamma}\,\Omega^{\alpha\ \ \ \beta}_{\ \mu)\kappa}-\delta^{\gamma}_{\kappa}\,\Omega^{\alpha  \ \ \beta}_{\ \lambda\mu } - 2 \delta^{\alpha}_{\kappa}\,\Omega_{(\lambda\mu)}^{\ \ \ \ \beta\gamma} - 2\delta^{\alpha}_{(\lambda}\, \Omega_{\mu)\kappa}^{\ \ \ \beta\gamma} +\nonumber \\
  &\quad +2\delta^{\alpha}_{(\lambda|}\,\Omega_{\kappa|\mu)}^{\ \ \ \ \beta\gamma} +\frac12\, \delta^{\gamma}_{\kappa}\, g_{\lambda\mu}\, \cO^{\alpha\beta} - \delta^{\gamma}_{(\lambda} \,g_{\mu)\kappa}\, \cO^{\alpha\beta} +\nonumber\\
  &\quad +2 g_{\kappa(\lambda}\, \Omega_{\mu)}^{\ \ \alpha\beta\gamma} - g_{\lambda\mu}\, \Omega_{\kappa}^{\ \alpha\beta\gamma}\, , 
  \label{Delta6}\\
    \cO_{\kappa}^{\ \nu} &:= \Omega_{\kappa}^{ \ \lambda\mu\nu}\, g_{\lambda\mu}\, . \label{def tcO}
\end{align}
where $\mGamma$ is a metric connection which conserves the metric tensor density $\pi^{\mu\nu}$~\eqref{pi2}:
\begin{align}
\mnabla_{\kappa}\pi^{\lambda\mu} =    \mnabla_{\kappa}\left(\frac{\sqrt{|\det g|}}{16\pi}\, g^{\lambda\mu}\right) =0\, .
\end{align}
\end{theorem}
\begin{proof}
The covariant derivatives $\nabla {\pi}$ and $\nabla\Omega$ in formula~\eqref{1row1} are given by the following expressions:
\begin{align}
   \nabla_{\kappa}\, \pi^{\lambda\mu}&= \underbrace{\mnabla_{\kappa} \pi^{\lambda \mu}}_{=0} -N^{\sigma}_{\ \sigma \kappa}\, \pi^{\lambda\mu} +N^{\lambda}_{\ \kappa\sigma}\, \pi^{\sigma\mu} +N^{\mu}_{\ \kappa\sigma}\, \pi^{\lambda\sigma } = \nonumber \\
   &= \frac{\sqrt{|\det g|}}{16\pi}\, \left( -N^{\sigma}_{\ \sigma \kappa}\, g^{\lambda\mu} +N^{\lambda}_{\ \kappa\sigma}\, g^{\sigma\mu} +N^{\mu}_{\ \kappa\sigma}\, g^{\lambda\sigma }  \right)\, , \label{cdevpi}\\
   \nabla_{\nu}\, \Omega_{\kappa}^{\ \lambda\mu\nu}&=  \mnabla_{\nu} \Omega_{\kappa}^{\ \lambda \mu \nu} \underline{ - N^{\sigma}_{\ \nu\sigma}\, \Omega_{\kappa}^{ \lambda\mu\nu}}  -N^{\sigma}_{\ \nu\kappa}\, \Omega_{\sigma}^{\ \lambda\mu\nu}+ \nonumber \\
   &\quad +N^{\lambda}_{\ \nu\sigma}\, \Omega_{\kappa}^{\ \sigma\mu\nu} + N^{\mu}_{\ \nu\sigma}\, \Omega_{\kappa}^{\ \lambda\sigma\nu}  \underline{+N^{\nu}_{\ \nu\sigma}\, \Omega_{\kappa}^{\ \lambda\mu\sigma}} = \nonumber \\
   &= \mnabla_{\nu} \Omega_{\kappa}^{\ \lambda \mu \nu}   -N^{\sigma}_{\ \nu\kappa}\, \Omega_{\sigma}^{\ \lambda\mu\nu}+  N^{\lambda}_{\ \nu\sigma}\, \Omega_{\kappa}^{\ \sigma\mu\nu} + N^{\mu}_{\ \nu\sigma}\, \Omega_{\kappa}^{\ \lambda\sigma\nu}  \, .
\end{align}
When the connection is split, the metric tensor can be used to lower and raise indices. Then, the first field equation~\eqref{1row1}, with the above formulas implemented, takes the following form:
\begin{align}
 \frac{\sqrt{|\det g|}}{16\pi}\, \left( -N^{\sigma}_{\ \sigma \kappa}\, g_{\lambda\mu} +N_{\lambda\mu \kappa} +N_{\mu\lambda \kappa }  \right) &=-\frac{1}{3}\, g_{\kappa\lambda}\, \cJ_{\mu}-\frac{1}{3}\, g_{\kappa\mu}\, \cJ_{\lambda } - \mnabla_{\nu}\, \Omega_{\kappa \lambda \mu}^{\ \ \ \ \nu} + \nonumber\\
 & \quad  -  N^{\lambda}_{\ \nu\sigma}\, \Omega_{\kappa}^{\ \sigma\mu\nu} - N^{\mu}_{\ \nu\sigma}\, \Omega_{\kappa}^{\ \lambda\sigma\nu}  + \nonumber \\
 &\quad +N^{\sigma}_{\ \nu\kappa}\, \Omega_{\sigma}^{\ \lambda\mu\nu}  \, .
 \label{fsodfn}
\end{align}
In the above expression, the trace of the non-metricity $N^{\sigma}_{\ \sigma\kappa}$ appears, which can be derived by contracting the entire equation with $g^{\lambda\mu}$:
    \begin{align}
         -\frac{\sqrt{|\det g|}}{8\pi}\, N^{\sigma}_{\ \sigma \kappa}  &= -\frac{2}{3}\,   \cJ_{\kappa} - \mnabla_{\nu}\underbrace{\Omega_{\kappa}^{\ \lambda\mu\nu}\, g_{\lambda\mu}}_{:=\cO_{\kappa}^{\ \nu}} - 2 N_{\alpha\beta\gamma}\, \Omega_{\kappa}^{\ \alpha\beta\gamma} + N^{\sigma}_{\ \nu\kappa}\, \underbrace{\Omega_{\sigma}^{\ \lambda\mu\nu}\, g_{\lambda\mu}}_{:=\cO_{\sigma}^{\ \nu}}=\nonumber \\
         &= -\frac{2}{3}\,   \cJ_{\kappa} - \mnabla_{\nu}\cO_{\kappa}^{\ \nu}  - 2 N_{\alpha\beta\gamma}\, \Omega_{\kappa}^{\ \alpha\beta\gamma} + N_{\alpha\beta\kappa}\,\cO^{\alpha\beta}\, .
    \end{align}
    Therefore, the formula~\eqref{fsodfn} takes the following form:
    \begin{align}
 \frac{\sqrt{|\det g|}}{16\pi}\, \left( N_{\lambda\mu \kappa} +N_{\mu\lambda \kappa }  \right) &=-\frac{1}{3}\, g_{\kappa\lambda}\, \cJ_{\mu}-\frac{1}{3}\, g_{\kappa\mu}\, \cJ_{\lambda } + \frac{1}{3}\,  g_{\lambda\mu}\, \cJ_{\kappa} - \mnabla_{\nu}\, \Omega_{\kappa \lambda \mu}^{\ \ \ \ \nu} + \nonumber\\
  & \quad  + \left( \frac 12\, \mnabla_{\nu} \, \cO_{\kappa}^{\ \nu}  +  N_{\alpha\beta\gamma}\, \Omega_{\kappa}^{\ \alpha\beta\gamma} -\frac 12\, N_{\sigma  \nu\kappa}\,  \cO^{\sigma \nu} \right)\, g_{\lambda\mu}+ \nonumber  \\
 &\quad   -  N_{\lambda  \sigma\nu}\, \Omega_{\kappa\mu}^{\ \ \sigma \nu} -  N_{\mu  \sigma\nu}\, \Omega_{\kappa\lambda}^{\ \ \sigma \nu} +N_{\sigma  \nu \kappa}\, \Omega^{\sigma \ \ \nu}_{\   \lambda\mu} \, . 
\end{align}
Rewriting this equation for commuted indices $\kappa\rightarrow \lambda \rightarrow \mu \rightarrow \kappa$, adding two of them and contracting the third one  produces the following result:
\begin{align}
     \frac{\sqrt{|\det g|}}{8\pi}\, N_{\kappa\lambda\mu} &= \frac{2}{3}\, g_{\kappa(\lambda}\, \cJ_{\mu)} - g_{\lambda\mu}\, \cJ_{\kappa} +\nonumber \\
     & \quad  +\mnabla_{\nu} \left(\Omega_{\kappa\lambda\mu}^{\ \ \ \ \nu} - 2\Omega_{(\lambda\mu)\kappa}^{\ \ \  \ \ \nu}+g_{\kappa(\lambda}\, \cO_{\mu)}^{\ \ \nu} - \frac 12 \, g_{\lambda\mu}\, \cO_{\kappa}^{\ \nu} \right) + \nonumber \\
     & \quad   +  \frac{\sqrt{|\det g|}}{8\pi}\, \Delta^{\alpha\beta\gamma}_{\kappa\lambda\mu}\, N_{\alpha\beta \gamma}\, ,
\end{align}
what finishes the proof.
\end{proof}
This demonstrates that deriving the non-metricity tensor $N$ involves inverting the operator $(\delta - \Delta)^{\alpha\beta\gamma}_{\kappa\lambda\mu}$, which can be represented as a $40 \times 40$ matrix. While it is indeed possible to invert this operator, doing so is unnecessary in this approach, as only linear terms are considered. However, one component of the non-metricity tensor can be calculated explicitly:
\begin{lemma} \label{lm rel h A}
     The equation~\eqref{1row}:
    \begin{align}
        \nabla_{\kappa} \pi^{\lambda \mu} =  - \frac 23 \, \delta^{(\lambda}_{\kappa} \mathcal{J}^{\mu)}- \nabla_{\nu} \Omega_{\kappa}^{\ \lambda \mu \nu}\, ,
    \end{align}
     for the non-metricity tensor $N^{\kappa}_{\ \lambda\mu}$ implies the following relation:
     \begin{align}
N^{\kappa \sigma}_{\ \ \sigma} =h^{\kappa} + \frac 45\, A^{\kappa} =-\frac{80\pi}{3\sqrt{|\det g|}}\, \mathcal{J}^{\kappa}\, ,
\label{eqNJ}
\end{align}
where $h^{\kappa}$ and $A^{\kappa}$ are traces of $N^{\kappa}_{\ \lambda\mu}$ ~\eqref{rozklad tensora N tot}.
\end{lemma}
\begin{proof}
    The proof relies on taking a contraction of indices $\kappa=\lambda$:
    \begin{align}
        \nabla_{\kappa} \pi^{\kappa\mu} =  - \frac 43 \,  \mathcal{J}^{\mu } - \frac13\, \cJ^{\mu} - \nabla_{\nu} \underbrace{\Omega_{\kappa}^{\ \kappa \mu \nu}}_{=0} = - \frac 53 \,  \mathcal{J}^{\mu }  \, .
    \end{align}
    The left-hand side, up to the formula for $\nabla_{\kappa}\pi^{\lambda\mu}$~\eqref{cdevpi}, is equal:
    \begin{align}
         \nabla_{\kappa} \pi^{\kappa\mu}  = \frac{\sqrt{|\det g|}}{16\pi}\, \left( -N^{\sigma}_{\ \sigma \kappa}\, g^{\kappa\mu} +N^{\kappa}_{\ \kappa\sigma}\, g^{\sigma\mu} +N^{\mu}_{\ \kappa\sigma}\, g^{\kappa\sigma }  \right) = \frac{\sqrt{|\det g|}}{16\pi}\, N^{\mu \sigma}_{\ \ \sigma}\, .
    \end{align}
    Comparing those two equations with the decomposition formula of the non-metricity tensor $N$~\eqref{rozklad tensora N tot} finishes the proof.
\end{proof}

Solving the equation~\eqref{1row}, or equivalently~\eqref{1row2}, with respect to the non-metricity tensor $N$, is split into two cases, when the momentum $\Omega_{\kappa}^{\ \lambda\mu\nu}$ vanishes, or not. Of course, such splitting correlates with the absence (or the presence) of the traceless part $U^{\kappa}_{\ \lambda\mu\nu}$ of curvature $K^{\kappa}_{\ \lambda\mu\nu}$~\eqref{dec kijowski}.

\subsubsection{Absence of the traceless part}
\label{1st feq abs}
In this case, the theory does not depend on the traceless tensor $U^{\kappa}_{\ \lambda\mu\nu}$ (cf. decomposition of the Kijowski tensor~\eqref{dec kijowski}), which is equivalent to taking $\Omega_{\kappa}^{\ \lambda\mu\nu}=0$ -- see~\eqref{rel Omega}. This implies that the operator $\Delta^{\alpha\beta\gamma}_{\kappa\lambda\mu}$~\eqref{Delta6} in \textbf{Theorem~\ref{th non-metricity}} automatically vanishes. Thus, the solution of the non-metricity equation~\eqref{1row2} is exact and can be written explicitly:
\begin{align}
      N_{\kappa\lambda\mu} &=\frac{8\pi}{\sqrt{|\det g|}}\, \left(\frac{2}{3}\, g_{\kappa(\lambda}\, \cJ_{\mu)} - g_{\lambda\mu}\, \cJ_{\kappa} \right)\,  .
      \label{ntens}
\end{align}
The non-metricity tensor $N$ can be decomposed (cf. \textbf{Chapter~\ref{non deco}}), as shown below:
\begin{lemma} 
\label{lem absence}
    The non-metricity tensor $N^{\kappa}_{\ \lambda\mu}$~\eqref{ntens} decomposes as follows:
\begin{align}
    A_{\mu}&= \frac 12\, N^{\kappa}_{\
     \kappa\mu}= \frac{8\pi}{3\sqrt{|\det g|}}\, \cJ_{\mu}\, , \label{A rel J}\\
    A^{\kappa}_{\ \lambda\mu}&=  N^{\kappa}_{\ \lambda\mu} - \frac 45\,  \delta^{\kappa}_{(\lambda}\, A_{\mu)} =  \frac{8\pi}{\sqrt{|\det g|}}\, \left(\frac25\, \delta^{\kappa}_{(\lambda}\, \cJ_{\mu)} - g_{\lambda\mu}\, \cJ^{\kappa} \right)\, , \\
    h^{\kappa}&=A^{\kappa}_{\ \lambda\mu}\, g^{\lambda\mu} =- \frac{18\cdot 8\pi}{5\sqrt{|\det g|}}\, \cJ^{\kappa}\, , \label{h rel J}\\
    \widetilde{A}_{ \kappa\lambda \mu}&={A}_{ \kappa\lambda \mu} + \frac{1}{18}\left(2g_{\kappa(\lambda}\,h_{\mu)} - 5g_{\lambda\mu}\, h_{\kappa} \right) = 0\, .
\end{align}
\end{lemma}
\begin{proof}
The proof consists of verifying the definitions presented in \textbf{Chapter~\ref{non deco}}: formulae (\ref{rozklad tensora N}-\ref{def tildeA}).
\end{proof}

Interestingly, equation~\eqref{A rel J} implies the “Lorenz gauge condition'' for the potential $A_{\mu}$, because the current $\cal J^{\mu}$~\eqref{cal J} is defined as the divergence of a skew-symmetric tensor density $\chi^{\mu\nu}$, hence its divergence necessarily vanishes:
 \begin{align}
     \mnabla_{\mu } A^{\mu} = \frac{8\pi}{3\sqrt{|\det g|}}\, \mnabla_{\mu} \cJ^{\mu} =  \frac{8\pi}{3\sqrt{|\det g|}}\, \partial_{\mu} \cJ^{\mu} = \frac{8\pi}{3\sqrt{|\det g|}}\, \partial_{\mu} \partial_{\nu}\chi^{\mu\nu}=0\, . 
     \label{lor gauge}
 \end{align}

The covariant derivative of the metric tensor is given by:
\begin{align}
    \nabla_{\lambda} g_{ \mu\nu} = -\frac{16\pi}{3\sqrt{|\det g|}}\,  \left(g_{\mu \nu}\, \cJ_\lambda - g_{\lambda \nu}\, \cJ_\mu -
   g_{\lambda\mu }\, \cJ_\nu \right)\, .
\end{align}
This implies that the general affine metric structure is not conformal, meaning that it cannot be expressed for any $\omega_{\lambda}$:
\begin{align}
    \nexists\, \omega_{\lambda}:\  \nabla_{\lambda} g_{ \mu\nu} = \omega_{\lambda}\, g_{\mu\nu}\, .
\end{align}

\subsubsection{Presence of the traceless part} 
If the Lagrangian depends on the traceless part $U^{\kappa}_{\ \lambda\mu\nu}$ and the momentum $\Omega_{\kappa}^{\ \lambda\mu\nu}$ does not vanish~\eqref{rel Omega}, then the solution of the equation~\eqref{1row2} becomes significantly more complicated. Specifically, the operator $(\delta - \Delta)^{\alpha\beta\gamma}_{\kappa\lambda\mu}$ must be inverted. Given that $\Delta$~\eqref{Delta6} is considered a small correction to the identity operator $\delta$, the following equality holds:
\begin{align}
    \left(\delta - \Delta\right)^{-1} =\delta + \sum_{k=1}^{\infty} \Delta^k\, ,
\end{align}
where the right-hand side is known as a \textit{Neumann series}, which is a natural generalisation of a geometric series. However, in this dissertation, the non-metricity tensor $N$ is treated as a small correction or deviation from the metric connection $\mGamma$. Furthermore, the order of the correction is at least two, corresponding to the first-order expansion in the  Neumann series acting on the right-hand side of the non-metricity equation~\eqref{1row2}:
\begin{align}
      N_{\alpha\beta\gamma} &=\frac{8\pi}{\sqrt{|\det g|}}\, \left(\delta_{\alpha}^{\kappa}\, \delta_{\beta}^{\lambda}\, \delta_{\gamma}^{\mu} + \Delta_{\alpha\beta\gamma}^{\kappa\lambda\mu} +\left[o(\Delta^2)\right]_{\alpha\beta\gamma}^{\kappa\lambda\mu} \right)\, \left[\frac{2}{3}\, g_{\kappa(\lambda}\, \cJ_{\mu)} - g_{\lambda\mu}\, \cJ_{\kappa} + \right. \nonumber  \\
    & \quad   \left. +\mnabla_{\nu} \left(\Omega_{\kappa\lambda\mu}^{\ \ \ \nu} - 2\Omega_{(\lambda\mu)\kappa}^{\ \ \ \nu}+g_{\kappa(\lambda}\, \cO_{\mu)}^{\ \ \nu} - \frac 12 \, g_{\lambda\mu}\, \cO_{\kappa}^{\ \nu} \right)\right]\, , 
    \label{neumann}
\end{align}
where $\left[o(\Delta^2)\right]_{\alpha\beta\gamma}^{\kappa\lambda\mu}$ represents terms of order $\Delta^2$ or higher. The reason for this restriction relates to the structure of curvature tensors, which contain linear and quadratic terms of the connection. Additionally, in the Einstein equation, the source of curvature is the stress-energy tensor, which is also quadratic in matter fields. This implies that the quadratic terms are the first non-trivial components in describing the interaction between matter and gravity. However, for describing the dynamics of matter fields and their interactions, it is sufficient to consider only linear terms. This is because, in the stress-energy tensor, matter fields appear quadratically, so a second-order correction would result in fourth-order terms. 

\

The linear part of the non-metricity tensor is denoted as $\No$ and is defined as the zeroth-order term in formula~\eqref{neumann}:
\begin{align}
      \No_{\kappa\lambda\mu} &= \frac{8\pi}{\sqrt{|\det g|}}\, \left[\frac{2}{3}\, g_{\kappa(\lambda}\, \cJ_{\mu)} - g_{\lambda\mu}\, \cJ_{\kappa} + \right. \nonumber  \\
    & \quad  \left. +\mnabla_{\nu} \left(\Omega_{\kappa\lambda\mu}^{\ \ \ \nu} - 2\Omega_{(\lambda\mu)\kappa}^{\ \ \ \ \ \nu}+g_{\kappa(\lambda}\, \cO_{\mu)}^{\ \ \nu} - \frac 12 \, g_{\lambda\mu}\, \cO_{\kappa}^{\ \nu} \right)\right]\, .
    \label{nsaifhy}
\end{align}
As in the previous subsection, the above non-metricity tensor can be decomposed:
\begin{lemma}
    \label{lem presence}
The linearised non-metricity tensor $\No^{\kappa}_{\ \lambda\mu}$~\eqref{nsaifhy} has the following decomposition:
    \begin{align}
\Ao_{\kappa} &=\frac 12\, \No^{\sigma}_{\ \sigma\kappa }= \frac{4\pi}{3\sqrt{|\det g|}}\left(2\, \mathcal{J}_{\kappa} + 3\mnabla_{\nu} \mathcal{O}_{\kappa}^{\ \nu}   \right)\, , 
\label{pot Ao}\\
 \Ao^{\kappa}_{\ \lambda\mu} &=  \No^{\kappa}_{\ \lambda\mu} - \frac 45\,  \delta^{\kappa}_{(\lambda}\, \Ao_{\mu)}  = \frac{8\pi}{ \sqrt{|\det g|}}\left[  \mnabla_{\nu} \left(\Omega_{\ \lambda\mu}^{\kappa \ \ \ \nu}   
 - 2\Omega_{(\lambda\mu)}^{\ \ \ \ \kappa\nu}  \right)   + \right.\nonumber \\
 &  \quad  \left. +\frac{2}{5}\, \delta^{\kappa}_{(\lambda} \, \mathcal{J}_{\mu)} - g_{\lambda\mu}\, \mathcal{J}^{\kappa}  + \frac 12\, \mnabla_{\nu}\left( \frac65 \, \delta^{\kappa}_{(\lambda }\, {\cal O}_{\mu)}^{\ \ \nu} - g_{\lambda \mu}\, {\cal O}^{\kappa  \nu} \right)   \right]\, ,
\label{pot tAo}\\
\ho_{\kappa}&=\Ao^{\kappa}_{\ \lambda\mu}\,g^{\lambda\mu}=-\frac{16\pi}{5\sqrt{|\det g|}}\, \left(9\cJ_{\kappa} + \mnabla_{\nu}\cO_{\kappa}^{\ \nu} \right)\, , \label{pot h}\\
\tAo_{\kappa \lambda\mu} &=\Ao_{ \kappa\lambda \mu} + \frac{1}{18}\left(2g_{\kappa(\lambda}\,\ho_{\mu)} - 5g_{\lambda\mu}\, \ho_{\kappa} \right) = \nonumber \\
&=\frac{8 \pi}{\sqrt{|\det g|}}\, \mnabla_{\nu}\left[\Omega_{\kappa\lambda\mu}^{\ \ \ \ \nu} - 2\Omega_{(\lambda\mu)\kappa}^{\ \ \ \ \ \nu} + \frac{5}{9}\, g_{\kappa(\lambda}\cO_{\mu)}^{\ \ \nu} - \frac{7}{18}\, g_{\lambda\mu}\, \cO_{\kappa}^{ \ \nu}  \right]\, .\label{pot tA}
\end{align}  
\end{lemma}

\begin{proof}
    The proof relies on calculations presented above definitions -- see also \textbf{Chapter~\ref{non deco}}.
\end{proof}

The explicit formula for the quadratic term $\Nt$~\eqref{neumann} is much more complicated, due to the appearance of $\Delta$ operator:

\begin{align}
    \Nt_{\kappa\lambda\mu} &:=  \Delta_{\kappa\lambda\mu}^{\alpha\beta\gamma}\, \No_{\alpha\beta\gamma} = \nonumber \\
    &=\frac{8\pi}{\sqrt{|\det g|}}\, \left[2\cJ^{\alpha}\,\Omega_{\alpha\lambda\mu\kappa} + \cO_{\kappa(\lambda}\, \cJ_{\mu)} - \cJ_{\kappa}\, \cO_{(\lambda\mu)} - \cJ_{(\lambda}\,\cO_{\mu)\kappa} +\right.\nonumber \\
    & \quad   +\cJ^{\alpha}\,\cO_{\alpha(\lambda}\, g_{\mu)\kappa} -\frac12\, \cJ^{\alpha}\,\cO_{\alpha\kappa}\, g_{\lambda\mu} + g_{\kappa(\lambda}\, \cO_{\mu)\alpha}\, \cJ^{\alpha} - \frac 12\, g_{\lambda\mu}\, \cO_{\kappa\alpha}\, \cJ^{\alpha}+\nonumber \\
    & \quad  +4 \left( \mnabla_{\nu}\Omega_{[\alpha\beta](\lambda}^{\ \ \ \ \ \ \nu }\right)\Omega^{\alpha \ \ \ \beta}_{\ \mu)\kappa} - 2\left( \mnabla_{\nu}\Omega_{[\alpha\beta]\kappa}^{  \ \ \ \ \ \nu }\right)\Omega^{\alpha \ \   \beta}_{\ \lambda\mu }+\nonumber \\
    & \quad  + \left( \mnabla_{\nu}\Omega_{[\alpha\beta]\kappa}^{  \ \ \ \ \ \nu }\right)\cO^{\alpha  \beta}\,g_{\lambda \mu } -2 \left( \mnabla_{\nu}\Omega_{[\alpha\beta](\lambda}^{\ \ \ \ \ \ \nu }\right)\cO^{\alpha \beta}\,g_{\mu)\kappa}+\nonumber \\
    & \quad  -2\left( \mnabla_{\nu}\Omega_{\alpha\beta\kappa}^{  \ \ \ \ \nu }\right)\Omega_{(\lambda \ \  \mu)}^{\ \ \alpha\beta} -2\left( \mnabla_{\nu}\Omega_{\alpha\beta(\lambda}^{ \  \ \ \ \ \nu }\right)\Omega_{ \mu) \ \ \kappa}^{\ \ \alpha\beta} + 2 \left( \mnabla_{\nu}\Omega_{\alpha\beta(\lambda|}^{\ \ \ \ \ \nu }\right)\Omega_{\kappa \ \  |\mu)}^{  \ \alpha\beta} + \nonumber \\
    & \quad  +2\left( \mnabla_{\nu}\Omega_{\alpha\beta\gamma}^{  \ \ \ \ \nu }\right)\,g_{\kappa(\lambda}\, \Omega_{\mu) }^{\ \  \alpha\beta\gamma}- \left( \mnabla_{\nu}\Omega_{\alpha\beta\gamma}^{  \ \ \ \ \nu }\right)\,g_{ \lambda\mu}\, \Omega_{\kappa }^{\  \alpha\beta\gamma} + \nonumber \\
    & \quad  +\left( \mnabla_{\nu}\Omega_{\kappa\alpha\beta }^{  \ \ \ \ \nu }\right) \Omega_{(\lambda \ \ \mu) }^{\ \  \alpha\beta } + \left( \mnabla_{\nu}\Omega_{(\lambda|\alpha\beta }^{ \  \ \ \ \ \nu }\right) \Omega_{\mu) \ \ \kappa }^{\ \  \alpha\beta } - \left( \mnabla_{\nu}\Omega_{(\lambda|\alpha\beta }^{\  \ \ \ \ \nu }\right) \Omega_{\kappa \ \ |\mu) }^{\  \alpha\beta }+ \nonumber \\
    & \quad  + \frac 12 \, \left( \mnabla_{\nu}\Omega_{ \alpha\beta\gamma }^{  \ \ \ \ \nu }\right) \, g_{ \lambda\mu}\, \Omega_{\kappa }^{\   \alpha\beta \gamma} -\left( \mnabla_{\nu}\Omega_{ \alpha\beta\gamma }^{  \ \ \ \ \nu }\right) \, g_{\kappa(\lambda}\, \Omega_{\mu) }^{\ \  \alpha\beta \gamma} + \nonumber \\
    & \quad  + \left( \mnabla_{\nu}\Omega_{\kappa\alpha\beta }^{  \ \ \ \ \nu }\right) \Omega^{\alpha \ \ \beta}_{\ \lambda\mu } -2  \left( \mnabla_{\nu}\Omega_{(\lambda|\alpha\beta }^{  \ \ \ \ \ \nu }\right) \Omega^{\alpha \ \ \  \beta}_{\  |\mu)\kappa } + \nonumber\\
    & \quad   +  \left( \mnabla_{\nu}\Omega_{(\lambda|\alpha\beta }^{  \ \ \ \ \ \nu }\right) \cO^{\alpha \beta}\, g_{ \mu)\kappa }-\frac 12\, \left( \mnabla_{\nu}\Omega_{\kappa\alpha\beta }^{ \ \ \ \ \nu }\right) \cO^{\alpha \beta}\, g_{ \lambda\mu  } +\nonumber\\
    & \quad  +\left(\mnabla_{\nu}\cO_{\alpha}^{\ \nu} \right) \Omega^{\alpha}_{\ \lambda\mu\kappa}+\frac 12 \, \left(\mnabla_{\nu}\cO_{\alpha}^{\ \nu} \right) \, g_{\kappa(\lambda}\,\cO_{\mu)}^{\ \ \alpha} - \frac 14\,  \left(\mnabla_{\nu}\cO_{\alpha}^{\ \nu} \right) \, g_{ \lambda\mu}\,\cO_{\kappa}^{\ \alpha}+ \nonumber\\
    & \quad   +\frac 12\,\left(\mnabla_{\nu}\cO_{(\lambda|}^{\ \ 
     \nu} \right) \cO_{\kappa|\mu)} - \frac 12\,\left(\mnabla_{\nu}\cO_{(\lambda}^{\ \ \nu} \right) \cO_{\mu)\kappa} - \frac 12\, \left(\mnabla_{\nu}\cO_{\kappa}^{\ \nu} \right) \cO_{(\lambda\mu)} + \nonumber \\
     & \quad  \left. +\frac 12\,\left(\mnabla_{\nu}\cO_{\alpha}^{\ \nu} \right) \cO^{\alpha}_{\ (\lambda}\,g_{\mu)\kappa}-\frac 14\,  \left(\mnabla_{\nu}\cO_{\alpha}^{\ \nu} \right) \cO^{\alpha}_{\ \kappa}\, g_{\lambda\mu} \right]\, .
     \label{Nt formula}
\end{align}
As mentioned earlier, the above second-order correction will not be considered in the subsequent analysis, but it was derived to illustrate the complexity of the problem.

\subsection{The remaining field equations}

As it was written before, the variation of an affine Lagrangian $\delta\Lag_A$~\eqref{varLA0} generates the relation between momenta and configurations which will play roles of  the field equations -- see (\ref{rel pi}-\ref{rel Omega}). However, the affine Lagrangian will be constructed as a function of the Riemann tensor $R^{\kappa}_{\ \lambda\mu\nu}$~\eqref{rozklad: tensor riemanna}, as a better established object in literature than the Kijowski tensor $K^{\kappa}_{\ \lambda\mu\nu}$, whose traceless part $W^{\kappa}_{\ \lambda\mu\nu}$ has different symmetries than the traceless part of the Kijowski tensor $U^{\kappa}_{\ \lambda\mu\nu}$~\eqref{rel UW}. Therefore, it is necessary to introduce the momentum $\Sigma_{\kappa}^{\ \lambda\mu\nu}$ canonically conjugated to $W^{\kappa}_{\ \lambda\mu\nu}$ and find the relation between  $\Sigma_{\kappa}^{\ \lambda\mu\nu}$ and $\Omega_{\kappa}^{\ \lambda\mu\nu}$. The simplest option relies on the symplectic relation:
\begin{align} 
\Omega_{\kappa}^{\ \lambda\mu\nu} \,\delta U^{\kappa}_{\ \lambda\mu\nu} = -\frac{2}{3}\, \Omega_{\kappa}^{\ \lambda\mu\nu} \,\delta W^{\kappa}_{\ (\lambda\mu)\nu} =-\frac{2}{3}\, \Omega_{\kappa}^{\ \lambda[\mu\nu]} \,\delta W^{\kappa}_{\ \lambda\mu\nu} =\Sigma_{\kappa}^{\ \lambda\mu\nu} \,\delta W^{\kappa}_{\ \lambda\mu\nu}\, .
\end{align}
Then:
\begin{align}
    \Sigma_{\kappa}^{\ \lambda\mu\nu}  :=\frac{\partial \Lag_A}{\partial W^{\kappa}_{\ \lambda\mu\nu}} =-\frac{2}{3}\, \Omega_{\kappa}^{\ \lambda[\mu\nu]}\, ,
    \label{rel Sigma}
\end{align}
or inversely:
\begin{align}
    \Omega_{\kappa}^{\ \lambda\mu\nu}  =-2  \Sigma_{\kappa}^{\ (\lambda\mu)\nu}  \, .
    \label{rel Omega Sigma}
\end{align}
It shows that  momenta $\Omega$ and $\Sigma$ satisfy the dual relation to this one between tensors $U$ and $W$~\eqref{rel UW}.   Of course, as it was for the Riemann tensor and the Kijowski tensor, both of them, $\Omega$ and $\Sigma$, contain the same information and are equivalent. 

\

For future purposes, it is useful to present the metric decomposition formulae for the momenta $\Omega$ and $\Sigma$:
\begin{lemma}
\label{lem dec SigOm}
    For the tensor densities $\Sigma^{\kappa\lambda\mu\nu},\Omega^{\kappa\lambda\mu\nu}$, which satisfy:
    \begin{align}
    \Sigma^{\kappa}_{\ [\lambda\mu\nu]}&=0\, , & \Sigma^{\kappa}_{\ \lambda\mu\nu} &= -\Sigma^{\kappa}_{\ \lambda\nu\mu}\, , & \Sigma^{\kappa}_{\ \kappa\mu\nu}&=0\, , & \Sigma^{\kappa}_{\ \mu\nu\kappa}&=0\, , \\
    \Omega^{\kappa}_{\ (\lambda\mu\nu)}&=0\, , & \Omega^{\kappa}_{\ \lambda\mu\nu} &= \Omega^{\kappa}_{\ \lambda\nu\mu}\, , & \Omega^{\kappa}_{\ \kappa\mu\nu}&=0\, , &\Omega^{\kappa}_{\ \mu\nu\kappa}&=0\, , \\
    \Omega_{\kappa}^{\ \lambda\mu\nu}  &=-2  \Sigma_{\kappa}^{\ (\lambda\mu)\nu}\, ,
\end{align}
    the following equalities hold:
    \begin{align}
      \Sigma^{\kappa\lambda\mu\nu} &:= \Sigma_{\sigma}^{\ \lambda\mu\nu}\, g^{\sigma\kappa}=  \widetilde{\Sigma}^{\kappa\lambda\mu\nu} -\frac{1}{6}\, g^{\kappa\lambda}\, \Sigma^{[\mu\nu]} + \frac{1}{8}\, \left(g^{\kappa\nu}\, \Sigma^{(\lambda\mu)} - g^{\kappa\mu}\, \Sigma^{(\lambda\nu)} \right)+\nonumber\\
        & \quad   + \frac{1}{12}\, \left(g^{\kappa\nu}\, \Sigma^{[\lambda\mu]} - g^{\kappa\mu}\, \Sigma^{[\lambda\nu]} \right) +\frac{3}{8}\, \left(\Sigma^{(\kappa\nu)}\,g^{\lambda\mu} - \Sigma^{(\kappa\mu)}\, g^{ \lambda\nu } \right) +\nonumber \\
        & \quad  + \frac{5}{12}\, \left(\Sigma^{[  \kappa\nu]}\,g^{\lambda\mu} - \Sigma^{[\kappa\mu]}\, g^{ \lambda\nu } \right)\, , \label{Sigma decomposition} \\
      \Omega^{\kappa\lambda\mu\nu}  &:= \Omega_{\sigma}^{\ \lambda\mu\nu}\, g^{\sigma\kappa}=  \widetilde{\Omega}^{\kappa\lambda\mu\nu}  + \frac 18 g^{\kappa\nu}\cO^{(\lambda\mu)} -\frac 18 \left(g^{\kappa\lambda} \cO^{[\mu\nu]} + g^{\kappa\mu}\cO^{[\lambda\nu]} \right) +\nonumber\\
      & \quad   - \frac 1{16} \left(g^{\kappa\lambda} \cO^{(\mu\nu)} + g^{\kappa\mu}\cO^{(\lambda\nu)} \right)  - \frac{5}{24}\left(\cO^{[\kappa\lambda]} g^{\mu\nu} + \cO^{[\kappa\mu]} g^{\lambda\nu} - 2\cO^{[\kappa\nu]} g^{\lambda\mu}   \right) +\nonumber\\
      & \quad  - \frac{3}{16}\left(\cO^{(\kappa\lambda)} g^{\mu\nu} + \cO^{(\kappa\mu)} g^{\lambda\nu} - 2\cO^{(\kappa\nu)} g^{\lambda\mu}   \right)\, ,
        \label{Omega decomposition}
    \end{align} 
    where
    \begin{align}
        \Sigma^{\mu\nu}&:=\Sigma^{\mu \alpha\beta\nu}g_{\mu\nu}\, ,& \cO^{\mu\nu}&:=\Omega^{\mu \alpha\beta\nu}g_{\mu\nu}\, ,
        \label{def tSigma}
    \end{align}
    and $\widetilde{\Sigma}^{\kappa\lambda\mu\nu}, \widetilde{\Omega}^{\kappa\lambda\mu\nu}$ are totally traceless parts of $\Sigma^{\kappa\lambda\mu\nu}, \Omega^{\kappa\lambda\mu\nu}$.
    \end{lemma}
    \begin{proof}
        The equality~\eqref{Sigma decomposition} is analogous to the decomposition formula of tensor $W^{\kappa\lambda\mu\nu}$ -- see ~\eqref{W decomposition} in  \textbf{Lemma~\ref{lemm W dec}}, whereas~\eqref{Omega decomposition} is obtained from the relation between $\Omega$ and $\Sigma$.
    \end{proof}

\section{Construction of affine Lagrangians}
\label{how to construct}
As previously discussed, the affine framework lacks a metric structure. Moreover, the Lagrangian must be a scalar density, making its construction in the affine approach non-trivial. The available geometrical objects are quite limited, primarily consisting of the Kronecker delta $\delta^{\alpha}_{\beta}$, the Levi-Civita symbol $\epsilon^{\alpha\beta\kappa\lambda}$, and the Riemann tensor $R^{\kappa}_{\ \lambda\mu\nu}$. Without a metric, even a fundamental quantity such as the Ricci scalar cannot be defined. Consequently, it is both insightful and instructive to revisit the affine formulation of standard vacuum gravity with a cosmological constant $\Lambda$, which can be explored analytically. This discussion yields several important conclusions and introduces a systematic method for constructing affine Lagrangians.

\subsection{Example: affine description of the \texorpdfstring{$\Lambda$}{Lambda}-vacuum gravity}
\label{Lambda vacuum grav}
The standard metric Lagrangian for vacuum with a cosmological constant $\Lambda$ is given by:
\begin{align}
\Lag_{\Lambda} = \frac{\sqrt{|\det g|}}{16\pi}\, \kolo{K}_{\mu\nu} \,g^{\mu\nu} - \frac{\Lambda\, \sqrt{|\det g|}}{8\pi}\, .
\label{lagLam}
\end{align}
The corresponding field equation is the well-known Einstein  $\Lambda$-vacuum equation:
\begin{align}
    \kolo{G}_{\mu\nu} = -\Lambda\, g_{\mu\nu}\, ,
\end{align}
or, equivalently:
\begin{align}
    \kolo{K}_{\mu\nu} = \Lambda\, g_{\mu\nu}\, .
    \label{eineq0}
\end{align}
In the presence of matter fields, it often happens that the metric Lagrangian does not depend upon the metric covariant derivatives~$\mnabla$ of the matter fields and involves only the metric Ricci tensor $\kolo{K}_{\mu\nu}$. Then, the affine Lagrangian remains equal numerically to the metric Lagrangian, but must be expressed in an affine control mode\footnote{These calculations were the central focus of the author's Master's thesis~\cite{mag} and were later published in~\cite{nonmetricity}.}. The precise transition from the affine to the metric picture is presented in \textbf{Chapter~\ref{chap passage}}.

As will be seen in the sequel, in this case, field equations imply that the general affine connection $\Gamma$ must be the metric connection due to the absence of covariant derivatives, aligning it with the standard Palatini approach. Consequently, using the field equation~\eqref{eineq0}, the metric tensor must be "replaced" by the curvature. This is somewhat analogous to the transition from the Lagrangian to the Hamiltonian formulation, where velocities are replaced by momenta. 

Then the affine Lagrangian equals:  
\begin{align}
    \Lag_{A} =\frac{\sqrt{\left|\det K\right|}}{8\pi \Lambda}\, .
    \label{L0}
\end{align}
The theory described by the above Lagrangian is sometimes called  \textit{Eddington theory} -- cf.~\cite{banados, eddington}.

To verify that the affine Lagrangian above reproduces the same theory as the metric Lagrangian~\eqref{lagLam}, the corresponding field equations will be explicitly derived. The general variational formula for the affine Lagrangian~\eqref{varLA0} must be restricted to its dependence on the symmetric Ricci tensor $K_{\mu\nu}$ and connection $\Gamma^{\kappa}_{\ \lambda\mu}$:  
\begin{align}
 \delta \Lag_A= \left(\nabla_{\nu}\cP_{\kappa}^{\ \lambda\mu\nu} \right)\, \delta\Gamma^{\kappa}_{\ \lambda\mu} + \pi^{\mu\nu}\, \delta K_{\mu\nu}\, ,
\end{align}
where the momentum $\cP$~\eqref{eq: rozklad pedu P} is restricted to 
\begin{align}
    \cP_{\kappa}^{\ \lambda \mu \nu}  =  \pi_{\kappa}^{\ \lambda\mu\nu} = \delta_\kappa^\nu \,
{\pi}^{\lambda\mu} -\delta_{\kappa}^{(\lambda}\, \pi^{\mu)\nu} \, ,
\label{res cP}
\end{align}
since the remaining terms vanish due to the absence of other components of the Riemann curvature tensor $R^{\kappa}_{\ \lambda\mu\nu}$~\eqref{rozklad: tensor riemanna}.  

\ 

The first field equation~\eqref{1fieldeq}:
\begin{align}
    \nabla_{\nu}\cP_{\kappa}^{\ \lambda\mu\nu} =0\, ,
\end{align}
stays the metricity equation -- see~\eqref{1row} and~\eqref{meteq}:
\begin{align}
    \nabla_{\nu}\pi_{\kappa}^{\ \lambda\mu\nu}=0\ \Longrightarrow \ \nabla_{\kappa}\pi^{\lambda\mu}=0\ \Longrightarrow \ \nabla_{\kappa}g_{\lambda\mu}=0\ \Longrightarrow  \ \Gamma^{\kappa}_{ \ \lambda\mu}= \mGamma^{\kappa}_{ \ \lambda\mu}\, .
    \label{ffe}
\end{align}
This ensures that, in this theory, the affine connection $\Gamma$ coincides with the metric connection $\mGamma$.  

\ 

The second field equation~\eqref{rel pi} establishes a relationship between the Ricci tensor $K_{\mu\nu}$ and the momentum~$\pi^{\mu\nu}$:  
\begin{align}
    \pi^{\mu\nu}  = \frac{\partial \Lag_A}{\partial K_{\mu\nu}} = \frac{\sqrt{\left|\det K\right|}}{16\pi \Lambda}\, \left(K^{-1}\right)^{\mu\nu}\, .
\end{align}
Recalling how the momentum $\pi^{\mu\nu}$ relates to the metric tensor $g_{\mu\nu}$~\eqref{pi2}, we obtain the following equalities:
\begin{align}
    \frac{\sqrt{|\det g|}}{16\pi}\, g^{\mu\nu} &= \frac{\sqrt{\left|\det K\right|}}{16\pi \Lambda}\, \left(K^{-1}\right)^{\mu\nu} \, ,\\
    \det K &= \Lambda^4\, \det g\, , \\
    K_{\mu\nu} &= \Lambda\, g_{\mu\nu}\, .
\end{align}  
Since the connection becomes metric due to the first field equation~\eqref{ffe}, the Ricci tensor is also metric. Thus,  
\begin{align}
    \kolo{K}_{\mu\nu} =  \Lambda\, g_{\mu\nu}\, ,
\end{align}  
which is precisely the $\Lambda$-vacuum Einstein equation~\eqref{eineq0}. These simple calculations confirm that the affine Lagrangian~\eqref{L0} and the metric Lagrangian~\eqref{lagLam} describe the same theory, albeit in different formulations.

\subsection{Conclusions}
To sum up, the simplest form of the affine Lagrangian is given by the determinant of the symmetric Ricci curvature~\eqref{L0}, which corresponds to the standard $\Lambda$-vacuum spacetimes. A natural extension involves taking the determinant of the full Ricci tensor\footnote{The interaction between fields via the square root of the determinant~\eqref{detR} is sometimes referred to as “\textit{determinantal coupling}” or “\textit{Born-Infeld coupling}.” Such interactions also appear in other areas of theoretical physics — see \cite{olmo, indie1, string}.}:  
\begin{align}
    \Lag_{A} =\frac{\sqrt{\left|\det R\right|}}{8\pi \Lambda}  =\frac{\sqrt{\left|\det (K+F)\right|}}{8\pi \Lambda}\, ,
    \label{detR}
\end{align}  
which, in the weak-field approximation, leads to the Born-Infeld theory (see \cite{born-infeld, lic, mag, APP,  kij2024}). In this formulation, the cosmological constant $\Lambda$ is also related to the Born-Infeld coupling constant. Naturally, the next-order approximation recovers the Einstein-Maxwell theory.  

The connection between these theories arises from the relation between the skew-symmetric Ricci tensor $F_{\mu\nu}$ and the Faraday 2-form $f_{\mu\nu}$, up to a suitable constant. This example will be analysed in \textbf{Chapter~\ref{affine ricci}}.  

However, there also exist other Lagrangians that lead to the Born-Infeld theory and the Einstein-Maxwell theory (cf. \cite{kij2024}). Unfortunately, these alternative forms somewhat compromise the natural simplicity of the model:  

\begin{align}
    \Lag_{A} =C_1\, \sqrt{\left|\det K\right|}+C_2\, \sqrt{\left|\det (K+F)\right|}\, .
\end{align}  

Unfortunately, this type of affine Lagrangians are restricted to tensors with two indices due to the definition of the determinant. In the approach described above, the dependence on the Riemann curvature was limited to its trace, namely, the Ricci tensor. As a result, the interaction with the traceless part of the curvature was initially excluded. Therefore, a new functional is needed — one that behaves as a scalar density, incorporates the full curvature, and reduces to the above result as a special case when the traceless part vanishes. To achieve this, only the Levi-Civita symbol $\epsilon^{\nu_1 \nu_2 \nu_3 \nu_4}$ can be used to preserve the traceless part of the curvature while ensuring the density character of the Lagrangian. Indeed, the determinant used in the previous examples can be expressed through suitable contractions with Levi-Civita symbols. Specifically, the determinant of the Ricci tensor $R_{\mu\nu}$, represented as a $4 \times 4$ matrix, is defined as follows:
\begin{align}
    \det R &=  \frac{1}{4!}\, R_{\mu_1 \nu_1}\, R_{\mu_2 \nu_2}\, R_{\mu_3 \nu_3}\, R_{\mu_4 \nu_4}\, \epsilon^{\mu_1 \mu_2 \mu_3 \mu_4}\, \epsilon^{\nu_1 \nu_2 \nu_3 \nu_4} =\nonumber  \\
&= \frac{1}{4!}\, R^{\textcolor{red}\alpha}_{\ \mu_1 {\textcolor{red}\alpha} \nu_1}\, R^{\textcolor{blue}\beta}_{\ \mu_2 {\textcolor{blue}\beta} \nu_2}\, R^{\textcolor{magenta}\gamma}_{\ \mu_3 {\textcolor{magenta}\gamma} \nu_3}\, R^{\textcolor{cyan}\delta}_{\ \mu_4 {\textcolor{cyan}\delta} \nu_4}\, \epsilon^{\mu_1 \mu_2 \mu_3 \mu_4}\, \epsilon^{\nu_1 \nu_2 \nu_3 \nu_4}\, .
\label{riccidet}
\end{align}
The crucial issue lies in taking the trace of the Riemann tensor, which is achieved through an “inner” contraction with the Kronecker delta:
\begin{align}
    R^{\textcolor{red}\alpha}_{\ \mu_1 {\textcolor{red}\alpha} \nu_1}:= R^{\textcolor{red}\alpha}_{\ \mu_1 {\textcolor{blue}\beta} \nu_1}\, \delta^{\textcolor{blue}\beta}_{\textcolor{red}\alpha}\, .
\end{align}
Therefore, the simplest way to extend this formula is to replace these “inner” contractions with “mutual” contractions, e.g.:  
\begin{align}
 R^{\textcolor{red}\alpha}_{\ \mu_1 {\textcolor{blue} \beta} \nu_1}\, R^{{\textcolor{blue} \beta}}_{\ \mu_2 {\textcolor{magenta}\gamma} \nu_2}\, R^{{\textcolor{magenta}\gamma}}_{\ \mu_3 {\textcolor{cyan}\delta} \nu_3}\, R^{{\textcolor{cyan}\delta}}_{\ \mu_4 {\textcolor{red}\alpha} \nu_4}\, \epsilon^{\mu_1 \mu_2 \mu_3 \mu_4}\, \epsilon^{\nu_1 \nu_2 \nu_3 \nu_4} \, .
\end{align}
Of course, there are other options, like contracting the upper index with the first lower index: 
\begin{align}
 R^{\textcolor{red}\alpha}_{\  {\textcolor{blue} \beta} \mu_1 \nu_1}\, R^{{\textcolor{blue} \beta}}_{\ {\textcolor{magenta}\gamma} \mu_2 \nu_2}\, R^{{\textcolor{magenta}\gamma}}_{\  {\textcolor{cyan}\delta} \mu_3 \nu_3}\, R^{{\textcolor{cyan}\delta}}_{\ {\textcolor{red}\alpha} \mu_4 \nu_4}\, \epsilon^{\mu_1 \mu_2 \mu_3 \mu_4}\, \epsilon^{\nu_1 \nu_2 \nu_3 \nu_4} \, ,
\end{align}
or, some mixture of those two options. It produces many possible combinations. Happily, not all of them are independent due to the first Bianchi identity~\eqref{1bianchirieman}. The list of all independent possibilities is presented below:
\begin{align}
V_0&= R^{\textcolor{red}\alpha}_{\ \mu_1 {\textcolor{red}\alpha} \nu_1}\, R^{{\textcolor{blue} \beta}}_{\ \mu_2 {\textcolor{blue} \beta} \nu_2}\, R^{{\textcolor{magenta}\gamma}}_{\ \mu_3 {\textcolor{magenta}\gamma} \nu_3}\, R^{{\textcolor{cyan}\delta}}_{\ \mu_4 {\textcolor{cyan}\delta} \nu_4}\, \epsilon^{\mu_1 \mu_2 \mu_3 \mu_4}\, \epsilon^{\nu_1 \nu_2 \nu_3 \nu_4}\, , \label{w0} \\
V_1&=R^{\textcolor{red}\alpha}_{\ \mu_1 {\textcolor{blue} \beta}  \nu_1}\, R^{{\textcolor{blue} \beta}}_{\  \mu_2 {\textcolor{magenta}\gamma} \nu_2}\,R^{{\textcolor{magenta}\gamma}}_{\ \mu_3  {\textcolor{cyan}\delta} \nu_3}\, R^{{\textcolor{cyan}\delta}}_{\  \mu_4 {\textcolor{red}\alpha}  \nu_4}\, \epsilon^{\mu_1 \mu_2 \mu_3 \mu_4}\, \epsilon^{\nu_1 \nu_2 \nu_3 \nu_4} \, , \label{w1} \\
V_2&=R^{\textcolor{red}\alpha}_{\ \mu_1 {\textcolor{blue} \beta}  \nu_1}\, R^{{\textcolor{blue} \beta}}_{\  \mu_2 {\textcolor{magenta}\gamma} \nu_2}\,R^{{\textcolor{magenta}\gamma}}_{\ \mu_3  {\textcolor{cyan}\delta} \nu_3}\,R^{{\textcolor{cyan}\delta}}_{\ {\textcolor{red}\alpha} \mu_4  \nu_4}\, \epsilon^{\mu_1 \mu_2 \mu_3 \mu_4}\, \epsilon^{\nu_1 \nu_2 \nu_3 \nu_4} \, , \label{w2} \\
V_3&=R^{\textcolor{red}\alpha}_{\ \mu_1 {\textcolor{blue} \beta}  \nu_1}\, R^{{\textcolor{blue} \beta}}_{\  \mu_2 {\textcolor{magenta}\gamma} \nu_2}\,R^{{\textcolor{magenta}\gamma}}_{\ {\textcolor{cyan}\delta} \mu_3  \nu_3}\, R^{{\textcolor{cyan}\delta}}_{\ {\textcolor{red}\alpha} \mu_4  \nu_4}\, \epsilon^{\mu_1 \mu_2 \mu_3 \mu_4}\, \epsilon^{\nu_1 \nu_2 \nu_3 \nu_4}\, ,\label{w3} \\
V_4&=R^{\textcolor{red}\alpha}_{\ \mu_1 {\textcolor{blue} \beta}  \nu_1}\,R^{{\textcolor{blue} \beta}}_{\ {\textcolor{magenta}\gamma} \mu_2  \nu_2}\,R^{{\textcolor{magenta}\gamma}}_{\ {\textcolor{cyan}\delta} \mu_3  \nu_3}\, R^{{\textcolor{cyan}\delta}}_{\ {\textcolor{red}\alpha} \mu_4  \nu_4}\, \epsilon^{\mu_1 \mu_2 \mu_3 \mu_4}\, \epsilon^{\nu_1 \nu_2 \nu_3 \nu_4} \, ,\label{w4}  \\
V_5&=R^{\textcolor{red}\alpha}_{\ {\textcolor{blue} \beta} \mu_1  \nu_1}\, R^{{\textcolor{blue} \beta}}_{\ {\textcolor{magenta}\gamma} \mu_2  \nu_2}\,R^{{\textcolor{magenta}\gamma}}_{\ {\textcolor{cyan}\delta} \mu_3  \nu_3}\, R^{{\textcolor{cyan}\delta}}_{\ {\textcolor{red}\alpha} \mu_4  \nu_4}\, \epsilon^{\mu_1 \mu_2 \mu_3 \mu_4}\, \epsilon^{\nu_1 \nu_2 \nu_3 \nu_4}\, .
\ \label{w5}
\end{align}
Obviously, the example $V_0$ is equivalent to the determinant of the Ricci tensor~\eqref{riccidet}. The other option is based on the following construction:
\begin{align}
C^{\alpha \kappa}_{\beta \lambda}&=R^{\alpha}_{\ \beta \mu_1 \mu_2}\, R^{\kappa}_{\ \lambda \mu_3 \mu_4}\, \epsilon^{\mu_1 \mu_2 \mu_3 \mu_4}\, , \nonumber \\
V_6 &=C^{\alpha \kappa}_{\beta \lambda}\, C^{\beta \lambda}_{\alpha \kappa}\, .
\label{w6}
\end{align}
As above, here also could be taken into account the contraction with respect to the first or second lower index. The last concept uses the deformation of the trace as a contraction of the Riemann tensor with the Kronecker delta. It is  perturbed by an  extra traceless matrix $\Delta$:
\begin{align}
    R_{\mu\nu} = R^{\kappa}_{\ \lambda\mu\nu}\, \delta^{\mu}_{\kappa} \longmapsto  R^{\kappa}_{\ \lambda\mu\nu}\, \left(\delta^{\mu}_{\kappa} + \Delta^{\mu}_{\kappa}\right)  = R_{\mu\nu} + R^{\kappa}_{\ \lambda\mu\nu}\,   \Delta^{\mu}_{\kappa}\, .
    \label{pertRDel}
\end{align}
However, this approach allows for many possible forms of the matrix $\Delta^{\mu}_{\kappa}$ — it has fifteen independent components, and there is no natural object that could be represented by this matrix. Of course, it could be used in a phenomenological approach, where elements of $\Delta^{\mu}_{\kappa}$ would be “fitted'' to some effective models and theories. Therefore, this idea will not be studied further.

\ 

All previous propositions are based on the definition of the determinant of four-dimensional matrices, which is a weighted scalar density of weight “2'' (due to the double appearance of the Levi-Civita symbol \( \epsilon^{\kappa\lambda\mu\nu} \)) and a fourth-order polynomial in the curvature. Therefore, the affine Lagrangian, being the square root of such determinants, is effectively a standard scalar density (with weight “1'') and a quadratic polynomial in curvature.  

Surprisingly, it is possible to construct an object that \textit{a priori} satisfies these properties (a scalar density of weight “1'' and a quadratic expression in curvature):  
\begin{align}
    R^{ \alpha}_{\ { \beta}\kappa\lambda} \, R^{ \beta}_{\ { \alpha}\mu \nu}\,\epsilon^{\kappa\lambda\mu\nu}\, .
\end{align}
However, an affine Lagrangian obtained in this way does not generate any dynamics, since the variation of the above quantity results in a pure divergence. Therefore, no field equations arise from such a theory. This example is presented in the theorem below:
\begin{theorem}
\label{div}
    For the Riemann tensor $R^{\kappa}_{\ \lambda\mu\nu}$~\eqref{def: tensor riemanna} of the symmetric affine connection~$\Gamma^{\kappa}_{\ \lambda\mu}$ the following equality holds:
\begin{align}
    \delta \left(  R^{ \alpha}_{\ { \beta}\kappa\lambda} \, R^{ \beta}_{\ { \alpha}\mu \nu}\,\epsilon^{\kappa\lambda\mu\nu}\right) =-4 \partial_{\nu} \left( R^{ \alpha}_{\  \beta \kappa\lambda}\,\epsilon^{\kappa\lambda\mu\nu} \, \delta \Gamma^{\beta}_{\ \alpha\mu}\right)\, ,
\end{align}
where $\delta$ denotes the “variation'' operator and $\epsilon^{\kappa\lambda\mu\nu}$ is the Levi-Civita symbol.
\end{theorem}
\begin{proof}
    The proof is practically straightforward:
    \begin{align}
         \delta \left(  R^{ \alpha}_{\ { \beta}\kappa\lambda} \, R^{ \beta}_{\ { \alpha}\mu \nu}\,\epsilon^{\kappa\lambda\mu\nu}\right) =&  \,2 R^{ \alpha}_{\ { \beta}\kappa\lambda} \,\epsilon^{\kappa\lambda\mu\nu}\, \delta   R^{ \beta}_{\ { \alpha}\mu \nu} = -4  R^{ \alpha}_{\ { \beta}\kappa\lambda} \,\epsilon^{\kappa\lambda\mu\nu}\, \delta   \left(\Gamma^{ \beta}_{\ { \alpha}\mu, \nu} + \Gamma^{\sigma}_{\ \alpha\mu}\, \Gamma^{\beta}_{\ \nu\sigma}\right) = \nonumber \\
         =& -4 \partial_{\nu}\left(R^{ \alpha}_{\ { \beta}\kappa\lambda} \,\epsilon^{\kappa\lambda\mu\nu}\, \delta   \Gamma^{ \beta}_{\ { \alpha}\mu} \right) + 4R^{ \alpha}_{\ { \beta}\kappa\lambda,\nu} \,\epsilon^{\kappa\lambda\mu\nu}\, \delta   \Gamma^{ \beta}_{\ { \alpha}\mu} + \nonumber\\
         &-4R^{ \alpha}_{\ { \beta}\kappa\lambda} \,\epsilon^{\kappa\lambda\mu\nu}\,    \left(  \Gamma^{\sigma}_{\ \alpha\mu}\, \delta \Gamma^{\beta}_{\ \nu\sigma} + \Gamma^{\beta}_{\ \nu\sigma}\, \delta\Gamma^{\sigma}_{\ \alpha\mu}\right) \, .
    \end{align}
    All parts proportional to $\delta \Gamma$ will be proportional to the covariant derivative of the Riemann tensor, which is the following:
    \begin{align}
        \nabla_{\nu} R^{\alpha}_{\ \beta\kappa\lambda} = R^{\alpha}_{\ \beta\kappa\lambda,\nu} + \Gamma^{\alpha}_{\ \nu\sigma}\, R^{\sigma}_{\ \beta\kappa\lambda}  - \Gamma^{\sigma}_{\ \nu\beta}\, R^{\alpha}_{\ \sigma\kappa\lambda}  - \Gamma^{\sigma}_{\ \nu\kappa}\, R^{\alpha}_{\ \beta \sigma \lambda}  - \Gamma^{\sigma}_{\ \nu\lambda}\, R^{\alpha}_{\ \beta\kappa\sigma} \, .
    \end{align}
    Contraction of the above equality with the Levi-Civita symbol $\epsilon^{\kappa\lambda\mu\nu}$ produces zero on the left-hand side, due to the second Bianchi identity:
    \begin{align}
        \nabla_{[\nu|} R^{\alpha}_{\ \beta|\kappa\lambda]} = 0\, ,
    \end{align}
    whereas on the right-hand side, some terms vanish due to the symmetry of the connection:
     \begin{align}
        \underbrace{\nabla_{\nu} R^{\alpha}_{\ \beta\kappa\lambda} \, \epsilon^{\kappa\lambda\mu\nu}}_{=0 \, \text{(II Bianchi identity)}}
       &= R^{\alpha}_{\ \beta\kappa\lambda,\nu}\, \epsilon^{\kappa\lambda\mu\nu} + \left( \Gamma^{\alpha}_{\ \nu\sigma}\, R^{\sigma}_{\ \beta\kappa\lambda}  - \Gamma^{\sigma}_{\ \nu\beta}\, R^{\alpha}_{\ \sigma\kappa\lambda}\right)  \, \epsilon^{\kappa\lambda\mu\nu}+\nonumber\\
       & \quad  + \underbrace{\left(-\Gamma^{\sigma}_{\ \nu\kappa}\, R^{\alpha}_{\ \beta \sigma \lambda}  - \Gamma^{\sigma}_{\ \nu\lambda}\, R^{\alpha}_{\ \beta\kappa\sigma}\right)\,   \epsilon^{\kappa\lambda\mu\nu}}_{=0\, \text{(symmetry)}} \, .
    \end{align}
    Contraction with the $\delta \Gamma^{\beta}_{\ \alpha\mu}$ implies:
    \begin{align}
        0=&\, R^{\alpha}_{\ \beta\kappa\lambda,\nu}\, \epsilon^{\kappa\lambda\mu\nu}\, \delta \Gamma^{\beta}_{\ \alpha\mu} + \left( \Gamma^{\alpha}_{\ \nu\sigma}\, R^{\sigma}_{\ \beta\kappa\lambda}  - \Gamma^{\sigma}_{\ \nu\beta}\, R^{\alpha}_{\ \sigma\kappa\lambda}\right)\,  \epsilon^{\kappa\lambda\mu\nu}\, \delta \Gamma^{\beta}_{\ \alpha\mu} = \nonumber\\
        =&R^{\alpha}_{\ \beta\kappa\lambda,\nu}\, \epsilon^{\kappa\lambda\mu\nu}\, \delta \Gamma^{\beta}_{\ \alpha\mu} -R^{ \alpha}_{\ { \beta}\kappa\lambda} \,\epsilon^{\kappa\lambda\mu\nu}\,    \left(  \Gamma^{\sigma}_{\ \alpha\mu}\, \delta \Gamma^{\beta}_{\ \nu\sigma} + \Gamma^{\beta}_{\ \nu\sigma}\, \delta\Gamma^{\sigma}_{\ \alpha\mu}\right) \, .
    \end{align}
   As a result of the above equality, the initial statement is proven.
\end{proof}

\section{The scheme of deriving the approximated affine Lagrangians and field equations}
\sectionmark{The scheme}
\label{scheme}

In the affine theory of full Riemann curvature, the initial Lagrangian takes the form of the square root of four contracted Riemann tensors with two Levi-Civita symbols — see~(\ref{w0}--\ref{w6}). Moreover, it is assumed that the skew-symmetric part of the Ricci tensor $F_{\mu\nu}$ and the traceless part $W^{\kappa}_{\ \lambda\mu\nu}$ are small perturbations (with the same weight) of the symmetric part $K_{\mu\nu}$. Hence, the initial Lagrangian $\Lag_F$ will be restricted to at most quadratic terms in $F_{\mu\nu}$ and $W^{\kappa}_{\ \lambda\mu\nu}$ and will be denoted as $\Lag_A$ — the appropriate affine Lagrangian used in further analysis.  

All Lagrangians examined in this dissertation share the same structure. Therefore, presenting a general procedure for obtaining the quadratic approximation and deriving the field equations significantly simplifies the content of the following chapters.

\subsection{Structure of Lagrangians} 
Schematically, the  affine Lagrangian $\Lag_A$ depends on the  full Riemann tensor~$R$ in the following way:  
\begin{align}  
 \Lag_A = \alpha\, \sqrt{|RRRR|} \, ,
\label{lagF}
\end{align}  
where $\alpha$ is a constant coefficient chosen to match the specific variant (\ref{w0}--\ref{w6}) represented by four contracted Riemann tensors $RRRR$ with two Levi-Civita symbols. However, to simplify the notation, those Levi-Civita symbols are omitted. 

The Riemann tensor decomposes into three irreducible parts -- see~\eqref{rozklad: tensor riemanna} -- the symmetric part of the Ricci tensor $K$, the skew-symmetric part of the Ricci tensor $F$, and the algebraically traceless part of the Riemann tensor $W$. Therefore, the expression $RRRR$ decomposes in the following way:
\begin{align}
    RRRR =& KKKK + KKKF + KKKW+\nonumber\\
    &+KKFF+KKFW+KKWW + o\left(F,W\right)\, , \label{RRRR}
\end{align}
where $o\left(F,W\right)$ represents higher-order terms in $F$ and $W$. Thus, collective terms (e.g., $KKFF$) represent all possible contractions involving the indicated number of components (e.g., two $K$ and two $F$, with two Levi-Civita symbols omitted). While this notation may initially seem unconventional, it helps to manage the complexity of precise calculations and prevents getting lost in a maze of symbols. 

Hopefully, some terms vanish \textit{ab initio}:
\begin{equation}
KKKF =0\, ,
\end{equation}
due to the theorem presented below:
\begin{theorem}
   There does not exist a non-zero contraction involving three symmetric tensors $K_{\mu\nu}$, one skew-symmetric tensor $F_{\mu\nu}$, and two Levi-Civita symbols $\epsilon^{\alpha\beta\gamma\delta}$.
\end{theorem}
\begin{proof}
    Since \( K_{\mu\nu} \) is symmetric and the Levi-Civita symbol \( \epsilon^{\alpha\beta\gamma\delta} \) is totally skew-symmetric, the only non-trivial possibility is to contract each \( K_{\mu\nu} \) with both Levi-Civita symbols. Consequently, each Levi-Civita symbol retains one free index, requiring the skew-symmetric tensor \( F_{\mu\nu} \) to be contracted in the same manner as \( K_{\mu\nu} \). Hence:
    \begin{align}
    KKKF \approx K_{\mu_1\nu_1}\, K_{\mu_2\nu_2}\, K_{\mu_3\nu_3}\, F_{\mu_4\nu_4}\, \epsilon^{\mu_1\mu_2\mu_3\mu_4}\, \epsilon^{\nu_1\nu_2\nu_3\nu_4}\, .
    \end{align}
   Of course, permuting the sequence of \( K \) and \( F \) changes only the sign of the final result, due to the total skew-symmetry of \( \epsilon^{\mu_1\mu_2\mu_3\mu_4} \). Now, using the skew-symmetry of the tensor \( F_{\mu\nu} \) and the symmetry of the tensors \( K_{\mu\nu} \), it can be shown that the above quantity is equal to itself with an opposite sign, and therefore, it vanishes.

\end{proof}
Moreover, inspired by the above theorem, it is easy to show that the only non-trivial contraction of four symmetric tensors \( K_{\mu\nu} \) with two Levi-Civita symbols \( \epsilon^{\alpha\beta\gamma\delta} \) is proportional to the determinant of \( K \) — see formula~\eqref{riccidet}:
\begin{align}
KKKK = \sigma \gamma^2 \, \det K = \sigma \sigma_K \gamma^2 \, \left|\det K\right| =  \sigma \sigma_K\left|KKKK\right|\, ,
\label{reldetK}
\end{align}
where \( \sigma_K := \sgn \left( \det K\right) \), while the constants \( \sigma = \pm 1 \) and \( \gamma^2 \) depend on the chosen variant (\ref{w0}--\ref{w6}) and will be determined later. 

This term forms the core of the theory, as it represents the square of the affine Lagrangian of the \( \Lambda \)-vacuum~\eqref{L0}, around which the quadratic extension in \( F \) and \( W \) will be developed. Furthermore, it affects the modulus of $RRRR$~\eqref{RRRR}, which appears in the formula for the affine Lagrangian~\eqref{lagF}:
\begin{align}
   \left| RRRR\right| =&  \left|KKKK + o(F,W) \right| = \left|KKKK\right|\left| 1 + \frac{o(F,W)}{KKKK} \right| =  \nonumber \\    
   =&\left|KKKK\right|\left( 1 + \frac{o(F,W)}{KKKK} \right) = \sigma \sigma_K \left(KKKK + o(F,W) \right)\, ,  
    \label{znakLA}
\end{align}
where \( o(F,W) \) denotes all terms containing \( F \) and \( W \) (see~\eqref{RRRR}) and is assumed to be a small perturbation of the \( KKKK \) term. Finally, the affine Lagrangian \( \Lag_A \), which will be used in the sequel, is defined as the restriction of the  $|RRRR|$~\eqref{RRRR} to at most quadratic terms in \( F \) and \( W \) under the square root:
\begin{align}
\Lag_A&= \alpha \sqrt{|KKKK  +KKKW+ KKFF + KKFW + KKWW  |} = \nonumber \\
&= \alpha\sqrt{ \sigma \sigma_K \left(KKKK  +KKKW+ KKFF + KKFW + KKWW \right) }\, . 
\label{lagmodel}
\end{align}

\subsection{Einstein equation}

The first step contains the derivation of the Einstein equation, which, in the affine picture, is given by the relation between the momentum $\pi$,  and the symmetric Ricci tensor~$K$~\eqref{rel pi}:
\begin{align}
\pi&=\frac{\partial \Lag_A}{\partial K }=\frac{1}{2\Lag_A}\, \frac{\partial \Lag_A^2}{\partial K } =\frac{\alpha^2 \sigma \sigma_K}{2\Lag_A}\, \left( KKK + KKW + KFF+KFW + KWW \right)\, .  
\label{pi1}
\end{align}
The objects inside the bracket are matrices (with two upper indices), obtained via taking the derivative with respect to the Ricci tensor $K$ (which has two lower indices),~e.g.:
\begin{align}
    KKK&:=\frac{\partial }{\partial K} (KKKK)\, , & KKW&:=\frac{\partial }{\partial K}(KKKW)\, .
    \label{derK}
\end{align}
Moreover, these quantities are tensor densities of weight “2'' due to the presence of two Levi-Civita symbols, which are not explicitly written for practical reasons. As it was for scalar densities, these symbols denote all objects of given structure. Thus, the equation  ~\eqref{pi1} is a complicated, non-linear tensorial equation for $K$. However, it could be written in the following way: 
\begin{align}
    \frac{2\, \Lag_A}{\alpha ^2\sigma \sigma_K } \, \pi= KKK + KKW + KFF+KFW+ KWW\, .
    \label{row1}
\end{align}
Now, the weight “2'' tensor densities on the right-hand side are in coherence with the left-hand side, where the standard tensor density $\pi$ is multiplied by the scalar density $\Lag_A$, which produces the “double'' density character. 

\subsection{Perturbative method}
To solve the Einstein equation~\eqref{row1}, it is necessary to introduce the perturbation of the  symmetric curvature $K$ as follows:
\begin{equation}
    K= \Kz+\Ko+\Kt\, ,
    \label{rachperp}
\end{equation}
where $\Kz$ corresponds to the non-perturbed solution ($\Lambda$-vacuum), whereas  $\Ko$ and $\Kt$ contain the first- and second-order corrections depending on  $F$ and $W$. It should be implemented into the equation~\eqref{row1}, but it will make it too long and absolutely messy. Therefore, the perturbation method is divided into a few steps: at first, the affine Lagrangian $\Lag_A$ which appears on the left-hand side of the Einstein equation is analysed. Next, the right-hand side which contains the double tensor densities will be expanded. Finally, the solution will be obtained by deriving the non-perturbed solution, and then, the next-order corrections related to the “power'' of $F$ and $W$. 

\

Applying the expansion of tensor $K$~\eqref{rachperp} into the formula of the affine Lagrangian~\eqref{lagmodel} and limiting it to at most quadratic terms produces:
\begin{align}
    \Lag_A&=\alpha  \left|\Kz \Kz \Kz \Kz +\Kz \Kz \Kz \Ko + \Kz \Kz \Kz W +  \Kz \Kz \Ko \Ko + \Kz \Kz \Kz \Kt +\right. \nonumber \\
    & \quad  \left.+\Kz\Kz\Ko W  + \Kz \Kz F F + \Kz \Kz F W + \Kz \Kz W W \right|^{1/2}  \, . 
    \label{LagAperp}
\end{align}
As before, the expansion is done around the dominating zeroth order term $\Kz\Kz\Kz\Kz$, which is proportional to the   $\det \Kz$ -- see formula~\eqref{reldetK}:
\begin{align}
    \Lag_A&=\alpha\gamma\, \sqrt{\left|\det \Kz\right|}\, \left[1 +\frac1{2\sigma \gamma^2\, \det \Kz } \left(\Kz \Kz \Kz \Ko + \Kz \Kz \Kz W  + \Kz \Kz \Ko \Ko +  \right. \right. \nonumber  \\
    &  \quad   \left.\left. + \Kz \Kz \Kz \Kt   + \Kz \Kz \Ko W + \Kz \Kz F F +\Kz \Kz F W + \Kz \Kz W W + o (F,W) \right)\right]\, ,
    \label{rozwiniecie}
\end{align}
where $ o (F,W)$ denotes the higher-order terms in $F$ and $W$. The terms appearing on the right-hand side of the equation ~\eqref{row1} are expanded (up to the quadratic terms) as follows:
\begin{equation}
\begin{split}
    KKK&= \Kz \Kz \Kz + \Kz \Kz \Ko +\Kz \Ko \Ko+ \Kz \Kz \Kt\, ,\\ 
    KKW&= \Kz \Kz W +\Kz \Ko W \, ,\\
    KFF&=  \Kz F F \, ,\\
    KFW&=\Kz F W \, ,\\
    KWW &=  \Kz W W\, .
\end{split}
    \label{rozwiniecierhs}
    \end{equation}
Formulae~\eqref{reldetK},~\eqref{derK} imply that double tensor density $\Kz \Kz \Kz$ is a derivative of the determinant $\det \Kz$ with respect to the tensor  $\Kz$. Hence: 
\begin{equation}
    \Kz \Kz \Kz := \frac{\partial}{\partial \Kz} (\Kz \Kz \Kz \Kz)  =\sigma \gamma^2\, \,   \frac{\partial \det \Kz}{\partial \Kz}   = \sigma \gamma^2\, \left(\det \Kz\right)\, \Kz^{-1}\,.
\end{equation}
Therefore, the zeroth order solution (non-perturbative) is obtained as a solution of the equation~\eqref{row1}, with the applied extensions of the affine Lagrangian~\eqref{rozwiniecie}, and double tensor densities from~\eqref{rozwiniecierhs}, when $F$, $W$, $\Ko$, $\Kt$ vanish. Thus, the equation~\eqref{row1} is limited to the following form:
\begin{align}
   \frac{ 2\gamma}{\alpha \sigma \sigma_K}\, \sqrt{\left|\det \Kz\right|}\, \pi  = \sigma \gamma^2\, \left(\det \Kz\right)\, \Kz^{-1}\,.
   \label{shgsg}
\end{align}
Using the relation between the momentum $\pi$ and the metric tensor~\eqref{pi2}:
\begin{align}
    \pi = \frac{\sqrt{|\det g|}}{16\pmb{\pi}}\, g^{-1}\, ,
    \label{pidada}
\end{align}
the equation~\eqref{shgsg} is equivalent to:
\begin{equation}
    \Kz = \frac{ 1}{8\pmb{\pi}\alpha \gamma}\, g\, .
    \label{eq00}
\end{equation} 
The above expression, as it was noticed before, is the Einstein $\Lambda$-vacuum equation~\eqref{eineq0}, since
\begin{align}
\alpha := \frac{1}{8\pmb{\pi}\Lambda  \gamma} \, ,
    \label{alpha const} 
\end{align} 
and then:
\begin{align}
    \Kz =\Lambda\, g\,.
    \label{eq0}
\end{align}
Unfortunately, the above “symbolical'' notation makes a little confusion, because the quantity $\pi$ on the left-hand side of~\eqref{pidada} is a tensor density, whereas  $\pmb{\pi}$ on the right-hand side of (\ref{pidada}-\ref{alpha const})  is the mathematical constant $\pmb{\pi} \approx3.14$. Such an embarrassment will not appear in the exact examples, where all tensorial quantities will have indices.

\ 

The explicit derivation of the first-order perturbation is impossible in the above schematic manner. Although, observations presented below will be very useful in the sequel. The Einstein equation~\eqref{row1}  extended to the first-order perturbations has the following form:
\begin{align}
      \frac{  \Kz \Kz \Kz \Ko + \Kz \Kz \Kz W}{\alpha\gamma\, \sqrt{\left|\det \Kz\right|} }\,  \, \pi  =  \Kz \Kz \Ko + \Kz \Kz W\, .
      \label{eqein1}
\end{align}
Using the already obtained non-perturbed solution, the components of the above equation can be written as follows:
\begin{align}
        \frac{1}{\alpha\gamma}&= 8\pmb{\pi}\Lambda\, , & \Kz \Kz \Kz \Ko &= \Lambda^3\,  ggg\Ko\, , \label{fo1}\\
        \pi &= \frac{\sqrt{|\det g|}}{16 \pmb{\pi}}\, g^{-1}\, , & \Kz \Kz \Kz W &= \Lambda^3\,  gggW\, , \label{fo2}\\
        \Kz&= \Lambda\, g\, , & \Kz \Kz \Ko&= \Lambda^2\, gg\Ko\, , \label{fo3}\\
        \sqrt{\left|\det \Kz\right|} &= \Lambda^2\, \sqrt{|\det g|}\, , & \Kz \Kz W&= \Lambda^2\, ggW\, .\label{fo4}
\end{align}
Then, the equation~\eqref{eqein1} is equivalent to:
\begin{align}
     \frac12\, \left(ggg \Ko + ggg W \right)\, g^{-1} =  g g \Ko + g g W\, .
     \label{eq1}
\end{align}
As it was mentioned before, such an equation cannot be solved without knowledge of numerical coefficients and an explicit structure of all symbolical terms. However, in standard general relativity formulation, the Einstein equation prescribes  a relation between the Ricci tensor (or Einstein tensor)  and the stress-energy tensor, which contains only quadratic terms of fields. This heuristic observation suggests that the linear correction $\Ko$ should vanish.  

\ 

In the analogue to the first-order perturbation presented above, the second-order expansion is given by:
\begin{align}
& \frac{ \Kz \Kz \Ko \Ko + \Kz \Kz \Kz \Kt + \Kz \Kz \Ko W + \Kz \Kz FF+ \Kz \Kz FW+ \Kz \Kz WW}{\alpha\gamma\, \sqrt{\left|\det \Kz\right|} }\,  \pi= \nonumber  \\
&\qquad \qquad =\Kz \Ko \Ko + \Kz \Kz \Kt + \Kz  \Ko W + \Kz FF+   \Kz FW+   \Kz WW \, , 
\end{align}
which, after implementing the zeroth order solution $\Kz=\Lambda\,  g$~\eqref{eq0}, simplifies to the following form:
\begin{align}
  &\frac{  1}{2}\, \left(g g \Ko \Ko + \Lambda\,  g g g \Kt +gg\Ko W +  g g  FF+ g g F W+ g g  WW \right)\, g^{-1}= \nonumber  \\
&\qquad \qquad =  g \Ko \Ko + \Lambda\, g g \Kt  +g\Ko W+  g FF+     g  FW+     g WW \, .
\label{eq2}   
\end{align}

Then, the Einstein equation is given by the formula~\eqref{rachperp}, but now $\Kz=\Lambda\, g $, $\Ko$, $\Kt$ are functions of $g$, $F$ and $W$:
\begin{align}
    K =  \Lambda\, g + \Ko (g,F,W) +  \Kt (g,F,W)\, .   
    \label{eineqgen}
\end{align}

However, the above formula is not precisely the well-known form of the Einstein equation, due to the general, non-metric Ricci tensor $K$ on the left-hand side. Thus, the decomposition of the general symmetric Ricci tensor $K$ for the metric part $\kolo{K}$ and non-metricity terms has to be implied -- see~\eqref{K}, then:
\begin{align}
    \kolo{K} =\Lambda\, g + \Ko  +  \Kt  - Q  \, ,
    \label{scheinmet}
\end{align}
where $Q$ denotes the difference between the general symmetric Ricci tensor $K$ and the metric one $\kolo{K}$:
\begin{align}
    Q_{\mu\nu} := K_{\mu\nu} - \kolo{K}_{\mu\nu} =  \mnabla_{\kappa} N^{\kappa}_{\ \mu \nu} - \mnabla_{(\mu}  N^{\kappa}_{\ \nu) \kappa}  + N^{\sigma}_{\ \mu \nu}\, N^{\kappa}_{\ \kappa  \sigma} - N^{\sigma}_{\ \kappa \mu }\, N^{\kappa}_{\ \nu \sigma}\, .
    \label{def Q}
\end{align}

\subsection{Effective cosmological parameter}

After the derivation of the Einstein equation, there could be discussed a proposition of \textit{the effective cosmological parameter} $\Lambda_{\rm eff}$, which could be defined by  the metric Ricci scalar $\kolo{R}$~\eqref{def: tensor Z}, obtained from the Einstein equation~\eqref{scheinmet}:
\begin{align}
   \Lambda_{\rm eff} :=\frac 14\,  \kolo{R}=\frac 12 \kolo{K}_{\mu\nu} g^{\mu\nu} = \Lambda + \frac14\,\left(\Ko_{\mu\nu} + \Kt_{\mu\nu} - Q_{\mu\nu} \right)g^{\mu\nu} \, .
      \label{def lam eff}
\end{align}
Of course, such an object can be defined for any theory, whenever the stress-energy tensor (interpreted as a right-hand side of the Einstein equation) has a non-vanishing trace.  

\subsection{Remaining field equations}
\label{chap fieldeq scheme}
The derivation of the remaining  field equations for the skew-symmetric part of the Ricci tensor $F$~\eqref{rel chi} and the traceless part of the Riemann tensor $W$~\eqref{rel Sigma} is as follows: 
\begin{align}
\chi&=\frac{\partial\Lag_A}{\partial F} =\frac{\alpha^2\sigma \sigma_K }{2\Lag_A}\, \left(\Kz\Kz F + \Kz\Kz W \right)  \, , \\
    \Sigma&=\frac{\partial\Lag_A}{\partial W} = \frac{\alpha^2\sigma \sigma_K}{2\Lag_A}\, \left(\Kz\Kz\Kz + \Kz\Kz\Ko + \Kz\Kz F + \Kz\Kz W \right) \, , 
\end{align}
where the tensors written on the right-hand sides of the above equations are obtained via derivatives of the affine Lagrangian $\Lag_A$~\eqref{rozwiniecie} with respect to the tensors $F$ and $W$. Due to the perturbative method, the above equations have to be linear in   $F$ and $W$, because at least quadratic terms appear in the approximated Einstein equation. Although the first term on the right-hand side of the formula for $\Sigma$ contains only zeroth order terms $\Kz\Kz\Kz$, therefore, to have a linear formula, the inverse of the Lagrangian $\Lag_A$ has to be expanded to the first order corrections. For other linear terms: $\Kz\Kz\Ko,\, \Kz\Kz F,\,  \Kz\Kz W$, the inverse of the affine Lagrangian  is simply restricted to the inverse of $\Lambda$-vacuum affine Lagrangian~\eqref{L0}, where  $\Kz=\Lambda\, g$ via equation~\eqref{eq0}. Therefore, above equations take the following form:
\begin{align}
    \chi&= \frac{\sigma \sigma_g }{ 16 
 \pi \Lambda \gamma^2 \sqrt{|\det g|}}\, \left( ggF + ggW\right) \, ,
    \label{chi}\\
    \Sigma&= \frac{ \sigma \sigma_g }{16\pi \gamma^2  \sqrt{\left|\det g\right|}}\,  \left[1 -\frac1{2\sigma\gamma^2\, \Lambda  \, \det g } \left(  g g g \Ko +   g g g W   \right) \right]\,  ggg + \nonumber \\
    & \quad  +\frac{\sigma \sigma_g  }{ 16 
 \pi \Lambda \gamma^2 \sqrt{|\det g|}}\,  \left( gg \Ko + gg F + gg W \right)\, , 
 \label{sigma} 
\end{align} 
where $\sigma_g$ denotes the sign of the determinant of the metric tensor $g$. This arises from the following consideration:
\begin{align}
    \sigma_K = \text{sgn}(\det K) \approx \text{sgn}(\det \Kz) = \text{sgn}(\Lambda^4 \det g) = \text{sgn}(\det g) := \sigma_g\, .
    \label{def sigg}
\end{align}

Of course, the right-hand sides of these equations must satisfy the same properties as the momenta on the left-hand sides: $\chi$ is skew-symmetric, whereas $\Sigma$ is algebraically traceless, satisfies the first Bianchi identity, and is skew-symmetric in its last upper indices. This is a typical situation in which derivatives with respect to a tensorial object (possessing certain symmetries) must be taken. Importantly, in formulae~\eqref{chi} and~\eqref{sigma}, the terms $ggF$, $ggW$, etc., correspond to the derivatives of the scalar densities $ggFF$, $ggFW$, etc., with respect to the tensors $F$ and $W$, respectively.

Some deviations from these rules may occur — in particular, in the case of the Bianchi identity, which need not be satisfied by the momentum $\Sigma$ if the Lagrangian depends not on the full tensor $W$ but only on certain parts of it. This is precisely the situation in Variant~$V_1$, discussed in the sequel, which makes the derivation of $\Sigma$ non-trivial.

To address this issue, it is necessary to recall the decomposition formula~\eqref{W decomposition} for the tensor $W$, in order to derive the correct momenta associated with its irreducible parts:
\begin{align}
    \Sigma^{\kappa \lambda\mu\nu}\,\delta W_{\kappa\lambda\mu\nu} = \frac 56\, \Sigma^{\mu\nu}\, \delta W_{[\mu\nu]} + \frac 34\, \Sigma^{\mu\nu}\, \delta W_{(\mu\nu)} + \widetilde{\Sigma}^{\kappa \lambda\mu\nu}\,\delta \widetilde{W}_{\kappa\lambda\mu\nu}\, ,
\end{align}
where $\Sigma^{\mu\nu}=\Sigma^{\mu\kappa\lambda\nu}g_{\kappa\lambda}$ and $\widetilde{\Sigma}^{\kappa \lambda\mu\nu}$ denote the totally traceless part of the  momentum $\Sigma^{\kappa \lambda\mu\nu}$ (see \textbf{Lemma~\ref{lem dec SigOm}} formula \eqref{Sigma decomposition}), which also decomposes in the following way:
\begin{align}
    \widetilde{\Sigma}^{\kappa \lambda\mu\nu}\,\delta \widetilde{W}_{\kappa\lambda\mu\nu} = \widetilde{\Sigma}^{\kappa \lambda\mu\nu}\,\delta \widetilde{W}_{(\kappa\lambda)\mu\nu} + \widetilde{\Sigma}^{\kappa \lambda\mu\nu}\,\delta \widetilde{W}_{[\kappa\lambda]\mu\nu}\, .
\end{align}
Collecting all terms leads to:
\begin{align}
    \Sigma^{\kappa \lambda\mu\nu}\,\delta W_{\kappa\lambda\mu\nu} = \frac 56\, \Sigma^{\mu\nu}\, \delta W_{[\mu\nu]} + \frac 34\, \Sigma^{\mu\nu}\, \delta W_{(\mu\nu)} + \widetilde{\Sigma}^{\kappa \lambda\mu\nu}\,\delta \widetilde{W}_{(\kappa\lambda)\mu\nu} + \widetilde{\Sigma}^{\kappa \lambda\mu\nu}\,\delta \widetilde{W}_{[\kappa\lambda]\mu\nu}\, ,
    \label{sympl Sigma}
\end{align}
which induces the following field equations:
\begin{align}
   \frac 56\, \Sigma^{[\mu\nu]}&=\frac{\partial \Lag_A}{\partial W_{[\mu\nu]}}\, , 
   \label{sigma feq1}\\
    \frac 34\, \Sigma^{(\mu\nu)}&=\frac{\partial \Lag_A}{\partial W_{(\mu\nu)}}\, , 
    \label{sigma feq2}\\
    \widetilde{\Sigma}^{(\kappa \lambda)\mu\nu} &=\frac{\partial \Lag_A}{\partial\widetilde{W}_{(\kappa\lambda)\mu\nu}}\, , 
    \label{sigma feq3}\\
    \widetilde{\Sigma}^{[\kappa \lambda]\mu\nu}&= \frac{\partial \Lag_A}{\partial \widetilde{W}_{[\kappa\lambda]\mu\nu}}\, .
    \label{sigma feq4}
\end{align}

\subsection{Potential equations }
The previous subsection presents the field equations expressed as relations between the momenta ($\chi$ and $\Omega$) and the fields ($F$ and $W$), and in this dissertation, these relations are restricted to the linear case. These fields, as components of the general affine curvature, contain both metric and non-metric parts — see formulae~\eqref{rozklad pelny F},~\eqref{rozklad pelny W}. Importantly, by definition, the tensor $W$ also contains terms quadratic in the non-metricity tensor, which will be neglected under the linearisation assumption.  On the other hand, the non-metricity tensor $N$ is defined in terms of the covariant derivatives of the momenta — see \textbf{Theorem~\ref{th non-metricity}}. Therefore, these equations can be used to derive a system of second-order differential equations describing the non-metricity. Schematically, this takes the following form:
\begin{align}
    \chi &= \frac{\partial \Lag_A}{\partial F} = F + W\, , & F &= \mnabla N\, ,& \Sigma \simeq \Omega\, ,\quad \mnabla\chi = \cJ\, ,\\
    \Sigma &= \frac{\partial \Lag_A}{\partial W} = F + W\, , & W &= \kolo{W} + \mnabla N\, ,& N = \cJ + \mnabla \Omega\, .
\end{align}
Taking covariant divergences and using the definition of non-metricity gives the following structure:
\begin{align}
    N = \cJ + \mnabla \Omega = \mnabla F + \mnabla W = \mnabla\mnabla N +\mnabla \kolo{W} \, .
\end{align}
This type of equation will be referred to as a \textit{potential equation}, due to the fact that the non-metricity tensor $N^{\kappa}_{\ \lambda\mu}$ decomposes into components involving $A_{\mu}$ and $A^{\kappa}_{\ \lambda\mu}$ (see \textbf{Chapter~\ref{non deco}}), which act as potentials for the tensors $F_{\mu\nu}$ and $W^{\kappa}_{\ \lambda\mu\nu}$, respectively. Naturally, the right-hand sides of the field equations depend on the specific variant chosen, whereas the decomposition of the non-metricity tensor $N^{\kappa}_{\ \lambda\mu}$ does not. Therefore, some general remarks are presented below.

\

Firstly, the non-metricity tensor $N^{\kappa}_{\ \lambda\mu}$ decomposes into irreducible parts $A_{\mu}$, $h_{\mu}$, $\widetilde{A}_{\kappa\lambda\mu}$ -- see \textbf{Chapter~\ref{non deco}}. All of those components, up to the first field equation, depend on the covariant derivatives of the momenta $\chi$ and $\Omega$ -- see \textbf{Chapter~\ref{first field eq}}. To derive the potential equations, those relations have to be inverted. To do that, there is a needed extra decomposition formula:
\begin{lemma}
\label{lem dec Omega}
    The covariant derivative $\mnabla_{\nu}\Omega_{\kappa}^{\ \lambda\mu\nu}$ decomposes as follows
    \begin{align}
        \mnabla_{\nu}\Omega_{\kappa}^{\ \lambda\mu\nu} = \mnabla_{\nu}\left[\mathfrak{O}_{\kappa}^{\ \lambda\mu\nu} - \frac{1}{18}\left( \delta_{\kappa}^{\lambda}\, \cO^{\mu\nu} +\delta_{\kappa}^{\mu}\,\cO^{\lambda\nu} - 5g^{\lambda\mu}\, \cO_{\kappa}^{\ \nu} \right) \right]\, ,
        \label{mathfrakO}
    \end{align}
    where $\cO_{\kappa}^{\ \nu}:=\Omega_{\kappa}^{\ \lambda\mu\nu} g_{\lambda\mu}$, and $ \mnabla_{\nu}\mathfrak{O}_{\kappa}^{\ \lambda\mu\nu}$ is a  totally  traceless part of $\mnabla_{\nu}\Omega_{\kappa}^{\ \lambda\mu\nu}$:
    \begin{align}
        \mnabla_{\nu}\mathfrak{O}_{\kappa}^{\ \lambda\mu\nu} \, g_{\lambda\mu}&=0\, , &\mnabla_{\nu}\mathfrak{O}_{\kappa}^{\ \kappa\mu\nu}=0\, .
    \end{align}    
\end{lemma}
\begin{proof}
    The proof is based on the analogous equality for the algebraically traceless part of the  non-metricity tensor $A^{\kappa}_{\ \lambda\mu}$~\eqref{def tildeA}.
\end{proof}
Furthermore, combining the decomposition of $\Omega_{\kappa}^{\ \lambda\mu\nu}$~\eqref{Omega decomposition} from \textbf{Lemma~\ref{lem dec SigOm}} with the above decomposition of $\mnabla_{\nu}\Omega_{\kappa}^{\ \lambda\mu\nu}$~\eqref{mathfrakO} gives the relation between $\mathfrak{O}_{\kappa}^{\ \lambda\mu\nu}$ and $\widetilde{\Omega}^{\kappa  \lambda\mu\nu}$:
\begin{align}
     \widetilde{\Omega}^{\kappa\lambda\mu\nu}&= \mathfrak{O}^{\kappa  \lambda\mu\nu}  +\frac 5{72} \left(g^{\kappa\lambda} \cO^{[\mu\nu]} + g^{\kappa\mu}\cO^{[\lambda\nu]} \right) + \frac 1{144} \left(g^{\kappa\lambda} \cO^{(\mu\nu)} + g^{\kappa\mu}\cO^{(\lambda\nu)} \right) + \nonumber\\
      & - \frac 18 g^{\kappa\nu}\cO^{(\lambda\mu)}  + \frac{5}{24}\left(\cO^{[\kappa\lambda]} g^{\mu\nu} + \cO^{[\kappa\mu]} g^{\lambda\nu}   \right) - \frac{5}{36}\cO^{[\kappa\nu]} g^{\lambda\mu}  +\nonumber\\
      &+ \frac{3}{16}\left(\cO^{(\kappa\lambda)} g^{\mu\nu} + \cO^{(\kappa\mu)} g^{\lambda\nu}   \right) - \frac{7}{72}\cO^{(\kappa\nu)} g^{\lambda\mu} \, .
\end{align}
Due to the above decomposition of $\mnabla_{\nu}\Omega_{\kappa}^{\ \lambda\mu\nu}$~\eqref{mathfrakO}, \textbf{Lemma~\ref{lem presence}} can be reformulated as follows:

\begin{lemma}[\textbf{Reformulation of  Lemma~\ref{lem presence}}]
    \label{lem presence2}
The linearised non-metricity tensor $N^{\kappa}_{\ \lambda\mu}$~\eqref{nsaifhy} has the following form:
\begin{align}
      N_{\kappa\lambda\mu} &= \frac{8\pi}{\sqrt{|\det g|}}\, \left[\mnabla_{\nu} \left(\mathfrak{O}_{\kappa\lambda\mu}^{\ \ \ \nu} - 2\mathfrak{O}_{(\lambda\mu)\kappa}^{\ \ \ \ \ \nu} +\frac 49  g_{\kappa(\lambda}\, \cO_{\mu)}^{\ \ \nu} - \frac 19   \, g_{\lambda\mu}\, \cO_{\kappa}^{\ \nu} \right) + \right. \nonumber  \\
    & \quad  \left. +\frac{2}{3}\, g_{\kappa(\lambda}\, \cJ_{\mu)} - g_{\lambda\mu}\, \cJ_{\kappa}\right]\, .
\end{align}
and decomposes as follows:
    \begin{align}
A_{\kappa} &=\frac 12\, N^{\sigma}_{\ \sigma\kappa }= \frac{4\pi}{3\sqrt{|\det g|}}\left(2\, \mathcal{J}_{\kappa} + 3\mnabla_{\nu} \mathcal{O}_{\kappa}^{\ \nu}   \right)\, , 
\label{pot Ao2}\\
 A^{\kappa}_{\ \lambda\mu} &=  N^{\kappa}_{\ \lambda\mu} - \frac 45\,  \delta^{\kappa}_{(\lambda}\, A_{\mu)}  = \nonumber \\
 &=\frac{8\pi}{ \sqrt{|\det g|}} \mnabla_{\nu} \left[ \mathfrak{O}_{\ \lambda\mu}^{\kappa \ \ \ \nu}   
 - 2\mathfrak{O}_{(\lambda\mu)}^{\ \ \ \ \kappa\nu}+  \frac{1}{45} \left(\delta^{\kappa}_{\lambda}\cO_{\mu}^{\ \nu} +\delta^{\kappa}_{\mu}\cO_{\lambda}^{\ \nu}  - 5g_{\lambda\mu}\cO^{\kappa\nu}\right)    \right] + \nonumber \\
 & \quad   + \frac{8\pi}{5\sqrt{|\det g|}}\left( \delta^{\kappa}_{\lambda}\cJ_{\mu}  +\delta^{\kappa}_{\mu}\cJ_{\lambda}  - 5g_{\lambda\mu}\cJ^{\kappa } \right)   \, ,
\label{pot tAo2}\\
h_{\kappa}&=A^{\kappa}_{\ \lambda\mu}\,g^{\lambda\mu}=-\frac{16\pi}{5\sqrt{|\det g|}}\, \left(9\cJ_{\kappa} + \mnabla_{\nu}\cO_{\kappa}^{\ \nu} \right)\, , \label{pot h2}\\
\tA_{\kappa \lambda\mu} &=A_{ \kappa\lambda \mu} + \frac{1}{18}\left(2g_{\kappa(\lambda}\,h_{\mu)} - 5g_{\lambda\mu}\, h_{\kappa} \right) =\frac{8 \pi}{\sqrt{|\det g|}}\, \mnabla_{\nu}\left[ \mathfrak{O}_{\kappa\lambda\mu}^{\ \ \ \ \nu} -2 \mathfrak{O}_{(\lambda\mu)\kappa}^{\ \ \ \ \ \nu}  \right]\, .\label{pot tA2}
\end{align}  
\end{lemma}

Therefore, the  inversion of   \textbf{Lemma~\ref{lem presence2}} is as follows:
\begin{lemma}[\textbf{Inverse Lemma~\ref{lem presence2}}]
\label{inv lem presence2}
    The following equalities hold:
    \begin{align}
        \cJ^{\kappa}&= -\frac{3\sqrt{|\det g|}}{400\pi}\left(5h^{\kappa} + 4 A^{\kappa} \right)\, , 
        \label{dec cJ}\\
        \mnabla_{\nu}\cO^{\kappa\nu}&= \frac{\sqrt{|\det g|}}{200\pi}\left(5h^{\kappa} + 54 A^{\kappa}\right)\, , 
        \label{dec cO}\\
        \mnabla_{\nu}\mathfrak{O}_{\kappa}^{\ \lambda\mu\nu}&= -\frac{\sqrt{|\det g|}}{8\pi}\, \widetilde{A}^{(\lambda\mu)}_{\ \ \ \ \kappa}\, . 
        \label{dec tilOm} 
    \end{align}
\end{lemma}
\begin{proof}
The most challenging part is finding an explicit inverse formula for the divergence $\mnabla_{\nu} \mathfrak{O}_{\kappa}^{\ \lambda\mu\nu}$. However, the following symmetrisation resolves this issue:
\begin{align}
    \tA_{(\kappa \lambda)\mu}  
    &= \frac{8 \pi}{\sqrt{|\det g|}}\, \mnabla_{\nu}\left[
        \mathfrak{O}_{(\kappa\lambda)\mu}^{\ \ \ \ \ \nu} 
        - \mathfrak{O}_{(\lambda|\mu|\kappa)}^{\ \ \ \ \ \ \nu} 
        - \mathfrak{O}_{\mu (\lambda \kappa)}^{\ \ \ \ \ \nu}
    \right] = -\frac{8 \pi}{\sqrt{|\det g|}}\, \mnabla_{\nu} \mathfrak{O}_{\mu\kappa\lambda}^{\ \ \ \nu}\, ,
\end{align}
which completes the proof.
\end{proof}

Moreover, up to the formulae~\eqref{dec cJ}, the following relation is satisfied:
\begin{align}
0=\mnabla_{\kappa} \cJ^{\kappa}= -\frac{3\sqrt{|\det g|}}{400\pi}\left(5\mnabla
_{\kappa}h^{\kappa} + 4 \mnabla_{\kappa}A^{\kappa} \right)\, , 
\end{align} 
due to the definition of $\cJ^{\kappa}$ as a covariant divergence of the skew-symmetric tensor density $\chi^{\mu\nu}$ -- see~\eqref{cal J} and~\eqref{lor gauge}. It automatically gives a relation between divergences of those two vector potentials:
\begin{align}
 \mnabla
_{\kappa}h^{\kappa} =- \frac 45\, \mnabla_{\kappa}A^{\kappa} \, .
\end{align}

On the right-hand sides of the field equations~\eqref{chi} and~\eqref{sigma}, the tensors $F_{\mu\nu}$~\eqref{rozklad pelny F} and $W^{\kappa}_{\ \lambda\mu\nu}$~\eqref{rozklad pelny W} appear linearly. Importantly, the tensor $W$ has to  be restricted to its linear part due to the assumptions made. Therefore, it will be necessary to compute their covariant derivatives. To do this, the following  lemma is needed:
\begin{lemma}
\label{lem comut}
    For any vector field $X^{\mu}$, tensor $Y^{\kappa\lambda\mu}$ and metric connection $\mGamma^{\kappa}_{\ \lambda\mu}$, the following equalities hold:
    \begin{align}        
        \left(\mnabla_{\alpha}  \mnabla_{\beta} - \mnabla_{\beta}\mnabla_{\alpha}  \right)  X^{\mu}  & = X^{\sigma}  \kolo{R\, }^{\mu}_{ \ \sigma \alpha \beta} \, ,\\
        \left(\mnabla_{\alpha}  \mnabla_{\beta} - \mnabla_{\beta}\mnabla_{\alpha}  \right) Y^{\kappa\lambda\mu} &= Y^{\sigma\lambda\mu} \kolo{R\, }^{\kappa}_{ \ \sigma \alpha \beta} + Y^{\kappa\sigma\mu} \kolo{R\, }^{\lambda}_{ \ \sigma \alpha \beta} + Y^{\kappa\lambda\sigma} \kolo{R\, }^{\mu}_{ \ \sigma \alpha \beta}\, ,
    \end{align} 
\end{lemma}
where $\kolo{R\, }^{\mu}_{ \ \sigma \alpha \beta}$ denotes the metric Riemann tensor.
\begin{proof}
    The proof relies on the definition of covariant derivatives of tensors and the formula for the metric  Riemann tensor -- cf. formula~\eqref{def: tensor riemanna} and take $\Gamma=\mGamma$.
\end{proof}
Finally, the covariant divergence of $F_{\mu\nu}$~\eqref{rozklad pelny F} is the following:
\begin{align}
    \mnabla_{\nu}F^{\mu\nu} &=\mnabla_{\nu}\left( \mnabla^{\mu}A^{\nu} - \mnabla^{\nu}A^{\mu} \right)=\mnabla^{\mu}\mnabla_{\nu}A^{\nu} + A^{\sigma}\,\kolo{K}_{\sigma}^{\ \mu} - \mBox A^{\mu}\, ,
    \label{mnabla F}
\end{align} 
where $\mBox$ denotes the metric D'Alembert operator.

All above formulae were \textit{a priori} linear, whereas the definition of the tensor $W$~\eqref{rozklad pelny W} also contains the quadratic terms in potential $A$. Therefore, it will be useful to introduce the following  tensor:
\begin{align}
    C_{\kappa  \lambda\mu\nu}&:=  \mnabla_{\mu}{A}_{\kappa  \nu \lambda } - \mnabla_{\nu}{A}_{\kappa  \mu \lambda } + \frac 13\,\left( g_{\kappa \nu} \mnabla_{\sigma} {A}^{\sigma}_{\ \mu\lambda} - g_{\kappa \mu} \mnabla_{\sigma} {A}^{\sigma}_{\ \nu\lambda}\right)     \, ,
    \label{def tensor C}
\end{align}
Then, the traceless Riemann tensor  $W$~\eqref{rozklad pelny W} equals:
\begin{align}
    W_{\kappa \lambda\mu\nu}= \kolo{W}_{\kappa  \lambda\mu\nu} + C_{\kappa  \lambda\mu\nu} + o\left(A^2\right)\, .
    \label{lin W}
\end{align}
and terms $o\left(A^2\right)$ will be neglected in this subsection. Furthermore, the potential ${A}^{\kappa}_{\ \mu\lambda}$ also admits a decomposition~\eqref{def tildeA}, so the tensor $C_{\kappa \lambda\mu\nu}$~\eqref{def tensor C} takes the form:
\begin{align}
    C_{\kappa  \lambda\mu\nu}&=  \mnabla_{\mu}{\tA}_{\kappa  \nu \lambda } - \mnabla_{\nu}{\tA}_{\kappa  \mu \lambda } + \frac 13\,\left( g_{\kappa \nu} \mnabla_{\sigma} {\tA}^{\sigma}_{\ \mu\lambda} - g_{\kappa \mu} \mnabla_{\sigma} {\tA}^{\sigma}_{\ \nu\lambda}\right) -\frac{1}{27}\left(3g_{\kappa\lambda}\mnabla_{[\mu}h_{\nu]}+\right. \nonumber\\
    & \quad  \left. -  4 g_{\kappa[\mu}\mnabla_{\nu]}h_{\lambda} -g_{\kappa[\mu|}\mnabla_{\lambda}h_{|\nu]}+15g_{\lambda[\mu}\mnabla_{\nu]}h_{\kappa}+5g_{\kappa[\mu}g_{\nu]\lambda}\mnabla_{\sigma}h^{\sigma} \right)\, .
    \label{decomp tensor C}
\end{align}
It is also useful to compute the only non-vanishing metric trace $C_{\kappa\nu}$ — cf. the tensor $W^{\kappa}_{\ \nu}$~\eqref{trace W}:
\begin{align}
    C_{\kappa\nu}&:=C_{\kappa  \lambda\mu\nu}g^{\lambda\mu} =  \mnabla_{\sigma}{\tA}_{\kappa  \nu }^{\ \ \sigma}  -  \frac 13   \mnabla_{\sigma} {\tA}^{\sigma}_{\ \kappa\nu }  -\frac{1}{27}\left( \mnabla_{\kappa}h_{\nu}+ 19 \mnabla_{\nu}h_{\kappa} - 5g_{\kappa\nu }\mnabla_{\sigma}h^{\sigma} \right)\, .
    \label{decomp tensor C2}
\end{align}

For the  tensor  $W^{\kappa}_{\ \lambda\mu\nu}$~\eqref{lin W} there are two kinds of possible divergences due to the first Bianchi identity~\eqref{1bianchirieman}, and skew-symmetry in the last two indices. Therefore:
\begin{align}
    \mnabla^{\nu}W_{\kappa  \lambda\mu\nu}&=  \mnabla^{\nu}\kolo{W}_{\kappa  \lambda\mu\nu} + \mnabla^{\nu}C_{\kappa  \lambda\mu\nu} =\nonumber\\
    &=\mnabla^{\nu}\kolo{W}_{\kappa  \lambda\mu\nu} + \mnabla_{\nu}\mnabla_{\mu}\tA_{\kappa\lambda}^{\ \ \nu} -\mBox \tA_{\kappa\lambda\mu} +\frac13 \left( \mnabla_{\kappa}\mnabla_{\sigma}\tA^{\sigma}_{\ \lambda\mu} - g_{\kappa\mu}\mnabla_{\nu}\mnabla_{\sigma}\tA^{\sigma\nu}_{\ \ \lambda}\right)+\nonumber \\
    &-\frac{1}{54}\left[\mnabla_{\kappa}\mnabla_{\lambda}h_{\mu}+4\mnabla_{\kappa}\mnabla_{\mu}h_{\lambda}-15\mnabla_{\lambda}\mnabla_{\mu}h_{\kappa} +5g_{\lambda\mu}\left( 3\mBox h_{\kappa} -\mnabla_{\kappa}\mnabla_{\sigma}h^{\sigma}\right) +\right.\nonumber\\
    &\left.- 3g_{\kappa\lambda}\left(\mBox h_{\mu} -\mnabla_{\mu}\mnabla_{\sigma}h^{\sigma} - h^{\sigma}\kolo{K}_{\sigma\mu}\right)  -g_{\kappa\mu}\left(4\mBox h_{\lambda} -4\mnabla_{\lambda}\mnabla_{\sigma}h^{\sigma} + h^{\sigma}\kolo{K}_{\sigma\lambda}\right) \right] \, , \label{mnabla W1}
\\
\mnabla^{\kappa}W_{\kappa  \lambda\mu\nu}&=\mnabla^{\kappa} \kolo{W}_{\kappa  \lambda\mu\nu} + \mnabla^{\kappa} C_{\kappa  \lambda\mu\nu} =\nonumber \\
&= \mnabla^{\kappa} \kolo{W}_{\kappa  \lambda\mu\nu} +  \mnabla^{\kappa}\mnabla_{\mu}{\tA}_{\kappa  \nu \lambda } - \mnabla^{\kappa}\mnabla_{\nu}{\tA}_{\kappa  \mu \lambda } + \frac 13\,\left( \mnabla_{ \nu} \mnabla_{\sigma} {\tA}^{\sigma}_{\ \mu\lambda} - \mnabla_{ \mu} \mnabla_{\sigma} {\tA}^{\sigma}_{\ \nu\lambda}\right) + \nonumber\\
& \quad  -\frac{1}{27}\left(3\mnabla_{ \lambda}\mnabla_{[\mu}h_{\nu]}  -\mnabla_{[\mu|}\mnabla_{\lambda}h_{|\nu]} +10g_{\lambda[\mu}\mnabla_{\nu]}\mnabla_{\sigma}h^{\sigma}- 2h^{\sigma}\kolo{R}_{\lambda\sigma\mu\nu} +\right. \nonumber \\
& \quad  \left.  +15g_{\lambda[\mu} \kolo{K}_{\nu]\sigma}h^{\sigma} \right)\, .\label{mnabla W2}
\end{align}
The divergences of the  metric tensor $\kolo{W}_{\kappa  \lambda\mu\nu}$~\eqref{rel W mathfrakw} are the following:
\begin{align}
 \mnabla^{\nu} \kolo{W}_{\kappa  \lambda\mu\nu} &=  \mnabla^{\nu}\kolo{\mathfrak{w} }_{\kappa\lambda\mu\nu} - \frac{1}{12}g_{\mu\lambda}\mnabla_{\kappa}\kolo{R} - \frac{1}{12}g_{\mu\kappa}\mnabla_{\lambda}\kolo{R}+\frac 12 \mnabla_{\lambda}\kolo{K}_{\kappa\mu} - \frac 16 \mnabla_{\kappa}\kolo{K}_{\lambda\mu} =\nonumber\\
 &=\mnabla_{\lambda}\kolo{K}_{\kappa\mu} - \frac 23 \mnabla_{\kappa} \kolo{K}_{\lambda \mu } - \frac 16 g_{\mu \kappa} \mnabla_{\lambda} \kolo{R} \, , \\
    \mnabla^{\kappa} \kolo{W}_{\kappa  \lambda\mu\nu} &=  \mnabla^{\kappa}\kolo{\mathfrak{w} }_{\kappa\lambda\mu\nu} + \frac{1}{12}\left(g_{\lambda \nu}\mnabla_{\mu}\kolo{R} -g_{\lambda \mu}\mnabla_{\nu}\kolo{R} \right) +\frac{1}{6}\left( \mnabla_{\mu}\kolo{K}_{\nu\lambda} - \mnabla_{\nu}\kolo{K}_{\mu\lambda}  \right) = \nonumber \\
    &= \frac 23 \left( \mnabla_{\mu}\kolo{K}_{\nu\lambda} - \mnabla_{\nu}\kolo{K}_{\mu\lambda}  \right) \, ,
\end{align}
and  the  contracted second Bianchi identity~\eqref{cont 2Bian} with the equality~\eqref{cont 2Bian 2} were used.

Also, divergences of the metric trace $W_{\kappa\nu}$  will be useful:
\begin{align}
    \mnabla^{\nu}W_{\kappa  \nu}&=\mnabla^{\nu}W_{\kappa  \lambda\mu\nu}  \,g^{\lambda \mu} = \mnabla^{\nu}\kolo{W}_{\kappa  \nu} + \mnabla_{\mu}\mnabla_{\nu}{\tA}_{\kappa}^{\ \mu \nu}  - \frac 13\,  \mnabla_{\nu}\mnabla_{\mu} {\tA}^{\mu \nu}_{\ \ \kappa}  + \nonumber\\
    & \quad  -\frac{1}{27}\left[19 \mBox h_{\kappa }  - 4  \mnabla_{\kappa}\mnabla_{\mu}h^{\mu}   +     h^{\sigma }\kolo{K}_{\sigma  \kappa }  \right] \, , \\
     \mnabla^{\kappa}W_{\kappa \nu}&= \mnabla^{\kappa}W_{\kappa \lambda\mu\nu} \, g^{\lambda\mu} =  \mnabla^{\kappa}\kolo{W}_{\kappa  \nu} +  \mnabla_{\kappa} \mnabla_{\lambda}{\tA}^{\kappa\lambda}_{\ \  \nu }   - \frac 13  \mnabla_{\lambda} \mnabla_{\kappa} {\tA}^{\kappa\lambda }_{\ \ \nu}  +\nonumber\\
     & \quad  -\frac{1}{27}\left[ \mBox h_{\nu} +14 \mnabla_{\nu} \mnabla_{\lambda}h^{\lambda}  +  19  h^{\sigma}\!\kolo{K}_{\sigma\nu}     \right]\, .  
\end{align}
Interestingly, the divergences of the tensor $\kolo{W}_{\kappa\nu}$~\eqref{def metric W2} simplify upon the contracted second Bianchi identity~\eqref{cont 2Bian}:
\begin{align}
    \mnabla^{\kappa}\kolo{W}_{\kappa\nu} = \mnabla^{\kappa}\kolo{W}_{\nu\kappa} =  -\frac 13 \mnabla_{\kappa}\kolo{R}\, .  
    \label{div metric W2}
\end{align}
Therefore:
\begin{align}
    \mnabla^{\nu}W_{\kappa  \nu}&=   -\frac 13 \mnabla_{\kappa}\kolo{R} + \mnabla_{\mu}\mnabla_{\nu}{\tA}_{\kappa}^{\ \mu \nu}  - \frac 13\,  \mnabla_{\nu}\mnabla_{\mu} {\tA}^{\mu \nu}_{\ \ \kappa}  + \nonumber\\
    & \quad  -\frac{1}{27}\left[19 \mBox h_{\kappa }  - 4  \mnabla_{\kappa}\mnabla_{\mu}h^{\mu}   +     h^{\sigma }\kolo{K}_{\sigma  \kappa }  \right] \, , \label{mnabla W3}\\
     \mnabla^{\kappa}W_{\kappa \nu}&=   -\frac 13 \mnabla_{\nu}\kolo{R} +  \mnabla_{\kappa} \mnabla_{\lambda}{\tA}^{\kappa\lambda}_{\ \  \nu }   - \frac 13  \mnabla_{\lambda} \mnabla_{\kappa} {\tA}^{\kappa\lambda }_{\ \ \nu}  +\nonumber\\
     & \quad  -\frac{1}{27}\left[ \mBox h_{\nu} +14 \mnabla_{\nu} \mnabla_{\lambda}h^{\lambda}  +  19  h^{\sigma}\!\kolo{K}_{\sigma\nu}     \right]\, . \label{mnabla W4}
\end{align}

\chapter{Affine Lagrangians}

Below are presented a few examples of affine Lagrangians and associated field equations describing different theories. Precisely, \textbf{Chapter~\ref{affine ricci}} contains the simplest non-trivial theory, where curvature is represented by the full Ricci tensor and coincides with the Born-Infeld theory. In \textbf{Chapters~\ref{VAR V1}} and \textbf{\ref{VAR V6}} are presented two theories of the full Riemann curvature, whereas in \textbf{Chapter~\ref{fixed back}} is presented a model, where all irreducible parts of the Riemann curvature appear, but the algebraically traceless part $W^{\kappa}_{\ \lambda\mu\nu}$ is treated as a  fixed background field.

\section{Affine Lagrangian depending on the full Ricci tensor as a model of the unified theory of gravity and electromagnetism}

\sectionmark{Affine Lagrangian depending on the full Ricci tensor}

\sectionmark{The full Ricci tensor theory}
\label{affine ricci}

Results presented in this section were already written in \cite{lic, mag, kij2024}, although it will be very useful to rewrite them in this dissertation too, especially to demonstrate the formalism developed in \textbf{Chapter~\ref{scheme}}. Secondly, it is the easiest non-trivial theory in the affine picture. Finally, those results will be used  as a \textit{reference theory} for other extensions obtained from variants $V_1\!-\!V_6$ (\ref{w1}--\ref{w6}).

\subsection{Lagrangian} 

The affine Lagrangian $\Lag_{A}$~\eqref{lagF} is defined as a natural extension of the $\Lambda$--vacuum Lagrangian -- see~\eqref{detR}. By the scheme from \textbf{Chapter~\ref{scheme}}, it is given by the square root of four Riemann tensors contracted with two Levi-Civita symbols, denoted $RRRR$, which in this case coincides with the determinant of the full Ricci tensor $R_{\mu\nu}$ (cf. variant~$V_0$~\eqref{w0}):

\begin{align}
    RRRR = \det(R_{\mu\nu}) =\det (K_{\mu\nu}+F_{\mu\nu}) = KKKK + KKFF + o(F) \, ,
    \label{Lag F0}
\end{align}
where $o(F)$ denotes the high-order terms in $F$, and
\begin{align}
    KKKK&=\det K\, , \label{KKKK v0}\\
    KKFF&=\frac{1}{4}\, K_{\mu_1\nu_1} \,K_{\mu_2\nu_2}\, F_{\mu_3\nu_3} \,  F_{\mu_4\nu_4}\, \epsilon^{\mu_1\mu_2\mu_3\mu_4} \, \epsilon^{\nu_1\nu_2\nu_3\nu_4}\label{KKFF v0}\, .
\end{align}
Therefore, the  affine Lagrangian $\Lag_{A}$~\eqref{lagmodel}  of this theory  is given by:
\begin{align}
    \Lag_{A}=\frac{\sqrt{|KKKK + KKFF|}}{8\pi\Lambda}\, .
    \label{lag iksdfbosgv}
\end{align}
The global constant $\alpha=\frac{1}{8\pi\Lambda}$ -- see~\eqref{lag iksdfbosgv} -- is already determined. Because the term $KKKK$ is precisely a determinant of $K$, constants $\sigma=\gamma=1$. All of those characteristic constants are written below:
\begin{align}
    \alpha &= \frac{1}{8\pi\Lambda}\, , & \gamma&=1\, , & \sigma&=1\, .
    \label{constsV0}
\end{align}

\subsection{Non-metricity equation}

Due to the fact that the Lagrangian~\eqref{lag iksdfbosgv} is not dependent upon the traceless part of the curvature $W^{\kappa}_{\ \lambda\mu\nu}$, the first field equation, which is used to find the non-metricity tensor $N$, is the same as in the \textbf{Chapter~\ref{1st feq abs}}, equation~\eqref{ntens}. To sum up and remind those results, the non-metricity tensor $N_{\kappa \lambda\mu}$~\eqref{ntens} is:
\begin{align}
    N^{\kappa}_{\ \lambda\mu} =  \frac{8\pi}{3\sqrt{|\det g|}}\, \left(2\delta^{\kappa}_{(\lambda}\, \cJ_{\mu)}  - 3 g_{\lambda\mu}\, \cJ^{\kappa} \right)\, . 
    \label{nononono}
\end{align}
Its decomposition is the following -- see \textbf{Lemma~\ref{lem absence}}:
\begin{align}
     A_{\mu}&=   \frac{8\pi}{3\sqrt{|\det g|}}\, \cJ_{\mu}\, , &  A^{\kappa}_{\ \lambda\mu}&=    \frac{8\pi}{\sqrt{|\det g|}}\, \left(\frac25\, \delta^{\kappa}_{(\lambda}\, \cJ_{\mu)} - g_{\lambda\mu}\, \cJ^{\kappa} \right)\, ,\label{A rel J 1}    \\
    h^{\kappa}&= - \frac{18\cdot 8\pi}{5\sqrt{|\det g|}}\, \cJ^{\kappa}\, , & \widetilde{A}_{ \kappa\lambda \mu}&= 0\label{h rel J 1}    \, .
\end{align}
To find the remaining field equations, the scheme developed in the previous chapter will be used. 

\subsection{Einstein equation}
The non-perturbed solution~\eqref{eq0} is purely the Einstein $\Lambda$--vacuum equation~\eqref{eineq0}, due to the already chosen global constant $\alpha$~\eqref{constsV0}:
\begin{align}
    \Kz_{\mu\nu} = \Lambda\, g_{\mu\nu}\, .
    \label{Kzv0}
\end{align}
The first-order correction~\eqref{eq1} provides for the vanishing of  $\Ko$:
\begin{align}
     \frac12\, \left(ggg \Ko  \right)\, g^{-1}= g g \Ko\, , 
\end{align}
where
\begin{align}
     ggg\Ko &=-\Ko^{\alpha}_{\ \alpha}\, |\det g|\, , \\
    gg\Ko &= \left( \Ko^{\mu\nu} -  \Ko^{\alpha}_{\ \alpha}\, g^{\mu\nu}\right) \, |\det g|\, . \label{ajrgnasog}
\end{align}
Hence:
\begin{align}
    \Ko_{\mu\nu}=0\, .
\end{align}
The non-trivial terms appear in the second-order perturbation~\eqref{eq1}:
\begin{align}
 \frac{  1}{2}\, \left( \Lambda\,  g g g \Kt  +  g g  FF \right)\, g^{-1} =     \Lambda\, g g \Kt   +  g FF  \, , 
\end{align}
where:
\begin{align}
    ggFF &=-\frac 12\, F_{\alpha\beta}\, F^{\alpha\beta}\, |\det g|\, , \label{ggFFV0} \\
    ggg\Kt &=-\Kt^{\alpha}_{\ \alpha}\, |\det g|\, , \\
     gg\Kt  &= \left( \Kt^{\mu\nu} -   \Kt^{\alpha}_{\ \alpha}\, g^{\mu\nu}\right) \, |\det g| \, , \\ 
     gFF  &= \left(F^{\mu\alpha}\, F^{\nu}_{\ \alpha} - \frac 12\, F_{\alpha\beta}\, F^{\alpha\beta}\, g^{\mu\nu}\right) \, |\det g|\, .
\end{align}
The solution is
\begin{align}
    \Lambda\, \Kt^{\mu\nu} = - F^{\mu\alpha}\, F^{\nu}_{\ \alpha} + \frac 14\, g^{\mu\nu}\, F_{\alpha\beta}\, F^{\alpha\beta}\, .
\end{align}
Therefore, the Einstein equation obtained via the perturbative method~\eqref{eineqgen} is the following:
\begin{align}
    K_{\mu\nu} =  \Lambda\, g_{\mu\nu} -\frac{1}{\Lambda}\, \left(F_{\mu\alpha}\, F_{\nu}^{\ \alpha} - \frac 14\, g_{\mu\nu}\,  F_{\alpha\beta}\, F^{\alpha\beta} \right)\, .
    \label{eqq1}
\end{align}
However, this equation is not precisely the well-known form of the Einstein equation~\eqref{scheinmet}, due to the general, non-metric Ricci tensor $K_{\mu\nu}$ on the left-hand side which decomposes into the purely metric part $\kolo{K}$ and the rest $Q_{\mu\nu}$~\eqref{def Q}. Using the exact form of the non-metricity tensor $N^{\kappa}_{\ \lambda\mu}$~\eqref{nononono}, the tensor $Q_{\mu\nu}$ equals:
\begin{align}
     Q_{\mu\nu} =K_{\mu\nu} - \kolo{K}_{\mu\nu} = -6\, \left(\frac{8\pi}{3\sqrt{|\det g|}}\right)^2\, \cJ_{\mu}\,\cJ_{\nu}\, .
\end{align}
Moreover, implementing the  relation~\eqref{A rel J 1} between the potential~$A_{\mu}$ and the current~${\cal J}_{\mu}$ yields:
\begin{align}
     Q_{\mu\nu} = -6 \, A_{\mu}\,A_{\nu} \, .
     \label{Q v0}
\end{align}
Finally, the Einstein equation~\eqref{scheinmet} is the following:
\begin{align}
    \kolo{K}_{\mu\nu} = \Lambda\, g_{\mu\nu} -\frac{1}{\Lambda}\, \left(F_{\mu\alpha}\, F_{\nu}^{\ \alpha} - \frac 14\, g_{\mu\nu}\,  F_{\alpha\beta}\, F^{\alpha\beta} \right) + 6 \, A_{\mu}\,A_{\nu}\, ,
    \label{eqq2}
\end{align}  
or, using the Einstein tensor $\kolo{G}_{\mu\nu}$~\eqref{def ein tensor}:
\begin{align}
   \kolo{G}_{\mu\nu}   = -\Lambda\, g_{\mu\nu} -\frac{1}{\Lambda}\, \left(F_{\mu\alpha}\, F_{\nu}^{\ \alpha} - \frac 14\, g_{\mu\nu}\,  F_{\alpha\beta}\, F^{\alpha\beta} \right) + 6 \,\left( A_{\mu}\,A_{\nu}- \frac12\, g_{\mu\nu}\, A_{\sigma}A^{\sigma}\right)\, .
        \label{eqq21}
\end{align}
The right-hand side of this equation corresponds with the stress-energy tensor density of the Proca field -- see \textbf{Appendix~\ref{Proca th}}.

\subsection{Effective cosmological parameter}
\label{eff lam v0}
The cosmological parameter $\Lambda_{\rm eff}$~\eqref{def lam eff} associated with this theory is given by
\begin{align}
     \Lambda_{\rm eff} :=\frac 14\,  \kolo{K}_{\mu\nu} g^{\mu\nu} = \Lambda + \frac{3}{2}\, A_{\kappa}A^{\kappa} \, .
      \label{lam effv0}
\end{align}
Then, the Einstein equation~\eqref{eqq2} with introduced $\Lambda_{\rm eff}$~\eqref{lam effv0} takes the following form:
\begin{align}
\kolo{K}_{\mu\nu} =\Lambda_{\rm eff} \, g_{\mu\nu} +6\left(A_{\mu}A_{\nu} - \frac14\, g_{\mu\nu} \, A_{\kappa}A^{\kappa}\right)-\frac{1}{\Lambda} \left(F_{\mu\alpha}\, F_{\nu}^{\ \alpha} - \frac 14\, g_{\mu\nu}\, F_{\alpha\beta}\, F^{\alpha\beta}  \right)\, ,
\label{eqq22}
\end{align}
or, equivalently to equation~\eqref{eqq21}, where was used the Einstein tensor $\kolo{G}_{\mu\nu}$:
\begin{align}
   \kolo{G}_{\mu\nu}   = -\Lambda_{\rm eff} \, g_{\mu\nu} +6\left(A_{\mu}A_{\nu} - \frac14\, g_{\mu\nu} \, A_{\kappa}A^{\kappa}\right)-\frac{1}{\Lambda} \left(F_{\mu\alpha}\, F_{\nu}^{\ \alpha} - \frac 14\, g_{\mu\nu}\, F_{\alpha\beta}\, F^{\alpha\beta}  \right) \, .
        \label{eqq211}
\end{align}

\subsection{Field equation for the skew-symmetric Ricci tensor}
\label{feq F v1}
The  field equation for  $F_{\mu\nu}$ was determined by the symplectic relation~\eqref{rel chi} and derived schematically in the previous chapter~\eqref{chi}:
\begin{align}
    \chi = \frac{ \sigma\, \sigma_g }{ 16 
 \pi \Lambda \gamma^2 \sqrt{|\det g|}}\,   ggF  \, ,
\end{align}
where
\begin{align}
   ggF  = -F^{\mu\nu}\, |\det g|\, ,
\end{align}
due to the already calculated  term $ggFF$ -- see~\eqref{ggFFV0}. The metric signature is assumed to be Lorentzian, thus, $\sigma_g=-1$. Upon substituting all characteristic constants~\eqref{constsV0}, the above field equation becomes:
\begin{align}
    \chi^{\mu\nu} =\frac{ \sqrt{|\det g|}}{16\pi \Lambda} \, F^{\mu\nu}\, ,
     \label{rel konst0}
\end{align}
and will be called \textit{the constitutive relation between $\chi$ and $F$}. 

\subsection{Potential equation}
\label{eq pot v0}

The above expression may appear quite simple, but it encodes a much deeper structure. The left-hand side, the tensor density \( \chi^{\mu\nu} \), is related to the non-metricity tensor \( N^{\kappa}_{\ \lambda\mu} \) via the current \( \cJ^{\mu} \). Specifically, the divergence of \( \chi^{\mu\nu} \) equals \( \cJ^{\mu} \) — see formula~\eqref{cal J}. On the other hand, the skew-symmetric Ricci tensor \( F_{\mu\nu} \)~\eqref{rozklad pelny F} is constructed from derivatives of the trace of non-metricity \( A_{\mu} \), which is also proportional to the current \( \cJ^{\mu} \) ~\eqref{A rel J 1}. Therefore, taking the covariant derivative of the above expression yields a differential equation for the potential \( A_{\mu} \). Firstly:
\begin{align}
    \cJ^{\mu}:=\mnabla_{\nu}\chi^{\mu\nu}   = \frac{\sqrt{|\det g|}}{16\pi \Lambda} \,\mnabla_{\nu} F^{\mu\nu}  \, .
\end{align}
Then, replacing the current ${\cal J}_\mu$ by  the potential $A_{\mu}$~\eqref{A rel J 1}, and using the formula for $\mnabla_{\nu} F^{\mu\nu}$~\eqref{mnabla F}, one has:
\begin{align}
 \mBox A^{\mu}=A^{\sigma}\, \kolo{K}_{\sigma}^{\ \mu} -6\Lambda\, A^{\mu}   \, .
    \label{jdshdfbs2} 
\end{align}
An above equation is known in literature as the \textit{Proca equation} \cite{Proca} -- the generalisation of the Klein-Gordon equation for “massive'' vector fields, which represent massive bosons, or is treated as an extension of the standard Maxwellian electromagnetism. Although, to be precise, in the Proca equation the constant “$-6 \Lambda$'' should be positive and is related to the mass parameter~\eqref{eq P2}:
\begin{align}
    m^2=-6\hbar^2 \Lambda\, .
\end{align}
But even if $m^2>0$, the mass interpretation of this constant is not well-posed in curved spacetimes. 

\ 

To confirm that the above equation does not imply any other constraints, the covariant divergence of the above equation is derived:
\begin{align}
    6\Lambda \underbrace{\mnabla_{\mu}A^{\mu}}_{=0}  = \mnabla_{\mu} \left( A^{\sigma}\, \kolo{K}_{\sigma}^{\ \mu}  \right) - \mnabla_{\mu} \mBox A^{\mu}\, . 
    \label{aisfos}
\end{align}
Of course, the above divergence term vanishes upon the Lorenz gauge condition~\eqref{lor gauge}. Using \textbf{Lemma~\ref{lem comut}}, the divergence of the d'Alembert operator $\mBox $ is the following:
\begin{align}
  \mnabla_{\mu} \mBox A^{\mu} &= \mnabla_{\mu}\mnabla_{\alpha}\mnabla^{\alpha}A^{\mu} = \mnabla_{\alpha}\mnabla_{\mu}\mnabla^{\alpha}A^{\mu} + \mnabla^{\sigma}A^{\mu} \kolo{R}^{\alpha}_{\ \sigma \mu \alpha}  + \mnabla^{\alpha}A^{\sigma} \kolo{R}^{\mu}_{\ \sigma \mu \alpha} = \nonumber \\
  &= \mnabla_{\alpha}\mnabla_{\mu}\mnabla^{\alpha}A^{\mu} =\mnabla^{\alpha}\left( \mnabla_{\alpha}\underbrace{\mnabla_{\mu}A^{\mu}}_{=0} + A^{\sigma} \kolo{R}^{\mu}_{\ \sigma \mu \alpha} \right) = \mnabla^{\alpha} \left( A^{\sigma}  \kolo{K}_{\sigma \alpha}  \right) \, . 
  \label{com box}
\end{align}
Therefore, the equation~\eqref{aisfos} vanishes automatically and does not produce any extra constraints on potential $A_{\mu}$. It means that $A_{\mu}$ has to satisfy the equation~\eqref{jdshdfbs2} with the Lorenz gauge condition~\eqref{lor gauge}. A priori, the equation~\eqref{jdshdfbs2} is not linear, due to the quadratic terms in $\kolo{K}$~\eqref{eqq2}. However, in this equation appears only the contraction $A^{\sigma} \!\kolo{K}\!_{\sigma}^{\ \mu}$, which produces a third-order term, which could be neglected, up to the taken assumptions. Hence, $\kolo{K}\!_{\mu\nu}\approx\Lambda\, g_{\mu\nu}$, and  equation~\eqref{jdshdfbs2} takes the following (approximated form):
\begin{align}
 \mBox A^{\mu} =  -5 \Lambda\, A^{\mu}  \, .
    \label{jdshdfbs3} 
\end{align}

\section{Affine Lagrangians depending on the full curvature}

The next chapters contain  theories represented by proposed variants $V_1$~\eqref{w1} and $V_6$~\eqref{w6}, with derived field equations via the scheme presented in \textbf{Chapter~\ref{scheme}}. However, the precise calculations of many formulae are absent,  due to the high level of complexity and the length of those formulae. Moreover, most of  the expressions  presented below  were calculated in \textbf{Wolfram Mathematica 13.2} with \textbf{Package xAct`xTensor`  version 1.2.0}, prepared by Jose M. Martin-Garcia, under the Public Licence. Without the support of this program, the correct calculations would not be possible.

\section{Variant \texorpdfstring{$V_1$}{V1}}
\label{VAR V1}
\subsection{Lagrangian}
By the scheme from \textbf{Chapter~\ref{scheme}}, the affine Lagrangian~\eqref{lagF} is given by the square root of four Riemann tensors contracted with two Levi-Civita symbols, denoted as~$RRRR$ (cf. variant $V_1$~\eqref{w1}):

\begin{align}
 RRRR &= R^{\textcolor{red}\alpha}_{\ \mu_1 {\textcolor{blue} \beta}  \nu_1}\, R^{{\textcolor{blue} \beta}}_{\  \mu_2 {\textcolor{magenta}\gamma} \nu_2}\,R^{{\textcolor{magenta}\gamma}}_{\ \mu_3  {\textcolor{cyan}\delta} \nu_3}\, R^{{\textcolor{cyan}\delta}}_{\  \mu_4 {\textcolor{red}\alpha}  \nu_4}\, \epsilon^{\mu_1 \mu_2 \mu_3 \mu_4}\, \epsilon^{\nu_1 \nu_2 \nu_3 \nu_4} =\label{LagF1}   \\
 &=KKKK  +KKKW + KKFF + KKFW + KKWW   + o(F,W)\, ,   \nonumber
\end{align}
where $o(F,W)$ denotes high-order terms in $F$ and $W$, and
\begin{align}
KKKK&=-\frac{88}{27}\, \det K\, , \label{KKKKv1}  \\
KKKW&=-\frac{28}{27}\,  K_{\alpha \beta}\, K_{\gamma \delta}\, K_{\mu_1\mu_2}\, W^{\alpha}_{\ \mu_3\mu_4\nu} \, \epsilon^{\beta\gamma\mu_1\mu_3}\, \epsilon^{\delta \mu_2\mu_4\nu} \, ,  \label{KKKWv1}\\
KKFF &= -F_{\alpha \beta}\,  F_{\gamma\delta}\,  K_{\mu_1\mu_2} \, K_{\mu_3\mu_4}\,  \left[\frac{22}{75}\,\epsilon^{\alpha\gamma\mu_1\mu_3}\,  \epsilon^{\beta\delta \mu_2\mu_4} +\frac{56}{225}\, \epsilon^{\alpha\beta\mu_1\mu_3}\,  \epsilon^{\gamma\delta \mu_2\mu_4}\right] \, ,   \\
KKFW &=F_{\alpha \beta}\,  K_{\gamma \delta}\,  K_{\mu_1\mu_2} \, W^{\alpha}_{\ \mu_3\mu_4\nu} \,  \left[\frac{8}{15}\, \epsilon^{\beta\gamma\mu_1\mu_4} \, \epsilon^{\delta \mu_2\mu_3\nu} + \frac{28}{45}\,  \epsilon^{\beta\gamma\mu_1\mu_3}\,  \epsilon^{\delta \mu_2\mu_4\nu}\right]+ \nonumber \\
& \quad  +F_{\alpha \beta} \, K_{\gamma \delta}\,  K_{\mu_1\mu_2}\,  W^{\gamma}_{\ \mu_3\mu_4\nu} \, \left[\frac{56}{45}\,  \epsilon^{\alpha \delta \mu_1\mu_3} \, \epsilon^{\beta \mu_2\mu_4\nu} +\frac{32}{45} \,\epsilon^{\alpha \beta \mu_1\mu_4} \,  \epsilon^{\delta \mu_2\mu_3\nu} \right] \, ,   \\
KKWW&=\frac{2}{9}\, K_{\alpha \beta}\,  K_{\gamma \delta}\,  W^{\alpha}_{\ \mu_1\mu_2\mu_3}\, W^{\gamma}_{\ \mu_4\nu_1\nu_2} \, \epsilon^{\beta\delta\mu_1\mu_4} \,  \epsilon^{ \mu_2 \mu_3\nu_1\nu_2} +\nonumber \\
& \quad  -\frac{16}{9}\,  K_{\alpha \beta}\,  K_{\gamma \delta}\,  W^{\alpha}_{\ \mu_1\mu_2\mu_3}\,  W^{\mu_2}_{\ \mu_4\nu_1\nu_2} \, \epsilon^{\beta\gamma\mu_1\mu_4} \, \epsilon^{ \delta\mu_3\nu_1\nu_2} +\nonumber \\
& \quad  +\frac{2}{3}\, K_{\alpha \beta}\,  K_{\gamma \delta}\,  W^{\mu_1}_{\ \mu_2\mu_3\mu_4}\,  W^{\mu_3}_{\ \nu_1 \mu_1\nu_2}\, \epsilon^{\alpha\gamma\mu_2\nu_1}\,  \epsilon^{\beta\delta \mu_4\nu_2} \,  .
\label{rozklad w1}
\end{align}

Therefore, the  affine Lagrangian $\Lag_{A}$~\eqref{lagmodel}  of this theory  is given by:
\begin{align}
    \Lag_{A}=\alpha\, \sqrt{|KKKK  +KKKW + KKFF + KKFW + KKWW|}\, . 
    \label{lag A1}
\end{align}
The equality~\eqref{KKKKv1} determines two characteristic constants $\sigma$ and $\gamma$~\eqref{reldetK}, whereas the global constant $\alpha$~\eqref{alpha const} is  fitted to reconstruct the standard  Einstein equation with the cosmological constant $\Lambda$ for the unperturbed theory:
\begin{align}
    \sigma&=-1\, ,& \gamma^2 &= \frac{88}{27}\, , & \alpha&= \frac{1}{8\pi\Lambda } \,  \sqrt{\frac{27}{88}}\,.
    \label{const1}
\end{align}

\subsection{Non-metricity equation}
\label{non-met eq v1}
In this case, the Lagrangian depends upon the whole curvature. Therefore, the non-metricity tensor $N_{\kappa \lambda\mu}$ is restricted to $\No_{\kappa\lambda\mu}$  in \textbf{Chapter~\ref{first field eq}}. To make the discussion self-consistent, obtained results are rewritten below. However, the symbol “1'' over other letters related  to the non-metricity tensor will be omitted. Thus, non-metricity tensor $N_{\kappa\lambda\mu}$~\eqref{nsaifhy} is:
\begin{align}
      N_{\kappa\lambda\mu} &= \frac{8\pi}{\sqrt{|\det g|}}\, \left[\mnabla_{\nu} \left(\Omega_{\kappa\lambda\mu}^{\ \ \ \nu} - 2\Omega_{(\lambda\mu)\kappa}^{\ \ \ \  \ \nu}+g_{\kappa(\lambda}\, \cO_{\mu)}^{\ \ \nu} - \frac 12 \, g_{\lambda\mu}\, \cO_{\kappa}^{\ \nu} \right) +\right.\nonumber\\
      & \quad  \left.  +\frac{2}{3}\, g_{\kappa(\lambda}\, \cJ_{\mu)} - g_{\lambda\mu}\, \cJ_{\kappa} \right]\, , 
    \label{nsaifhyv1}        
    \end{align}
    whereas its decomposition -- see \textbf{Lemma~\ref{lem presence}} -- is the following:
    \begin{align}
A_{\kappa} &=  \frac{4\pi}{3\sqrt{|\det g|}}\left(2\, \mathcal{J}_{\kappa} + 3\mnabla_{\nu} \mathcal{O}_{\kappa}^{\ \nu}   \right)\, , 
\label{pot Aov1}\\
 A^{\kappa}_{\ \lambda\mu} &=   \frac{8\pi}{ \sqrt{|\det g|}}\left[  \mnabla_{\nu} \left(\Omega_{\ \lambda\mu}^{\kappa \ \ \ \nu}   
 - 2\Omega_{(\lambda\mu)}^{\ \ \ \ \kappa\nu}  \right) + \frac 12\, \mnabla_{\nu}\left( \frac65 \, \delta^{\kappa}_{(\lambda }\, {\cal O}_{\mu)}^{\ \ \nu} - g_{\lambda \mu}\, {\cal O}^{\kappa  \nu} \right)   + \right.\nonumber \\
 &  \quad  \left. +\frac{2}{5}\, \delta^{\kappa}_{(\lambda} \, \mathcal{J}_{\mu)} - g_{\lambda\mu}\, \mathcal{J}^{\kappa}    \right]\, ,
\label{pot tAov1}\\
h_{\kappa}&= -\frac{16\pi}{5\sqrt{|\det g|}}\, \left(9\cJ_{\kappa} + \mnabla_{\nu}\cO_{\kappa}^{\ \nu} \right)\, , 
\label{pot hv1}\\
\widetilde{A}_{\kappa \lambda\mu} &=\frac{8 \pi}{\sqrt{|\det g|}}\, \mnabla_{\nu}\left[\Omega_{\kappa\lambda\mu}^{\ \ \ \ \nu} - 2\Omega_{(\lambda\mu)\kappa}^{\ \ \ \ \ \nu} + \frac{5}{9}\, g_{\kappa(\lambda}\cO_{\mu)}^{\ \ \nu} - \frac{7}{18}\, g_{\lambda\mu}\, \cO_{\kappa}^{ \ \nu}  \right]\, .
\label{pot ttAov1}
\end{align} 

\subsection{Einstein equation}
 \label{eineq v1}
As in the previous example, the unperturbed solution $\Kz$ must be the Einstein $\Lambda$-vacuum equation~\eqref{eq0}:
\begin{align}
    \Kz_{\mu\nu} =  \frac{1}{8\pi \alpha\, \gamma}\, g_{\mu\nu}=\Lambda\, g_{\mu\nu}\, .
    \label{Kzv1}
\end{align}

The first-order correction $\Ko$ reads as follows~\eqref{eq1}:
\begin{align}
     \frac12\, \left(ggg \Ko + ggg W \right)\, g^{-1}   =  g g \Ko + g g W\, . 
\end{align}
where
\begin{align}
    ggg \Ko &=\frac{88}{27}\, \Ko^{\alpha}_{\ \alpha}\, |\det g|\,  , \\
     ggg W &=  0\, , \\
      g g \Ko &=\frac{88}{27}\,\left(\Ko^{\alpha}_{\ \alpha}\, g^{\mu\nu} - \Ko^{\mu\nu} \right)\, |\det g| \, , \\
      g g W  &=   0\,.
\end{align}
It is quite interesting that, \textit{a priori}, there appears a term $\Kz\Kz\Kz W$~\eqref{KKKWv1}, but the Einstein equation $\Kz = \Lambda\, g$ implies that the associated term $gggW$ vanishes:
\begin{align}
    gggW &= -\frac{28}{27}\,  g_{\alpha \beta}\, g_{\gamma \delta}\, g_{\mu_1\mu_2}\, W^{\alpha}_{\ \mu_3\mu_4\nu}\, \epsilon^{\beta\gamma\mu_1\mu_3}\, \epsilon^{\delta \mu_2\mu_4\nu} \\
    &= -\frac{28}{27}\, W^{\alpha\mu_3}_{\ \ \ \ \mu_4\nu}\, \epsilon_{\alpha\delta\mu_2\mu_3}\, \epsilon^{\delta\mu_2\mu_4\nu} \nonumber \\
    &= -\frac{28}{27}\, W^{\alpha\mu_3}_{\ \ \ \ \mu_4\nu}\, \delta_{\alpha}^{[\mu_4}\, \delta_{\mu_3}^{\nu]} = 0\, , \label{gggW van}
\end{align}
because, by definition~\eqref{trless W}, the tensor $W$ is algebraically traceless. Analogously, the term $ggW$ also vanishes:
\begin{align}
    ggW = \frac{\partial}{\partial g}\, (gggW) = 0\, .
\end{align}
Thus, as in the theory of the full Ricci tensor (see \textbf{Chapter~\ref{affine ricci}}), $\Ko_{\mu\nu}$ vanishes:
\begin{align}
    \Ko_{\mu\nu} = 0\, .
\end{align}

The first non-trivial corrections appear at the level of the second-order perturbation~\eqref{eq2}:
\begin{align}
&\frac{  1}{2}\, \left(  \Lambda\,  g g g \Kt  +  g g  FF+ g g F W+ g g  WW \right)\, g^{-1}= \nonumber  \\
&=   \Lambda\, g g \Kt  +  g FF+     g  FW+     g WW \, ,
\label{eq2v2}
\end{align}
where:
\begin{align}
    g g g \Kt&=\frac{88}{27}\ \Kt^{\alpha}_{\ \alpha}\, |\det g| \,    , \\
    g g  FF&= \frac{356}{225}\, F_{\alpha\beta}\, F^{\alpha\beta}\, |\det g| \,  ,\label{ggFF v1}\\
    g g F W&=- \frac{16}{45}\, F_{\alpha\beta}\, W^{\alpha  \beta} \, |\det g| \,  , \label{ggFW v1}\\
    g g  WW&=\frac43 \,\left( W_{\alpha \beta }\, W^{ \beta \alpha}  - \, W_{\alpha\beta \kappa\lambda}\, W^{\kappa\lambda \alpha\beta}\right)\, |\det g|\,  , \label{ggWW v1}  \\
    g g \Kt  &=\frac{88}{27}\,\left(\Kt^{\alpha}_{\ \alpha}\, g^{\mu\nu} -  \Kt^{\mu\nu} \right)\, |\det g|\,  ,  
    \end{align}
    
    \begin{align}
     g FF &=- \frac{712}{225}\, \left(F^{\mu\alpha}\, F^{\nu}_{\ \alpha} - \frac{1}{2}\, F^{\alpha\beta}\, F_{\alpha\beta}\, g^{\mu\nu}\right)\, |\det g|\,  , \\
     g  FW &=\frac {16}{45}\, \left(F^{\ (\mu|}_{ \alpha}\, W^{\alpha |\nu)}  +  F_{\alpha\beta}\, W^{\alpha (\mu\nu)\beta} -  F_{\alpha\beta}\, W^{\alpha \beta} \, g^{\mu\nu} \right)\, |\det g|\, , \\
      g WW &= \frac{40}{9}\,\left( W^{\alpha (\mu| \gamma\sigma}\, W_{\gamma\sigma\alpha}^{\ \ \ \ |\nu)} -  W_{\alpha\beta} W^{\beta (\mu\nu)\alpha}\right)\, |\det g|  +\nonumber \\
     &  \quad  +\frac{20}{9}\, \left(W^{\alpha \beta} W_{\beta\alpha}-W^{\alpha\beta\gamma\sigma}\, W_{\gamma\sigma\alpha
     \beta}\right)\, g^{\mu\nu}\, |\det g| +\nonumber  \\
     & \quad  -\frac{16}9\, \left( W^{\alpha(\mu} W^{\nu)}_{\ \ \alpha} + W^{\alpha\beta\gamma (\mu} W^{\nu)}_{\ \ \gamma\alpha\beta} + W_{\alpha\beta}\, W^{(\mu|\beta\alpha|\nu)} \right)\, |\det g|\,  ,
\end{align}
and the tensor $W^{\kappa}_{\ \nu}$ is defined as the remaining metric trace of $W^{\kappa}_{\ \lambda\mu\nu}$:
\begin{align}
    W^{\kappa}_{\ \nu} :=W^{\kappa}_{\ \lambda\mu\nu}\, g^{\lambda\mu}\, .
    \label{def ten W2}
\end{align}
Contracting the  equation~\eqref{eq2v2} with the metric tensor $g_{\mu\nu}$ implies:
\begin{align}
    \Kt^{\alpha}_{\ \alpha}=0\, .
    \label{trKtv1}
\end{align}
Then, the second-order correction $\Kt_{\mu\nu}$ equals:
\begin{align}
- \frac{88}{27}\,\Lambda\ \Kt^{\mu\nu}&=       \frac{712}{225}\, \left(F^{\mu\alpha}\, F^{\nu}_{\ \alpha}-\frac 14\,  F^{\alpha\beta}\, F_{\alpha\beta}\, g^{\mu\nu} \right) +\nonumber  \\
& \quad  -\frac {16}{45}\, \left( F^{\ (\mu|}_{  \alpha}\, W^{\alpha |\nu)}  +  F_{\alpha\beta}\, W^{\alpha (\mu\nu)\beta} - \frac{1}{2}\, F_{\alpha\beta}\, W^{\alpha \beta} \, g^{\mu\nu}\right) +  \nonumber\\
&  \quad    -  \frac{40}{9}\,\left( W^{\alpha (\mu| \gamma\sigma}\, W_{\gamma\sigma\alpha}^{\ \ \ \ |\nu)} -  W_{\alpha\beta} W^{\beta (\mu\nu)\alpha}\right) +\nonumber \\
     &  \quad  -\frac{14}{9}\, \left(W^{\alpha \beta} W_{\beta\alpha}-W^{\alpha\beta\gamma\sigma}\, W_{\gamma\sigma\alpha  \beta}\right)\, g^{\mu\nu}  +\nonumber  \\
     & \quad  +\frac{16}9\, \left( W^{\alpha(\mu} W^{\nu)}_{\ \ \alpha} + W^{\alpha\beta\gamma (\mu} W^{\nu)}_{\ \ \gamma\alpha\beta} + W_{\alpha\beta}\, W^{(\mu|\beta\alpha|\nu)} \right)  \, .
     \label{v1 k2}
\end{align} 
 Therefore, the  Einstein equation~\eqref{eineqgen} is the following:
 \begin{align}
     K_{\mu\nu} = \Lambda\, g_{\mu\nu} + \Kt_{\mu\nu}\, .
     \label{ein eq0 v1}
 \end{align}
However, this equation is not exactly in the standard form of the Einstein equation~\eqref{scheinmet}, due to the presence of the general, non-metric, symmetric Ricci tensor $K_{\mu\nu}$ on the left-hand side, which decomposes into the purely metric part $\kolo{K}$ and the remainder $Q_{\mu\nu}$~\eqref{def Q}. Using the explicit form of the non-metricity tensor $N^{\kappa}_{\ \lambda\mu}$~\eqref{nsaifhyv1}, the tensor $Q_{\mu\nu}$ takes the form:
\begin{align}
     Q_{\mu\nu} &=  \frac{8\pi}{\sqrt{|\det g|}}\, \mnabla_{\kappa} \mnabla_{\lambda}\left(\Omega^{\kappa \ \ \lambda}_{\ \mu\nu} - 2 \Omega_{(\mu\nu)}^{\ \ \ \    \kappa\lambda} -\frac 12\, g_{\mu\nu}\, \cO^{\kappa\lambda}\right) + \nonumber\\
     & \quad  -\left(\frac{8\pi}{\sqrt{|\det g|}} \right)^2\, \bigg\{ \left(\mnabla_{\alpha}\Omega_{\mu\sigma}^{\ \ \kappa \alpha}\right) \left(\mnabla_{\beta}\Omega_{\nu\kappa}^{\ \ \sigma \beta}\right) - 2\left(\mnabla_{\alpha}\Omega_{\sigma\mu}^{\ \ \kappa \alpha}\right) \left(\mnabla_{\beta}\Omega_{\ \nu\kappa}^{\sigma \ \ \beta}\right) + \nonumber\\
     & \quad   + 2\left(\mnabla_{\alpha}\Omega_{\sigma\mu}^{\ \ \kappa \alpha}\right) \left(\mnabla_{\beta}\Omega_{\kappa\nu}^{\ \ \sigma \beta}\right) +\left(\mnabla_{\alpha}\Omega_{\kappa\mu\nu}^{\ \ \  \alpha}\right)\left( \mnabla_{\beta}\cO^{\kappa\beta}\right) - \frac 12\, \left(\mnabla_{\alpha}\cO_{ \mu}^{\  \alpha}\right)\left( \mnabla_{\beta}\cO_{\nu}^{\ \beta}\right)+ \nonumber\\
     & \quad  +2\cJ^{\sigma}\mnabla_{\alpha}\Omega_{\sigma\mu\nu}^{\ \ \ \alpha}+\frac 23\, \cJ_{\mu}\cJ_{\nu}\bigg\} \, .
     \label{tensor Q}
\end{align}
Then, the Einstein equation~\eqref{scheinmet} is the following:
\begin{align}
    \kolo{K}_{\mu\nu} = \Lambda\, g_{\mu\nu} + \Kt_{\mu\nu}-Q_{\mu\nu}\, ,
    \label{eqq2v1}
\end{align} 
or, using the Einstein tensor $\kolo{G}_{\mu\nu}$~\eqref{def ein tensor}:
\begin{align}
   \kolo{G}_{\mu\nu}   = -\Lambda\, g_{\mu\nu} +\Kt_{\mu\nu} - \left( Q_{\mu\nu} - \frac12\, g_{\mu\nu}\, Q_{\sigma}^{\ \sigma}\right)\, ,
   \label{eqq21v1}
\end{align}
where
\begin{align}
      Q_{\alpha}^{\ \alpha} &= - \frac{8\pi}{\sqrt{|\det g|}}\, \left(\mnabla_{\kappa}\mnabla_{\lambda}\cO^{\kappa\lambda} \right) + \left( \frac{8\pi}{\sqrt{|\det g|}}\right)^2\bigg\{\left(\mnabla_{\alpha}\Omega^{\kappa\lambda\mu\alpha}\right) \left(\mnabla_{\beta}\Omega_{\kappa\lambda\mu}^{\ \ \ \beta}\right)+ \nonumber\\
     & \quad   - 2 \left(\mnabla_{\alpha}\Omega^{\kappa\lambda\mu\alpha}\right) \left(\mnabla_{\beta}\Omega_{\lambda\mu\kappa}^{\ \ \ \beta}\right)-\frac 12  \left(\mnabla_{\alpha}\cO_{\kappa}^{\ \alpha}\right)\left(\mnabla_{\beta}\cO^{\kappa \beta}\right)+\nonumber\\
     & \quad  - 2\cJ^{\kappa}\left(\mnabla_{\alpha}\cO_{\kappa}^{\ \alpha}\right)-\frac 23\, \cJ_{\kappa}\cJ^{\kappa}\bigg\}\, .
     \label{trace Q}
\end{align}
Due to the fact that the non-metricity tensor $N^{\kappa}_{\ \lambda \mu}$ is quite complicated in this theory~\eqref{nsaifhyv1}, the simplification presented in the Ricci tensor theory cannot be performed -- cf.~\eqref{eqq2}.

 \subsection{Effective cosmological parameter}
\label{eff lam v1} 

 The cosmological parameter $\Lambda_{\rm eff}$~\eqref{def lam eff} associated with this theory is given by
\begin{align}
     \Lambda_{\rm eff} :=\frac 14\,  \kolo{K}_{\mu\nu} g^{\mu\nu} = \Lambda  -\frac 14\, Q_{\alpha}^{\ \alpha}\, ,
      \label{lam effv1}
\end{align}
because $\Kt_{\mu\nu}$ is metrically traceless~\eqref{trKtv1}. Then, the Einstein equation~\eqref{eqq2v1} with introduced $\Lambda_{\rm eff}$~\eqref{lam effv1} takes the following form:
\begin{align}
\kolo{K}_{\mu\nu} =\Lambda_{\rm eff} \, g_{\mu\nu} + \Kt_{\mu\nu} - \left(Q_{\mu\nu}-\frac 14\, g_{\mu\nu}\,Q_{\sigma}^{\ \sigma} \right)\, ,
\end{align}
or, equivalently to equation~\eqref{eqq21v1}, where the Einstein tensor $\kolo{G}_{\mu\nu}$  was used :
\begin{align}
   \kolo{G}_{\mu\nu}   = -\Lambda_{\rm eff} \, g_{\mu\nu} +\Kt_{\mu\nu} - \left(Q_{\mu\nu}-\frac 14\, g_{\mu\nu}\,Q_{\sigma}^{\ \sigma} \right) \, . 
\end{align}

\subsection{Field equation for the skew-symmetric Ricci tensor}

The  field equation for  $F_{\mu\nu}$ was determined by the symplectic relation~\eqref{rel chi} and was schematically  derived in the \textbf{Chapter~\ref{scheme}} in formula~\eqref{chi}:
\begin{align}
    \chi&= \frac{\sigma\sigma_g }{ 16 
 \pi \Lambda \gamma^2 \sqrt{|\det g|}}\, \left( ggF + ggW\right)   \, ,
 \label{rel konst11}
\end{align}
where
\begin{align}
     g g F &= \frac{\partial}{\partial F}(ggFF)= \frac{712}{225}\, F^{\mu\nu}\, |\det g|\, , \\
     g g W  &= \frac{\partial}{\partial F}(ggFW)= -\frac{16}{45}\, W^{[\mu \nu]} \, |\det g|\,.
\end{align}
The metric signature is assumed to be Lorentzian, thus, $\sigma_g=-1$. Upon substituting all  characteristic constants $\alpha, \sigma, 
\gamma$~\eqref{const1}, the above  formula~\eqref{rel konst11} stays:
\begin{align}
    \chi^{\mu\nu} =\frac{27\, \sqrt{|\det g|} }{88 \cdot 16 
 \pi \Lambda }\,  \left( \frac{712}{225}\, F^{\mu\nu}  -\frac{16}{45}\, W^{[\mu \nu]} \right) \,  ,
     \label{rel konst1}
\end{align}
and will be called \textit{the constitutive relation between} $\chi$ \textit{and} $F$. 

In contrast to the theory based on the full Ricci tensor from \textbf{Chapter~\ref{feq F v1}}, it is not straightforward to derive an equation for the potential $A_{\mu}$~\eqref{jdshdfbs2}, due to the much more complicated structure of the non-metricity tensor~\eqref{nsaifhyv1}. Specifically, the presence of covariant derivatives $\mnabla \Omega$ significantly complicates the calculations.

\subsection{Field equation for the traceless Riemann tensor}

The  field equation for  $W^{\kappa}_{\ \lambda\mu\nu}$ was determined by the symplectic relation~\eqref{rel Sigma} and was schematically  derived in the \textbf{Chapter~\ref{scheme}} in formula~\eqref{sigma}:
\begin{align}
    \Sigma&= \frac{\sigma \sigma_g  }{ 16 
 \pi \Lambda \gamma^2 \sqrt{|\det g|}}\,  \left( gg F + gg W \right)\, ,
 \label{rel Sigma11}
\end{align}
where 
\begin{align}
     g g F &=\frac{\partial}{\partial W}(ggFW)   \, , &
     g g W  &=\frac{\partial}{\partial W}(ggWW) \,.
\end{align}
However, it was split into four equations (\ref{sigma feq1}–\ref{sigma feq4}), each corresponding to an independent component of the tensor $W$ — see \textbf{Lemma~\ref{lemm W dec}} and formula~\eqref{W decomposition}. To simplify the derivation of these equations, the quantities $ggFW$~\eqref{ggFW v1} and $ggWW$~\eqref{ggWW v1} above are rewritten, taking into account the decomposition of $W$ given in~\eqref{W decomposition}:
\begin{align}
    g g F W&=- \frac{16}{45}\, F_{\alpha\beta}\, W^{[\alpha  \beta]} \, |\det g| \,  , \label{ggFW2 v1}\\
    g g  WW&= \left(  - \frac{2}{3} W_{[\alpha\beta]}W^{[\alpha\beta]}    + \frac{2}{3} W_{(\alpha\beta)}W^{(\alpha\beta)} - \frac{4}{3} \widetilde{W}_{[\alpha\beta]\kappa\lambda} \widetilde{W}^{[\kappa\lambda]\alpha\beta}\right)\, |\det g|\,  . \label{ggWW2 v1}
\end{align}
Then, implying  all characteristic constants~\eqref{const1} with $\sigma_g=-1$, the corresponding field equations (\ref{sigma feq1}–\ref{sigma feq4})  are the following:
\begin{align}
    \frac{5}{6}\, \Sigma^{[\mu\nu]} &= \frac{\partial \Lag_A}{\partial W_{[\mu\nu]}} = \frac{27\sqrt{|\det g|}}{88 \cdot 16 \pi \Lambda}\left(-\frac{4}{3} W^{[\mu\nu]} - \frac{16}{45}F^{\mu\nu}\right)\, ,
    \label{sympl sigma1 v1}\\
    \frac{3}{4}\, \Sigma^{(\mu\nu)} &= \frac{\partial \Lag_A}{\partial W_{(\mu\nu)}} = \frac{27\sqrt{|\det g|}}{88 \cdot 16 \pi \Lambda} \cdot \frac{4}{3} W^{(\mu\nu)}\, , 
    \label{sympl sigma2 v1}\\
    \widetilde{\Sigma}^{(\kappa \lambda)\mu\nu} &= \frac{\partial \Lag_A}{\partial \widetilde{W}_{(\kappa\lambda)\mu\nu}} = 0\, , 
    \label{sympl sigma3 v1}\\
    \widetilde{\Sigma}^{[\kappa \lambda]\mu\nu} &= \frac{\partial \Lag_A}{\partial \widetilde{W}_{[\kappa\lambda]\mu\nu}} = \frac{27\sqrt{|\det g|}}{88 \cdot 16 \pi \Lambda}\cdot \left(  -\frac{8}{3}\right) \widetilde{W}^{[\mu\nu]\kappa\lambda}\, .
    \label{sympl sigma4 v1}
\end{align}
It is now clear that the affine Lagrangian~\eqref{lag A1} does not depend on the tensor $\widetilde{W}_{(\kappa\lambda)\mu\nu}$.

\subsection{Potential equations}
\label{pot eqs v1}
The first equation is obtained by taking the covariant divergence of $\chi^{\mu\nu}$~\eqref{rel konst1}:
\begin{align}
    \mnabla_{\nu}\chi^{\mu\nu } = \cJ^{\mu } =\frac{27\, \sqrt{|\det g|} }{88 \cdot 16 
 \pi \Lambda }\,  \left( \frac{712}{225}\, \mnabla_{\nu}F^{\mu\nu}  -\frac{16}{45}\, \mnabla_{\nu}W^{[\mu \nu]} \right) \,  .
\end{align}
Then, applying the formulae for $\cJ$~\eqref{dec cJ}, $\mnabla F$~\eqref{mnabla F}, and $\mnabla W$~\eqref{mnabla W3},~\eqref{mnabla W4}, yields the following equality:
\begin{align}
   - \frac{ 11 \Lambda}{5}\left(5h^{\mu}+ 4 A^{\mu}\right) &=  \mnabla_{\nu}\underbrace{\kolo{W}^{[\mu \nu]}}_{=0}+ \frac{89}{5}\, \left(  \mnabla^{\mu}\mnabla_{\sigma}A^{\sigma}+ A^{\sigma}\,\kolo{K}_{\sigma}^{\ \mu} - \mBox A^{\mu} \right) + \nonumber \\
    & \quad  -\frac{2}{3}\left(3\mnabla_{\nu}\mnabla_{\sigma}\widetilde{A}^{[\mu\nu]\sigma}+   \mnabla^{\mu}\mnabla_{\sigma}h^{\sigma}+ h^{\sigma}\,\kolo{K}_{\sigma}^{\ \mu} - \mBox h^{\mu}\right)  \, ,
    \label{pot cJ v1}
\end{align}
and   $ \kolo{W}^{[\mu \nu]}$ vanishes due to the symmetry -- see formula~\eqref{def metric W2}.

The second equation appears as a divergence of the momentum $\Omega$~\eqref{rel Omega Sigma}:
\begin{align}
\mnabla_{\nu} \Omega_{\kappa}^{\ \lambda\mu\nu} &= -2\mnabla_{\nu} \Sigma_{\kappa}^{\ (\lambda\mu)\nu} = -2\mnabla_{\nu}\widetilde{\Sigma}_{\kappa}^{\ (\lambda\mu)\nu} - \frac 14 \mnabla_{\kappa}\Sigma^{(\lambda\mu)} +\frac 38\left(\mnabla^{\lambda}\Sigma^{(\kappa\mu)} + \mnabla^{\mu}\Sigma^{(\kappa\lambda)} \right)  + \nonumber\\
&\quad +\frac{5}{12} \left(\mnabla^{\lambda}\Sigma^{[\kappa\mu]} + \mnabla^{\mu}\Sigma^{[\kappa\lambda]} \right) -g^{\lambda\mu}g_{\kappa\sigma}\mnabla_{\nu}\left( \frac 34 \Sigma^{(\sigma\nu)} + \frac 56 \Sigma^{[\sigma\nu]}\right)  +\nonumber \\
&\quad +\frac 18 \mnabla_{\nu}\left(\delta_{\kappa}^{\lambda}\Sigma^{(\mu\nu)} + \delta_{\kappa}^{\mu}\Sigma^{(\lambda\nu)} \right) + \frac 14 \mnabla_{\nu}\left(\delta_{\kappa}^{\lambda}\Sigma^{[\mu\nu]} + \delta_{\kappa}^{\mu}\Sigma^{[\lambda\nu]} \right)  = \nonumber \\
&=\frac{27\sqrt{|\det g|}}{88\cdot 16\pi \Lambda}\bigg\{-\frac 83 \mnabla_{\nu}\left(\widetilde{W}^{[\mu\nu]\lambda}_{\ \ \ \ \ \kappa} + \widetilde{W}^{[\lambda\nu]\mu}_{\ \ \ \ \ \kappa} \right) -\frac 49 \mnabla_{\kappa}W^{(\lambda\mu)} +\nonumber\\
&\quad +\frac23\left(\mnabla^{\lambda}W^{\mu}_{\ \kappa} + \mnabla^{\mu}W^{\lambda}_{\ \kappa} \right)  +\frac{8}{45} \left(\mnabla^{\lambda}F^{\mu}_{\ \kappa} + \mnabla^{\mu}F^{\lambda}_{\ \kappa} \right) +\nonumber\\
&\quad - g^{\lambda\mu}  \mnabla_{\nu}\left( \frac 43 W^{\nu}_{\ \kappa} + \frac{16}{45}F^{\nu}_{\ \kappa}\right) + \frac 29 \mnabla_{\nu}\left( \delta^{\lambda}_{\kappa}W^{(\mu\nu)} + \delta_{\kappa}^{\mu}W^{(\lambda\nu)}\right)+\nonumber \\
&\quad -  \frac 25 \mnabla_{\nu}\left( \delta^{\lambda}_{\kappa}W^{[\mu\nu]} + \delta_{\kappa}^{\mu}W^{[\lambda\nu]}\right) - \frac{8}{75}  \mnabla_{\nu}\left( \delta^{\lambda}_{\kappa}F^{\mu\nu} + \delta_{\kappa}^{\mu}F^{\lambda\nu}\right) \bigg\}\, ,
\end{align}
where the decomposition formula~\eqref{Sigma decomposition} of the momentum $\Sigma$, and field equations~(\ref{sympl sigma1 v1}-\ref{sympl sigma4 v1}) were applied.

\ 

It could be divided into two independent parts: the trace $\mnabla_{\nu}\cO_{\kappa}^{\ \nu}$ and the traceless part $\mnabla_{\nu}\mathfrak{O}_{\kappa}^{\ \lambda\mu\nu}$ (cf. \textbf{Lemma~\ref{lem dec Omega}}). The divergence of the  trace $\cO_{\kappa}^{\ \nu}$ is the following:
\begin{align}
    \mnabla_{\nu}\cO_{\kappa}^{\ \nu}&=  \frac{27\, \sqrt{|\det g|} }{88 \cdot 16 
 \pi \Lambda }\,   g_{\kappa\sigma} \mnabla_{\nu}\left[ -\frac{32}{9} W^{(\sigma\nu)} + \frac{16}{5} W^{[\sigma\nu]} +\frac{64}{75}F^{\sigma\nu}\right] \, .
\end{align}
The substitution of formulae for $\mnabla\cO$ ~\eqref{dec cO}, $\mnabla F$~\eqref{mnabla F} and $\mnabla W$~\eqref{mnabla W3},~\eqref{mnabla W4} induces:
\begin{align}
    \frac{22 \Lambda}{45}\left(5h_{\kappa} + 54 A_{\kappa}\right)  &=  \frac{8 }{5} \left(  \mnabla_{\kappa} \mnabla_{\sigma}A^{\sigma}+ A^{\sigma}\,\kolo{K}_{\sigma\kappa} - \mBox A_{\kappa} \right) + \frac{20}{9} \mnabla_{\kappa}\kolo{R}+ \nonumber\\
    & \quad -\frac{1}{3}\mnabla_{\lambda}\mnabla_{\mu}\left(19\tA^{\lambda\mu}_{\ \ \kappa}-\frac{20}{3}\tA^{\mu\lambda}_{\ \ \kappa}+\tA_{\kappa}^{\ \lambda\mu} \right) + \nonumber \\
    & \quad + \frac{2}{81}\left(131 \mnabla_{\kappa}\mnabla_{\lambda}h^{\lambda} + 181 h^{\sigma} \kolo{K}_{\sigma\kappa} + 19\mBox h_{\kappa}\right)\,.
    \label{pot cO v1}
\end{align}
The traceless part $\mnabla_{\nu}\mathfrak{O}_{\kappa}^{\ \lambda\mu\nu}$~\eqref{mathfrakO} is much more complicated:
 \begin{align}
        \mnabla_{\nu}\mathfrak{O}_{\kappa}^{\ \lambda\mu\nu} &= \mnabla_{\nu}\left[\Omega_{\kappa}^{\ \lambda\mu\nu} + \frac{1}{18}\left( \delta_{\kappa}^{\lambda}\, \cO^{\mu\nu} +\delta_{\kappa}^{\mu}\,\cO^{\lambda\nu} - 5g^{\lambda\mu}\, \cO_{\kappa}^{\ \nu} \right) \right] = \nonumber\\
        &=\frac{27\, \sqrt{|\det g|} }{88 \cdot 16  \pi \Lambda }\,   \left[-\frac 83 \mnabla_{\nu} \left(\widetilde{W}^{[\mu\nu]\lambda}_{\ \ \ \ \ \kappa} + \widetilde{W}^{[\lambda\nu]\mu}_{\ \ \ \ \ \kappa} \right) -\frac 49 \mnabla_{\kappa}W^{(\lambda\mu)} + \right. \nonumber\\
&\quad +\frac23\left(\mnabla^{\lambda}W^{\mu}_{\ \kappa} + \mnabla^{\mu}W^{\lambda}_{\ \kappa} \right)  +\frac{8}{45} \left(\mnabla^{\lambda}F^{\mu}_{\ \kappa} + \mnabla^{\mu}F^{\lambda}_{\ \kappa} \right) +\nonumber\\
&\quad - g^{\lambda\mu}  g_{\kappa\sigma}\mnabla_{\nu}\left( \frac {28}{81} W^{(\sigma\nu)}  - \frac 49 W^{[\sigma\nu]}- \frac{16}{15\cdot 9} F^{\sigma\nu} \right) + \nonumber\\
&\quad +\frac 2{81} \mnabla_{\nu}\left( \delta^{\lambda}_{\kappa}W^{(\mu\nu)} + \delta_{\kappa}^{\mu}W^{(\lambda\nu)}\right)  -  \frac 29 \mnabla_{\nu}\left( \delta^{\lambda}_{\kappa}W^{[\mu\nu]} + \delta_{\kappa}^{\mu}W^{[\lambda\nu]}\right)+\nonumber \\
 & \quad  \left.  - \frac{8}{15\cdot 9}  \mnabla_{\nu}\left( \delta^{\lambda}_{\kappa}F^{\mu\nu} + \delta_{\kappa}^{\mu}F^{\lambda\nu}\right)  \right] \, . 
    \end{align}     
The application of the formulae for $\mnabla\mathfrak{O}$~\eqref{dec tilOm}, $\mnabla F$~\eqref{mnabla F}, and $\mnabla W$~(\ref{mnabla W1}-\ref{mnabla W4}), leads to the  very long and complicated equation, which could be symbolically written in the following way:
\begin{align}
    \tA^{\kappa\lambda\mu} = \textbf{F}_1(\mnabla^2 \tA^{\kappa\lambda\mu},\mnabla^2 A^{\kappa},\mnabla^2 h^{\kappa}, \mnabla \kolo{K}_{\mu\nu}, \mnabla\kolo{R})\, ,
    \label{pot tOm v1}
\end{align}
where $\textbf{F}_1$ denotes a linear function (with respect to all arguments) depending on second-order derivatives of potentials $\tA^{\kappa\lambda\mu}, A^{\kappa}, h^{\kappa}$ and first-order derivatives of the metric curvature components: $\kolo{R},\kolo{K}_{\mu\nu}$.

The above equations~\eqref{pot cJ v1},~\eqref{pot cO v1}, and~\eqref{pot tOm v1}, referred to as “potential equations", explicitly illustrate how much more complicated the theory becomes when the traceless part of the Riemann tensor $W^{\kappa}_{\ \lambda\mu\nu}$ is included -- cf. the potential equation \eqref{jdshdfbs2}  for the theory of the full Ricci tensor.

\section{Variant \texorpdfstring{$V_6$}{V6}}
\label{VAR V6}

\subsection{Lagrangian}
By the scheme from \textbf{Chapter~\ref{scheme}}, the affine Lagrangian~\eqref{lagF} is given by the square root of four Riemann tensors contracted with two Levi-Civita symbols, denoted as~$RRRR$ (cf. variant $V_6$~\eqref{w6}):
\begin{align}
 RRRR &=  R^{\alpha}_{\ \beta \mu_1 \mu_2}\, R^{\kappa}_{\ \lambda \mu_3 \mu_4}\, \epsilon^{\mu_1 \mu_2 \mu_3 \mu_4} \ R^{\beta}_{\ \alpha \nu_1 \nu_2}\, R^{\lambda}_{\ \kappa \nu_3 \nu_4}\, \epsilon^{\nu_1 \nu_2 \nu_3 \nu_4} =\nonumber\\
 &=KKKK  +KKKW + KKFF + KKFW + KKWW   + o(F,W)\, ,   
    \label{LagF6}    
\end{align}
where $o(F,W)$ denotes high-order terms in $F$ and $W$, and
\begin{align}
KKKK&=\frac{128}{27}\, \det K\, , \label{KKKKv6}  \\
KKKW&=\frac{32}{27}\,  K_{\alpha \beta}\,  K_{\gamma \delta}\,  K_{\mu_1\mu_2}\,  W^{\alpha}_{\ \mu_3\mu_4\nu}\,  \epsilon^{\beta\gamma\mu_1\mu_3}\, \epsilon^{\delta \mu_2\mu_4\nu} \, ,\label{KKKWv6} \\
KKFF &= -F_{\alpha \beta}\, F_{\gamma\delta}\,  K_{\mu_1\mu_2}\,  K_{\mu_3\mu_4}\,  \left[\frac{32}{225}\,\epsilon^{\alpha\gamma\mu_1\mu_3} \,\epsilon^{\beta\delta \mu_2\mu_4} +\frac{64}{75}\, \epsilon^{\alpha\beta\mu_1\mu_3}\, \epsilon^{\gamma\delta \mu_2\mu_4}\right] \, ,  \\
KKFW &=\frac{32}{45}\, F_{\alpha \beta}\, K_{\gamma \delta}\, K_{\mu_1\mu_2}\, W^{\alpha}_{\ \mu_3\mu_4\nu}\,    \epsilon^{\beta\gamma\mu_1\mu_3}\, \epsilon^{\delta \mu_2\mu_4\nu}  \, ,  \\
KKWW&= K_{\alpha \beta}\, K_{\gamma \delta} \,W^{\alpha}_{\ \mu_1\mu_2\mu_3} \, W^{\gamma}_{\ \mu_4\nu_1\nu_2}\,\left[\frac 89\,  \epsilon^{\beta\mu_1\nu_1\nu_2} \, \epsilon^{\delta \mu_2 \mu_3\mu_4} - \frac 89\,  \epsilon^{\beta\delta\mu_1\mu_4} \, \epsilon^{ \mu_2 \mu_3\nu_1\nu_2}\right] +\nonumber \\
& \quad -\frac{8}{9}\,   K_{\alpha \beta}\,  K_{\gamma \delta}\,  W^{\mu_1}_{\ \mu_2\mu_3\mu_4}\,  W^{\mu_2}_{\ \mu_1 \nu_1 \nu_2}\, \epsilon^{\alpha\gamma\mu_3\mu_4}\, \epsilon^{\beta\delta \nu_1\nu_2} \,  .
\label{rozklad w6}
\end{align} 
Therefore, the  affine Lagrangian $\Lag_{A}$~\eqref{lagmodel}  of this theory  is given by:
\begin{align}
    \Lag_{A}=\alpha\, \sqrt{|KKKK  +KKKW + KKFF + KKFW + KKWW|}\, . 
    \label{lag A6}
\end{align}
The equality~\eqref{KKKKv6} determines two characteristic constants $\sigma$ and $\gamma$~\eqref{reldetK}, whereas the global constant $\alpha$~\eqref{alpha const} is  fitted to reconstruct the standard  Einstein equation with the cosmological constant $\Lambda$ for the unperturbed theory:
\begin{align}
    \sigma&=1\, ,& \gamma^2 &=\frac{128}{27}\, , & \alpha&= \frac{1}{8\pi\Lambda } \,  \sqrt{\frac{27}{128}}\,.
    \label{const6}
\end{align}

\subsection{Non-metricity equation}
\label{sub noneq v6}
The discussion about the non-metricity equation is exactly the same as for the variant~$V_1$ presented in the \textbf{Chapter~\ref{non-met eq v1}}.

\subsection{Einstein equation}
 As in the previous example, the unperturbed solution $\Kz$ must be the Einstein $\Lambda$-vacuum equation~\eqref{eq0}:
\begin{align}
    \Kz_{\mu\nu} =  \frac{1}{8\pi \alpha \gamma}\, g_{\mu\nu}=\Lambda\, g_{\mu\nu}\, .
    \label{Kzv6}
\end{align}

The first-order correction $\Ko$ reads as follows~\eqref{eq1}:
\begin{align}
     \frac12\, \left(ggg \Ko + ggg W \right)\, g^{-1}   =  g g \Ko + g g W\, . 
\end{align}
where
\begin{align}
    ggg \Ko &=-\frac{128}{27}\, \Ko^{\alpha}_{\ \alpha}\, |\det g|\,  , \\
     ggg W &=  0\, , \\
      g g \Ko &=-\frac{128}{27}\,\left(\Ko^{\alpha}_{\ \alpha}\, g^{\mu\nu} - \Ko^{\mu\nu} \right)\, |\det g| \, , \\
      g g W  &=   0\,.
\end{align}
As it was in the previous variant, there appears a term $\Kz\Kz\Kz W$~\eqref{KKKWv6}, but the Einstein equation $\Kz=\Lambda \, g$ implies that the associated  terms $gggW$ and $ggW$ vanish -- see~\eqref{gggW van}. Thus, as in the theory of the full Ricci tensor (see \textbf{Chapter~\ref{affine ricci}}), $\Ko_{\mu\nu}$ vanishes:
\begin{align}
    \Ko_{\mu\nu}=0\, .
\end{align}

The first non-trivial corrections appear at the level of the second-order perturbation~\eqref{eq2}:
\begin{align}
&\frac{  1}{2}\, \left(  \Lambda\,  g g g \Kt  +  g g  FF+ g g F W+ g g  WW \right)\, g^{-1}= \nonumber  \\
&=   \Lambda\, g g \Kt  +  g FF+     g  FW+     g WW \, ,
\label{eq2v6}
\end{align}
where:
\begin{align}
    g g g \Kt&=-\frac{128}{27}\, \Kt^{\alpha}_{\ \alpha}\, |\det g|\,    , \\
    g g  FF&= \frac{832}{225}\, F^{\alpha\beta}\, F_{\alpha\beta}\, |\det g|\,  ,
     \label{ggFF v6}\\
    g g F W&= \frac{128}{45}\, F_{\alpha\beta}\, W^{\alpha  \beta} \, |\det g|\,  , 
    \label{ggFW v6} \\
    g g  WW&= \left(\frac{64}{9}\,  W_{(\alpha \beta) }\, W^{ \alpha\beta }-\frac{32}9\,   W_{\alpha\beta \kappa\lambda}\, W^{\kappa\lambda \alpha\beta} +\right. \nonumber\\
    & \quad  \left.- \frac{16}{9}\,   W_{\alpha\beta \kappa\lambda}\, W^{ \alpha\beta\kappa\lambda} + \frac{16}{3}\,   W_{\alpha\beta \kappa\lambda}\, W^{ \beta\alpha\kappa\lambda}\right)\, |\det g|  \,  ,    
      \label{ggWW v6}\\
    g g \Kt  &=-\frac{128}{27}\,\left(\Kt^{\alpha}_{\ \alpha}\, g^{\mu\nu} -  \Kt^{\mu\nu} \right)\, |\det g|\,  ,  \\
     g FF &=- \frac{1664}{225}\, \left(F^{\mu\alpha}\, F^{\nu}_{\ \alpha} - \frac{1}{2}\, F^{\alpha\beta}\, F_{\alpha\beta}\, g^{\mu\nu}\right)\, |\det g|\,  , \\
     g  FW &=-\frac {128}{45}\, \left(F^{\ (\mu|}_{ \alpha}\, W^{\alpha |\nu)}  +  F_{\alpha\beta}\, W^{\alpha (\mu\nu)\beta} -  F_{\alpha\beta}\, W^{\alpha \beta} \, g^{\mu\nu} \right)\, |\det g|\, , \\
      g WW &= \frac{32}{9}\, \left(W_{\alpha \ \ \beta\gamma}^{\ (\mu} \, W^{\nu)\alpha\beta\gamma} +2\, W_{\alpha\beta\gamma}^{\ \ \  (\mu}\, W^{\nu) \gamma\alpha\beta} +\right. \nonumber \\
      & \quad +4\, W_{(\alpha\beta)}\, W^{(\mu|\alpha\beta|\nu)} + 2\, W^{\mu}_{\ \alpha}\, W^{\nu \alpha} + 2\, W_{\alpha}^{\ (\mu}\, W^{\nu)\alpha} +\nonumber \\
      & \quad  \left. -2\, W_{\alpha\beta\gamma}^{\ \ \ \ \mu}\, W^{\beta\alpha\gamma\nu} -  W^{\mu}_{\ \alpha\beta\gamma}\, W^{\nu\alpha\beta\gamma}+   g^{\mu\nu}\, \,W_{\alpha\beta\kappa\lambda}\,W^{\beta\alpha\kappa\lambda} \right)\, |\det g|\,  .
\end{align}
 and tensor $W^{\kappa}_{\ \nu}$~\eqref{def ten W2} is defined as  a remaining metric trace of $W^{\kappa}_{\ \lambda\mu\nu}$. Contracting the  equation~\eqref{eq2v6} with the metric tensor $g_{\mu\nu}$ implies:
\begin{align}
    \Kt^{\alpha}_{\ \alpha}=0\, .
    \label{trKtv6}
\end{align}
Then, the second-order correction $\Kt_{\mu\nu}$ equals:
\begin{align}
 \frac{128}{27}\,\Lambda\ \Kt^{\mu\nu}&=      \frac{1664}{225}\, \left(F^{\mu\alpha}\, F^{\nu}_{\ \alpha}-\frac 14\,  F^{\alpha\beta}\, F_{\alpha\beta}\, g^{\mu\nu} \right) +\nonumber  \\
&  \quad +\frac {128}{45}\, \left( F^{\ (\mu|}_{  \alpha}\, W^{\alpha |\nu)}  +  F_{\alpha\beta}\, W^{\alpha (\mu\nu)\beta} - \frac{1}{2}  F_{\alpha\beta}\, W^{\alpha \beta} \, g^{\mu\nu}\right) +  \nonumber\\
& \quad  + \frac{32}{9}\left(2 W^{\alpha\beta\gamma\mu}W_{\beta\alpha\gamma}^{\ \ \ \nu} + W^{\mu}_{\ \alpha\beta\gamma}W^{\nu \alpha\beta\gamma} -2 W^{\mu}_{\ \alpha}W^{\nu\alpha} \right. -2W^{\alpha(\mu} W^{\nu)}_{\ \ \alpha}+\nonumber\\
& \quad   - W^{\ (\mu}_{\alpha \ \ \beta\gamma} W^{\nu)\alpha\beta\gamma} -2W^{\alpha\beta\gamma(\mu} W^{\nu)\gamma\alpha\beta}  - 4W_{\alpha\beta} W^{(\mu|\alpha\beta|\nu)}  + \nonumber \\
& \quad \left.+W_{(\alpha\beta)}W^{\alpha\beta}\,g^{\mu\nu} -\frac 12\, W_{(\alpha\beta) \kappa\lambda}W^{\alpha\beta\kappa\lambda}\, g^{\mu\nu} - \frac{1}{2} W_{\alpha\beta\kappa\lambda}W^{\kappa\lambda\alpha\beta}\, g^{\mu\nu} \right)\, .
     \label{v6 k2}
\end{align} 
 Therefore, the  Einstein equation~\eqref{eineqgen} reads:
 \begin{align}
     K_{\mu\nu} = \Lambda\, g_{\mu\nu} + \Kt_{\mu\nu}\, .
     \label{ein eq0 v6}
 \end{align}
However, this equation is not exactly in the standard form of the Einstein equation~\eqref{scheinmet}, due to the presence of the general, non-metric, symmetric Ricci tensor $K_{\mu\nu}$ on the left-hand side, which decomposes into the purely metric part $\kolo{K}$ and the remainder $Q_{\mu\nu}$~\eqref{tensor Q}, as already presented in \textbf{Chapter~\ref{eineq v1}}. Then, the Einstein equation~\eqref{scheinmet} takes the following form:
\begin{align}
    \kolo{K}_{\mu\nu} = \Lambda\, g_{\mu\nu} + \Kt_{\mu\nu}-Q_{\mu\nu}\, ,
    \label{eqq2v6}
\end{align} 
or, using the Einstein tensor $\kolo{G}_{\mu\nu}$~\eqref{def ein tensor}:
\begin{align}
   \kolo{G}_{\mu\nu}   = -\Lambda\, g_{\mu\nu} +\Kt_{\mu\nu} - \left( Q_{\mu\nu} - \frac12\, g_{\mu\nu}\, Q_{\sigma}^{\ \sigma}\right)\, .
   \label{eqq21v6}
\end{align}
Due to the fact that the non-metricity tensor $N^{\kappa}_{\ \lambda \mu}$ is quite complicated in this theory~\eqref{nsaifhyv1}, the simplification presented in the Ricci tensor theory cannot be performed -- cf. equation~\eqref{eqq2}.

 \subsection{Effective cosmological parameter}

The cosmological parameter $\Lambda_{\rm eff}$~\eqref{def lam eff} associated with this theory is given by
\begin{align}
\Lambda_{\rm eff} := \frac{1}{4} \, \kolo{K}_{\mu\nu} g^{\mu\nu} = \Lambda - \frac{1}{4} \, Q_{\alpha}^{\ \alpha} \, ,
\label{lam effv6}
\end{align}
and coincides with the expression for $\Lambda_{\rm eff}$ presented in \textbf{Chapter~\ref{eff lam v1}}, since $Q_{\alpha}^{\ \alpha}$ depends only on the non-metricity tensor, which is identical in both theories.

\subsection{Field equation for the skew-symmetric Ricci tensor}
The  field equation for  $F_{\mu\nu}$ was determined by the symplectic relation~\eqref{rel chi} and was schematically  derived in the \textbf{Chapter~\ref{scheme}} in formula~\eqref{chi}:
\begin{align}
    \chi&= \frac{\sigma \sigma_g }{ 16 
 \pi \Lambda \gamma^2 \sqrt{|\det g|}}\, \left( ggF + ggW\right)   \, ,
 \label{rel konst16}
\end{align}
where
\begin{align}
     g g F &=  \frac{1664}{225}\, F^{\mu\nu}\, |\det g|\, , \\
     g g W  &=  \frac{128}{45}\, W^{[\mu \nu]} \, |\det g|\,.
\end{align}
The metric signature is assumed to be Lorentzian, thus, $\sigma_g=-1$. Upon substituting all  characteristic constants~\eqref{const6}, the above  formula~\eqref{rel konst16} stays:
\begin{align}
     \chi^{\mu\nu} =- \frac{27\, \sqrt{|\det g|} }{128 \cdot 16 
 \pi \Lambda }\,  \left(\frac{1664}{225}\, F^{\mu\nu} +\frac{128}{45}\, W^{[\mu \nu]} \right) \,  .
     \label{rel konst6}
\end{align}
and will be called \textit{the constitutive relation between} $\chi$ \textit{and} $F$. 

In contrast to the theory based on the full Ricci tensor from \textbf{Chapter~\ref{feq F v1}}, it is not straightforward to derive an equation for the potential $A_{\mu}$~\eqref{jdshdfbs2}, due to the much more complicated structure of the non-metricity tensor~\eqref{nsaifhyv1}. Specifically, the presence of covariant derivatives $\mnabla \Omega$ significantly complicates the calculations.

\subsection{Field equation for the traceless Riemann tensor}
\label{chap sigma v6}
The  field equation for  $W^{\kappa}_{\ \lambda\mu\nu}$ was determined by the symplectic relation~\eqref{rel Sigma} and was schematically  derived in the \textbf{Chapter~\ref{scheme}} in formula~\eqref{sigma}:
\begin{align}
    \Sigma&= \frac{\sigma \sigma_g  }{ 16 
 \pi \Lambda \gamma^2 \sqrt{|\det g|}}\,  \left( gg F + gg W \right)\, ,
 \label{rel Sigma16}
\end{align}
where
\begin{align}
     g g F &=\frac{\partial}{\partial W}(ggFW)   \, , &
     g g W  &=\frac{\partial}{\partial W}(ggWW) \,.
\end{align}
However, it was split into four equations (\ref{sigma feq1}–\ref{sigma feq4}), each corresponding to an independent component of the tensor $W$ — see \textbf{Lemma~\ref{lemm W dec}} and formula~\eqref{W decomposition}. To simplify the derivation of these equations, the quantities $ggFW$~\eqref{ggFW v6} and $ggWW$~\eqref{ggWW v6} above are rewritten, taking into account the decomposition of $W$ given in~\eqref{W decomposition}:
\begin{align}
   g g F W&= \frac{128}{45}\, F_{\alpha\beta}\, W^{[\alpha  \beta]} \, |\det g|\,  , 
    \label{ggFW2 v6} \\
    g g  WW&= \left(  -\frac{16}{27}W_{[\alpha\beta]}W^{\alpha\beta}+\frac{8}{3}\,  W_{(\alpha \beta) }\, W^{ \alpha\beta }- \frac{32}{9}\widetilde{W}_{[\alpha\beta]\kappa\lambda} \widetilde{W}^{ \kappa\lambda\alpha \beta}  + \right. \nonumber\\
    & \quad   \left.  - \frac{64}{9} \widetilde{W}_{[\alpha\beta]\kappa\lambda} \widetilde{W}^{\alpha\beta \kappa\lambda } + \frac{32}{9} \widetilde{W}_{(\alpha\beta)\kappa\lambda} \widetilde{W}^{\alpha\beta \kappa\lambda } \right)\, |\det g|  \,  ,    
      \label{ggWW2 v6}
\end{align}
Then, implying  all characteristic constants~\eqref{const1} with $\sigma_g=-1$, the corresponding field equations (\ref{sigma feq1}–\ref{sigma feq4})  are the following:

\begin{align}
    \frac{5}{6}\, \Sigma^{[\mu\nu]} &= \frac{\partial \Lag_A}{\partial W_{[\mu\nu]}} = -\frac{27\sqrt{|\det g|}}{128 \cdot 16 \pi \Lambda}\left(-\frac{32}{27} W^{[\mu\nu]} + \frac{128}{45}F^{\mu\nu}\right)\, ,
    \label{sympl sigma1 v6}\\
    \frac{3}{4}\, \Sigma^{(\mu\nu)} &= \frac{\partial \Lag_A}{\partial W_{(\mu\nu)}} = -\frac{27\sqrt{|\det g|}}{128 \cdot 16 \pi \Lambda} \cdot \frac{16}{3} W^{(\mu\nu)}\, , 
    \label{sympl sigma2 v6}\\
    \widetilde{\Sigma}^{(\kappa \lambda)\mu\nu} &= \frac{\partial \Lag_A}{\partial \widetilde{W}_{(\kappa\lambda)\mu\nu}} =-\frac{27\sqrt{|\det g|}}{128 \cdot 16 \pi \Lambda} \cdot  \frac{64}{9}   \widetilde{W}^{(\kappa \lambda)\mu\nu}\, , 
    \label{sympl sigma3 v6}\\
    \widetilde{\Sigma}^{[\kappa \lambda]\mu\nu} &= \frac{\partial \Lag_A}{\partial \widetilde{W}_{[\kappa\lambda]\mu\nu}} = -\frac{27\sqrt{|\det g|}}{128 \cdot 16 \pi \Lambda}  \left(  -\frac{64}{9}\widetilde{W}^{[\mu\nu]\kappa\lambda} - \frac{128}{9} \widetilde{W}^{[\kappa \lambda]\mu\nu} \right)  \, .
    \label{sympl sigma4 v6}
\end{align}
In the opposite to the Variant $V_1$, this theory depends upon all components of $W_{\kappa\lambda\mu\nu}$ -- cf. formula \eqref{sympl sigma3 v1}.

\subsection{Potential equations}
 
\label{pot eqs v6}
The first equation is obtained by taking the covariant divergence of $\chi^{\mu\nu}$~\eqref{rel konst6}:
\begin{align}
    \mnabla_{\nu}\chi^{\mu\nu } = \cJ^{\mu } =- \frac{27\, \sqrt{|\det g|} }{128 \cdot 16 
 \pi \Lambda }\,  \left(\frac{1664}{225}\mnabla_{\nu}F^{\mu\nu} +\frac{128}{45}\mnabla_{\nu} W^{[\mu \nu]} \right)  \,  .
\end{align}
Then, applying the formulae for $\cJ$~\eqref{dec cJ}, $\mnabla F$~\eqref{mnabla F}, and $\mnabla W$~\eqref{mnabla W3},~\eqref{mnabla W4}, yields the following equality:
\begin{align}
     \Lambda \left(5h^{\mu}+ 4 A^{\mu}\right) &=  \mnabla_{\nu}\underbrace{\kolo{W}^{[\mu \nu]}}_{=0}+ 13\, \left(  \mnabla^{\mu}\mnabla_{\sigma}A^{\sigma}+ A^{\sigma}\,\kolo{K}_{\sigma}^{\ \mu} - \mBox A^{\mu} \right) + \nonumber \\
    & \quad  +\frac53 \left(3\mnabla_{\nu}\mnabla_{\sigma}\widetilde{A}^{[\mu\nu]\sigma}+   \mnabla^{\mu}\mnabla_{\sigma}h^{\sigma}+ h^{\sigma}\,\kolo{K}_{\sigma}^{\ \mu} - \mBox h^{\mu}\right)  \, ,
    \label{pot cJ v6}
\end{align}
and $\kolo{W}^{[\mu \nu]}$ vanishes due to the symmetry -- see formula~\eqref{def metric W2}.

The second equation appears as a divergence of the momentum $\Omega$~\eqref{rel Omega Sigma}:
\begin{align}
\mnabla_{\nu} \Omega_{\kappa}^{\ \lambda\mu\nu} &= -2\mnabla_{\nu} \Sigma_{\kappa}^{\ (\lambda\mu)\nu} = -2\mnabla_{\nu}\widetilde{\Sigma}_{\kappa}^{\ (\lambda\mu)\nu} - \frac 14 \mnabla_{\kappa}\Sigma^{(\lambda\mu)} +\frac 38\left(\mnabla^{\lambda}\Sigma^{(\kappa\mu)} + \mnabla^{\mu}\Sigma^{(\kappa\lambda)} \right)  + \nonumber\\
&\quad +\frac{5}{12} \left(\mnabla^{\lambda}\Sigma^{[\kappa\mu]} + \mnabla^{\mu}\Sigma^{[\kappa\lambda]} \right) -g^{\lambda\mu}g_{\kappa\sigma}\mnabla_{\nu}\left( \frac 34 \Sigma^{(\sigma\nu)} + \frac 56 \Sigma^{[\sigma\nu]}\right)  +\nonumber \\
&\quad +\frac 18 \mnabla_{\nu}\left(\delta_{\kappa}^{\lambda}\Sigma^{(\mu\nu)} + \delta_{\kappa}^{\mu}\Sigma^{(\lambda\nu)} \right) + \frac 14 \mnabla_{\nu}\left(\delta_{\kappa}^{\lambda}\Sigma^{[\mu\nu]} + \delta_{\kappa}^{\mu}\Sigma^{[\lambda\nu]} \right)  = \nonumber \\
&=-\frac{27\sqrt{|\det g|}}{128\cdot 16\pi \Lambda} \, g_{\kappa\sigma} \bigg\{-\frac{64}9\mnabla_{\nu} \left(\widetilde{W}^{(\sigma\lambda)\mu\nu} -\widetilde{W}^{[\mu\nu]\sigma\lambda} -2\widetilde{W}^{[\sigma\lambda]\mu\nu} + \right. \nonumber \\
&\quad \left. +\widetilde{W}^{(\sigma\mu)\lambda\nu} - \widetilde{W}^{[\lambda\nu]\sigma\mu} - 2\widetilde{W}^{[\sigma\mu]\lambda\nu} \right) -\frac {16}9 \mnabla^{\sigma} W^{(\lambda\mu)} + \nonumber\\
&\quad +\frac83\left(\mnabla^{\lambda}W^{(\mu\sigma)} + \mnabla^{\mu}W^{(\lambda\sigma)} \right) -\frac{16}{27} \left(\mnabla^{\lambda}W^{[\mu\sigma]} + \mnabla^{\mu}W^{[\lambda\sigma]} \right)  +\nonumber\\
&\quad  +\frac{64}{45} \left(\mnabla^{\lambda}F^{\mu\sigma} + \mnabla^{\mu}F^{\lambda\sigma} \right) - g^{\lambda\mu}  \mnabla_{\nu}\left( \frac {16}3 W^{(\sigma\nu)} - \frac {32}{27} W^{[\sigma\nu]} + \frac{128}{45}F^{\kappa\nu} \right)+\nonumber \\
&\quad  + \frac 89 \mnabla_{\nu}\left( g^{\sigma\lambda} W^{(\mu\nu)} + g^{\sigma\mu} W^{(\lambda\nu)}\right)-  \frac {16}{45} \mnabla_{\nu}\left( g^{\sigma \lambda} W^{[\mu\nu]} + g^{\sigma\mu} W^{[\lambda\nu]}\right) + \nonumber\\
&\quad + \frac{64}{75}  \mnabla_{\nu}\left( \delta^{\lambda}_{\kappa}F^{\mu\nu} + \delta_{\kappa}^{\mu}F^{\lambda\nu}\right) \bigg\}\, ,
\end{align}
where the decomposition formula~\eqref{Sigma decomposition} of the momentum $\Sigma$, and field equations (\ref{sympl sigma1 v6}-\ref{sympl sigma4 v6}) were applied.

\ 

It could be divided into two independent parts: the trace $\mnabla_{\nu}\cO_{\kappa}^{\ \nu}$ and the traceless part $\mnabla_{\nu}\mathfrak{O}_{\kappa}^{\ \lambda\mu\nu}$ -- cf. \textbf{Lemma~\ref{lem dec Omega}}. The divergence of the  trace $\cO_{\kappa}^{\ \nu}$ is the following:
\begin{align}
    \mnabla_{\nu}\cO_{\kappa}^{\ \nu}&=  -\frac{27\, \sqrt{|\det g|} }{128 \cdot 16 
 \pi \Lambda }\,   g_{\kappa\sigma} \mnabla_{\nu}\left[ -\frac{128}{9} W^{(\sigma\nu)} + \frac{128}{45} W^{[\sigma\nu]} -\frac{4\cdot 128 }{75}F^{\sigma\nu}\right] \, .
\end{align}
The substitution of formulae for $\mnabla\cO$ ~\eqref{dec cO}, $\mnabla F$~\eqref{mnabla F} and $\mnabla W$~\eqref{mnabla W3},~\eqref{mnabla W4} induces:
\begin{align}
    - 2 \Lambda\left(5h_{\kappa} + 54 A_{\kappa}\right)  &= -  36  \left(  \mnabla_{\kappa} \mnabla_{\sigma}A^{\sigma}+ A^{\sigma}\,\kolo{K}_{\sigma\kappa} - \mBox A_{\kappa} \right) + 25 \mnabla_{\kappa}\kolo{R}+ \nonumber\\
    & \quad -5 \mnabla_{\lambda}\mnabla_{\mu}\left(9\tA^{\lambda\mu}_{\ \ \kappa}-5 \tA^{\mu\lambda}_{\ \ \kappa}+ 6\tA_{\kappa}^{\ \lambda\mu} \right) + \nonumber \\
    & \quad + \frac{5}{9}\left(34 \mnabla_{\kappa}\mnabla_{\lambda}h^{\lambda} + 59 h^{\sigma} \kolo{K}_{\sigma\kappa} + 41\mBox h_{\kappa}\right)\,.
    \label{pot cO v6}
\end{align} 
The traceless part $\mnabla_{\nu}\mathfrak{O}_{\kappa}^{\ \lambda\mu\nu}$~\eqref{mathfrakO} is much more complicated:
 \begin{align}
        \mnabla_{\nu}\mathfrak{O}_{\kappa}^{\ \lambda\mu\nu} &= \mnabla_{\nu}\left[\Omega_{\kappa}^{\ \lambda\mu\nu} + \frac{1}{18}\left( \delta_{\kappa}^{\lambda}\, \cO^{\mu\nu} +\delta_{\kappa}^{\mu}\,\cO^{\lambda\nu} - 5g^{\lambda\mu}\, \cO_{\kappa}^{\ \nu} \right) \right] = \nonumber\\
        &=-\frac{27\sqrt{|\det g|}}{128\cdot 16\pi \Lambda} \, g_{\kappa\sigma} \bigg\{-\frac{64}9\mnabla_{\nu} \left(\widetilde{W}^{(\sigma\lambda)\mu\nu} -\widetilde{W}^{[\mu\nu]\sigma\lambda} -2\widetilde{W}^{[\sigma\lambda]\mu\nu} + \right. \nonumber \\
&\quad \left. +\widetilde{W}^{(\sigma\mu)\lambda\nu} - \widetilde{W}^{[\lambda\nu]\sigma\mu} - 2\widetilde{W}^{[\sigma\mu]\lambda\nu} \right) -\frac {16}9 \mnabla^{\sigma} W^{(\lambda\mu)} + \nonumber\\
&\quad +\frac83\left(\mnabla^{\lambda}W^{(\mu\sigma)} + \mnabla^{\mu}W^{(\lambda\sigma)} \right) -\frac{16}{27} \left(\mnabla^{\lambda}W^{[\mu\sigma]} + \mnabla^{\mu}W^{[\lambda\sigma]} \right)  +\nonumber\\
&\quad  +\frac{64}{45} \left(\mnabla^{\lambda}F^{\mu\sigma} + \mnabla^{\mu}F^{\lambda\sigma} \right) - g^{\lambda\mu}  \mnabla_{\nu}\left( \frac {7\cdot 16}{81} W^{(\sigma\nu)} - \frac {32}{81} W^{[\sigma\nu]} + \frac{128}{3\cdot 45}F^{\kappa\nu} \right)+\nonumber \\
&\quad  + \frac 8{81} \mnabla_{\nu}\left( g^{\sigma\lambda} W^{(\mu\nu)} + g^{\sigma\mu} W^{(\lambda\nu)}\right)-  \frac {16}{81} \mnabla_{\nu}\left( g^{\sigma \lambda} W^{[\mu\nu]} + g^{\sigma\mu} W^{[\lambda\nu]}\right) + \nonumber\\
&\quad + \frac{64}{3\cdot 45}  \mnabla_{\nu}\left( \delta^{\lambda}_{\kappa}F^{\mu\nu} + \delta_{\kappa}^{\mu}F^{\lambda\nu}\right) \bigg\}\, .
    \end{align}     
The application of the formulae for $\mnabla\mathfrak{O}$~\eqref{dec tilOm}, $\mnabla F$~\eqref{mnabla F}, and $\mnabla W$~(\ref{mnabla W1}-\ref{mnabla W4}), leads to the  very long and complicated equation, which could be symbolically written in the following way:
\begin{align}
    \tA^{\kappa\lambda\mu} = \textbf{F}_6(\mnabla^2 \tA^{\kappa\lambda\mu},\mnabla^2 A^{\kappa},\mnabla^2 h^{\kappa}, \mnabla \kolo{K}_{\mu\nu}, \mnabla\kolo{R})\, ,
    \label{pot tOm v6}
\end{align}
where $\textbf{F}_6$ denotes a linear function (with respect to all arguments) depending on second-order derivatives of potentials $\tA^{\kappa\lambda\mu}, A^{\kappa}, h^{\kappa}$ and first-order derivatives of the metric curvature components: $\kolo{R},\kolo{K}_{\mu\nu}$.

The above equations~\eqref{pot cJ v1},~\eqref{pot cO v1}, and~\eqref{pot tOm v1}, referred to as “potential equations", explicitly illustrate how much more complicated the theory becomes when the traceless part of the Riemann tensor $W^{\kappa}_{\ \lambda\mu\nu}$ is included -- cf. the potential equation \eqref{jdshdfbs2}  for the theory of the full Ricci tensor.

\section{Theory  of the full Ricci tensor with a fixed background field}
\label{fixed back}
\sectionmark{Full Ricci tensor with a background field}

\subsection{Lagrangian}
In the previous chapters, affine theories based on the full Riemann tensor were presented. It was shown that treating the algebraically traceless Riemann tensor $W^{\kappa}_{\ \lambda\mu\nu}$ as a dynamical field leads to a highly complex theory with non-trivial equations, whereas the theory involving only the Ricci tensor, discussed in \textbf{Chapter~\ref{affine ricci}}, remains relatively simple. The tensor $W^{\kappa}_{\ \lambda\mu\nu}$ is believed to describe the dark matter field, which at our scale is weak and slowly varying. 

As a simplified model, a theory with a Lagrangian based on the full Riemann tensor is presented, in which $W^{\kappa}_{\ \lambda\mu\nu}$ is treated as a fixed “background” field.  This means that the variational structure corresponds to that of the full Ricci tensor theory, while the field equations remain those of the full Riemann curvature theory. Such a model represents a compromise between a purely mathematical formulation and a phenomenological description. For further simplification, the following affine Lagrangian will be used:
\begin{align}
    \Lag_A:= \frac{1}{8\pi \Lambda}\, \sqrt{\left|\det K + KKFF+KKFW +KKWW  \right|}\, ,
    \label{lagA ricci W}
\end{align}
where
\begin{align}
    KKFF &= C_{F}\, K_{\mu_1\nu_1} \,K_{\mu_2\nu_2}\, F_{\mu_3\nu_3} \,  F_{\mu_4\nu_4}\, \epsilon^{\mu_1\mu_2\mu_3\mu_4} \, \epsilon^{\nu_1\nu_2\nu_3\nu_4}\, ,\label{KKFF v10} \\
    KKFW &= I_{FW}\, K_{\mu_1\nu_1}\,K_{\mu_2\nu_2}\, F_{\mu_3\alpha}\, W^{\alpha}_{\ \nu_3\mu_4\nu_4}\, \epsilon^{\mu_1\mu_2\mu_3\mu_4} \, \epsilon^{\nu_1\nu_2\nu_3\nu_4}\, ,\label{KKFW v10} \\
    KKWW&= C_{W}\, K_{\mu_1\nu_1} \,K_{\mu_2\nu_2}\, W^{\alpha}_{\ \mu_3\beta\nu_3}\, W^{\beta}_{\ \mu_4\alpha\nu_3}\, \epsilon^{\mu_1\mu_2\mu_3\mu_4} \, \epsilon^{\nu_1\nu_2\nu_3\nu_4}\, ,\label{KKWW v10}
\end{align}
where $C_F$, $I_{FW}$ and $C_W$ are numerical constants.  
The above theory is precisely a phenomenological generalisation of the theory based on the full Ricci tensor (cf.~\textbf{Chapter~\ref{affine ricci}}), with corrections drawn from the theory of the full Riemann tensor (cf.~\textbf{Chapter~\ref{how to construct}}). In fact, setting $C_F = \frac{1}{4}$ and $I_{FW}=C_W = 0$ recovers the affine Lagrangian~\eqref{lag iksdfbosgv}. Of course, the above proposition does not encompass all possible forms of interaction between the curvature tensors $K_{\mu\nu}$, $F_{\mu\nu}$, and $W^{\kappa}_{\ \lambda\mu\nu}$, but it is introduced to illustrate the idea of a background field $W^{\kappa}_{\ \lambda\mu\nu}$ coexisting with the dynamical fields $K_{\mu\nu}$ and $F_{\mu\nu}$.

The global constant $\alpha$ -- see~\eqref{lag iksdfbosgv} -- is already determined. Because the term $KKKK$ is precisely a determinant of $K$, constants $\sigma=\gamma=1$. All of those characteristic constants are written below:
\begin{align}
    \alpha &= \frac{1}{8\pi\Lambda}\, , & \gamma&=1\, , & \sigma&=1\, .
    \label{constsV10}
\end{align}

\subsection{The non-metricity equation}
\label{sub noneq v10}
The non-metricity tensor \( N^{\kappa}_{\ \lambda\mu} \) does not depend on the explicit form of the affine Lagrangian but only on the choice of the configuration space — see \textbf{Chapter~\ref{first field eq}}. Therefore, \( N^{\kappa}_{\ \lambda\mu} \) is exactly the same as in the theory based on the full Ricci tensor, as presented in formula~\eqref{ntens}, and it decomposes as in \textbf{Lemma~\ref{lem absence}}. The traceless tensor \( W^{\kappa}_{\ \lambda\mu\nu} \) is “absent” in the variational sense — it represents a background field and does not contribute to the symplectic structure. As mentioned earlier, this situation is analogous to standard electrodynamics, where the metric tensor \( g \) is present but treated as a background field, not interacting even with very strong electromagnetic fields.

\subsection{Einstein equation}
The non-perturbed solution $\Kz$~\eqref{eq0} is purely the Einstein $\Lambda$--vacuum equation~\eqref{eineq0}:
\begin{align}
    \Kz_{\mu\nu} = \Lambda\, g_{\mu\nu}\, .
\end{align}
The first-order coefficient $\Ko$~\eqref{eq1} vanishes, as in the case of the Ricci tensor theory, due to the absence of linear terms in the perturbations. Non-trivial corrections appear only at second order~\eqref{eq2}:
\begin{align}
&\frac{  1}{2}\, \left(  \Lambda\,  g g g \Kt  +  g g  FF+ g g F W+ g g  WW \right)\, g^{-1}= \nonumber  \\
&=   \Lambda\, g g \Kt  +  g FF+     g  FW+     g WW \, ,
\end{align}
where:
\begin{align}
ggg\Kt &=-\Kt^{\alpha}_{\ \alpha}\, |\det g|\, ,  \\
ggFF &=-2C_F\, F_{\alpha\beta}\, F^{\alpha\beta}\, |\det g|\, , \label{ggFFV10}\\
    g g F W&=-2 I_{FW}\, F_{\alpha\beta}\, W^{\alpha  \beta}\, |\det g| \,  , \label{ggFWV10}\\
    g g  WW&=-2\,C_W \left( W_{\alpha \beta }\, W^{ \beta \alpha}-   W_{\alpha\beta \kappa\lambda}\, W^{\kappa\lambda \alpha\beta}\right)\, |\det g|\,  , \label{ggWWV10} \\
    gg\Kt  &= \left(  \Kt^{\mu\nu} -   \Kt^{\alpha}_{\ \alpha}\, g^{\mu\nu}\right) \, |\det g| \, , \\ 
     gFF  &= 4C_F\left(F^{\mu\alpha}\, F^{\nu}_{\ \alpha} - \frac 12\, F_{\alpha\beta}\, F^{\alpha\beta}\, g^{\mu\nu}\right)\, |\det g|\, ,     \\
     g  FW &=2I_{FW}\, \left( F_{\alpha\beta}\, W^{\alpha (\mu\nu)\beta} + F^{\alpha (\mu} \, W_{\alpha}^{ \ \nu)}   - g^{\mu\nu}\, F_{\alpha\beta}\, W^{\alpha \beta}  \right)\, |\det g|\, , \\
      g WW &=4C_W\,\left[W_{\alpha\beta}\, W^{\beta(\mu\nu)\alpha}-W_{\kappa\lambda\alpha}^{\ \ \ (\mu|}\, W^{\alpha|\nu)\kappa\lambda}  + \right. \nonumber\\
      & \quad \left. +\frac 12\, g^{\mu\nu}  \left(W_{\alpha\beta\kappa\lambda}\,W^{\kappa\lambda\alpha\beta} - W_{\alpha\beta}\, W^{\beta\alpha}\right) \right]\, |\det g|\,  .
\end{align}
As before, contraction with the metric tensor $g^{\mu\nu}$ implies vanishing of the trace $\Kt^{\alpha}_{\ \alpha}$. Thus:
 \begin{align}
      \Kt^{\mu\nu} &= - \frac{4C_F}{\Lambda}\, \left(F^{\mu\alpha}\, F^{\nu}_{\ \alpha} - \frac 14\, F_{\alpha\beta}\, F^{\alpha\beta}\, g^{\mu\nu}\right) + \nonumber \\
       & \quad -\frac{4C_W}{\Lambda}\,\left[ W_{\alpha\beta}\, W^{\beta(\mu\nu)\alpha} 
 - W^{\alpha(\mu|\kappa\lambda}\, W_{\kappa\lambda\alpha}^{\ \ \ |\nu)}  -\frac 14\, g^{\mu\nu}  \left(W_{\alpha\beta}\, W^{\beta\alpha} - W_{\alpha\beta\kappa\lambda}\,W^{\kappa\lambda\alpha\beta} \right) \right] + \nonumber \\
     & \quad -\frac{2I_{FW}}{\Lambda}\  \left( F_{\alpha\beta}\, W^{\alpha (\mu\nu)\beta} + F^{\alpha (\mu} \, W_{\alpha}^{ \ \nu)}   - \frac 12\,  g^{\mu\nu}\, F_{\alpha\beta}\, W^{\alpha \beta}   \right)  \, . 
     \label{Kt v10}
 \end{align}
 Therefore, the Einstein equation obtained via the perturbative method~\eqref{rachperp} is the following:
\begin{align}
    K_{\mu\nu}=  \Lambda\, g_{\mu\nu} +\Kt_{\mu\nu}\, .\nonumber \\
    \label{eqq13}
\end{align}
However, this equation is not precisely the well-known form of the Einstein equation~\eqref{scheinmet}, due to the general, non-metric Ricci tensor $K_{\mu\nu}$ on the left-hand side which decomposes into the purely metric part $\kolo{K}$ and the rest $Q_{\mu\nu}$~\eqref{def Q}. The tensor $Q_{\mu\nu}$ is exactly the same as in the theory of the full Ricci tensor -- see~\eqref{Q v0}:
\begin{align}
     Q_{\mu\nu} =  -6 \, A_{\mu}\,A_{\nu} \, ,
\end{align}
due to the “absence'' (in the variational sense) of the traceless Riemann tensor $W^{\kappa}_{\ \lambda\mu\nu}$. Finally, the Einstein equation~\eqref{scheinmet} is the following:
\begin{align}
    \kolo{K}_{\mu\nu} = \Lambda\, g_{\mu\nu} + \Kt_{\mu\nu} + 6 \, A_{\mu}\,A_{\nu}\, ,
    \label{eqq2v10}
\end{align}  
or, using the Einstein tensor $\kolo{G}_{\mu\nu}$~\eqref{def ein tensor}:
\begin{align}
   \kolo{G}_{\mu\nu}   = -\Lambda\, g_{\mu\nu} + \Kt_{\mu\nu}+ 6 \,\left( A_{\mu}\,A_{\nu}- \frac12\, g_{\mu\nu}\, A_{\sigma}A^{\sigma}\right)\, .
        \label{eqq21v10}
\end{align}

\subsection{Effective cosmological parameter}

The cosmological parameter $\Lambda_{\rm eff}$~\eqref{def lam eff} associated with this theory is given by
\begin{align}
\Lambda_{\rm eff} := \frac{1}{4} \, \kolo{K}_{\mu\nu} g^{\mu\nu}= \Lambda + \frac{3}{2}\, A_{\kappa}A^{\kappa}\, ,
\label{lam effv10}
\end{align}
and coincides with the expression for $\Lambda_{\rm eff}$ presented in \textbf{Chapter~\ref{eff lam v0}}, since $Q_{\alpha}^{\ \alpha}$ depends only on the non-metricity tensor, which is identical in both theories.

\subsection{Field equation for the skew-symmetric Ricci tensor}
\label{feq F v10}
The  field equation for  $F_{\mu\nu}$ was determined by the symplectic relation~\eqref{rel chi} and was schematically  derived in the \textbf{Chapter~\ref{scheme}} in formula~\eqref{chi}:
\begin{align}
    \chi&= \frac{\sigma \sigma_g }{ 16 
 \pi \Lambda \gamma^2  \sqrt{|\det g|}}\, \left( ggF + ggW\right)   \, ,
\end{align}
where
\begin{align}
     g g F & :=\frac{\partial}{\partial F} ( ggFF) = -4C_F\, F^{\mu\nu}\, |\det g|\, , \\
     g g W  & :=\frac{\partial}{\partial F} ( ggFW)=  -2I_{FW}\, W^{[\mu \nu]} \, |\det g|\,,
\end{align}
due to the already calculated  terms $ggFF$ ~\eqref{ggFFV10} and $ggFW$~\eqref{ggFWV10}. The metric signature is assumed to be Lorentzian, thus, $\sigma_g=-1$. Upon substituting all characteristic constants~\eqref{constsV10}, the above field equation becomes:
\begin{align}
    \chi^{\mu\nu} = \frac{\sqrt{|\det g|} }{8 \pi \Lambda   }\, \left( 2C_F\, F^{\mu\nu} + I_{FW}\, W^{[\mu \nu]}\right)\, ,
     \label{rel konst10}
\end{align}
and will be called \textit{the constitutive relation between $\chi$ and $F$}.

\subsection{Potential equation}

The situation is very similar to the theory of the full Ricci tensor presented in \textbf{Chapter~\ref{eq pot v0}}. Indeed, taking the covariant derivative of equation~\eqref{rel konst10} yields:
\begin{align}
   \mnabla_{\nu} \chi^{\mu\nu} = \cJ^{\mu} = \frac{\sqrt{|\det g|} }{8 \pi \Lambda}\, \left( 2C_F\, \mnabla_{\nu}F^{\mu\nu} + I_{FW}\,\mnabla_{\nu} W^{[\mu \nu]} \right)\, ,
\end{align}
Then, substituting the current ${\cal J}_\mu$ with the potential $A_{\mu}$~\eqref{A rel J 1}, and using the formula for $\mnabla_{\nu} F^{\mu\nu}$~\eqref{mnabla F}, one obtains:
\begin{align}
 3\Lambda\, A^{\mu} = 2C_F\left( A^{\sigma}\, \kolo{K}_{\sigma}^{\ \mu} - \mBox A^{\mu} \right) + I_{FW}\,\mnabla_{\nu} W^{[\mu \nu]} \, ,
\end{align}
or equivalently:
\begin{align}
   \mBox A^{\mu} =-  \frac{3\Lambda}{2C_F}A^{\mu}  + A^{\sigma}\, \kolo{K}_{\sigma}^{\ \mu}  +\frac{I_{FW}}{2C_F}\,\mnabla_{\nu} W^{[\mu \nu]}\, .
   \label{poteq v10}
\end{align}
This formula is  a non-homogeneous Proca equation~\eqref{eq P2} with the following mass parameter:
\begin{align}
    m^2 = -  \frac{3\hbar^2\Lambda}{2C_F}\, .
    \label{mpar v10}
\end{align}
The above potential equation is very similar to the one obtained for the full Ricci tensor theory — see equation~\eqref{jdshdfbs2}. The only difference lies in the non-homogeneous term involving the tensor $\mnabla_{\nu} W^{[\mu\nu]}$. As before, the metric Ricci tensor is approximated by $\kolo{K}_{\mu\nu} \approx \Lambda g_{\mu\nu}$ due to the Einstein equation~\eqref{eqq13}, leading to:
\begin{align}
    \mBox A^{\mu} = \left(\Lambda-  \frac{3\Lambda}{2C_F}\right) A^{\mu}  +\frac{I_{FW}}{2C_F}\,\mnabla_{\nu} W^{[\mu \nu]}\, .
\end{align}

\chapter{Metric Lagrangians}

 \section{Passage from the affine picture to the metric picture -- variational calculus}
\label{chap passage}
\sectionmark{Passage -- variational calculus}

The variational formulation of the affine theory of the full Riemann curvature was derived and analysed in \textbf{Chapter~\ref{chap var str aff}}. Examples of such theories, along with their field equations, were also presented. Although it would be both interesting and valuable to construct a corresponding metric theory that reproduces the same field equations and allows for comparison with other models in the literature, this has so far only been accomplished in a special case - when the theory depends solely on the symmetric Ricci tensor $K_{\mu\nu}$. The most recent treatment of this case can be found in~\cite{nonmetricity}, written by the  author together with one of the supervisors, J.~Kijowski. The extension of this correspondence to the full Riemann tensor $R^{\kappa}_{\ \lambda\mu\nu}$ represents a new and, as yet, unpublished result.

\ 

The passage to the metric picture starts from  reminding the affine symplectic formula~$\delta \Lag_A$~\eqref{var1}:
\begin{align}
     \delta \Lag_A =   \partial_{\nu}\left( \cP_{\kappa}^{\ \lambda\mu\nu}\, \delta \Gamma^{\kappa}_{\ \lambda\mu }\right) = \left( \nabla_{\nu}\cP_{\kappa}^{\ \lambda\mu\nu}\right)\, \delta \Gamma^{\kappa}_{\ \lambda\mu } + \cP_{\kappa}^{\ \lambda\mu\nu}\, \delta K^{\kappa}_{\ \lambda\mu\nu}\, .
     \label{vareq}
\end{align}
The first field equation~\eqref{1fieldeq} $\nabla \cP=0$  induces the decomposition of the connection for the metric part and the non-metricity: $\Gamma=\mGamma + N$ -- see \textbf{Theorem~\ref{th non-metricity}}. This decomposition is used to divide the first  boundary term in the following way:
\begin{align}
    \delta \Lag_A =   \partial_{\nu}\left( \cP_{\kappa}^{\ \lambda\mu\nu}\, \delta \Gamma^{\kappa}_{\ \lambda\mu }\right) =  \partial_{\nu}\left( \cP_{\kappa}^{\ \lambda\mu\nu}\, \delta \mGamma^{\kappa}_{\ \lambda\mu }\right) +   \partial_{\nu}\left( \cP_{\kappa}^{\ \lambda\mu\nu}\, \delta N^{\kappa}_{\ \lambda\mu }\right)\, .
\end{align}
Now, implementing the decomposition of the momentum $\cP_{\kappa}^{\ \lambda\mu\nu}$~\eqref{eq: rozklad pedu P}, as presented in \textbf{Lemma~\ref{lemma dec cP}}, into the above variation yields:
\begin{align}
    \delta \Lag_A &= \partial_{\nu} \left( \pi_{\kappa}^{\ \lambda \mu \nu} \, \delta \mGamma^{\kappa}_{\ \lambda \mu} 
    + \Omega_{\kappa}^{\ \lambda \mu \nu} \, \delta \mGamma^{\kappa}_{\ \lambda \mu} 
    - \chi^{\mu\nu} \, \delta \mGamma^{\kappa}_{\ \kappa\mu}+ \right. \nonumber \\
    &\left. \quad + \pi_{\kappa}^{\ \lambda \mu \nu} \, \delta N^{\kappa}_{\ \lambda \mu} 
    + \Omega_{\kappa}^{\ \lambda \mu \nu} \, \delta A^{\kappa}_{\ \lambda\mu} 
    - 2 \chi^{\mu\nu} \, \delta A_{\mu} \right) \, .
    \label{vareq00}
\end{align}
To obtain the metric picture\footnote{The transformation between affine and metric pictures was discussed extensively in \cite{nonmetricity}, albeit for a slightly different class of theories.}, the metric tensor $g_{\mu\nu}$ must be treated as a \textit{control parameter}, whereas here it appears only as a \textit{response parameter}, encoded in $\pi_{\kappa}^{\ \lambda\mu\nu}$ — see~\eqref{pi2} and~\eqref{pi40}.

The analysis begins with the part involving the non-metricity tensor, namely the term $\partial(\pi\, \delta N)$, whose Legendre transformation is presented in the following lemma:
\begin{lemma}
    \label{leg trans piN}
    The following equality holds:
    \begin{align}
\partial_{\nu}\left(  \pi_{\kappa}^{\ \lambda\mu\nu}\, \delta {N^{\kappa}}_{\lambda\mu}\right) = \partial_{\kappa}\left(\mathcal{R}^{\mu\nu\kappa}\, \delta g_{\mu\nu} \right) - \delta \left[\mnabla_{\kappa} \mathcal{R}_{\sigma}^{\ \sigma\kappa} \right]\, ,
\end{align}
where: 
\begin{align}
\mathcal{R}^{\mu\nu\kappa} &=\frac{\sqrt{|\det g|}}{16\pi} \left[N^{\kappa\mu\nu} - N_{\sigma}^{\ \sigma(\mu}\, g^{\nu)\kappa} + \frac12 \left( N_{\sigma}^{\ \sigma \kappa}- N^{\kappa\sigma}_{\ \ \sigma} \right)\,g^{\mu\nu}\right] = \nonumber\\
&=\frac{\sqrt{|\det g|}}{16\pi} \left[ A^{\kappa\mu\nu} - \frac{6}{5}\, g^{\kappa(\mu} A^{\nu)}+\frac 12\, \left(\frac 65\, A^{\kappa} - h^{\kappa} \right)\, g^{\mu\nu} \right]\, ,
\label{cal R} \\
  \mnabla_{\kappa}{\cal R}_{\sigma}^{\ \sigma\kappa} &= \frac{\sqrt{|\det g|}}{16\pi} \,  \left(\mnabla_{\kappa}N_{\sigma}^{\ \sigma\kappa} - \mnabla_{\kappa}N^{\kappa\sigma}_{\ \ \sigma}\right) = \frac{\sqrt{|\det g|}}{16\pi} \, \mnabla_{\kappa}\left(\frac{6}{5}\, A^{\kappa} - h^{\kappa}\right)\, .  
  \label{divR}
\end{align}
\end{lemma}
\begin{proof}
    The proof relies on the tensorial calculus and starts as follows:
    \begin{align}
        \partial_{\nu}\left(  \pi_{\kappa}^{\ \lambda\mu\nu}\, \delta {N^{\kappa}}_{\lambda\mu}\right) &= \delta \partial_{\nu} \left(  \pi_{\kappa}^{\ \lambda\mu\nu}\,   {N^{\kappa}}_{\lambda\mu} \right) -\partial_{\nu}\left(  {N^{\kappa}}_{\lambda\mu} \, \delta\pi_{\kappa}^{\ \lambda\mu\nu}\right)  \, .
    \end{align}
    The commutation of $\delta$ and $\partial_{\nu}$ was discussed in \textbf{Chapter~\ref{varcal}}. Using the definition of $\pi_{\kappa}^{\ \lambda\mu\nu}$~\eqref{pi40}, the total variation equals:
    \begin{align}
        \pi_{\kappa}^{\ \lambda\mu\nu}\,   {N^{\kappa}}_{\lambda\mu} &= \frac{\sqrt{|\det g|}}{16\pi}\,\left(N^{\nu \sigma}_{\ \ \sigma} - N_{\sigma}^{\ \sigma\nu}\right)\, , 
    \end{align}
    and then:
    \begin{align}
         \partial_{\nu} \left(  \pi_{\kappa}^{\ \lambda\mu\nu}\,   {N^{\kappa}}_{\lambda\mu} \right) &= \mnabla_{\nu}  \left(  \pi_{\kappa}^{\ \lambda\mu\nu}\,   {N^{\kappa}}_{\lambda\mu} \right)=  -\mnabla_{\kappa}{\cal R}_{\sigma}^{\ \sigma\kappa} \, ,
    \end{align}
    where the first equality holds due to the vector-density character of the object inside the bracket, whereas the second one corresponds with the formula~\eqref{divR} presented in this thesis.  The second term transforms as follows:
    \begin{align}
        {N^{\kappa}}_{\lambda\mu} \, \delta\pi_{\kappa}^{\ \lambda\mu\nu}  = \left(N^{\nu}_{\ \lambda\mu} - \delta^{\nu}_{\lambda}\, N^{\sigma}_{\ \sigma \mu} \right)\, \delta\pi^{\lambda\mu}\, .
    \end{align}
    The variation of momentum $\pi^{\lambda\mu}$~\eqref{pi2} is the following:
    \begin{align}
        \delta \pi^{\lambda\mu} &=\delta\left(\frac{\sqrt{|\det g|}}{16\pi}\, g^{\lambda\mu} \right) =\frac{\sqrt{|\det g|}}{16\pi}\,\left[\frac 12\, g^{\alpha \beta}\, g^{\lambda\mu} - g^{\lambda \alpha}\, g^{\mu \beta }\right] \delta g_{\alpha\beta}\, .
    \end{align}
    Hence:
    \begin{align}
        {N^{\kappa}}_{\lambda\mu} \, \delta\pi_{\kappa}^{\ \lambda\mu\nu} &=-\frac{\sqrt{|\det g|}}{16\pi}\,\left[N^{\nu\alpha\beta} - N_{\sigma}^{\ \sigma\alpha}\, g^{\beta\nu} + \frac 12\, \left(N_{\sigma}^{\ \sigma\nu} - N^{\nu\sigma}_{\ \ \sigma}\right)\, g^{\alpha\beta} \right]\, \delta g_{\alpha\beta} =\nonumber \\
        &= - {\cal R}^{\alpha\beta \nu}\, \delta g_{\alpha\beta}\, ,
    \end{align}
    what finishes the proof.
\end{proof}

Then, the variation of the affine Lagrangian $\delta \Lag_A$~\eqref{vareq00} equals:
\begin{align}
    \delta \Lag_A &=    \partial_{\nu}\bigg(\pi_{\kappa}^{\ \lambda \mu \nu}\, \delta \mGamma^{\kappa}_{\ \lambda \mu} + \Omega_{\kappa}^{\ \lambda \mu \nu}\, \delta \mGamma^{\kappa}_{\ \lambda \mu} - \chi^{\mu\nu}\, \delta \mGamma^{\kappa}_{\ \kappa\mu}  +\nonumber \\
    &  \quad  +{\cal R}^{ \lambda \mu \nu}\, \delta g_{ \lambda \mu}  +\Omega_{\kappa}^{\ \lambda \mu \nu} \, \delta A^{\kappa}_{\ \lambda\mu} - 2  \chi^{\mu\nu}   \, \delta A_{ \mu} \bigg)  - \delta \left[\mnabla_{\kappa} \mathcal{R}_{\sigma}^{\ \sigma\kappa} \right]\, .
    \label{vareq01}
\end{align}
The metric picture \textit{on shell} is described by the metric Lagrangian $\Lag_g$, which is defined as:
\begin{align}
    \Lag_g := \Lag_A + \mnabla_{\kappa} \mathcal{R}_{\sigma}^{\ \sigma\kappa}\, ,
    \label{def lagg}
\end{align}
whereas its symplectic structure is given by:
\begin{align}
    \delta \Lag_g &=    \partial_{\nu}\left(\pi_{\kappa}^{\ \lambda \mu \nu}\, \delta \mGamma^{\kappa}_{\ \lambda \mu} + \Omega_{\kappa}^{\ \lambda \mu \nu}\, \delta \mGamma^{\kappa}_{\ \lambda \mu} - \chi^{\mu\nu}\, \delta \mGamma^{\kappa}_{\ \kappa\mu} \right.+\nonumber \\
    & \quad \left. +{\cal R}^{ \lambda \mu \nu}\, \delta g_{ \lambda \mu}  +\Omega_{\kappa}^{\ \lambda \mu \nu} \, \delta A^{\kappa}_{\ \lambda\mu} - 2  \chi^{\mu\nu}   \, \delta A_{ \mu} \right)  \, .
    \label{varlagg0}
\end{align}
In the “standard'' theories, the metric Lagrangian is defined as the sum of the Hilbert Lagrangian $\Lag_H$ and the matter Lagrangian $\Lag_{\rm matt}$, which in this case corresponds to the assumption that $\chi^{\mu\nu} = 0 = \Omega_{\kappa}^{\ \lambda \mu\nu}$, along with the addition of the extra boundary term $\partial_{\nu}(p^{\nu}\, \delta \phi)$ associated with the matter field $\phi$. Such theories were presented and thoroughly explored in \cite{mag, nonmetricity}. However, the situation described above is much more complicated. Therefore, to extract the formula that will be unquestionably responsible for the metric picture description, several transformations must be performed. First, the formula~\eqref{varlagg0} can be written in the following manner:
\begin{align}
    \delta \Lag_g =    \partial_{\nu}\left({\cal P}_{\kappa}^{\ \lambda\mu\nu}\, \delta \mGamma^{\kappa}_{\ \lambda\mu}   +{\cal R}^{ \lambda \mu \nu}\, \delta g_{ \lambda \mu}  +\Omega_{\kappa}^{\ \lambda \mu \nu} \, \delta A^{\kappa}_{\ \lambda\mu} - 2  \chi^{\mu\nu}   \, \delta A_{ \mu} \right)  \, ,
    \label{varlagg1}
\end{align}
where was used the decomposition of the momentum ${\cal P}_{\kappa}^{\ \lambda\mu\nu}$ -- see formula~\eqref{eq: rozklad pedu P}  in \textbf{Lemma~\ref{lemma dec cP}}. Of course, the equality~\eqref{vareq} holds also for the metric connection $\mGamma$ (as a special example of the affine connection $\Gamma$), therefore:
\begin{align}
    \delta \Lag_g &=  \left(\mnabla_{\nu }{\cal P}_{\kappa}^{\ \lambda\mu\nu}\right) \, \delta \mGamma^{\kappa}_{\ \lambda\mu} + {\cal P}_{\kappa}^{\ \lambda\mu\nu}\, \delta \kolo{K}^{\kappa}_{\ \lambda\mu\nu}+ \nonumber \\
    & \quad +\partial_{\nu}\left(  {\cal R}^{ \lambda \mu \nu}\, \delta g_{ \lambda \mu}  +\Omega_{\kappa}^{\ \lambda \mu \nu} \, \delta A^{\kappa}_{\ \lambda\mu} - 2  \chi^{\mu\nu}   \, \delta A_{ \mu} \right)  \, .
    \label{varlagg2}
\end{align}
Next, there is implemented the decomposition of the Kijowski tensor  $\kolo{K}^{\kappa}_{\ \lambda\mu\nu}$~\eqref{dec kijowski}:
\begin{align}
    \delta \Lag_g &=   \left(\mnabla_{\nu }{\cal P}_{\kappa}^{\ \lambda\mu\nu}\right) \, \delta \mGamma^{\kappa}_{\ \lambda\mu} + \pi^{\mu\nu}\, \delta \kolo{K}_{\mu\nu} + \Omega_{\kappa}^{\ \lambda\mu\nu}\, \delta \kolo{U}^{\kappa}_{\ \lambda\mu\nu}+ \nonumber \\
    & \quad +\partial_{\nu}\left(  {\cal R}^{ \lambda \mu \nu}\, \delta g_{ \lambda \mu}  +\Omega_{\kappa}^{\ \lambda \mu \nu} \, \delta A^{\kappa}_{\ \lambda\mu} - 2  \chi^{\mu\nu}   \, \delta A_{ \mu} \right)  \, .
    \label{varlagg3}
\end{align} 
Now, the result looks much better, but the work is still not complete. The boundary term $\partial(\mathcal{R} \, \delta g)$ is addressed first, and its treatment is presented in the following lemma:
\begin{lemma}
\label{lemma cR}
The following equality holds:
\begin{align}
 \partial_{\kappa} \left(\mathcal{R}^{\mu\nu\kappa}\, \delta g_{\mu\nu} \right)  = \mathcal{Y}^{\lambda\mu}_{\ \ \kappa} \, \delta \mGamma^{\kappa}_{\ \lambda\mu} + \left(\mnabla_{\kappa}\mathcal{R}^{\mu\nu\kappa} \right)\delta g_{\mu\nu}\, ,
 \label{sgsfgdfg}
\end{align}
where:
\begin{align}
\mathcal{Y}^{\lambda\mu  \kappa} &:= \mathcal{R}^{\kappa\lambda\mu} + \mathcal{R}^{\kappa \mu\lambda} = \frac{\sqrt{|\det g|}}{16\pi}\, \left[ 2N^{(\lambda\mu)\kappa} - N_{\sigma}^{\ \sigma\kappa}\, g^{\lambda\mu} -g^{\kappa(\lambda}\, N^{\mu)\sigma}_{\ \ \ \sigma}\right] = \nonumber \\
&=\frac{\sqrt{|\det g|}}{16\pi}\, \left[2A^{(\lambda\mu)\kappa} - \frac 65 \, g^{\lambda\mu}\,  A^{\kappa}  - h^{(\lambda} g^{\mu)\kappa }   \right] \, .
\label{cal Y}
\end{align}    
\end{lemma}
\begin{proof}
    The proof is purely algebraic, so:
    \begin{align}
        \mathcal{Y}^{\lambda\mu}_{\ \ \kappa} \, \delta \mGamma^{\kappa}_{\ \lambda\mu} &=-  \mathcal{Y}^{\lambda\mu \alpha}\, \mGamma^{\beta}_{\ \lambda\mu}\, \delta g_{\alpha\beta} + \frac 12 \left(2{\cal Y}^{\nu\alpha\beta} - {\cal Y}^{\alpha\beta \nu} \right)\, \delta g_{\alpha\beta,\nu} = \nonumber \\
        &= - \left(\mathcal{R}^{\alpha\lambda\mu} + \mathcal{R}^{\alpha \mu\lambda} \right)\, \mGamma^{\beta}_{\ \lambda\mu}\, \delta g_{\alpha\beta}  + {\cal R}^{\alpha\beta \nu}\, \delta g_{\alpha\beta,\nu}= \nonumber \\
        &= \partial_{\nu}\left( {\cal R}^{\alpha\beta \nu}\, \delta g_{\alpha\beta } \right) - \left( \partial_{\nu}\, {\cal R}^{\alpha\beta\nu} +\mathcal{R}^{\alpha\lambda\mu} \, \mGamma^{\beta}_{\ \lambda\mu}  + \mathcal{R}^{\alpha \mu\lambda}  \, \mGamma^{\beta}_{\ \lambda\mu} \right)\, \delta g_{\alpha\beta} = \nonumber \\
        &= \partial_{\nu}\left( {\cal R}^{\alpha\beta \nu}\, \delta g_{\alpha\beta } \right) - \left( \mnabla_{\nu} {\cal R}^{\alpha\beta \nu}\right)\, \delta g_{\alpha\beta}\, ,
    \end{align}
    what finishes the proof.
\end{proof}
Secondly, the boundary term $\partial(\Omega\, \delta A)$ can be written as:
\begin{align}
    \partial_{\nu}\left[\Omega_{\kappa}^{\ \lambda\mu\nu}\, \delta A^{\kappa}_{\ \lambda\mu}\right] &=   \left(\mnabla_{\nu}\, \Omega_{\kappa}^{\ \lambda\mu\nu} \right)\, \delta A^{\kappa}_{\ \lambda\mu} + \Omega_{\kappa}^{\ \lambda\mu\nu}\, \delta D^{\kappa}_{\ \lambda\mu\nu} + \nonumber \\
    &\quad -\left(\Omega_{\kappa}^{\ \alpha\beta \lambda}\, A^{\mu}_{\ \alpha\beta} + \Omega_{\sigma}^{\ \lambda\mu\nu}\, A^{\sigma}_{\ \nu\kappa} \right)\, \delta\mGamma^{\kappa}_{\ \lambda\mu} \, ,
    \label{aosdhaoifg}
\end{align}
where 
\begin{align}
    D^{\kappa}_{\ \lambda\mu\nu} &:= \mnabla_{\nu}{A}^{\kappa}_{\ \lambda  \mu  } -  \mnabla_{(\nu}{A}^{\kappa}_{\ \lambda  \mu)  } -\frac 13\, \mnabla_{\sigma}\left(\delta^k_{\nu} A^{\sigma}_{\ \lambda\mu} -\delta^k_{(\nu} A^{\sigma}_{\ \lambda\mu)}  \right)= \nonumber \\
    &=  \frac 23\, \left( \mnabla_{\nu}{A}^{\kappa}_{\ \lambda  \mu  } -  \mnabla_{(\lambda }A^{\kappa}_{\ \mu) \nu }\right) - \frac 29\, \mnabla_{\sigma}\left( \delta^{\kappa}_{\nu}  {A}^{\sigma}_{\ \lambda\mu} - \delta^{\kappa}_{(\lambda }   {A}^{\sigma}_{\ \mu)\nu }\right) \, ,
    \label{def tensor D}
\end{align}
what is precisely a linear part of the tensor  $U^{\kappa}_{\ \lambda\mu\nu}$  ~\eqref{rozklad pelny U}. Of course, the validity of the equality~\eqref{aosdhaoifg} could be proven analogously as it was done in \textbf{Theorem~\ref{th cPdG}}.

Thirdly,  derivatives of potential $A_{\mu}$ will appear only via the skew-symmetric Ricci tensor $F_{\mu\nu}=A_{\nu,\mu} - A_{\mu,\nu}$~\eqref{rozklad pelny F}. Thus:
\begin{align}
   - 2  \partial_{\nu}\left( \chi^{\mu\nu}   \, \delta A_{ \mu} \right) = -2\cJ^{\mu}\, \delta A_{ \mu} + \chi^{\mu\nu}\, \delta F_{\mu\nu}\, .
\end{align}
Then, the variational formula~\eqref{varlagg3} takes the following form:
\begin{align}
    \delta \Lag_g &=  \left(\mnabla_{\nu }{\cal P}_{\kappa}^{\ \lambda\mu\nu} 
 + {\cal Y}^{\lambda\mu}_{\ \ \kappa} - \Omega_{\kappa}^{\ \alpha\beta \lambda}\, A^{\mu}_{\ \alpha\beta} - \Omega_{\sigma}^{\ \lambda\mu\nu}\, A^{\sigma}_{\ \nu\kappa}\right) \, \delta \mGamma^{\kappa}_{\ \lambda\mu} + \delta \Lag_H + \nonumber \\
 & \quad +\Omega_{\kappa}^{\ \lambda\mu\nu}\, \delta \left(\kolo{U}^{\kappa}_{\ \lambda\mu\nu} + D^{\kappa}_{\ \lambda\mu\nu} \right) + \left(\mnabla_{\nu} \Omega_{\kappa}^{\ \lambda\mu\nu}\right)\, \delta A^{\kappa}_{\ \lambda\mu} -2\partial_{\nu}\left(    \chi^{\mu\nu}   \, \delta A_{ \mu} \right) + \nonumber\\
 & \quad + \left( \frac{1}{16\pi}\, \kolo{\cal G}^{\mu\nu} + \mnabla_{\kappa}{\cal R}^{\mu\nu\kappa}\right)\, \delta g_{\mu\nu}\, .
    \label{varlagg4}
\end{align} 
Interestingly, it can be shown (using the techniques presented in \textbf{Theorem~\ref{th cPdG}}) that:
\begin{align}
 - \Omega_{\kappa}^{\ \alpha\beta (\lambda}\, A^{\mu)}_{\ \alpha\beta} - \Omega_{\sigma}^{\ \lambda\mu\nu}\, A^{\sigma}_{\ \nu\kappa} = \nabla_{\nu} \Omega_{\kappa}^{\ \lambda\mu\nu}   -\mnabla_{\nu }\Omega_{\kappa}^{\ \lambda\mu\nu} \, .
\end{align}
In the same manner, the following equality holds -- see the definition of ${\cal Y}^{\lambda\mu}_{\ \ \kappa}$~\eqref{cal Y}:
\begin{align}
    {\cal Y}^{\lambda\mu}_{\ \ \kappa} =\nabla_{\nu} \pi_{\kappa}^{\ \lambda\mu\nu}\, .
\end{align}
Whence, the symplectic formula $\delta \Lag_g$~\eqref{varlagg4} drastically simplifies, because:
\begin{align}
    \frac{\partial \Lag_g}{\partial \mGamma^{\kappa}_{\ \lambda\mu}}&= \mnabla_{\nu }{\cal P}_{\kappa}^{\ \lambda\mu\nu} 
 + {\cal Y}^{\lambda\mu}_{\ \ \kappa} - \Omega_{\kappa}^{\ \alpha\beta \lambda}\, A^{\mu}_{\ \alpha\beta} - \Omega_{\sigma}^{\ \lambda\mu\nu}\, A^{\sigma}_{\ \nu\kappa} = \nonumber\\
 &=\mnabla_{\nu }{\cal P}_{\kappa}^{\ \lambda\mu\nu} 
 + \nabla_{\nu} \pi_{\kappa}^{\ \lambda\mu\nu} +  \nabla_{\nu} \Omega_{\kappa}^{\ \lambda\mu\nu}   -\mnabla_{\nu }\Omega_{\kappa}^{\ \lambda\mu\nu} =  \nabla_{\nu}{\cal P}_{\kappa}^{\ \lambda\mu\nu} =0  \, ,
\end{align} 
due to the first field equation~\eqref{1fieldeq}. Finally:
\begin{align}
    \delta \Lag_g &=   \pi^{\mu\nu}\, \delta\kolo{K}_{\mu\nu} +\Omega_{\kappa}^{\ \lambda\mu\nu}\, \delta \left(\kolo{U}^{\kappa}_{\ \lambda\mu\nu} + D^{\kappa}_{\ \lambda\mu\nu} \right) + \left(\mnabla_{\nu} \Omega_{\kappa}^{\ \lambda\mu\nu}\right)\, \delta A^{\kappa}_{\ \lambda\mu} + \nonumber\\
 & \quad + \chi^{\mu\nu}\, \delta F_{\mu\nu} -2\cJ^{\mu}\, \delta A_{ \mu}  + \left(   \mnabla_{\kappa}{\cal R}^{\mu\nu\kappa}\right)\, \delta g_{\mu\nu}\, .
    \label{varlagg5}
\end{align}
Indeed, the metric tensor $g_{\mu\nu}$ and its derivatives, organised into the curvature tensors $\kolo{U}^{\kappa}_{\ \lambda\mu\nu}$ and $\kolo{K}_{\mu\nu}$, are now under control, whereas $A^{\kappa}_{\ \lambda\mu}$ and $A_{\mu}$ play the role of “matter” potentials. This is compatible with the standard understanding of the metric picture, although a few comments are still necessary.

\ 

The sum  $\kolo{U}^{\kappa}_{\ \lambda\mu\nu} + D^{\kappa}_{\ \lambda\mu\nu}$ in the above variation could look strange, although, it is not so surprising. Effectively, the same happened with the symmetric Ricci tensor $K_{\mu\nu}$, but it was done in parts. Indeed, from formula~\eqref{vareq} the below quantity could be extracted and rewritten as follows:
\begin{align}
    \partial_{\nu}\left[\pi_{\kappa}^{\ \lambda\mu\nu} \, \delta \left(\mGamma^{\kappa}_{\ \lambda\mu} + N^{\kappa}_{\ \lambda\mu} \right)\right] = \pi^{\mu\nu}\, \delta\left(  \kolo{K}_{\mu\nu} +    \mnabla_{\kappa} A^{\kappa}_{\ \mu \nu}  - \frac 65\,  \mnabla_{\mu}  A_{\nu} \right)\, ,
\end{align}
where the covariant derivatives of $A^{\kappa}_{\ \lambda\mu}$ and $A_{\mu}$ correspond with the linear part in the formula for $K_{\mu\nu}$~\eqref{rozklad pelny K}, thus, it is analogous to the term $\Omega\, \delta (\kolo{U}+D)$. Then, there was made a Legendre transformation:
\begin{align}
    \partial_{\nu}\left[\pi_{\kappa}^{\ \lambda\mu\nu} \, \delta \left(\mGamma^{\kappa}_{\ \lambda\mu} + N^{\kappa}_{\ \lambda\mu} \right)\right] &= \pi^{\mu\nu}\, \delta \kolo{K}_{\mu\nu} + \delta \left[ \pi^{\mu\nu}\, \left(\mnabla_{\kappa} A^{\kappa}_{\ \mu \nu}  - \frac 65\,  \mnabla_{\mu}  A_{\nu}  \right)\right]+\nonumber\\
    & \quad  -\left(\mnabla_{\kappa} A^{\kappa}_{\ \mu \nu}  - \frac 65\,  \mnabla_{\mu}  A_{\nu}  \right)\, \delta \pi^{\mu\nu}\, ,
\end{align}
where was used the equality~\eqref{paigoaesg} between $\partial_{\nu}\left[\pi_{\kappa}^{\ \lambda\mu\nu} \, \delta  \mGamma^{\kappa}_{\ \lambda\mu}\right]$ and $ \pi^{\mu\nu}\, \delta\! \kolo{K}_{\mu\nu} $, and the decomposition of the non-metricity tensor $N$~\eqref{rozklad tensora N}. Next, the \textbf{Lemma~\ref{leg trans piN}} implies the following equalities -- see formulae~\eqref{cal R} and~\eqref{divR}:
\begin{align}
    \delta \left[ \pi^{\mu\nu}\, \left(\mnabla_{\kappa} A^{\kappa}_{\ \mu \nu}  - \frac 65\,  \mnabla_{\mu}  A_{\nu}  \right)\right]&= -\delta \left( \mnabla_{\kappa}{\cal R}_{\sigma}^{\ \sigma\kappa}\right)\, , \\
      -\left(\mnabla_{\kappa} A^{\kappa}_{\ \mu \nu}  - \frac 65\,  \mnabla_{\mu}  A_{\nu}  \right)\, \delta \pi^{\mu\nu}&= \left(\mnabla_{\nu} {\cal R}^{\mu\nu\kappa}\right)\, \delta g_{\mu\nu}\, ,
\end{align}
and finally:
\begin{align}
     \partial_{\nu}\left[\pi_{\kappa}^{\ \lambda\mu\nu} \, \delta \left(\mGamma^{\kappa}_{\ \lambda\mu} + N^{\kappa}_{\ \lambda\mu} \right)\right]  = \pi^{\mu\nu}\, \delta \kolo{K}_{\mu\nu} -\delta \left( \mnabla_{\kappa}{\cal R}_{\sigma}^{\ \sigma\kappa}\right) +  \left(\mnabla_{\nu} {\cal R}^{\mu\nu\kappa}\right)\, \delta g_{\mu\nu}\, .
\end{align}

\

To complete this passage, the symplectic formula for the matter Lagrangian $\Lag_{\rm matt}$ must be found. Firstly, from the variation of $\delta \Lag_H$~\eqref{paigoaesg}, it follows that:
\begin{align}
     \pi^{\mu\nu}\, \delta \kolo{K}_{\mu\nu} = \delta \Lag_H + \frac{1}{16\pi}\, \kolo{\cal G}^{\mu\nu}\, \delta g_{\mu\nu}\,,
\end{align}
which recovers the standard Hilbert Lagrangian, one of the ingredients of the “typical” metric Lagrangian. If $\Omega_{\kappa}^{\ \lambda\mu\nu} = 0$, corresponding to the theory of the full Ricci tensor $R_{\mu\nu} = K_{\mu\nu} + F_{\mu\nu}$, the matter Lagrangian is simply the difference between the metric Lagrangian $\Lag_g$ and the Hilbert Lagrangian, as is the case for matter fields coupled to gravity — cf. \cite{mag, APP, nonmetricity}. However, the presence of the traceless part of the Kijowski tensor $U^{\kappa}_{\ \lambda\mu\nu}$ slightly changes the situation, because the tensor $\kolo{U}^{\kappa}_{\ \lambda\mu\nu}$ must also be treated as a \textit{response} parameter. This automatically implies that $\Omega_{\kappa}^{\ \lambda\mu\nu}$ remains a control parameter, and to maintain consistency of the description, $D^{\kappa}_{\ \lambda\mu\nu}$ must also be switched to a response parameter:
\begin{align}
     \Omega_{\kappa}^{\ \lambda\mu\nu}  \delta \left(\kolo{U}^{\kappa}_{\ \lambda\mu\nu} + D^{\kappa}_{\ \lambda\mu\nu} \right) = \delta \left[ \Omega_{\kappa}^{\ \lambda\mu\nu} \left(\kolo{U}^{\kappa}_{\ \lambda\mu\nu} + D^{\kappa}_{\ \lambda\mu\nu} \right)\right] - \left(\kolo{U}^{\kappa}_{\ \lambda\mu\nu} + D^{\kappa}_{\ \lambda\mu\nu} \right) \delta \Omega_{\kappa}^{\ \lambda\mu\nu}\, .
\end{align}
But now, the variational description of potential $A^{\kappa}_{\ \lambda\mu}$ will be associated with “Hamiltonian'' rather than “Lagrangian'', because the symplectic structure has the following form:
\begin{align}
    - \left(\kolo{U}^{\kappa}_{\ \lambda\mu\nu} + D^{\kappa}_{\ \lambda\mu\nu} \right) \delta \Omega_{\kappa}^{\ \lambda\mu\nu} + \left(\mnabla_{\nu}\Omega_{\kappa}^{\ \lambda\mu\nu} \right)\, \delta A^{\kappa}_{\ \lambda\mu}\, .
\end{align}
Therefore, the extra Legendre transformation has to be implemented:
\begin{align}
     \left(\mnabla_{\nu}\Omega_{\kappa}^{\ \lambda\mu\nu} \right)\, \delta A^{\kappa}_{\ \lambda\mu} = \delta\left[  \left(\mnabla_{\nu}\Omega_{\kappa}^{\ \lambda\mu\nu} \right)\,  A^{\kappa}_{\ \lambda\mu}\right] -  A^{\kappa}_{\ \lambda\mu}\, \delta \left(\mnabla_{\nu}\Omega_{\kappa}^{\ \lambda\mu\nu} \right)\, .
\end{align}
Now, the momentum $\Omega_{\kappa}^{\ \lambda\mu\nu}$ plays the role of the matter field, and its dynamics is described in Lagrangian formalism. 

Whence, the matter Lagrangian $\Lag_{\rm matt}$ is defined as follows:
\begin{align}
\Lag_{\rm matt} &:= \Lag_g -  \Lag_H -  \Omega_{\kappa}^{\ \lambda\mu\nu}\,   \left(\!\kolo{U}^{\kappa}_{\ \lambda\mu\nu} + D^{\kappa}_{\ \lambda\mu\nu} \right) - \left(\mnabla_{\nu}\Omega_{\kappa}^{\ \lambda\mu\nu} \right)\,  A^{\kappa}_{\ \lambda\mu}= 
    \label{def lag matt}\\
&=\Lag_A + \mnabla_{\kappa} \mathcal{R}_{\sigma}^{\ \sigma\kappa} -  \pi^{\mu\nu}\kolo{K}_{\mu\nu} -  \Omega_{\kappa}^{\ \lambda\mu\nu}\,   \left(\!\kolo{U}^{\kappa}_{\ \lambda\mu\nu} + D^{\kappa}_{\ \lambda\mu\nu} \right) -\left(\mnabla_{\nu}\Omega_{\kappa}^{\ \lambda\mu\nu} \right)\,  A^{\kappa}_{\ \lambda\mu}\, ,\nonumber
\end{align}
 whereas the symplectic structure is the following:
 \begin{align}
     \delta \Lag_{\rm matt}   &=  -    \left(\!\kolo{U}^{\kappa}_{\ \lambda\mu\nu} +  D^{\kappa}_{\ \lambda\mu\nu} \right)\, \delta \Omega_{\kappa}^{\ \lambda\mu\nu}  -  A^{\kappa}_{\ \lambda\mu}\, \delta \left(\mnabla_{\nu}\Omega_{\kappa}^{\ \lambda\mu\nu} \right) + \nonumber\\
 & \quad + \chi^{\mu\nu}\, \delta F_{\mu\nu} -2\cJ^{\mu}\, \delta A_{ \mu}  + \left( \frac{1}{16\pi}\, \kolo{\cal G}^{\mu\nu} +   \mnabla_{\kappa}{\cal R}^{\mu\nu\kappa}\right)\, \delta g_{\mu\nu}\, .
 \label{var lag matt0}
 \end{align}

However, to obtain the above matter Lagrangian, a Legendre transformation was used, which requires inverting the relations between $\Omega_{\kappa}^{\ \lambda\mu\nu}$ and $\mnabla_{\nu}\Omega_{\kappa}^{\ \lambda\mu\nu}$ on the one hand, and $U^{\kappa}_{\ \lambda\mu\nu}$ and $A^{\kappa}_{\ \lambda\mu\nu}$ on the other. Technically, this inversion becomes significantly easier when the irreducible components are taken into account — and this approach will be adopted in the sequel. Therefore, the variational formula~\eqref{var lag matt} must also be refined. To this end, the following quantities are introduced:
 \begin{align}
     \mathfrak{U}^{\kappa}_{\ \lambda\mu\nu}&:= \kolo{U}^{\kappa}_{\ \lambda\mu\nu} +  D^{\kappa}_{\ \lambda\mu\nu}\, , 
     \label{def mathfrakU} \\
     \mathfrak{U}^{\kappa}_{\ \nu}&:= \mathfrak{U}^{\kappa}_{\ \lambda\mu\nu}\, g^{\lambda \mu} \, . 
     \label{def mathfrakU 2}
 \end{align}
 Using the decomposition of $\Omega_{\kappa}^{\ \lambda\mu\nu}$~\eqref{Omega decomposition} from \textbf{Lemma~\ref{lem dec SigOm}}, the following equalities hold:
 \begin{align}
     \left(\!\kolo{U}^{\kappa}_{\ \lambda\mu\nu} +  D^{\kappa}_{\ \lambda\mu\nu} \right)\, \delta \Omega_{\kappa}^{\ \lambda\mu\nu}  &= \mathfrak{U}^{\kappa}_{\ \lambda\mu\nu} \, \delta \Omega_{\kappa}^{\ \lambda\mu\nu}  = \mathfrak{U}^{\kappa}_{\ \lambda\mu\nu} \delta \left( g_{\kappa\sigma} \Omega^{\sigma \lambda\mu\nu} \right)=\nonumber \\
     &=\mathfrak{U}^{\kappa}_{\ \lambda\mu\nu}\Omega^{\sigma \lambda\mu\nu} \, \delta g_{\kappa\sigma} + \mathfrak{U}_{\kappa  \lambda\mu\nu}  \delta \Omega^{\kappa  \lambda\mu\nu} = \nonumber \\
     &= \left[ \frac 98 \mathfrak{U}^{\ \ \alpha\beta}_{  \mu\nu}  \cO^{(\mu\nu)}  + \frac 9{16} \mathfrak{U}^{\alpha}_{\ \nu} \cO^{(\beta \nu)}   +  \frac 54 \mathfrak{U}^{\ \ \alpha\beta}_{  \mu\nu}  \cO^{[\mu\nu]} + \frac 58 \mathfrak{U}^{\alpha}_{\ \nu} \cO^{[\beta \nu]}+\right. \nonumber \\
     &\quad \left.+\mathfrak{U}^{\alpha}_{\ \lambda\mu\nu}   \widetilde{\Omega}^{\beta \lambda\mu\nu}   \right]\delta g_{\alpha\beta } +  \frac 58 \mathfrak{U}_{\mu\nu}\, \delta \cO^{[\mu\nu]} + \frac 9{16} \mathfrak{U}_{\mu\nu}\, \delta \cO^{(\mu\nu)}+\nonumber \\
     &\quad +  \widetilde{\mathfrak{U}}_{\kappa \lambda\mu\nu}   \, \delta \widetilde{\Omega}^{\kappa  \lambda\mu\nu}  \, ,
     \label{mathUdelOm}
 \end{align}
 where $ \widetilde{\mathfrak{U}}$ denotes the totally traceless part of the tensor ${\mathfrak{U}}$. Since the tensor $\mathfrak{U}_{\kappa\lambda\mu\nu}$ possesses the same symmetries as the tensor density $\Omega^{\kappa\lambda\mu\nu}$, it admits an analogous decomposition -- see \textbf{Lemma~\ref{lem dec SigOm}} and equation~\eqref{Omega decomposition}:
\begin{align}
    \mathfrak{U}_{\kappa\lambda\mu\nu}  &:= \mathfrak{U}^{\sigma}_{\ \lambda\mu\nu}\, g_{\sigma\kappa}=  \widetilde{\mathfrak{U}}_{\kappa\lambda\mu\nu}  + \frac 18 g_{\kappa\nu}\,\mathfrak{U}_{(\lambda\mu)} -\frac 18 \left(g_{\kappa\lambda}\, \mathfrak{U}_{[\mu\nu]} + g_{\kappa\mu}\, \mathfrak{U}_{[\lambda\nu]} \right) +\nonumber\\
      & \quad   - \frac 1{16} \left(g_{\kappa\lambda}\, \mathfrak{U}_{(\mu\nu)} + g_{\kappa\mu}\, \mathfrak{U}_{(\lambda\nu)} \right)  - \frac{5}{24}\left(\mathfrak{U}_{[\kappa\lambda]}\, g_{\mu\nu} + \mathfrak{U}_{[\kappa\mu]}\, g_{\lambda\nu} - 2\mathfrak{U}_{[\kappa\nu]}\, g_{\lambda\mu}   \right) +\nonumber\\
      & \quad  - \frac{3}{16}\left(\mathfrak{U}_{(\kappa\lambda)}\,  g_{\mu\nu} + \mathfrak{U}_{(\kappa\mu)} \, g_{\lambda\nu} - 2\mathfrak{U}_{(\kappa\nu)}\, g_{\lambda\mu}   \right)\, .
        \label{mathfrakU decomposition}
\end{align}

 The term $A \, \delta\left(\mnabla \Omega\right)$ is treated analogously. In particular, $\mnabla \Omega$ decomposes as in formula~\eqref{mathfrakO} from \textbf{Lemma~\ref{lem dec Omega}}:
 \begin{align}
      A^{\kappa}_{\ \lambda\mu}\, \delta \left(\mnabla_{\nu}\Omega_{\kappa}^{\ \lambda\mu\nu} \right)&=  \tA^{\kappa}_{\ \lambda\mu}\, \delta \left(\mnabla_{\nu}\mathfrak{O}_{\kappa}^{\ \lambda\mu\nu} \right) + \frac{5}{18}  A^{\kappa}_{\ \lambda\mu}\, \delta \left(g^{\lambda\mu}\mnabla_{\nu}\cO_{\kappa}^{\ \nu} \right) = \nonumber\\
      &=\tA^{\kappa}_{\ \lambda\mu}\, \delta \left(\mnabla_{\nu}\mathfrak{O}_{\kappa}^{\ \lambda\mu\nu} \right) +  \frac{5}{18}  h^{\kappa} \, \delta \left( \mnabla_{\nu}\cO_{\kappa}^{\ \nu} \right)  - \left(A^{\kappa\alpha\beta} \mnabla_{\nu}\cO_{\kappa}^{\ \nu}\right)\delta g_{\alpha\beta}\, .
      \label{AdelCovOm}
 \end{align}

 Including above equations \eqref{mathUdelOm} and \eqref{AdelCovOm}, the variation of the matter Lagrangian \eqref{var lag matt0} takes the following form:
 \begin{align}
     \delta \Lag_{\rm matt}   & =\left( \frac{1}{16\pi}\, \kolo{\cal G}^{\alpha\beta} +   \mnabla_{\kappa}{\cal R}^{\alpha\beta\kappa} + A^{\kappa\alpha\beta} \mnabla_{\nu}\cO_{\kappa}^{\ \nu}-\frac 98 \mathfrak{U}^{\ \ \alpha\beta}_{  \mu\nu}  \cO^{(\mu\nu)}  - \frac 9{16} \mathfrak{U}^{\alpha}_{\ \nu} \cO^{(\beta \nu)} +\right. \nonumber \\
     &\quad \left.  -  \frac 54 \mathfrak{U}^{\ \ \alpha\beta}_{  \mu\nu}  \cO^{[\mu\nu]} - \frac 58 \mathfrak{U}^{\alpha}_{\ \nu} \cO^{[\beta \nu]}-\mathfrak{U}^{\alpha}_{\ \lambda\mu\nu}   \widetilde{\Omega}^{\beta \lambda\mu\nu}   \right)\delta g_{\alpha\beta }+ \chi^{\mu\nu}\, \delta F_{\mu\nu} -2\cJ^{\mu}\, \delta A_{ \mu}  + \nonumber\\ 
     &\quad -  \frac 58 \mathfrak{U}_{\mu\nu}\, \delta \cO^{[\mu\nu]} - \frac 9{16} \mathfrak{U}_{\mu\nu}\, \delta \cO^{(\mu\nu)}-  \widetilde{\mathfrak{U}}_{\kappa \lambda\mu\nu}   \, \delta \widetilde{\Omega}^{\kappa  \lambda\mu\nu} +\nonumber \\
     &\quad  -  \tA^{\kappa}_{\ \lambda\mu}\, \delta \left(\mnabla_{\nu}\mathfrak{O}_{\kappa}^{\ \lambda\mu\nu} \right) -  \frac{5}{18}  h^{\kappa} \, \delta \left( \mnabla_{\nu}\cO_{\kappa}^{\ \nu} \right)    \, .
 \label{var lag matt} 
 \end{align}

 \subsection{Field equations}
The variational formula $\delta\Lag_{\rm matt}$ \eqref{var lag matt}  generates the following field equations:
\begin{enumerate}
    \item standard Euler-Lagrange system for the potential $A_{\mu}$:
    \begin{align}
        \frac{\partial \Lag_{\rm matt}}{\partial A_{\mu}} &= -2 \cJ^{\mu} \, , & \frac{\partial \Lag_{\rm matt}}{\partial F_{\mu\nu}}&=  \chi^{\mu\nu}\, , 
        \label{mat eq A}
    \end{align}
    where:
    \begin{align}
        F_{\mu\nu} &=\mnabla_{\mu}A_{\nu} - \mnabla_{\nu}A_{\mu} \, , &      \cJ^{\mu}&=\mnabla_{\nu}\chi^{\mu\nu} \, , 
    \end{align}
    cf. formulae for $F_{\mu\nu}$ \eqref{rozklad pelny F} and $\cJ^{\mu}$ \eqref{cal J};

    \item  a specific Euler-Lagrange system with constraints for the tensor density $\Omega_{\kappa}^{\ \lambda\mu\nu}$:
    \begin{align}
        \frac{\partial \Lag_{\rm matt}}{\partial \left(\mnabla_{\nu} \cO_{\kappa}^{\ \nu}\right)} &= - \frac{5}{18} h^{\kappa}\, , & \frac{\partial \Lag_{\rm matt}}{\partial   \cO^{(\mu \nu)}} &= - \frac 9{16} \mathfrak{U}_{(\mu\nu)}\, ,&   \frac{\partial \Lag_{\rm matt}}{\partial   \cO^{[\mu \nu]}} &= - \frac 58 \mathfrak{U}_{[\mu\nu]}\, ,\nonumber \\
        \frac{\partial \Lag_{\rm matt}}{\partial\left(\mnabla_{\nu}\mathfrak{O}_{\kappa}^{\ \lambda\mu\nu} \right)} &=-  \tA^{\kappa}_{\ \lambda\mu}\, ,&   \frac{\partial \Lag_{\rm matt}}{\partial\widetilde{\Omega}^{\kappa  \lambda\mu\nu}}&= -\widetilde{\mathfrak{U}}_{\kappa \lambda\mu\nu} \, , 
        \label{mat eq Omega}
    \end{align}
    where:
    \begin{align}
         \mnabla_{\nu}\Omega_{\kappa}^{\ \lambda\mu\nu} &= \mnabla_{\nu}\left[\mathfrak{O}_{\kappa}^{\ \lambda\mu\nu} - \frac{1}{18}\left( \delta_{\kappa}^{\lambda}\, \cO^{\mu\nu} +\delta_{\kappa}^{\mu}\,\cO^{\lambda\nu} - 5g^{\lambda\mu}\, \cO_{\kappa}^{\ \nu} \right) \right]\, , \\
          \Omega^{\kappa\lambda\mu\nu}  & =  \widetilde{\Omega}^{\kappa\lambda\mu\nu}  + \frac 18 g^{\kappa\nu}\cO^{(\lambda\mu)} -\frac 18 \left(g^{\kappa\lambda} \cO^{[\mu\nu]} + g^{\kappa\mu}\cO^{[\lambda\nu]} \right) +\nonumber\\
      & \quad   - \frac 1{16} \left(g^{\kappa\lambda} \cO^{(\mu\nu)} + g^{\kappa\mu}\cO^{(\lambda\nu)} \right)  - \frac{5}{24}\left(\cO^{[\kappa\lambda]} g^{\mu\nu} + \cO^{[\kappa\mu]} g^{\lambda\nu} - 2\cO^{[\kappa\nu]} g^{\lambda\mu}   \right) +\nonumber\\
      & \quad  - \frac{3}{16}\left(\cO^{(\kappa\lambda)} g^{\mu\nu} + \cO^{(\kappa\mu)} g^{\lambda\nu} - 2\cO^{(\kappa\nu)} g^{\lambda\mu}   \right)\, , \\
      \mathfrak{U}^{\kappa}_{\ \lambda\mu\nu}&= \kolo{U}^{\kappa}_{\ \lambda\mu\nu} +  D^{\kappa}_{\ \lambda\mu\nu}\, ,  \\
       D^{\kappa}_{\ \lambda\mu\nu} &=  \frac 23\, \left( \mnabla_{\nu}{A}^{\kappa}_{\ \lambda  \mu  } -  \mnabla_{(\lambda }A^{\kappa}_{\ \mu) \nu }\right) - \frac 29\, \mnabla_{\sigma}\left( \delta^{\kappa}_{\nu}  {A}^{\sigma}_{\ \lambda\mu} - \delta^{\kappa}_{(\lambda }   {A}^{\sigma}_{\ \mu)\nu }\right) \, , \\
       \mathfrak{U}^{\kappa}_{\ \nu}&= \mathfrak{U}^{\kappa}_{\ \lambda\mu\nu}\, g^{\lambda \mu} \, , \\
       A^{\kappa}_{\ \lambda \mu} &=\widetilde{ A}^{\kappa}_{\ \lambda \mu} -\frac{1}{18}\left( \delta^{\kappa}_{\lambda}\, h_{\mu} + \delta^{\kappa}_{\mu}\, h_{\lambda} - 5g_{\lambda\mu}\, h^{\kappa}\right)\, ,
    \end{align}
    cf. formulae for  $ \mnabla_{\nu}\Omega_{\kappa}^{\ \lambda\mu\nu}$~\eqref{mathfrakO}, $\Omega^{\kappa\lambda\mu\nu} $~\eqref{Omega decomposition}, $\mathfrak{U}^{\kappa}_{\ \lambda\mu\nu}$~\eqref{def mathfrakU}, $D^{\kappa}_{\ \lambda\mu\nu}$~\eqref{def tensor D}, $\mathfrak{U}^{\kappa}_{\ \nu}$~\eqref{def mathfrakU 2}, and $A^{\kappa}_{\ \lambda \mu}$~\eqref{def tildeA};

    \item Einstein equation:
    \begin{align}
        \frac{\partial \Lag_{\rm matt}}{\partial g_{\alpha\beta}} &= \frac{1}{16\pi}\, \kolo{\cal G}^{\alpha\beta} +   \mnabla_{\kappa}{\cal R}^{\alpha\beta\kappa} + A^{\kappa\alpha\beta} \mnabla_{\nu}\cO_{\kappa}^{\ \nu}-\frac 98 \mathfrak{U}^{\ \ (\alpha\beta)}_{  \mu\nu}  \cO^{(\mu\nu)}  -  \frac 54 \mathfrak{U}^{\ \ (\alpha\beta)}_{  \mu\nu}  \cO^{[\mu\nu]} +  \nonumber \\
     &\quad- \frac 9{32} \mathfrak{U}^{\alpha}_{\ \nu} \cO^{(\beta \nu)} - \frac 9{32} \mathfrak{U}^{\beta}_{\ \nu} \cO^{(\alpha \nu)}    - \frac 5{16} \mathfrak{U}^{\alpha}_{\ \nu} \cO^{[\beta \nu]} - \frac 5{16} \mathfrak{U}^{\beta}_{\ \nu} \cO^{[\alpha \nu]} - \mathfrak{U}^{(\alpha}_{\ \lambda\mu\nu}   \widetilde{\Omega}^{\beta) \lambda\mu\nu}  \, , 
        \label{passeineq}
    \end{align}
    where the “extra term'' $\mnabla {\cal R}$ is given by:
    \begin{align}
        \mnabla_{\kappa}\mathcal{R}^{\mu\nu\kappa} &=\frac{\sqrt{|\det g|}}{16\pi} \mnabla_{\kappa}\left[ A^{\kappa\mu\nu} - \frac{6}{5}\, g^{\kappa(\mu} A^{\nu)}+\frac 12\, \left(\frac 65\, A^{\kappa} - h^{\kappa} \right)\, g^{\mu\nu} \right]\, , 
    \end{align}
    cf. formula \eqref{cal R}.
    \end{enumerate}

The symplectic formula $\delta \Lag_{\rm matt}$~\eqref{var lag matt} induces that the configuration space is given by  $(A_{\mu}, F_{\mu\nu}, \cO^{(\mu\nu)}, \cO^{[\mu\nu]},  \widetilde{\Omega}_{\kappa}^{\ \lambda\mu\nu},   \mnabla_{\nu} \mathfrak{O}_{\kappa}^{\ \lambda\mu\nu} ,\mnabla_{\nu}\cO^{\mu\nu} ,g_{\mu\nu})$ and the matter Lagrangian $\Lag_{\rm matt}$~\eqref{def lag matt} is a function of those quantities.

\sectionmark{Passage -- examples}
 
\section{Passage from the affine picture to the metric picture -- examples}
\sectionmark{Passage -- examples}

\subsection{Theory of the full Ricci tensor}
\label{pass full ricci tensor}

The metric picture is obtained via the Legendre transformation from the affine picture, which was the main topic of the previous section. Specifically, the corresponding matter Lagrangian $\Lag_{\rm matt}$~\eqref{def lag matt} must be derived. However, in this theory, the traceless part $U^{\kappa}_{\ \lambda\mu\nu}$ does not appear, which simplifies the transition considerably:
\begin{align}
    \Lag_{\rm matt} =  \Lag_A + \mnabla_{\kappa} {\cal R}_{\sigma}^{\ \sigma\kappa} - \Lag_H \, ,
    \label{lag matt ricci v0}
\end{align}
where the divergence part $\mnabla {\cal R}$ vanishes -- see formulae~\eqref{divR},~\eqref{A rel J},~\eqref{h rel J}, and~\eqref{lor gauge}:
\begin{align}
     \mnabla_{\kappa}{\cal R}_{\sigma}^{\ \sigma\kappa} =  \frac{\sqrt{|\det g|}}{16\pi} \, \mnabla_{\kappa}\left(\frac{6}{5}\, A^{\kappa} - h^{\kappa}\right) = 2\sqrt{|\det g|}\, \mnabla_{\mu}\, \cJ^{\mu}=0\, .
\end{align}
Of course, the above quantities have to be written in a proper control mode  $(A_{\mu}, F_{\mu\nu} ,g_{\mu\nu})$  -- cf. the symplectic formula in the metric picture~\eqref{var lag matt}. 
Firstly, the affine Lagrangian~\eqref{lag iksdfbosgv} equals:
\begin{align}
    \Lag_{A} = \frac{1}{8\pi\Lambda}\, \sqrt{|KKKK  + KKFF  |}\, ,
\end{align}
where the terms $KKKK$ and $KKFF$ are defined in equations~(\ref{KKKK v0}–\ref{KKFF v0}). However, based on the full analysis presented in \textbf{Chapter~\ref{affine ricci}}, and in particular the Einstein equation~\eqref{eqq1}, the above affine Lagrangian takes the form:
\begin{align}
    \Lag_A &= \frac{1}{8\pi\Lambda}\, \sqrt{|\Lambda^4  \det g + \Lambda^2\, ggFF  |} = \frac{\Lambda\sqrt{|\det g|}}{8\pi}\, \sqrt{\left|1 + \frac{1}{  \Lambda^2 \det g}\, ggFF \right|}=\nonumber\\
    &= \frac{\Lambda \sqrt{|\det g|}}{8\pi} \, \sqrt{\left| 1 +\frac{1}{2\Lambda^2}\,F_{\mu\nu}F^{\mu\nu}  \right|}\, . \label{pass LagA v0}
\end{align}
Here, the term $ggFF$ was defined in~\eqref{ggFFV0}.

\

Next is the Hilbert Lagrangian~\eqref{LagH}, where the value of the metric Ricci curvature $\kolo{K}_{\mu\nu}$ is derived from the  Einstein equation~\eqref{eqq2}. Thus:
\begin{align}
    \Lag_H = \frac{\sqrt{|\det g|}}{8\pi}\, \left(2\Lambda + 3 A_{\sigma}A^{\sigma} \right)\, .
    \label{LagH uni v0}
\end{align}
Then, the corresponding matter Lagrangian~\eqref{lag matt ricci v0} is the following:
\begin{align}
    \Lag_{\rm matt} &=\frac{\Lambda\sqrt{|\det g|}}{8\pi }\,  \sqrt{\left|1+\frac{1}{2\Lambda^2}\,F_{\mu\nu}F^{\mu\nu}\right|} - \frac{\sqrt{|\det g|}}{8\pi}\, \left(2\Lambda + 3 A_{\sigma}A^{\sigma} \right)\, . 
    \label{Lmatt v0}
\end{align}
It could also be approximated (expanded around $\Lambda$-vacuum solution) in the following way:
\begin{align}
    \Lag_{\rm matt} &\approx   - \frac{\sqrt{|\det g|}}{8\pi}\, \left(\Lambda + 3 A_{\sigma}A^{\sigma} \right)  +\frac{ \sqrt{|\det g|}}{ 32\pi\Lambda}\,  F_{\alpha\beta}\, F^{\alpha\beta}    \, . 
    \label{app Lmatt v0}
\end{align}

\subsubsection{Field equations}
The field equations associated with the matter Lagrangian~\eqref{app Lmatt v0} are as follows -- cf. symplectic formula \eqref{var lag matt}:
\begin{enumerate}
    \item standard Euler-Lagrange system for the potential $A_{\mu}$ \eqref{mat eq A}, where:
    \begin{align}
        \frac{\partial \Lag_{\rm matt}}{\partial A_{\mu}}&= -\frac{3\sqrt{|\det g|}}{4\pi}\, A^{\mu}= -2 \cJ_{\mu}\, , \\
        \frac{\partial \Lag_{\rm matt}}{\partial F_{\mu\nu}}&=\frac{ \sqrt{|\det g|}}{ 16\pi\Lambda}\,  F^{\mu\nu}=\chi^{\mu\nu}\, ,
    \end{align}
    which precisely reproduce the non-metricity equation  \eqref{A rel J 1} and the constitutive relation \eqref{rel konst0};
    
    \item Einstein equation \eqref{passeineq}, where:
    \begin{align}
        \frac{\partial \Lag_{\rm matt}}{\partial g_{\mu\nu}} &=- \frac{\sqrt{|\det g|}}{16\pi}\, \left(\Lambda + 3 A_{\sigma}A^{\sigma} \right) g^{\mu\nu}  +\frac{ \sqrt{|\det g|}}{ 64\pi\Lambda}\,  F_{\alpha\beta}\, F^{\alpha\beta} \, g^{\mu\nu}+ \nonumber\\
        &\quad +\frac{3\sqrt{|\det g|}}{8\pi}\, A^{\mu}A^{\nu} - \frac{ \sqrt{|\det g|}}{ 16\pi\Lambda}\,  F^{\mu}_{\ \alpha }\, F^{\nu\alpha}\, , \\
        \mnabla_{\kappa}{\cal R}^{\mu\nu\kappa}&=0\, ,
    \end{align}
    which precisely reproduce the previously obtained Einstein equation \eqref{eqq21}.
\end{enumerate}

\subsubsection{Unification}
\label{uni0}
 The structure of the above theory is very similar to the Einstein-Maxwell theory with a cosmological constant \(\Lambda\). First, the skew-symmetric Ricci tensor \(F_{\mu\nu}\) is, \textit{by definition}, a closed 2-form — see~\eqref{rozklad pelny F}:
\begin{align}  
    F_{\mu\nu} = A_{\nu,\mu}-A_{\mu,\nu}\, ,  
\end{align}  
where \(A_{\mu}:=\frac{1}{2} \, N^{\sigma}_{\ \mu\sigma}\) is a potential — see~\eqref{slad N}.  

Secondly, the constitutive relation~\eqref{rel konst0} between the momentum \(\chi^{\mu\nu}\) and \(F_{\mu\nu}\), determined by the symplectic structure, is analogous to linear vacuum electrodynamics — cf. \textbf{Appendix~\ref{electro}}, equation~\eqref{constitutive}.  

Third argument is based on the Einstein equation~\eqref{eqq1}, where the right-hand side has an identical structure to the stress-energy tensor for electromagnetic fields (cf. formula~\eqref{calT ed}):
\begin{align}
     { T}^{\mu\nu}   =  f^{\mu \alpha}f^{\nu}_{\ \alpha} - \frac{1}{4}\, g^{\mu\nu}\, f_{\alpha\beta}\,f^{\alpha\beta} \, .
     \label{SET}
\end{align}
Finally, the matter Lagrangian \eqref{app Lmatt v0} contains a term, which is quadratic in the tensor $F$ and is very similar to the electromagnetic Lagrangian \eqref{Led}:
\begin{align}
    \Lag_{ed} = -\frac{\sqrt{|\det g|} }{4}\, f_{\alpha\beta}f^{\alpha\beta}\, .
\end{align}

However, there is a difference in the coupling constant, particularly in its sign and unit. The Faraday 2-form $f_{\mu\nu}$ has a length dimension ($[\textbf{cm}]$ in the geometrical unit system \cite{Gravitation}), whereas the skew-symmetric Ricci tensor $F_{\mu\nu}$ is dimensionless. This observation suggests that the relation between the skew-symmetric Ricci tensor \(F_{\mu\nu}\) and the Faraday 2-form~\(f_{\mu\nu}\) must be the following:
\begin{align}  
    F_{\mu\nu} := \pm \sqrt{8\pi |\Lambda|}\, f_{\mu\nu}\, ,  
    \label{uniF v0}
\end{align}  
under the assumption that:
\begin{align}
\Lambda <0\qquad   \Longrightarrow \qquad  \Lambda=-|\Lambda|\, .
    \label{uniconst v0}
\end{align}
Of course, the chosen “\(\pm\)'' sign does not matter, since only quadratic terms in \( F_{\mu\nu} \) appear in the Lagrangian. Accordingly, this implies an identical relation between the potential \( A_{\mu} \) and the electromagnetic potential \( a_{\mu} \)~\eqref{def: dwuforma faradaya}:  
\begin{align}
    A_{\mu} :=  \pm \sqrt{8\pi |\Lambda|}\,  a_{\mu} \, .
    \label{uniA v0}
\end{align}  
To reconstruct the same symplectic structure as in Maxwellian electrodynamics, the momentum \( \chi^{\mu\nu} \)~\eqref{rel konst0} should be related to the dual electromagnetic tensor density \( \cF^{\mu\nu} \)~\eqref{constitutive} as follows:  
\begin{align}
    \chi^{\mu\nu} =\frac{ \sqrt{|\det g|}}{16\pi \Lambda} \, F^{\mu\nu} := \mp \frac{1}{\sqrt{32\pi|\Lambda|}}\, \cF^{\mu\nu}\, .
\end{align}  
Thus,  
\begin{align}
    \chi^{\mu\nu}\, \delta F_{\mu\nu} = -\frac{1}{2}\, \cF^{\mu\nu}\, \delta f_{\mu\nu}\, ,
\end{align}  
which is precisely the same as in electrodynamics~\eqref{var lag em}.  

\ 

The obtained Einstein equation~\eqref{eqq1} for the general Ricci tensor \( K_{\mu\nu} \) takes the form of the standard Einstein-Maxwell equation with a negative cosmological constant:
\begin{align}
    K_{\mu\nu} =  -|\Lambda|\, g_{\mu\nu} +8\pi\, \left(f^{\mu\alpha}\, f^{\nu}_{\ \alpha} - \frac 14\, f_{\alpha\beta}\, f^{\alpha\beta}\, g^{\mu\nu} \right)\, ,
    \label{eqq11}
\end{align}
whereas the metric Einstein equation~\eqref{eqq21}  is more closely related to the Einstein-Proca theory \eqref{calT P}, due to the explicit appearance of the potential~\( a_{\mu} \):
\begin{align}
   \kolo{G}_{\mu\nu}   = |\Lambda|\, g_{\mu\nu} +8\pi\, \left[\left(f^{\mu\alpha}\, f^{\nu}_{\ \alpha} - \frac 14\, f_{\alpha\beta}\, f^{\alpha\beta}\, g^{\mu\nu} \right) +   6 |\Lambda| \,\left( a_{\mu}\,a_{\nu}- \frac12\, g_{\mu\nu}\, a_{\sigma}a^{\sigma}\right)\right]\, .  
   \label{ein eq4}
\end{align}
Of course, the unification statements can also be incorporated into the approximated matter Lagrangian~\eqref{app Lmatt v0}:
\begin{align}
    \Lag_{\rm matt} &\approx   \frac{|\Lambda|\, \sqrt{\left| \det g\right|}}{8\pi } -\frac{ \sqrt{|\det g|}}{4}\,\left( f_{\alpha\beta}\, f^{\alpha\beta}   + 2\frac{m^2}{\hbar^2} \, a_{\sigma}\,a^{\sigma} \right)  \, 
    \label{app  uni Lmatt v0}
\end{align}
and as before, it is rather Einstein-Proca than Einstein-Maxwell theory -- cf. \textbf{Appendices \ref{electro} and \ref{Proca th}}. However, the  mass parameter $m$ is very small, and given by the following formula:
\begin{align}
    \frac{m^2}{\hbar^2} &= 6|\Lambda|\, , \\
    m &= \hbar \sqrt{6|\Lambda|} \approx 2 \cdot 10^{-94}~[\textbf{cm}] \approx 2.7 \cdot 10^{-69}~[{\rm kg}]\, ,
    \label{mass v0}
\end{align}
cf. the formula for the stress-energy tensor density ${\cal T}^{\mu\nu}$ in equation~\eqref{calT P}, or the formula for the Proca Lagrangian \eqref{LagP}. The $[\textbf{cm}]$ unit refers to the geometrical unit system (cf.~\cite{Gravitation}), whereas $[\text{kg}]$ refers to the SI unit system. The value of $\hbar$ in geometrical units is provided in \textbf{Appendix~\ref{Proca th}}, formula~\eqref{hbar}. The cosmological constant $\Lambda$ (in geometrical units), as proposed by Ya. Zel'dovich, is taken to be (cf.~\cite{Zeldovich1, Zeldovich2}, or \cite{Gravitation}, p.~411, Ex.~17.5):
\begin{align}
    \Lambda \approx 10^{-57}~[\textbf{cm}^{-2}]\, .
    \label{lambda value}
\end{align}

The effective cosmological parameter  \( \Lambda_{\rm eff} \)~\eqref{lam effv0} is the following:
\begin{align}
      \Lambda_{\rm eff}  =-|\Lambda|\left( 1 -12\pi \, a_{\kappa}a^{\kappa} \right)\, .
 \end{align}

\

Interestingly, the affine formulation of the standard Einstein-Maxwell theory (without cosmological constant $\Lambda$) also can be considered -- cf. \cite{Kij-Fer}. However, the electromagnetic tensor $f_{\mu\nu}$ is not related to the skew-symmetric Ricci tensor $F_{\mu\nu}$, but as an external field.

\subsubsection{Born-Infeld theory}

It is very interesting to see what happens if the non-approximated Lagrangian is used in this passage -- cf. variant $V_0$ \eqref{w0} and formula \eqref{Lag F0}:
\begin{align}
    \Lag_A &= \frac{\sqrt{|\det (K+F)|}}{8\pi \Lambda} \, .
\end{align}
Then, the non-perturbed solution of the Einstein equation~\eqref{Kzv0} and the unification formulae~(\ref{uniF v0}-\ref{uniconst v0}) imply:
\begin{align}
    \Lag_A = \frac{\sqrt{\left|\det \left(\Lambda g \pm \sqrt{8\pi|\Lambda|}\, f\right)\right|}}{8\pi \Lambda} = - \frac{|\Lambda|}{8\pi}\,  \sqrt{\left|\det \left( g \mp \frac 1{\mathfrak{b}}\, f\right)\right|} \, ,
\end{align}
where 
\begin{align}
    \mathfrak{b}:=\sqrt{\frac{|\Lambda|}{8\pi}}\, .
    \label{BI const}
\end{align}
The Hilbert Lagrangian~\eqref{LagH uni v0}  equals:
\begin{align} 
    \Lag_H =  - \frac{|\Lambda|\, \sqrt{|\det g|}}{4\pi}\, \left( 1- 12\pi  a_{\mu}a^{\mu}\right) = -2\mathfrak{b}^2\sqrt{|\det g|}\left( 1- 12\pi  a_{\mu}a^{\mu}\right) \, .
\end{align} 
Then, the matter Lagrangian \eqref{lag matt ricci v0} is given by:
\begin{align}
    \Lag_{\rm   matt} = \Lag_A   - \Lag_H =- \mathfrak{b}^2\, \sqrt{\left|\det \left( g \mp  \frac{1}{\mathfrak{b}}\, f\right)\right|} +2\mathfrak{b}^2\,\sqrt{|\det g|}\, \left( 1- 12\pi a_{\mu}a^{\mu}\right)  \, .
\end{align}
The above theory can be viewed as an extension of the standard Born-Infeld electromagnetism~\cite{born-infeld}, in which the cosmological constant $\Lambda$ (encoded in $\mathfrak{b}$) plays the role of a coupling constant. This Lagrangian differs slightly from the standard Born-Infeld form, as it also includes couplings between the gravitational and electromagnetic fields, and additionally contains potential terms. As a result, it describes a non-trivial interaction between the gravitational field and a “massive” bosonic field, with the mass parameter $m^2$ given by:
\begin{align}
    m^2 = 6\hbar^2|\Lambda| = 48\pi \hbar^2 \mathfrak{b}^2\, .
\end{align}
As it was mentioned before, this mass parameter is very small -- see \eqref{mass v0}.

\subsection{Variant \texorpdfstring{$V_1$}{V1}}

This transition, in the opposite to the previous one, involves non-trivial terms related to the traceless Riemann tensor $W^{\kappa}_{\ \lambda\mu\nu}$, which is the main source of difficulty.

\ 

The first step is to rewrite the affine Lagrangian~\eqref{lag A1} in the proper control mode (cf. the symplectic formula in the metric picture~\eqref{var lag matt}):
\begin{align}
    \Lag_{A} = \alpha\, \sqrt{|KKKK + KKKW + KKFF + KKFW + KKWW|}\, ,
\end{align}
where the terms $KKKK$, $KKKW$, $KKFF$, $KKFW$, and $KKWW$ are defined in equations~(\ref{KKKKv1}–\ref{rozklad w1}). However, based on the full analysis presented in \textbf{Chapter~\ref{VAR V1}}, and in particular the Einstein equation~\eqref{ein eq0 v1}, the above affine Lagrangian takes the form:
\begin{align}
    \Lag_A &= \alpha\, \sqrt{|\Lambda^4 \sigma \gamma^2 \det g + \Lambda^2\, ggFF + \Lambda^2\, ggFW + \Lambda^2\, ggWW|} = \nonumber \\
    &=\frac{\Lambda\sqrt{|\det g|}}{8\pi}\, \sqrt{\left|1 + \frac{1}{\sigma\gamma^2 \Lambda^2 \det g}\left(ggFF+ggFW + ggWW \right) \right|}=\nonumber\\
    &= \frac{\Lambda \sqrt{|\det g|}}{8\pi} \, \bigg| 1 + \frac{27}{88\Lambda^2} \left( \frac{356}{225}\, F_{\alpha\beta}F^{\alpha\beta} - \frac{16}{45}\, F_{\alpha\beta}\, W^{[\alpha  \beta]} + \right. \nonumber\\
    & \quad  \left. - \frac{2}{3} W_{[\alpha\beta]}W^{[\alpha\beta]}    + \frac{2}{3} W_{(\alpha\beta)}W^{(\alpha\beta)} - \frac{4}{3} \widetilde{W}_{[\alpha\beta]\kappa\lambda} \widetilde{W}^{[\kappa\lambda]\alpha\beta}\right)\bigg|^{1/2}\, . \label{pass LagA v1}
\end{align}
Here, the characteristic constants $\alpha$, $\gamma^2$, and $\sigma$ are given in~\eqref{const1}, while the terms $ggFF$, $ggFW$, and $ggWW$ are defined in~\eqref{ggFF v1}, \eqref{ggFW2 v1} and \eqref{ggWW2 v1} respectively.

To obtain the appropriate matter Lagrangian $\Lag_{\text{matt}}$~\eqref{def lag matt}, the tensor $W$ must be expressed in terms of the momentum $\Omega$, which is equivalent to $\Sigma$~\eqref{rel Omega Sigma} — the momentum canonically conjugate to $W$. The momentum $\Sigma$ is decomposed into four independent components, generating four field equations (\ref{sympl sigma1 v1}–\ref{sympl sigma4 v1}), all of which can be inverted, except for the vanishing one~\eqref{sympl sigma3 v1}:
\begin{align}
   W_{[\mu\nu]} &= \frac{88 \cdot 16 \pi \Lambda}{27\sqrt{|\det g|}} \cdot \left(- \frac {5}{8}\right)\, \Sigma_{[\mu\nu]} - \frac{4}{15}F_{\mu\nu}\, , 
   \label{eqskewW v1}\\
   W_{(\mu\nu)} &=\frac{88 \cdot 16 \pi \Lambda}{27\sqrt{|\det g|}} \cdot   \frac{9}{16} \Sigma_{(\mu\nu)}\, , 
   \label{eqsymW v1}\\
   \widetilde{W}_{[\mu\nu]\kappa\lambda}&=\frac{88 \cdot 16 \pi \Lambda}{27\sqrt{|\det g|}}  \cdot\left(-\frac 38 \right)\widetilde{\Sigma}_{[\kappa\lambda]\mu\nu}\, .
   \label{symp w v1}
\end{align}
Therefore, the affine Lagrangian~\eqref{pass LagA v1} equals:
\begin{align}
    \Lag_A &=   \frac{\Lambda \sqrt{|\det g|}}{8\pi} \left|1 + \frac{1}{2  \Lambda^2}  F_{\alpha\beta}F^{\alpha\beta}  -  \frac{88\cdot 200\pi^2}{81|\det g|}  \Sigma_{[\alpha\beta]}\Sigma^{[\alpha\beta]} +   \right.\nonumber \\
    & \quad \left. +\frac{11\cdot 16\pi^2}{|\det g|}   \Sigma_{(\alpha\beta)}\Sigma ^{(\alpha\beta)}   - \frac{88\cdot 16\pi^2}{9|\det g|} \widetilde{\Sigma}_{[\alpha\beta]\kappa\lambda} \widetilde{\Sigma}^{[\kappa\lambda]\alpha\beta}  \right|^{1/2}\, .
\end{align}
Surprisingly, the “mixing” term $F_{\alpha\beta}\Sigma^{\alpha\beta}$ vanishes. Finally, replacing the momentum~$\Sigma$ with the momentum $\Omega$~\eqref{rel Sigma} yields:
\begin{align}
    \Lag_A &=   \frac{\Lambda \sqrt{|\det g|}}{8\pi} \left|1 + \frac{1}{2  \Lambda^2}    F_{\alpha\beta}F^{\alpha\beta}  -\frac{88 \cdot 50\pi^2}{81|\det g|}  \cO_{[\alpha\beta]}\cO^{[\alpha\beta]} + \right.\nonumber \\
    & \quad \left. +\frac{44\pi^2}{|\det g|}   \cO_{(\alpha\beta)} \cO^{(\alpha\beta)}   - \frac{88\cdot 16\pi^2}{81|\det g|}   \left(\widetilde{\Omega}_{\kappa\lambda\mu\nu} \widetilde{\Omega}^{\lambda\kappa\nu\mu} + 2 \widetilde{\Omega}_{\kappa\lambda\mu\nu} \widetilde{\Omega}^{\nu\lambda\mu\kappa}\right)   \right|^{1/2}\, ,
    \label{pass LagA1 v1}
\end{align}
where $\widetilde{\Omega}$ denotes the totally traceless part of the momentum $\Omega$ -- see \textbf{Lemma~\ref{lem dec SigOm}} and formula \eqref{Omega decomposition}. 

\ 

The next step in the passage to the metric picture involves deriving the divergence term $\mnabla{\cal R}$~\eqref{divR}:
\begin{align}
      \mnabla_{\kappa}{\cal R}_{\sigma}^{\ \sigma\kappa} &= \frac{\sqrt{|\det g|}}{16\pi} \,  \left(\mnabla_{\kappa}N_{\sigma}^{\ \sigma\kappa} - \mnabla_{\kappa}N^{\kappa\sigma}_{\ \ \sigma}\right) \, .
\end{align}
The second term in the above formula vanishes via \textbf{Lemma~\ref{lm rel h A}}:
\begin{align}
\mnabla_{\kappa}N^{\kappa \sigma}_{\ \ \sigma} = -\frac{80\pi}{3\sqrt{|\det g|}}\, \mnabla_{\kappa}\mathcal{J}^{\kappa} =0\, .
\end{align} 
The remaining part was already derived in \textbf{Lemma~\ref{lem presence}}, formula~\eqref{pot Ao},~so:
\begin{align}
   \mnabla_{\kappa} N_{\sigma}^{\ \sigma\kappa }= \frac{8\pi}{3\sqrt{|\det g|}}\left(2\,\underbrace{  \mnabla_{\kappa} \mathcal{J}^{\kappa}}_{=0} + 3\mnabla_{\kappa}\mnabla_{\nu} \mathcal{O}^{\kappa  \nu}   \right) = \frac{8\pi}{\sqrt{|\det g|}}\, \mnabla_{\kappa}\mnabla_{\nu} \mathcal{O}^{(\kappa  \nu)}   \,.
\end{align} 
Interestingly, symmetrisation in the above result is not necessary, as the skew-symmetric part vanishes as a consequence of \textbf{Lemma~\ref{lem div}}. Whence:
\begin{align}
      \mnabla_{\kappa}{\cal R}_{\sigma}^{\ \sigma\kappa} &= \frac 12 \mnabla_{\kappa}\mnabla_{\nu} \mathcal{O}^{\kappa  \nu}   \,.
      \label{divpart v1}
\end{align}
Next, from the Einstein equation~\eqref{eqq2v1} the Hilbert Lagrangian $\Lag_H$ has to be derived:
\begin{align}
    \Lag_H &= \pi^{\mu\nu}\, \kolo{K}_{\mu\nu} = \frac{\sqrt{|\det g|}}{16\pi}\left(4\Lambda - Q_{\alpha}^{\ \alpha} \right) = \nonumber\\
    &=\frac{\Lambda \sqrt{|\det g|}}{4\pi} - \frac{4\pi}{\sqrt{|\det g|}}\bigg\{ \left(\mnabla_{\alpha}\Omega^{\kappa\lambda\mu\alpha}\right) \mnabla_{\beta}\left(\Omega_{\kappa\lambda\mu}^{\ \ \ \beta}-2\Omega_{\lambda\mu\kappa}^{\ \ \  \beta}\right)+ \nonumber\\
     &  \quad -\frac 12  \left(\mnabla_{\alpha}\cO_{\kappa}^{\ \alpha}\right)\left(\mnabla_{\beta}\cO^{\kappa \beta}\right)- 2\cJ^{\kappa}\left(\mnabla_{\alpha}\cO_{\kappa}^{\ \alpha}\right)-\frac 23\, \cJ_{\kappa}\cJ^{\kappa}\bigg\} +\frac{1}{2} \left(\mnabla_{\kappa}\mnabla_{\lambda}\cO^{\kappa\lambda} \right) \, . 
\end{align}
However, objects like $\cJ$ and $\!\mnabla \mnabla \cO$ are not allowed in this description — cf. the symplectic formula~\eqref{var lag matt} -- just as velocities are forbidden in the Hamiltonian formalism. The current $\cJ$ is easily eliminated using the formula~\eqref{pot Aov1} from the non-metricity decomposition:
\begin{align}
    \mathcal{J}_{\kappa} &= \frac{3\sqrt{|\det g|}}{8\pi} A_{\kappa} - \frac{3}{2} \mnabla_{\nu} \mathcal{O}_{\kappa}^{\ \nu} \, ,
    \label{cJ v1}
\end{align}
whereas the second-order derivative $\mnabla\mnabla \cO$ is exactly cancelled by the gradient term $\mnabla {\cal R}$~\eqref{divpart v1} in the final expression for the matter Lagrangian $\Lag_{\rm matt}$~\eqref{def lag matt}. Moreover, the divergence terms $\mnabla \Omega$ must be decomposed too -- cf. \textbf{Lemma \ref{lem dec Omega}} formula \eqref{mathfrakO}:
\begin{align}
\left(\mnabla_{\alpha}\Omega^{\kappa\lambda\mu\alpha}\right) \mnabla_{\beta}\left(\Omega_{\kappa\lambda\mu}^{\ \ \ \beta}-2\Omega_{\lambda\mu\kappa}^{\ \ \  \beta}\right)&= \left(\mnabla_{\alpha}\mathfrak{O}^{\kappa\lambda\mu\alpha}\right) \mnabla_{\beta}\left(\mathfrak{O}_{\kappa\lambda\mu}^{\ \ \ \beta}-2\mathfrak{O}_{\lambda\mu\kappa}^{\ \ \  \beta}\right) +\nonumber \\
&\quad +\frac{7}{18} \left(\mnabla_{\alpha}\cO^{\kappa \alpha}\right) \mnabla_{\beta}\left(\cO_{\kappa }^{ \ \beta} \right)\, ,
\label{mnablaOM2}
    \end{align}

Thus:
\begin{align}
    \Lag_H &= \frac{\Lambda \sqrt{|\det g|}}{4\pi} - \frac{4\pi}{\sqrt{|\det g|}}\bigg\{ \left(\mnabla_{\alpha}\mathfrak{O}^{\kappa\lambda\mu\alpha}\right) \mnabla_{\beta}\left(\mathfrak{O}_{\kappa\lambda\mu}^{\ \ \ \beta}-2\mathfrak{O}_{\lambda\mu\kappa}^{\ \ \  \beta}\right)+ \nonumber\\
     &  \quad + \frac{25}{18} \left(\mnabla_{\alpha}\cO_{\kappa}^{\ \alpha}\right)\left(\mnabla_{\beta}\cO^{\kappa \beta}\right) \bigg\} + \frac{3\sqrt{|\det g|}}{8\pi } A_{\kappa}A^{\kappa} +\frac{1}{2} \left(\mnabla_{\kappa} \mnabla_{\lambda} \cO^{\kappa\lambda} \right) \, .      
\label{pass LagH v1}
\end{align} 
Here, the “mixing” term $A(\! \mnabla\cO)$ vanishes.

\ 

The next two steps describe the Legendre transformation between the potential $A^{\kappa}_{\ \lambda\mu}$ and the momentum $\Omega_{\kappa}^{\ \lambda\mu\nu}$. First, the term $\left(\mnabla_{\nu} \Omega_{\kappa}^{\ \lambda\mu\nu} \right)\, A^{\kappa}_{\ \lambda\mu}$ will be computed using the expression for $A^{\kappa}_{\ \lambda\mu}$ from the decomposition of the non-metricity tensor given in equation~\eqref{pot tAov1}, along with the decomposition formula~\eqref{mathfrakO} for $\mnabla_{\nu} \Omega_{\kappa}^{\ \lambda\mu\nu}$ from \textbf{Lemma~\ref{lem dec Omega}}. Thus:
\begin{align}
    \left(\mnabla_{\nu}\Omega_{\kappa}^{\ \lambda\mu\nu} \right) A^{\kappa}_{\ \lambda\mu} &= \frac{8\pi}{ \sqrt{|\det g|}}\left[  \left(\mnabla_{\sigma}\mathfrak{O}_{\kappa}^{\ \lambda\mu\sigma} \right)\mnabla_{\nu} \left(\mathfrak{O}_{\ \lambda\mu}^{\kappa \ \  \nu}   
 - 2\mathfrak{O}_{\lambda\mu}^{\ \  \kappa\nu}  \right)
    + \right.\nonumber \\
 & \quad  \left. + \frac{25}{18}\left(\mnabla_{\sigma}\cO_{\kappa}^{\ \sigma} \right) \left(\mnabla_{\nu}{\cal O}^{\kappa  \nu} \right)  \right] - 3 A^{\kappa}\left(\mnabla_{\sigma}\cO_{\kappa}^{\ \sigma} \right) \, .
 \label{pass Legendre term1 v1}
\end{align} 

\ 

To derive the term $\Omega_{\kappa}^{\ \lambda\mu\nu}\, \left(\!\kolo{U}^{\kappa}_{\ \lambda\mu\nu} + D^{\kappa}_{\ \lambda\mu\nu} \right)$, the tensors $\kolo{U}$ and $D$~\eqref{def tensor D} must be expressed in terms of the momentum $\Omega$. Since the pair $\Omega$ and $U$ is equivalent to the pair $\Sigma$ and $W$ (see~\eqref{rel UW} and~\eqref{rel Omega Sigma}), the following equality holds:
\begin{align}
    \Omega_{\kappa}^{\ \lambda\mu\nu}\, \left(\!\kolo{U}^{\kappa}_{\ \lambda\mu\nu} + D^{\kappa}_{\ \lambda\mu\nu} \right) = \Sigma_{\kappa}^{\ \lambda\mu\nu}\, \left(\!\kolo{W}^{\kappa}_{\ \lambda\mu\nu} + C^{\kappa}_{\ \lambda\mu\nu} \right)\, ,
\end{align}
where $C^{\kappa}_{\ \lambda\mu\nu}$~\eqref{def tensor C} is a linearised part of $W^{\kappa}_{\ \lambda\mu\nu}$, and satisfies:
\begin{align}
    C^{\kappa}_{\ \lambda\mu\nu}&= -2 D^{\kappa}_{\ \lambda[\mu\nu]}     \, . 
\end{align} 
Furthermore, the field equations~(\ref{sigma feq1}-\ref{sigma feq4}) are linearised -- i.e., it includes only terms linear in $W$. After the decomposition into the metric term and the remainder~\eqref{rozklad pelny W}, this linearisation applies to the potential terms as well, which formed the core of \textbf{Chapter~\ref{pot eqs v1}}. Therefore, in light of the above considerations, the following equality holds:
\begin{align}
   \Omega_{\kappa}^{\ \lambda\mu\nu}\, \left(\!\kolo{U}^{\kappa}_{\ \lambda\mu\nu} + D^{\kappa}_{\ \lambda\mu\nu} \right) = \Sigma_{\kappa}^{\ \lambda\mu\nu}\, \left(\!\kolo{W}^{\kappa}_{\ \lambda\mu\nu} + C^{\kappa}_{\ \lambda\mu\nu} \right) = \Sigma_{\kappa}^{\ \lambda\mu\nu}\, {W}^{\kappa}_{\ \lambda\mu\nu}\, .
   \label{pass Legendre term21 v1}
\end{align}

Ultimately, the problem reduces to inverting the field equations~(\ref{sympl sigma1 v1}-\ref{sympl sigma4 v1}), which have already been done — see formulae (\ref{eqskewW v1}–\ref{symp w v1}). Hence, recalling the symplectic formula~\eqref{sympl Sigma}, the above term $\Sigma_{\kappa}^{\ \lambda\mu\nu}\, {W}^{\kappa}_{\ \lambda\mu\nu}$ equals:
\begin{align}
    \Sigma_{\kappa}^{\ \lambda\mu\nu}\, W^{\kappa}_{\ \lambda\mu\nu} &=\frac 56\, \Sigma^{\mu\nu}\,   W_{[\mu\nu]} + \frac 34\, \Sigma^{\mu\nu}\,   W_{(\mu\nu)} + \underbrace{\widetilde{\Sigma}^{\kappa \lambda\mu\nu}}_{=0}\,  \widetilde{W}_{(\kappa\lambda)\mu\nu} + \widetilde{\Sigma}^{\kappa \lambda\mu\nu}\,  \widetilde{W}_{[\kappa\lambda]\mu\nu} =  \nonumber\\
    &= \frac{88\cdot 16\pi\Lambda}{27 \sqrt{|\det g|}}\left(-\frac {25}{48}\, \Sigma^{\mu\nu}\,   \Sigma_{[\mu\nu]}+ \frac {27}{64}\, \Sigma^{\mu\nu}\,   \Sigma_{(\mu\nu)}    -\frac 38  \widetilde{\Sigma}^{\kappa \lambda\mu\nu}\,  \widetilde{\Sigma}_{[\mu\nu]\kappa\lambda}\right)- \frac 29\Sigma^{\mu\nu}\,F_{\mu\nu}\, .
\end{align}
The last step corresponds to replacing the momentum $\Sigma$ with the momentum $\Omega$ using~\eqref{rel Omega Sigma}. Then, the required term $ \Omega_{\kappa}^{\ \lambda\mu\nu}\, \left(\!\kolo{U}^{\kappa}_{\ \lambda\mu\nu} + D^{\kappa}_{\ \lambda\mu\nu} \right)$ equals
\begin{align}
 \Omega_{\kappa}^{\ \lambda\mu\nu}\, \left(\!\kolo{U}^{\kappa}_{\ \lambda\mu\nu} + D^{\kappa}_{\ \lambda\mu\nu} \right) &= \Sigma_{\kappa}^{\ \lambda\mu\nu}\, \left(\!\kolo{W}^{\kappa}_{\ \lambda\mu\nu} + C^{\kappa}_{\ \lambda\mu\nu} \right) =  \Sigma^{\kappa\lambda\mu\nu}\, W_{\kappa\lambda\mu\nu} = \nonumber \\
 &=\frac{88\cdot 16\pi\Lambda}{27 \sqrt{|\det g|}}\left[-\frac { 25}{3\cdot 64}\, \cO^{\mu\nu}\,   \cO_{[\mu\nu]} + \frac {27}{4\cdot 64}\, \cO^{\mu\nu}\,   \cO_{(\mu\nu)} +\right. \nonumber\\
 & \quad  \left.  -\frac {1}{24} \left(\widetilde{\Omega}_{\kappa\lambda\mu\nu} \widetilde{\Omega}^{\lambda\kappa\nu\mu} + 2 \widetilde{\Omega}_{\kappa\lambda\mu\nu} \widetilde{\Omega}^{\nu\lambda\mu\kappa}\right) \right] + \frac {1}{9}\cO^{\mu\nu}\,F_{\mu\nu}\, .
 \label{OmU v1}
\end{align} 
 Finally, the matter Lagrangian $\Lag_{\rm matt}$~\eqref{def lag matt} has the following form:
 \begin{align}
     \Lag_{\rm matt}&=\Lag_A + \mnabla_{\kappa} \mathcal{R}_{\sigma}^{\ \sigma\kappa} - \Lag_H-  \Omega_{\kappa}^{\ \lambda\mu\nu}\,   \left(\!\kolo{U}^{\kappa}_{\ \lambda\mu\nu} + D^{\kappa}_{\ \lambda\mu\nu} \right) -\left(\mnabla_{\nu}\Omega_{\kappa}^{\ \lambda\mu\nu} \right)\,  A^{\kappa}_{\ \lambda\mu} = \nonumber\\
     &=  \frac{\Lambda \sqrt{|\det g|}}{8\pi} \left|1 + \frac{1}{2  \Lambda^2}    F_{\alpha\beta}F^{\alpha\beta}  -\frac{88 \cdot 50\pi^2}{81|\det g|}  \cO_{[\alpha\beta]}\cO^{[\alpha\beta]} + \right.\nonumber \\
    & \quad \left. +\frac{44\pi^2}{|\det g|}   \cO_{(\alpha\beta)} \cO^{(\alpha\beta)}   - \frac{88\cdot 16\pi^2}{81|\det g|} \left(\widetilde{\Omega}_{\kappa\lambda\mu\nu} \widetilde{\Omega}^{\lambda\kappa\nu\mu} + 2 \widetilde{\Omega}_{\kappa\lambda\mu\nu} \widetilde{\Omega}^{\nu\lambda\mu\kappa}\right)    \right|^{1/2}+ \nonumber \\
     & \quad -\frac{\Lambda \sqrt{|\det g|}}{4\pi} - \frac{3\sqrt{|\det g|}}{8\pi } A_{\kappa}A^{\kappa} - \frac{88\cdot 16\pi\Lambda}{27 \sqrt{|\det g|}}\left[-\frac { 25}{3\cdot 64}\, \cO^{\mu\nu}\,   \cO_{[\mu\nu]} +\right. \nonumber \\
     &\quad \left. + \frac {27}{4\cdot 64}\, \cO^{\mu\nu}\,   \cO_{(\mu\nu)} -\frac {1}{24} \left(\widetilde{\Omega}_{\kappa\lambda\mu\nu} \widetilde{\Omega}^{\lambda\kappa\nu\mu} + 2 \widetilde{\Omega}_{\kappa\lambda\mu\nu} \widetilde{\Omega}^{\nu\lambda\mu\kappa}\right)  \right] +\nonumber \\
     & \quad  - \frac{4\pi}{\sqrt{|\det g|}}\bigg\{  \left(\mnabla_{\sigma}\mathfrak{O}_{\kappa}^{\ \lambda\mu\sigma} \right)\mnabla_{\nu} \left(\mathfrak{O}_{\ \lambda\mu}^{\kappa \ \  \nu}   
 - 2\mathfrak{O}_{\lambda\mu}^{\  \ \kappa\nu}  \right)  + \frac{25}{18}\left(\mnabla_{\sigma}\cO_{\kappa}^{\ \sigma} \right) \left(\mnabla_{\nu}{\cal O}^{\kappa  \nu} \right) \bigg\}+ \nonumber\\
     & \quad    +   3 A^{\kappa}\left(\mnabla_{\sigma}\cO_{\kappa}^{\ \sigma} \right) - \frac {1}{9}\cO^{\mu\nu}\,F_{\mu\nu}\, .
 \end{align}
The expanded form of the above matter Lagrangian is explicitly provided below:
\begin{align}
    \Lag_{\rm matt} &\approx - \frac{ \sqrt{|\det g|}}{8\pi}\left(\Lambda + 3 A_{\kappa}A^{\kappa}\right)  + \frac{ \sqrt{|\det g|}}{32\pi\Lambda} F_{\alpha\beta}F^{\alpha\beta} + \nonumber\\
    & \quad  - \frac{4\pi}{\sqrt{|\det g|}}\bigg\{ \left(\mnabla_{\sigma}\mathfrak{O}_{\kappa}^{\ \lambda\mu\sigma} \right)\mnabla_{\nu} \left(\mathfrak{O}_{\ \lambda\mu}^{\kappa \ \  \nu}   
 - 2\mathfrak{O}_{\lambda\mu}^{\  \ \kappa\nu}  \right)  + \frac{25}{18}\left(\mnabla_{\sigma}\cO_{\kappa}^{\ \sigma} \right) \left(\mnabla_{\nu}{\cal O}^{\kappa  \nu} \right)  \bigg\} + \nonumber\\
    & \quad  +   3 A^{\kappa}\left(\mnabla_{\sigma}\cO_{\kappa}^{\ \sigma} \right) - \frac {1}{9}\cO^{\mu\nu}\,F_{\mu\nu} + \frac{11\cdot25 \pi\Lambda}{  81 \sqrt{|\det g|}}\cO_{[\mu\nu]}\cO^{\mu\nu}-\frac{11\pi\Lambda}{4\sqrt{|\det g|}}\cO_{(\mu\nu)}\cO^{\mu\nu} + \nonumber \\
    &\quad +\frac{88\pi \Lambda}{81\sqrt{|\det g|}}   \left(\widetilde{\Omega}_{\kappa\lambda\mu\nu}\widetilde{\Omega}^{\lambda\kappa\nu\mu} + 2 \widetilde{\Omega}_{\kappa\lambda\mu\nu}\widetilde{\Omega}^{\nu\lambda\mu\kappa}\right)    \, .
    \label{app Lmatt v1}
\end{align}

\subsubsection{Field equations}
The field equations associated with the matter Lagrangian~\eqref{app Lmatt v1} are as follows -- cf. symplectic formula~\eqref{var lag matt}:
\begin{enumerate}
    \item  Euler-Lagrange system for the potential $A_{\mu}$~\eqref{mat eq A}:
    \begin{align}
        \frac{\partial \Lag_{\rm matt}}{\partial A_{\mu}} &= -\frac{3\sqrt{|\det g|}}{4\pi}\, A^{\mu} + 3\mnabla_{\nu}\cO^{\mu\nu}=-2 \cJ^{\mu} \, , 
        \label{eqA1 v1}\\ \frac{\partial \Lag_{\rm matt}}{\partial F_{\mu\nu}}&=\frac{ \sqrt{|\det g|}}{16\pi\Lambda}\, F^{\mu\nu} - \frac 19 \cO^{[\mu\nu]}  =\chi^{\mu\nu}\, , 
        \label{eqA2 v1}
    \end{align}
    where the first equation~\eqref{eqA1 v1} reproduces the part of non-metricity equation~\eqref{pot Aov1}. The second one~\eqref{eqA2 v1} has to be combined with the field equation for $\cO^{[\mu\nu]}$~\eqref{eqom3 v1}, which is written below.

    \item  Specific Euler-Lagrange system with constraints for the tensor density $\Omega_{\kappa}^{\ \lambda\mu\nu}$~\eqref{mat eq Omega}:
    \begin{align}
        \frac{\partial \Lag_{\rm matt}}{\partial \left(\mnabla_{\nu} \cO_{\kappa}^{\ \nu}\right)} &=-\frac{100\pi}{9\sqrt{|\det g|}}\mnabla_{\nu}\cO^{\kappa\nu} + 3A^{\kappa} =- \frac{5}{18} h^{\kappa}\, , 
        \label{eqom1 v1} \\ 
        \frac{\partial \Lag_{\rm matt}}{\partial   \cO^{(\mu \nu)}} &= -\frac{11\pi\Lambda}{2\sqrt{|\det g|}}\, \cO_{(\mu\nu)}=- \frac 9{16} \mathfrak{U}_{(\mu\nu)}\, ,
        \label{eqom2 v1}\\
        \frac{\partial \Lag_{\rm matt}}{\partial   \cO^{[\mu \nu]}} &= -\frac{11\cdot 50\pi\Lambda}{81\sqrt{|\det g|}}\, \cO_{[\mu\nu]} - \frac 19 F_{\mu\nu} =- \frac 58 \mathfrak{U}_{[\mu\nu]}\, ,  
        \label{eqom3 v1}\\
        \frac{\partial \Lag_{\rm matt}}{\partial\left(\mnabla_{\nu}\mathfrak{O}_{\kappa}^{\ \lambda\mu\nu} \right)} &=-\frac{8\pi}{\sqrt{|\det g|}} \mnabla_{\nu}\cO^{\kappa\ \ \nu}_{\ \lambda\mu} + \frac{16\pi}{\sqrt{|\det g|}} \mnabla_{\nu}\cO^{\ \ \ \ \kappa  \nu}_{(\lambda\mu)}=-  \tA^{\kappa}_{\ \lambda\mu}\, ,
        \label{eqom4 v1}\\
        \frac{\partial \Lag_{\rm matt}}{\partial\widetilde{\Omega}^{\kappa  \lambda\mu\nu}}&=\frac{  88\pi\Lambda}{81\sqrt{|\det g|}}\, \left(2\widetilde{\Omega}_{(\lambda|\kappa\nu|\mu)}+4\widetilde{\Omega}_{\nu\lambda \mu \kappa} \right)=-\widetilde{\mathfrak{U}}_{\kappa \lambda\mu\nu} \, , \label{eqom5 v1}
    \end{align}
     where the first~\eqref{eqom1 v1} and the fourth~\eqref{eqom4 v1} equations reproduce  parts of the non-metricity equation~\eqref{pot hv1} and~\eqref{pot ttAov1} respectively. The second equation~\eqref{eqom2 v1} is compatible with the field equation for $\Sigma_{(\mu\nu)}$~\eqref{sympl sigma2 v1}, whereas the third one~\eqref{eqom3 v1} combined with the equation~\eqref{eqA2 v1} are compatible with equations for $\Sigma_{[\mu\nu]}$~\eqref{sympl sigma1 v1} and $\chi^{\mu\nu}$~\eqref{rel konst1}, noting that under the assumed linearisation the equality $ \mathfrak{U}_{\kappa \lambda\mu\nu} = U_{\kappa \lambda\mu\nu}$ holds -- cf. formula \eqref{def mathfrakU}. The traceless part is more complicated, because the initial momentum $\widetilde{\Sigma}_{\kappa\lambda\mu\nu}$~\eqref{sympl sigma4 v1} possesses an additional symmetry — it is skew-symmetric with respect to the first two indices — which does not translate directly into the tensor density $\widetilde{\Omega}_{\kappa\lambda\mu\nu}$. Precisely, the momentum $\widetilde{\Sigma}_{\kappa\lambda\mu\nu}$ naturally decomposes into $\widetilde{\Sigma}_{[\kappa\lambda]\mu\nu}$ and $\widetilde{\Sigma}_{(\kappa\lambda)\mu\nu}$, whereas $\widetilde{\Omega}_{\kappa\lambda\mu\nu}$ into $\widetilde{\Omega}_{(\kappa|\lambda\mu|\nu)}$ and $\widetilde{\Omega}_{[\kappa|\lambda\mu|\nu]}$. The relation between those two decompositions is the following:
     \begin{align}
         \widetilde{\Sigma}_{(\kappa\lambda)\mu\nu} &=\frac 13 \left(\widetilde{\Omega}_{(\kappa|\lambda\mu|\nu)} +2 \widetilde{\Omega}_{(\kappa\lambda)\nu\mu}   \right) =\nonumber\\
         &=\frac 13 \left(\widetilde{\Omega}_{(\kappa|\lambda\mu|\nu)} + \widetilde{\Omega}_{(\kappa|\lambda\nu|\mu)} + \widetilde{\Omega}_{(\lambda|\kappa\nu|\mu)}  + \widetilde{\Omega}_{[\kappa|\lambda\nu|\mu]} + \widetilde{\Omega}_{[\lambda|\kappa\nu|\mu]} \right)\, ,\\
         \widetilde{\Sigma}_{[\kappa\lambda]\mu\nu} &=\frac 13 \left(\widetilde{\Omega}_{[\kappa|\lambda\mu|\nu]} +2 \widetilde{\Omega}_{[\kappa\lambda]\nu\mu}   \right) =\nonumber\\
         &=\frac 13 \left(\widetilde{\Omega}_{[\kappa|\lambda\mu|\nu]} + \widetilde{\Omega}_{(\kappa|\lambda\nu|\mu)} - \widetilde{\Omega}_{(\lambda|\kappa\nu|\mu)}  + \widetilde{\Omega}_{[\kappa|\lambda\nu|\mu]} - \widetilde{\Omega}_{[\lambda|\kappa\nu|\mu]} \right)\, .
     \end{align}
     Then, the field equation \eqref{sympl sigma3 v1} induces the extra symmetry for $\widetilde{\Omega}_{\kappa\lambda\mu\nu}$:
     \begin{align}
         \frac 13 \left(\widetilde{\Omega}_{(\kappa|\lambda\mu|\nu)} +2 \widetilde{\Omega}_{(\kappa\lambda)\nu\mu}   \right) =\widetilde{\Sigma}_{(\kappa\lambda)\mu\nu} = 0\, .
     \end{align} 
     
     The problematic term (from the matter Lagrangian~\eqref{app Lmatt v1}) reads:
\begin{align}
    \widetilde{\Sigma}_{[\alpha\beta]\kappa\lambda} \widetilde{\Sigma}^{[\kappa\lambda]\alpha\beta} = \frac{1}{9}\left(\Omega_{\kappa\lambda\mu\nu}\Omega^{\lambda\kappa\nu\mu} -2 \Omega_{\kappa\lambda\mu\nu}\Omega^{\mu\nu\lambda\kappa} +  \Omega_{\kappa\lambda\mu\nu}\Omega^{\nu\lambda\mu\kappa}\right)\, ,
    \label{uwaga}
\end{align}
where the relation between $\Sigma$ and $\Omega$ is used — see~\eqref{rel Sigma}. In formulae~\eqref{pass LagA1 v1}, \eqref{OmU v1}, and~\eqref{app Lmatt v1}, the following identity is employed:
\begin{align}
    -2 \widetilde{\Omega}_{\kappa\lambda\mu\nu} \widetilde{\Omega}^{\mu\nu\lambda\kappa} = \widetilde{\Omega}_{\kappa\lambda\mu\nu} \widetilde{\Omega}^{\nu\lambda\mu\kappa}\, .
    \label{rownosc 1}
\end{align}
However, the following identity also holds:
\begin{align}
    -2 \widetilde{\Omega}_{\kappa\lambda\mu\nu} \widetilde{\Omega}^{\mu\nu\lambda\kappa} = 2 \widetilde{\Omega}_{\kappa\lambda\mu\nu} \widetilde{\Omega}^{\lambda\kappa\nu\mu} + 2 \widetilde{\Omega}_{\kappa\lambda\mu\nu} \widetilde{\Omega}^{\lambda\mu\kappa\nu}\, .
    \label{rownosc 2}
\end{align}
Inserting both of the above identities (multiplied by $\frac{1}{2}$) into the initial expression yields:
\begin{align}
    \widetilde{\Sigma}_{[\alpha\beta]\kappa\lambda} \widetilde{\Sigma}^{[\kappa\lambda]\alpha\beta} = \frac{1}{9}\left(2 \widetilde{\Omega}_{\kappa\lambda\mu\nu} \widetilde{\Omega}^{\lambda\kappa\nu\mu}  + \widetilde{\Omega}_{\kappa\lambda\mu\nu} \widetilde{\Omega}^{\lambda\mu\kappa\nu} + \frac{3}{2}  \widetilde{\Omega}_{\kappa\lambda\mu\nu} \widetilde{\Omega}^{\nu\lambda\mu\kappa}\right)\, .
\end{align}
If this identity is applied in the Legendre transformation, the field equation~\eqref{eqom5 v1} takes the form:
\begin{align}
    \frac{\partial \Lag_{\rm matt}}{\partial\widetilde{\Omega}^{\kappa  \lambda\mu\nu}} = \frac{88\pi\Lambda}{81\sqrt{|\det g|}}\, \left(4 \widetilde{\Omega}_{(\lambda|\kappa\nu|\mu)} + 2\widetilde{\Omega}_{(\lambda\mu)\kappa\nu} + 3  \widetilde{\Omega}_{\nu\lambda\mu\kappa}\right) = -\widetilde{\mathfrak{U}}_{\kappa \lambda\mu\nu}\, ,
\end{align}
which is equivalent to equation~\eqref{sympl sigma4 v1}, noting that under the assumed linearisation the equality $\widetilde{\mathfrak{U}}_{\kappa \lambda\mu\nu} = \widetilde{U}_{\kappa \lambda\mu\nu}$ holds -- cf. formula \eqref{def mathfrakU}.

    \item The verification that the Einstein equation \eqref{passeineq} is numerically equivalent with the formula \eqref{eqq21v1} derived in the affine picture, is very complicated and time-consuming, especially due to the necessity of the invertion of field equations (\ref{eqskewW v1}-\ref{symp w v1}) and highly non-trivial appearance in the Lagrangian. Thus, for this case the calculations are omitted.

    \end{enumerate}

\subsubsection{Unification}
\label{chapuni v1}
Even though the above theory is derived from the affine Lagrangian~\eqref{lag A1}, which differs from the affine Lagrangian~\eqref{lag iksdfbosgv} used in the theory based on the full Ricci tensor, the matter Lagrangian in both theories contains the same terms — see formulae~\eqref{app Lmatt v0} and~\eqref{app Lmatt v1}:
\begin{align}
 \frac{ \sqrt{|\det g|}}{8\pi}\cdot  3 A_{\kappa}A^{\kappa}   + \frac{ \sqrt{|\det g|}}{32\pi\Lambda} F_{\alpha\beta}F^{\alpha\beta} \, .
\end{align}
The above terms suggest that quantities $F_{\mu\nu}$ and $A_{\mu}$ can be interpreted with $B_{\mu\nu}$ and $b_{\mu}$ from Proca theory -- see \textbf{Appendix~\ref{Proca th}} -- what was already mentioned in \textbf{Chapter~\ref{uni0}}. Due to the identical appearance in the matter Lagrangians, the unification procedure is identical -- cf. formula~\eqref{uniF v0}:
\begin{align}  
    F_{\mu\nu} := \pm \sqrt{8\pi |\Lambda|}\, B_{\mu\nu}\, ,  
    \label{uniF v1}
\end{align}  
under the assumption that:
\begin{align}
\Lambda <0\qquad   \Longrightarrow \qquad  \Lambda=-|\Lambda|\, .
    \label{uniconst v1}
\end{align}
The only difference relies on substituting the electromagnetic tensor $F_{\mu\nu}$ by the  Proca field $B_{\mu\nu}$. Accordingly, this implies an identical relation between the potential \( A_{\mu} \) and the Proca potential \( b_{\mu} \) \eqref{Proca field} -- cf. formula~\eqref{uniA v0}:  
\begin{align}
    A_{\mu} :=  \pm \sqrt{8\pi |\Lambda|}\,  b_{\mu} \, .
    \label{uniA v1}
\end{align}   
However, in the Variant $V_1$ appeared an extra skew-symmetric field $W_{[\mu\nu]}$, which was coupled with the skew-symmetric Ricci tensor $F_{\mu\nu}$ -- see formulae~\eqref{rel konst1} and~\eqref{sympl sigma1 v1}. This field, after the passage to the metric picture, is ``replaced'' by the tensor density field $\cO^{[\mu\nu]}$ \eqref{eqA2 v1}, which is not considered in the unification procedure.

\

To reconstruct the same symplectic structure as in Proca theory \eqref{var lag P}, the momentum \( \chi^{\mu\nu} \)~\eqref{eqA2 v1} should be related to the dual  tensor density \( {\cal B}^{\mu\nu} \)~\eqref{const Proca} as follows:  
\begin{align}
    \chi^{\mu\nu} := \mp \frac{1}{\sqrt{32\pi|\Lambda|}}\, {\cal B}^{\mu\nu}\, .
    \label{chi v1}
\end{align}  
Then:
\begin{align}
    {\cal B}^{\mu\nu} = \sqrt{|\det g|}\, B^{\mu\nu} \pm \frac{4\sqrt{2\pi|\Lambda|}}{9}\, \cO^{[\mu\nu]}\, .
    \label{eqB1 v1}
\end{align}
The same happens with the current $\cJ^{\mu}$~\eqref{cal J}:
\begin{align}
    \cJ^{\mu} = \partial_{\nu}\chi^{\mu\nu} = \mp \frac{1}{\sqrt{32\pi|\Lambda|}}\,\partial_{\nu} {\cal B}^{\mu\nu}\, .
\end{align}
Consequently, the field equation~\eqref{eqA1 v1} equals:
\begin{align}
    \partial_{\nu} {\cal B}^{\mu\nu} = -6|\Lambda| b^{\mu} \pm 6\sqrt{2\pi|\Lambda|}\mnabla_{\nu}\cO^{\mu\nu}\, .
    \label{eqB2 v1}
\end{align}
Together, equations \eqref{eqB1 v1} and~\eqref{eqB2 v1} produce the non-homogeneous Proca equation -- cf. formula~\eqref{eq P1}:
\begin{align}
    \mBox \,b^{\mu} -  \mnabla^{\mu}\mnabla_{\nu}  b^{\nu} - b^{\sigma}\, \kolo{K}_{\sigma}^{\ \mu} - 6|\Lambda|\, b^{\mu}= \pm \sqrt{\frac{2\pi|\Lambda|}{|\det g|}} \left(\frac{50}{9}\mnabla_{\nu}\cO^{[\mu\nu]} +6\mnabla_{\nu}\cO^{(\mu\nu)} \right)\, .
\end{align}
Here, the mass parameter  equals
\begin{align}
    \frac{m^2}{\hbar^2}=6|\Lambda|\, .
\end{align}

\subsection{Variant \texorpdfstring{$V_6$}{V6}}

This transition, as the previous one, involves non-trivial terms related to the traceless Riemann tensor $W^{\kappa}_{\ \lambda\mu\nu}$, which is the main source of difficulty.

\ 

The first step is to rewrite the affine Lagrangian~\eqref{lag A6} in the proper control mode (cf. the symplectic formula in the metric picture~\eqref{var lag matt}):
\begin{align}
    \Lag_{A} = \alpha\, \sqrt{|KKKK + KKKW + KKFF + KKFW + KKWW|}\, ,
\end{align}
where the terms $KKKK$, $KKKW$, $KKFF$, $KKFW$, and $KKWW$ are defined in equations~(\ref{KKKKv6}–\ref{rozklad w6}). However, based on the full analysis presented in \textbf{Chapter~\ref{VAR V6}}, and in particular the Einstein equation~\eqref{ein eq0 v6}, the above affine Lagrangian takes the form:
\begin{align}
    \Lag_A &= \alpha\, \sqrt{|\Lambda^4 \sigma \gamma^2 \det g + \Lambda^2\, ggFF + \Lambda^2\, ggFW + \Lambda^2\, ggWW|} = \nonumber \\
     &=\frac{\Lambda\sqrt{|\det g|}}{8\pi}\, \sqrt{\left|1 + \frac{1}{\sigma\gamma^2 \Lambda^2 \det g}\left(ggFF+ggFW + ggWW \right) \right|}=\nonumber\\
    &=\frac{\Lambda \sqrt{|\det g|}}{8\pi} \, \bigg| 1 - \frac{27}{128\Lambda^2} \left( \frac{832}{225}\, F_{\alpha\beta}F^{\alpha\beta} + \frac{128}{45} F_{\alpha\beta}W^{\alpha\beta} +\frac{64}{9}\,  W_{(\alpha \beta) }\, W^{ \alpha\beta }+ \right. \nonumber\\
    & \quad  \left.-\frac{32}9\,   W_{\alpha\beta \kappa\lambda}\, W^{\kappa\lambda \alpha\beta} - \frac{16}{9}\,   W_{\alpha\beta \kappa\lambda}\, W^{ \alpha\beta\kappa\lambda} + \frac{16}{3}\,   W_{\alpha\beta \kappa\lambda}\, W^{ \beta\alpha\kappa\lambda} \right)\bigg|^{1/2}\, . 
    \label{pass LagA v6}
\end{align}
Here, the characteristic constants $\alpha, \gamma^2, \sigma$ are given in~\eqref{const6}, while the terms $ggFF$, $ggFW$, and $ggWW$ are defined in~\eqref{ggFF v6}, \eqref{ggFW2 v6} and \eqref{ggWW2 v6} respectively.

To obtain the appropriate matter Lagrangian $\Lag_{\text{matt}}$~\eqref{def lag matt}, the tensor $W$ must be expressed in terms of the momentum $\Omega$, which is equivalent to $\Sigma$~\eqref{rel Omega Sigma} — the momentum canonically conjugate to $W$. The momentum $\Sigma$ is decomposed into four independent components, generating four field equations (\ref{sympl sigma1 v6}–\ref{sympl sigma4 v6}), all of which can be inverted:
\begin{align}
   W_{[\mu\nu]} &= -\frac{128 \cdot 16 \pi \Lambda}{27\sqrt{|\det g|}} \cdot \left(- \frac {45}{64}\right) \, \Sigma_{[\mu\nu]} + \frac{12}{5}F_{\mu\nu}\, , 
   \label{ajfna}\\
   W_{(\mu\nu)} &=-\frac{128 \cdot 16 \pi \Lambda}{27\sqrt{|\det g|}} \cdot   \frac{9}{64} \Sigma_{(\mu\nu)}\, , \\
    \widetilde{W}_{(\mu\nu)\kappa\lambda}&=-\frac{128 \cdot 16 \pi \Lambda}{27\sqrt{|\det g|}}\cdot \frac{9}{64} \widetilde{\Sigma}_{(\mu\nu)\kappa\lambda} \, , \\
   \widetilde{W}_{[\mu\nu]\kappa\lambda}&=-\frac{128 \cdot 16 \pi \Lambda}{27\sqrt{|\det g|}}   \left( -\frac{3}{32}\widetilde{\Sigma}_{[\mu\nu]\kappa\lambda} + \frac{3}{64} \widetilde{\Sigma}_{[\kappa \lambda]\mu\nu}\right)\, .
   \label{symp w v6}
\end{align}
Therefore, the affine Lagrangian~\eqref{pass LagA v6} equals:
\begin{align}
    \Lag_A &=   \frac{\Lambda \sqrt{|\det g|}}{8\pi} \left|1 - \frac{3}{2 \Lambda^2}  F_{\alpha\beta}F^{\alpha\beta}  +  \frac{25\cdot 128\pi^2}{9|\det g|}  \Sigma_{[\alpha\beta]}\Sigma^{[\alpha\beta]} -\frac{64\pi^2}{|\det g|}   \Sigma_{(\alpha\beta)}\Sigma ^{(\alpha\beta)} +   \right.\nonumber \\
    & \quad \left.   - \frac{ (16\pi)^2}{9|\det g|} \left(\widetilde{\Sigma}_{[\alpha\beta]\kappa\lambda} \widetilde{\Sigma}^{[\kappa\lambda]\alpha\beta} -2 \widetilde{\Sigma}_{[\alpha\beta]\kappa\lambda} \widetilde{\Sigma}^{[\alpha\beta]\kappa\lambda} +3 \widetilde{\Sigma}_{(\alpha\beta)\kappa\lambda} 
    \widetilde{\Sigma}^{(\alpha\beta)\kappa\lambda} \right) \right|^{1/2}\, .
\end{align}

As it was in the previous example, the “mixing” term $F_{\alpha\beta}\Sigma^{\alpha\beta}$ vanishes. Finally, replacing the momentum~$\Sigma$ with the momentum $\Omega$~\eqref{rel Sigma} yields:
\begin{align}
    \Lag_A &=   \frac{\Lambda \sqrt{|\det g|}}{8\pi} \left|1 - \frac{3}{2 \Lambda^2}  F_{\alpha\beta}F^{\alpha\beta}  +  \frac{25\cdot 32\pi^2}{9|\det g|}  \cO_{[\alpha\beta]}\cO^{[\alpha\beta]} -\frac{16\pi^2}{|\det g|}   \cO_{(\alpha\beta)}\cO^{(\alpha\beta)} +   \right.\nonumber \\
    & \quad \left.   -  \frac{(16\pi )^2}{27 |\det g|}  \left(\frac 12  \widetilde{\Omega }_{\kappa\lambda\mu\nu} \widetilde{\Omega}^{\kappa\lambda\mu\nu} +  \frac 12 \widetilde{\Omega}_{\kappa\lambda\mu\nu} \widetilde{\Omega}^{\nu\lambda\mu \kappa}  +  2\widetilde{\Omega}_{\kappa\lambda\mu\nu} \widetilde{\Omega}^{\lambda\mu\nu\kappa} - \widetilde{\Omega}_{\kappa\lambda\mu\nu} \widetilde{\Omega}^{ \lambda\kappa\nu\mu }\right) \right|^{1/2}\, ,
    \label{pass LagA1 v6}
\end{align}
where $\widetilde{\Omega}$ denotes the totally traceless part of the momentum $\Omega$.

\ 

The next step in the passage to the metric picture involves deriving the divergence term $\mnabla{\cal R}$~\eqref{divR}, which is the same as in the Variant $V_1$(see formula~\eqref{divpart v1}), since $\mnabla{\cal R}$ is constructed from the non-metricity tensor $N$, which is identical in both theories. Thus:
\begin{align}
      \mnabla_{\kappa}{\cal R}_{\sigma}^{\ \sigma\kappa} &= \frac 12 \mnabla_{\kappa}\mnabla_{\nu} \mathcal{O}^{\kappa  \nu} \,.
      \label{divpart v6}
\end{align}
For the same reason, the Hilbert Lagrangian $\Lag_H$  is given by the same expression as in the Variant $V_1$ -- see formula~\eqref{pass LagH v1}:
\begin{align}
      \Lag_H &= \frac{\Lambda \sqrt{|\det g|}}{4\pi} - \frac{4\pi}{\sqrt{|\det g|}}\bigg\{ \left(\mnabla_{\alpha}\mathfrak{O}^{\kappa\lambda\mu\alpha}\right) \mnabla_{\beta}\left(\mathfrak{O}_{\kappa\lambda\mu}^{\ \ \ \beta}-2\mathfrak{O}_{\lambda\mu\kappa}^{\ \ \  \beta}\right)+ \nonumber\\
     &  \quad + \frac{25}{18} \left(\mnabla_{\alpha}\cO_{\kappa}^{\ \alpha}\right)\left(\mnabla_{\beta}\cO^{\kappa \beta}\right) \bigg\} + \frac{3\sqrt{|\det g|}}{8\pi } A_{\kappa}A^{\kappa} +\frac{1}{2} \left(\mnabla_{\kappa} \mnabla_{\lambda} \cO^{\kappa\lambda} \right) \, .   
\label{pass LagH v6}
\end{align} 
Here, the “mixing” term $A(\! \mnabla\cO)$ vanishes.

\ 

The next two steps describe the Legendre transformation between the potential $A^{\kappa}_{\ \lambda\mu}$ and the momentum $\Omega_{\kappa}^{\ \lambda\mu\nu}$. First, the term $\left(\mnabla_{\nu} \Omega_{\kappa}^{\ \lambda\mu\nu} \right)\, A^{\kappa}_{\ \lambda\mu}$ will be derived using the expression for $A^{\kappa}_{\ \lambda\mu}$ from the decomposition of the non-metricity tensor given in equation~\eqref{pot tAov1}. Therefore, it is exactly the same as in the Variant $V_1$ -- see formula~\eqref{pass Legendre term1 v1}:
\begin{align}
      \left(\mnabla_{\nu}\Omega_{\kappa}^{\ \lambda\mu\nu} \right) A^{\kappa}_{\ \lambda\mu} &= \frac{8\pi}{ \sqrt{|\det g|}}\left[  \left(\mnabla_{\sigma}\mathfrak{O}_{\kappa}^{\ \lambda\mu\sigma} \right)\mnabla_{\nu} \left(\mathfrak{O}_{\ \lambda\mu}^{\kappa \ \  \nu}   
 - 2\mathfrak{O}_{\lambda\mu}^{\ \  \kappa\nu}  \right)
    + \right.\nonumber \\
 & \quad  \left. + \frac{25}{18}\left(\mnabla_{\sigma}\cO_{\kappa}^{\ \sigma} \right) \left(\mnabla_{\nu}{\cal O}^{\kappa  \nu} \right)  \right] - 3 A^{\kappa}\left(\mnabla_{\sigma}\cO_{\kappa}^{\ \sigma} \right) \, .
 \label{pass Legendre term1 v6}
\end{align}

\ 

The derivation of the term $\Omega_{\kappa}^{\ \lambda\mu\nu}\, \left(\!\kolo{U}^{\kappa}_{\ \lambda\mu\nu} + D^{\kappa}_{\ \lambda\mu\nu} \right)$ was discussed in the previous subsection, with the result summarised in equation~\eqref{pass Legendre term21 v1}, which is rewritten below:
\begin{align}
   \Omega_{\kappa}^{\ \lambda\mu\nu}\, \left(\!\kolo{U}^{\kappa}_{\ \lambda\mu\nu} + D^{\kappa}_{\ \lambda\mu\nu} \right) =   \Sigma_{\kappa}^{\ \lambda\mu\nu}\, \left(\!\kolo{W}^{\kappa}_{\ \lambda\mu\nu} + C^{\kappa}_{\ \lambda\mu\nu} \right) = \Sigma_{\kappa}^{\ \lambda\mu\nu}\, {W}^{\kappa}_{\ \lambda\mu\nu}\, .
\end{align}

Ultimately, the problem reduces to inverting the field equations~(\ref{sympl sigma1 v6}-\ref{sympl sigma4 v6}), which have already been done — see formulae (\ref{ajfna}–\ref{symp w v6}). Hence, recalling the symplectic formula~\eqref{sympl Sigma}, the above term $\Sigma_{\kappa}^{\ \lambda\mu\nu}\, {W}^{\kappa}_{\ \lambda\mu\nu}$ equals:
\begin{align}
    \Sigma_{\kappa}^{\ \lambda\mu\nu}\, W^{\kappa}_{\ \lambda\mu\nu} &=\frac 56\, \Sigma^{\mu\nu}\,   W_{[\mu\nu]} + \frac 34\, \Sigma^{\mu\nu}\,   W_{(\mu\nu)} +  \widetilde{\Sigma}^{\kappa \lambda\mu\nu} \,  \widetilde{W}_{(\kappa\lambda)\mu\nu} + \widetilde{\Sigma}^{\kappa \lambda\mu\nu}\,  \widetilde{W}_{[\kappa\lambda]\mu\nu} =  \nonumber\\
    &= 2\Sigma^{\mu\nu}\,F_{\mu\nu}-\frac{128\cdot 16\pi\Lambda}{27 \sqrt{|\det g|}}\left[-\frac {75}{128}\, \Sigma^{\mu\nu}\,   \Sigma_{[\mu\nu]}+ \frac {27}{256}\, \Sigma^{\mu\nu}\,   \Sigma_{(\mu\nu)}   +\right.\nonumber \\
    & \quad \left.+\frac{3}{64} \left(\widetilde{\Sigma}_{[\alpha\beta]\kappa\lambda} \widetilde{\Sigma}^{[\kappa\lambda]\alpha\beta} -2 \widetilde{\Sigma}_{[\alpha\beta]\kappa\lambda} \widetilde{\Sigma}^{[\alpha\beta]\kappa\lambda} +3 \widetilde{\Sigma}_{(\alpha\beta)\kappa\lambda} 
    \widetilde{\Sigma}^{(\alpha\beta)\kappa\lambda} \right)  \right] \, .  
\end{align}

The last step corresponds to replacing the momentum $\Sigma$ with the momentum $\Omega$ using~\eqref{rel Omega Sigma}. Then, the required term $ \Omega_{\kappa}^{\ \lambda\mu\nu}\, \left(\!\kolo{U}^{\kappa}_{\ \lambda\mu\nu} + D^{\kappa}_{\ \lambda\mu\nu} \right)$ equals
\begin{align}
 \Omega_{\kappa}^{\ \lambda\mu\nu}\, \left(\!\kolo{U}^{\kappa}_{\ \lambda\mu\nu} + D^{\kappa}_{\ \lambda\mu\nu} \right) &= \Sigma_{\kappa}^{\ \lambda\mu\nu}\, \left(\!\kolo{W}^{\kappa}_{\ \lambda\mu\nu} + C^{\kappa}_{\ \lambda\mu\nu} \right) =  \Sigma^{\kappa\lambda\mu\nu}\, W_{\kappa\lambda\mu\nu} = \nonumber \\
 &= -\frac{128\cdot 16\pi\Lambda}{27 \sqrt{|\det g|}}\left[-\frac {75}{4\cdot 128}\, \cO^{\mu\nu}\,   \cO_{[\mu\nu]}+ \frac {27}{8\cdot 128}\, \cO^{\mu\nu}\,   \cO_{(\mu\nu)}+\right.\nonumber \\
    & \quad  + \frac{1}{64}\left(\frac 12  \widetilde{\Omega }_{\kappa\lambda\mu\nu} \widetilde{\Omega}^{\kappa\lambda\mu\nu} +  \frac 12 \widetilde{\Omega}_{\kappa\lambda\mu\nu} \widetilde{\Omega}^{\nu\lambda\mu \kappa}  +  2\widetilde{\Omega}_{\kappa\lambda\mu\nu} \widetilde{\Omega}^{\lambda\mu\nu\kappa} + \right.\nonumber\\
    &\quad \left. \left.- \widetilde{\Omega}_{\kappa\lambda\mu\nu} \widetilde{\Omega}^{ \lambda\kappa\nu\mu }\right) \right]   -\cO^{\mu\nu}\,F_{\mu\nu}\, .
\end{align}

\

 Finally, the matter Lagrangian $\Lag_{\rm matt}$~\eqref{def lag matt} has the following form:
 \begin{align}
     \Lag_{\rm matt}&=\Lag_A + \mnabla_{\kappa} \mathcal{R}_{\sigma}^{\ \sigma\kappa} - \Lag_H-  \Omega_{\kappa}^{\ \lambda\mu\nu}\,   \left(\!\kolo{U}^{\kappa}_{\ \lambda\mu\nu} + D^{\kappa}_{\ \lambda\mu\nu} \right) -\left(\mnabla_{\nu}\Omega_{\kappa}^{\ \lambda\mu\nu} \right)\,  A^{\kappa}_{\ \lambda\mu} = \nonumber\\
     &=     \frac{\Lambda \sqrt{|\det g|}}{8\pi} \left|1 - \frac{3}{2 \Lambda^2}  F_{\alpha\beta}F^{\alpha\beta}  +  \frac{25\cdot 32\pi^2}{9|\det g|}  \cO_{[\alpha\beta]}\cO^{[\alpha\beta]} -\frac{16\pi^2}{|\det g|}   \cO_{(\alpha\beta)}\cO^{(\alpha\beta)} +   \right.\nonumber \\
    & \quad \left.  -  \frac{(16\pi )^2}{27 |\det g|}  \left(\frac 12  \widetilde{\Omega }_{\kappa\lambda\mu\nu} \widetilde{\Omega}^{\kappa\lambda\mu\nu} +  \frac 12 \widetilde{\Omega}_{\kappa\lambda\mu\nu} \widetilde{\Omega}^{\nu\lambda\mu \kappa}  +  2\widetilde{\Omega}_{\kappa\lambda\mu\nu} \widetilde{\Omega}^{\lambda\mu\nu\kappa} - \widetilde{\Omega}_{\kappa\lambda\mu\nu} \widetilde{\Omega}^{ \lambda\kappa\nu\mu }\right)  \right|^{1/2} + \nonumber\\
    & \quad - \frac{\Lambda \sqrt{|\det g|}}{4\pi} + \frac{128\cdot 16\pi\Lambda}{27 \sqrt{|\det g|}}\left[-\frac {75}{4\cdot 128}\, \cO^{\mu\nu}\,   \cO_{[\mu\nu]}+ \frac {27}{8\cdot 128}\, \cO^{\mu\nu}\,   \cO_{(\mu\nu)}+\right.\nonumber \\
    & \quad \left.+\frac{1}{64} \left(\frac 12  \widetilde{\Omega }_{\kappa\lambda\mu\nu} \widetilde{\Omega}^{\kappa\lambda\mu\nu} +  \frac 12 \widetilde{\Omega}_{\kappa\lambda\mu\nu} \widetilde{\Omega}^{\nu\lambda\mu \kappa}  +  2\widetilde{\Omega}_{\kappa\lambda\mu\nu} \widetilde{\Omega}^{\lambda\mu\nu\kappa} - \widetilde{\Omega}_{\kappa\lambda\mu\nu} \widetilde{\Omega}^{ \lambda\kappa\nu\mu }\right)  \right] +\nonumber\\ 
    & \quad - \frac{4\pi}{\sqrt{|\det g|}}\bigg\{  \left(\mnabla_{\sigma}\mathfrak{O}_{\kappa}^{\ \lambda\mu\sigma} \right)\mnabla_{\nu} \left(\mathfrak{O}_{\ \lambda\mu}^{\kappa \ \  \nu}   
 - 2\mathfrak{O}_{\lambda\mu}^{\ \  \kappa\nu}  \right)
    + \frac{25}{18}\left(\mnabla_{\sigma}\cO_{\kappa}^{\ \sigma} \right) \left(\mnabla_{\nu}{\cal O}^{\kappa  \nu} \right)  \bigg\}  + \nonumber\\
     & \quad  - \frac{3\sqrt{|\det g|}}{8\pi } A_{\kappa}A^{\kappa} + 3 A^{\kappa}\left(\mnabla_{\sigma}\cO_{\kappa}^{\ \sigma} \right) + \cO^{\mu\nu}\,F_{\mu\nu}\, .
 \end{align}
The expanded form of the matter Lagrangian is explicitly provided below:
\begin{align}
    \Lag_{\rm matt} &\approx  - \frac{\sqrt{|\det g|}}{8\pi } \left(\Lambda + 3A_{\kappa}A^{\kappa} \right) - \frac{3\sqrt{|\det g|}}{32\pi \Lambda} F_{\alpha\beta} F^{\alpha\beta} + \nonumber\\
    & \quad - \frac{4\pi}{\sqrt{|\det g|}}\bigg\{  \left(\mnabla_{\sigma}\mathfrak{O}_{\kappa}^{\ \lambda\mu\sigma} \right)\mnabla_{\nu} \left(\mathfrak{O}_{\ \lambda\mu}^{\kappa \ \  \nu}   
 - 2\mathfrak{O}_{\lambda\mu}^{\ \  \kappa\nu}  \right)
      + \frac{25}{18}\left(\mnabla_{\sigma}\cO_{\kappa}^{\ \sigma} \right) \left(\mnabla_{\nu}{\cal O}^{\kappa  \nu} \right) \bigg\}  + \nonumber\\
     & \quad  + \frac{16\pi \Lambda}{27\sqrt{|\det g|}}   \left(\frac 12  \widetilde{\Omega }_{\kappa\lambda\mu\nu} \widetilde{\Omega}^{\kappa\lambda\mu\nu} +  \frac 12 \widetilde{\Omega}_{\kappa\lambda\mu\nu} \widetilde{\Omega}^{\nu\lambda\mu \kappa}  +  2\widetilde{\Omega}_{\kappa\lambda\mu\nu} \widetilde{\Omega}^{\lambda\mu\nu\kappa} - \widetilde{\Omega}_{\kappa\lambda\mu\nu} \widetilde{\Omega}^{ \lambda\kappa\nu\mu }\right)  + \nonumber \\
     & \quad-\frac{50\pi \Lambda}{9\sqrt{|\det g|}} \cO_{[\alpha\beta]}\cO^{\alpha\beta} + \frac{\pi \Lambda}{\sqrt{|\det g|}} \cO_{(\alpha\beta)}\cO^{\alpha\beta}   + 3 A^{\kappa}\left(\mnabla_{\sigma}\cO_{\kappa}^{\ \sigma} \right)+\cO^{\mu\nu}\,F_{\mu\nu}\, .
     \label{app Lmatt v6}
\end{align}

\subsubsection{Field equations}
The field equations associated with the matter Lagrangian~\eqref{app Lmatt v6} are as follows -- cf. symplectic formula \eqref{var lag matt}:
\begin{enumerate}
    \item  Euler-Lagrange system for the potential $A_{\mu}$~\eqref{mat eq A}:
    \begin{align}
        \frac{\partial \Lag_{\rm matt}}{\partial A_{\mu}} &= -\frac{3\sqrt{|\det g|}}{4\pi}\, A^{\mu} + 3\mnabla_{\nu}\cO^{\mu\nu}=-2 \cJ^{\mu} \, , 
        \label{eqA1 v6}\\ 
        \frac{\partial \Lag_{\rm matt}}{\partial F_{\mu\nu}}&=-\frac{ 3\sqrt{|\det g|}}{16\pi\Lambda}\, F^{\mu\nu} + \cO^{[\mu\nu]}  =\chi^{\mu\nu}\, , 
        \label{eqA2 v6}
    \end{align}
     where the first equation~\eqref{eqA1 v6} reproduces the part of non-metricity equation~\eqref{pot Aov1}, which is the same as for the Variant $V_1$ -- cf. \textbf{Chapter~\ref{sub noneq v6}}. The second one~\eqref{eqA2 v6} has to be combined with the field equation for $\cO^{[\mu\nu]}$~\eqref{eqom3 v6}, which is written below.

     \item  a specific Euler-Lagrange system with constraints for the tensor density $\Omega_{\kappa}^{\ \lambda\mu\nu}$~\eqref{mat eq Omega}:
    \begin{align}
        \frac{\partial \Lag_{\rm matt}}{\partial \left(\mnabla_{\nu} \cO_{\kappa}^{\ \nu}\right)} &=-\frac{100\pi}{9\sqrt{|\det g|}}\mnabla_{\nu}\cO^{\kappa\nu} + 3A^{\kappa} =- \frac{5}{18} h^{\kappa}\, , 
        \label{eqom1 v6}\\ 
        \frac{\partial \Lag_{\rm matt}}{\partial   \cO^{(\mu \nu)}} &= -\frac{11\pi\Lambda}{2\sqrt{|\det g|}}\, \cO_{(\mu\nu)}=- \frac 9{16} \mathfrak{U}_{(\mu\nu)}\, ,
        \label{eqom2 v6}\\
        \frac{\partial \Lag_{\rm matt}}{\partial   \cO^{[\mu \nu]}} &= -\frac{100\pi\Lambda}{9\sqrt{|\det g|}}\, \cO_{[\mu\nu]} + F_{\mu\nu} =- \frac 58 \mathfrak{U}_{[\mu\nu]}\, , 
        \label{eqom3 v6} \\
        \frac{\partial \Lag_{\rm matt}}{\partial\left(\mnabla_{\nu}\mathfrak{O}_{\kappa}^{\ \lambda\mu\nu} \right)} &=-\frac{8\pi}{\sqrt{|\det g|}} \mnabla_{\nu}\cO^{\kappa\ \ \nu}_{\ \lambda\mu} + \frac{16\pi}{\sqrt{|\det g|}} \mnabla_{\nu}\cO^{\ \ \ \ \kappa  \nu}_{(\lambda\mu)}=-  \tA^{\kappa}_{\ \lambda\mu}\, ,
        \label{eqom4 v6}\\
        \frac{\partial \Lag_{\rm matt}}{\partial\widetilde{\Omega}^{\kappa  \lambda\mu\nu}}&=\frac{16\pi \Lambda}{27\sqrt{|\det g|}}   \left(   \widetilde{\Omega }_{\kappa\lambda\mu\nu} +    \widetilde{\Omega}_{\nu\lambda\mu \kappa}  +  4\widetilde{\Omega}_{(\lambda\mu)\nu\kappa} - 2\widetilde{\Omega}_{ (\lambda|\kappa\nu|\mu) }\right) =\nonumber \\
        &=-\widetilde{\mathfrak{U}}_{\kappa \lambda\mu\nu} \, ,
        \label{eqom5 v6}
    \end{align}
     where the first~\eqref{eqom1 v6} and the fourth~\eqref{eqom4 v6} equations reproduce  parts of the non-metricity equation~\eqref{pot hv1} and~\eqref{pot ttAov1} respectively, which are the same as for the Variant $V_1$ -- cf. \textbf{Chapter~\ref{sub noneq v6}}. The second equation~\eqref{eqom2 v6} is compatible with the field equation for $\Sigma_{(\mu\nu)}$~\eqref{sympl sigma2 v6}, whereas the third one~\eqref{eqom3 v6} combined with the equation~\eqref{eqA2 v6} are compatible with equations for $\Sigma_{[\mu\nu]}$~\eqref{sympl sigma1 v6} and $\chi_{\mu\nu}$~\eqref{rel konst6}. The last field equation \eqref{eqom5 v6} is equivalent to equation~\eqref{sympl sigma4 v6}, noting that under the assumed linearisation the equality $\widetilde{\mathfrak{U}}_{\kappa \lambda\mu\nu} = \widetilde{U}_{\kappa \lambda\mu\nu}$ holds -- cf. \eqref{def mathfrakU}.   
      
      \item The verification that the Einstein equation \eqref{passeineq} is numerically equivalent with the formula \eqref{eqq21v6} derived in the affine picture, is very complicated and time-consuming, especially due to the necessity of the invertion of field equations~(\ref{ajfna}-\ref{symp w v6}) and highly non-trivial appearance in the Lagrangian. Thus, for this case the calculations are omitted.
    
    \end{enumerate}

\subsubsection{Unification}

The unification procedure is essentially the same as for Variant~$V_1$ — see \textbf{Chapter~\ref{chapuni v1}}. The only difference lies in the coupling constants, since in this theory the matter Lagrangian~\eqref{app Lmatt v6} differs slightly — cf.~formula~\eqref{app Lmatt v1}:
\begin{align}
- \frac{\sqrt{|\det g|}}{8\pi } \cdot 3A_{\kappa}A^{\kappa}   - \frac{3\sqrt{|\det g|}}{32\pi \Lambda} F_{\alpha\beta} F^{\alpha\beta}  \, .
\end{align}

As before, the above terms suggest that the quantities $F_{\mu\nu}$ and $A_{\mu}$ can be interpreted as $B_{\mu\nu}$ and $b_{\mu}$ from Proca theory — see \textbf{Appendix~\ref{Proca th}}. Thus, the unification procedure proceeds as follows:
\begin{align}  
    F_{\mu\nu} := \pm \sqrt{\frac{8\pi |\Lambda|}{3}}\, B_{\mu\nu}\, ,  
    \label{uniF v6}
\end{align}  
under the assumption:
\begin{align}
\Lambda > 0 \qquad \Longrightarrow \qquad \Lambda = |\Lambda|\, .
    \label{uniconst v6}
\end{align}

The main difference lies in the opposite sign of the cosmological constant $\Lambda$. Accordingly, this leads to the same relation between the potential \( A_{\mu} \) and the Proca potential \( b_{\mu} \)~\eqref{Proca field} — cf.~formula~\eqref{uniA v1}:
\begin{align}
    A_{\mu} := \pm \sqrt{\frac{8\pi |\Lambda|}{3}}\, b_{\mu} \, .
    \label{uniA v6}
\end{align}  
As it was in the Variant $V_1$, there appeared an extra skew-symmetric field $W_{[\mu\nu]}$, which was coupled with the skew-symmetric Ricci tensor $F_{\mu\nu}$ -- see formulae~\eqref{rel konst1} and~\eqref{sympl sigma1 v1}. This field, after the passage to the metric picture, is ``replaced'' by the tensor density field $\cO^{[\mu\nu]}$~\eqref{eqA2 v6}, which is not considered in the unification procedure.

To reconstruct the same symplectic structure as in Proca theory \eqref{var lag P}, the momentum \( \chi^{\mu\nu} \)~\eqref{eqA2 v6} should be related to the dual  tensor density \( {\cal B}^{\mu\nu} \)~\eqref{const Proca} as follows:  
\begin{align}
    \chi^{\mu\nu} := \mp \sqrt{\frac{3}{32\pi|\Lambda|}}\, {\cal B}^{\mu\nu}\, .
    \label{chi v6}
\end{align}   
Then:
\begin{align}
    {\cal B}^{\mu\nu} = \sqrt{|\det g|}\, B^{\mu\nu} \mp \sqrt{\frac{32\pi|\Lambda|}3}\, \cO^{[\mu\nu]}\, .
    \label{eqB1 v6}
\end{align}
The same happens with the current $\cJ^{\mu}$~\eqref{cal J}:
\begin{align}
    \cJ^{\mu} = \partial_{\nu}\chi^{\mu\nu} = \mp \sqrt{\frac{3}{32\pi|\Lambda|}}\,\partial_{\nu} {\cal B}^{\mu\nu}\, .
\end{align}
Consequently, the field equation~\eqref{eqA1 v6} equals:
\begin{align}
    \partial_{\nu} {\cal B}^{\mu\nu} = -6|\Lambda| b^{\mu} \pm 6\sqrt{2\pi|\Lambda|}\mnabla_{\nu}\cO^{\mu\nu}\, .
    \label{eqB2 v6}
\end{align}
Together, equations \eqref{eqB1 v6} and~\eqref{eqB2 v6} produce the non-homogeneous Proca equation -- cf. formula~\eqref{eq P1}:
\begin{align}
    \mBox \,b^{\mu} -  \mnabla^{\mu}\mnabla_{\nu}  b^{\nu} - b^{\sigma}\, \kolo{K}_{\sigma}^{\ \mu} - 2|\Lambda|\, b^{\mu}= \mp \sqrt{\frac{8\pi|\Lambda|}{3|\det g|}} \left(5\mnabla_{\nu}\cO^{[\mu\nu]} +3\mnabla_{\nu}\cO^{(\mu\nu)} \right)\, .
\end{align}
Here, the mass parameter equals:
\begin{align}
    \frac{m^2}{\hbar^2}=2|\Lambda|\, .
\end{align}

\subsection{Theory of the full Ricci tensor with a background field}

The passage to the metric picture is very similar to this one presented in \textbf{Chapter~\ref{pass full ricci tensor}}, due to the “variational absence'' of the traceless part of the Riemann curvature. As it was there, the metric Lagrangian is equal to the affine Lagrangian, whereas the matter Lagrangian is given by the following formula -- see~\eqref{lag matt ricci v0}:
\begin{align}
    \Lag_{\rm matt} = \Lag_A - \Lag_H\, .
    \label{lag matt ricci v1}
\end{align}
Of course, the above quantities have to be written in a proper control mode -- cf., the symplectic formula in the metric picture~\eqref{var lag matt}. 
Firstly, the affine Lagrangian~\eqref{lagA ricci W} equals:
\begin{align}
    \Lag_{A} = \frac{1}{8\pi \Lambda}\, \sqrt{\left|\det K + KKFF+KKFW +KKWW  \right|}\, ,
\end{align}
where the terms $KKFF$, $KKFW$ and $KKWW$ are defined in equations~(\ref{KKFF v10}–\ref{KKWW v10}). However, based on the full analysis presented in \textbf{Chapter~\ref{fixed back}}, and in particular the Einstein equation~\eqref{eqq13}, the above affine Lagrangian takes the form:
\begin{align}
    \Lag_A &= \frac{1}{8\pi\Lambda}\, \sqrt{|\Lambda^4  \det g + \Lambda^2\,\left( ggFF +ggFW + ggWW\right) |} = \nonumber\\
    &=\frac{\Lambda\sqrt{|\det g|}}{8\pi}\, \sqrt{\left|1 + \frac{1}{  \Lambda^2 \det g}\, \left( ggFF +ggFW + ggWW\right) \right|}=\nonumber\\
    &= \frac{\Lambda \sqrt{|\det g|}}{8\pi} \, \bigg| 1 + \frac{2}{\Lambda^2}\,\left(C_F \, F_{\mu\nu}F^{\mu\nu} +I_{FW}\, F_{\alpha\beta}\, W^{\alpha  \beta} + \right.\nonumber \\
    & \quad \left.+C_W \,  W_{\alpha \beta }\, W^{ \beta \alpha}-   C_W\, W_{\alpha\beta \kappa\lambda}\, W^{\kappa\lambda \alpha\beta} \right)\bigg|^{1/2}\, . \label{pass LagA v10}
\end{align}
Here, the terms $ggFF$, $ggFW$, and $ggWW$ were defined in~(\ref{ggFFV10}-\ref{ggWWV10}) .

\

Next, the Hilbert Lagrangian~\eqref{LagH} has to be derived, where the value of the metric Ricci curvature $\kolo{K}_{\mu\nu}$ is taken from the  Einstein equation~\eqref{eqq2v10}. Thus:
\begin{align}
    \Lag_H = \frac{\sqrt{|\det g|}}{16\pi}\, \left(4\Lambda + 6 A_{\sigma}A^{\sigma} \right)\, .
\end{align}
Then, the corresponding matter Lagrangian~\eqref{lag matt ricci v1} is the following:
\begin{align}
    \Lag_{\rm matt} &= \frac{\Lambda \sqrt{|\det g|}}{8\pi} \, \bigg| 1 + \frac{2}{\Lambda^2}\,\left(C_F \, F_{\mu\nu}F^{\mu\nu} +I_{FW}\, F_{\alpha\beta}\, W^{\alpha  \beta} +C_W \,  W_{\alpha \beta }\, W^{ \beta \alpha}+ \right.\nonumber \\
    & \quad \left.-   C_W\, W_{\alpha\beta \kappa\lambda}\, W^{\kappa\lambda \alpha\beta} \right)\bigg|^{1/2}  - \frac{\sqrt{|\det g|}}{8\pi}\, \left(2\Lambda + 3 A_{\sigma}A^{\sigma} \right)  \, .
    \label{Lmatt v10}
\end{align}
It can also be approximated (expanded around the $\Lambda$-vacuum solution) as follows:
\begin{align}
    \Lag_{\rm matt} &\approx  - \frac{ \sqrt{\left| \det g\right|}}{8\pi }\left( \Lambda + 3 A_{\kappa}A^{\kappa}\right) +\frac{ \sqrt{|\det g|}}{ 8\pi\Lambda}\,\bigg( C_F\, F_{\alpha\beta}\, F^{\alpha\beta}  +I_{FW}\, F_{\alpha\beta}\, W^{\alpha  \beta} +  \nonumber \\
    & \quad  +C_W   \, W_{\alpha \beta }\, W^{ \beta \alpha}-  C_W\, W_{\alpha\beta \kappa\lambda}\, W^{\kappa\lambda \alpha\beta} \bigg)   \, . 
    \label{app Lmatt v10}
\end{align}

\subsubsection{Field equations}
The field equations associated with the matter Lagrangian~\eqref{app Lmatt v10} are as follows -- cf. symplectic formula \eqref{var lag matt}:
\begin{enumerate}
    \item standard Euler-Lagrange system for the potential $A_{\mu}$ \eqref{mat eq A}, where:
    \begin{align}
        \frac{\partial \Lag_{\rm matt}}{\partial A_{\mu}}&= -\frac{3\sqrt{|\det g|}}{4\pi}\, A^{\mu}= -2 \cJ_{\mu}\, ,
        \label{eqA1 v10}\\
        \frac{\partial \Lag_{\rm matt}}{\partial F_{\mu\nu}}&=\frac{ \sqrt{|\det g|}}{ 8\pi\Lambda}\,\left( 2C_F\,   F^{\mu\nu}  +I_{FW}\,   W^{[\mu\nu]}\right) =\chi^{\mu\nu}\, ,
        \label{eqA2 v10}
    \end{align}
     where the first equation~\eqref{eqA1 v10} precisely reproduce the non-metricity equation  \eqref{A rel J 1}, which is the same as for the theory of the full Ricci tensor -- cf. \textbf{Chapter~\ref{sub noneq v10}}. The second one~\eqref{eqA2 v10} is equivalent with the constitutive relation \eqref{rel konst10};
    
    \item Einstein equation \eqref{passeineq}, which has to be identical with the previously obtained Einstein equation \eqref{eqq21v10}. To reconcile those two equations, the contractions between the metric tensor $g_{\mu\nu}$ and  traceless Riemann tensor $W$ have to be specified. It is easy to check that those terms are given by the following formulae:
\begin{align}
    F_{\alpha\beta}\, W^{\alpha\beta}&= F_{\alpha\beta}\, W^{\alpha}_{\ \kappa\lambda\gamma}\, g^{\kappa\lambda}\, g^{\beta\gamma}\, , \\
    W_{\alpha\beta}\, W^{\beta\alpha}&=W^{\alpha}_{\ \kappa\lambda \beta}\, W^{\beta}_{\ \mu\nu\alpha}\,  g^{\kappa\lambda}\, g^{\mu\nu}\, , \\
    W_{\alpha\beta \kappa\lambda}\, W^{\kappa\lambda \alpha\beta} &= W^{\alpha}_{\ \kappa \beta  \lambda}\, W^{\beta}_{\ \nu \alpha\mu}\, g^{\kappa\mu}\, g^{\lambda\nu}\, .
\end{align}
Then:
\begin{align}
    \frac{\partial \Lag_{\rm matt}}{\partial g_{\mu\nu}}&= - \frac{C_F\sqrt{|\det g|}}{4\pi\Lambda}\, \left(F^{\mu\alpha}\, F^{\nu}_{\ \alpha} - \frac 14\, F_{\alpha\beta}\, F^{\alpha\beta}\, g^{\mu\nu}\right) + \nonumber \\
       & \quad -\frac{C_W\sqrt{|\det g|}}{4\pi\Lambda}\,\left[ W_{\alpha\beta}\, W^{\beta(\mu\nu)\alpha} 
 - W^{\alpha(\mu|\kappa\lambda}\, W_{\kappa\lambda\alpha}^{\ \ \ |\nu)}  + \right. \nonumber \\ 
 &\quad \left. -\frac 14\, g^{\mu\nu}  \left(W_{\alpha\beta}\, W^{\beta\alpha} - W_{\alpha\beta\kappa\lambda}\,W^{\kappa\lambda\alpha\beta} \right) \right] + \nonumber \\
     & \quad -\frac{I_{FW}\sqrt{|\det g|}}{8\pi\Lambda}\  \left( F_{\alpha\beta}\, W^{\alpha (\mu\nu)\beta} + F^{\alpha (\mu} \, W_{\alpha}^{ \ \nu)}   - \frac 12\,  g^{\mu\nu}\, F_{\alpha\beta}\, W^{\alpha \beta}   \right)  \, . 
\end{align}
\end{enumerate}

\subsubsection{Unification}

The unification procedure is essentially the same as for the theory of the full Ricci tensor — see \textbf{Chapter~\ref{uni0}}. The only difference lies in the coupling constants, since in this theory the matter Lagrangian~\eqref{app Lmatt v10} differs slightly — cf.~formula~\eqref{app Lmatt v0}:
\begin{align}
- \frac{\sqrt{|\det g|}}{8\pi } \cdot 3A_{\kappa}A^{\kappa}     +\frac{ \sqrt{|\det g|}}{ 8\pi\Lambda}\,  C_F\, F_{\alpha\beta}\, F^{\alpha\beta} \, .
\end{align}
As before, the above terms suggest that the quantities $F_{\mu\nu}$ and $A_{\mu}$ can be interpreted as $B_{\mu\nu}$ and $b_{\mu}$ from Proca theory — see \textbf{Appendix~\ref{Proca th}}. Thus, the unification procedure proceeds as follows:
\begin{align}  
      F_{\mu\nu} := \sqrt{2\pi \left|\frac{\Lambda}{C_F} \right|}\, B_{\mu\nu}\, ,
    \label{uniF v10}
\end{align}  
under the assumption:
\begin{align}
    \frac{C_F}{\Lambda} = -\left|\frac{C_F}{\Lambda} \right| <0\, .
    \label{uniconst v10}
\end{align}
Within this theory, the cosmological constant \(\Lambda\) can take either a positive or a negative value, but it automatically fixes the sign of the coupling constant \(C_F\).  Accordingly, this leads to the same relation between the potential \( A_{\mu} \) and the Proca potential \( b_{\mu} \)~\eqref{Proca field} — cf.~formula~\eqref{uniA v0}:
\begin{align}
    A_{\mu} := \sqrt{2\pi \left|\frac {\Lambda}{C_F} \right|}\, b_{\mu}\, .
    \label{uniA v10}
\end{align}
  
As it was in the Variants $V_1$ and $V_6$, there appears an extra skew-symmetric field $W_{[\mu\nu]}$, which is coupled with the skew-symmetric Ricci tensor $F_{\mu\nu}$ -- see \eqref{eqA2 v10}. This background field is not considered in the unification procedure.

To reconstruct the same symplectic structure as in Proca theory \eqref{var lag P}, the momentum \( \chi^{\mu\nu} \)~\eqref{eqA2 v10} should be related to the dual  tensor density \( {\cal B}^{\mu\nu} \)~\eqref{const Proca} as follows:  
\begin{align}
    \chi^{\mu\nu} :=  -\sqrt{\frac{1}{8\pi}\left|\frac {C_F}{\Lambda} \right|}\, {\cal B}^{\mu\nu}\, .
    \label{chi v10}
\end{align}   
Then:
\begin{align}
    {\cal B}^{\mu\nu} = \sqrt{|\det g|}\, B^{\mu\nu}  + \frac{I_{FW} }{C_F}\,\sqrt{\frac{|\det g|}{8\pi}\left|\frac {C_F}{\Lambda} \right|} \, W^{[\mu\nu]}   \, .
    \label{eqB1 v10}
\end{align}
The same happens with the current $\cJ^{\mu}$~\eqref{cal J}:
\begin{align}
    \cJ^{\mu} = -\sqrt{\frac{1}{8\pi}\left|\frac {C_F}{\Lambda} \right|}\,   \,\partial_{\nu} {\cal B}^{\mu\nu}\, .
\end{align}
Consequently, the field equation~\eqref{eqA1 v10} equals:
\begin{align}
    \partial_{\nu} {\cal B}^{\mu\nu} = -\frac32 \left|\frac {\Lambda}{C_F} \right| b^{\mu} \, .
    \label{eqB2 v10}
\end{align}
Together, equations \eqref{eqB1 v10} and~\eqref{eqB2 v10} produce the non-homogeneous Proca equation -- cf. formula~\eqref{eq P1}:
\begin{align}
    \mBox \,b^{\mu} -  \mnabla^{\mu}\mnabla_{\nu}  b^{\nu} - b^{\sigma}\, \kolo{K}_{\sigma}^{\ \mu} - \frac32 \left|\frac {\Lambda}{C_F} \right|\, b^{\mu}= \frac{I_{FW} }{C_F}\,\sqrt{\frac{|\det g|}{8\pi}\left|\frac {C_F}{\Lambda} \right|} \,\mnabla_{\nu} W^{[\mu\nu]} \, ,
\end{align}
and it corresponds with the already derived equation \eqref{poteq v10}. Here, the mass parameter~\eqref{mpar v10} equals:
\begin{align}
    \frac{m^2}{\hbar^2}=\frac32 \left|\frac {\Lambda}{C_F} \right|\, .
\end{align}

    \chapter{Summary}

In this dissertation, the affine theory based on the full Riemann tensor is considered. The theory is described by the \textit{affine Lagrangian} $\Lag_A$, which depends on the symmetric affine connection $\Gamma$ and its first partial derivatives $\partial \Gamma$, but only through the Riemann tensor $R$: 
\[
\Lag_A(\Gamma,\partial\Gamma)=\Lag_{A}(R)\, .
\] 
The Riemann tensor $R$  algebraically decomposes into three independent components: the trace, called the \textit{Ricci tensor}, which itself splits into the symmetric Ricci tensor $K$ and the skew-symmetric Ricci tensor $F$, and the remaining traceless part of the Riemann tensor $W$. All these objects represent physical fields. Specifically, the symmetric Ricci tensor is associated with gravity, while the skew-symmetric Ricci tensor and the traceless part are believed to correspond to electromagnetism and dark matter, respectively.  

\

The variational structure of the theory is examined, and the corresponding field equations are derived. One of them links the non-metricity of the affine connection with the dependence of the theory on the skew-symmetric Ricci tensor $F$ or the traceless Riemann tensor $W$. In other words, the affine connection remains metric if and only if the affine Lagrangian $\Lag_A$ depends solely on the symmetric Ricci tensor $K$.  

\

Next, the construction of affine Lagrangians is discussed. The starting point is the already mentioned special case, in which the Lagrangian depends only on the symmetric Ricci tensor $K$. This theory is equivalent to standard $\Lambda$-vacuum gravity and is given by  
\[
\Lag_A=\frac{\sqrt{|\det K|}}{8\pi \Lambda}\, .
\]  
This model was already studied in the author’s  Bachelor thesis \cite{lic}. However, as a simple and elegant example, it is recalled and commented on here as well.  

The first generalisation relies on taking the determinant of the full Ricci tensor $K+F$. In this case, the Lagrangian has the form  
\[
\Lag_A=\frac{\sqrt{|\det (K+F)|}}{8\pi \Lambda}\, .
\]  
It turns out that the skew-symmetric Ricci tensor $F$ can be related to the electromagnetic Faraday tensor $f$ through a coupling constant (proportional to $\sqrt{|\Lambda|}$). 
Thus, this theory is closely related to Born-Infeld electromagnetism coupled with $\Lambda$-vacuum gravity. 
This result was also discussed in the author’s Bachelor thesis, but is recalled here for didactic purposes.  

All these theories are based on Lagrangians proportional to the square root of the determinant of the Ricci tensor (the trace of the Riemann tensor), which has two indices and can therefore be represented by a quadratic matrix. To extend the framework to the full Riemann tensor, which has four indices, the determinant-of-trace construction was modified. Unfortunately, there exist many possible modifications, leading to slightly different theories. All identified proposals are listed, but only two of them are examined in detail. This limitation is due to the highly complicated structure of such models, whose analysis requires significant space and time. Furthermore, an ``intermediate'' model is proposed, in which the Lagrangian explicitly depends on the full Riemann tensor (i.e.\ on all of its components), while the traceless part \(W\) — the principal source of complexity — is regarded as a prescribed \textit{background field}. This framework may be employed as a phenomenological description of cosmological effects.

\

The final part of the dissertation concerns the passage from the affine to the more familiar metric picture. While such a transition was already known in special cases (e.g.\ for Lagrangians depending only on the Ricci tensor), for the full Riemann tensor theory it is presented here for the first time. 
The main difficulty arises from the traceless part of the Riemann tensor $W$. The procedure is illustrated using the previously introduced examples.

\ 

Some of the results presented in this dissertation are ready to be published: specifically, the complete variational structure of the affine theory of the full Riemann tensor, the procedure for constructing affine Lagrangians, and the passage to the metric picture where the non-metricity of the connection appears as extra matter fields coupled to the standard theory of gravity. This will be done in the near future.

\

Further research is still ongoing. Initially, the theory based on the full Ricci tensor with a fixed background field (represented by the traceless part of the Riemann tensor $W$) should be investigated in more depth. The current knowledge of the behaviour of “dark matter” and of the Universe on cosmological scales is extremely limited and largely beyond our control, which makes it far more challenging than any other branch of physics. Moreover, the time scale of human observations is incomparable with the time spans required for processes such as galaxy collisions or galaxy formation, which are crucial for a better understanding of phenomena currently interpreted as “dark matter” or “dark energy”. Therefore, treating the background field $W$ as essentially constant or only very slowly varying appears to be a promising approach. Of course, the next steps should allow the ``dynamical'' interaction between the field $W$ and other fields.

A model was also proposed in which the determinant of the Ricci tensor (as a trace of the Riemann tensor) is perturbed by an “extra” constant matrix $\Delta$ acting on the Riemann tensor -- see \eqref{pertRDel}:
\begin{align}
    R_{\mu\nu} = R^{\kappa}_{\ \lambda\mu\nu}\, \delta^{\mu}_{\kappa} \longmapsto  R^{\kappa}_{\ \lambda\mu\nu}\, \left(\delta^{\mu}_{\kappa} + \Delta^{\mu}_{\kappa}\right)  = R_{\mu\nu} + R^{\kappa}_{\ \lambda\mu\nu}\,   \Delta^{\mu}_{\kappa}\, .
\end{align}
As mentioned, this matrix is traceless and its components can be chosen quite freely, which also makes room for other phenomenological models.

Another interesting and worthy-of-investigation aspect concerns the similarities between the totally traceless part of the Riemann tensor \(\widetilde{W}\) and the Lanczos field, which can be interpreted as a spin-2 field believed to describe the \textit{graviton}, the hypothetical particle associated with gravity. Analogously, studying the relation between the skew-symmetric Ricci tensor $F$ and the electromagnetic tensor $f$ is a natural direction of exploration, especially the Born–Infeld theory as an intermediate step between standard electrodynamics coupled with gravity and the unified affine theory of the full Riemann tensor.


\appendix 
\chapter{Classical electrodynamics}
\label{electro}

In classical electrodynamics,  the configuration space contains the potential 1-form~$a_{\mu}$ and its first derivatives $a_{\mu,\nu}$:
\begin{align}
\delta \mathcal{L}_{ed}\left(a_{\mu}, \, a_{\mu,\nu} \right) =\partial_{\nu}\left(\mathcal{F}^{\mu\nu}\, \delta a_{\mu} \right) =  \left(\partial_{\nu}\mathcal{F}^{\mu\nu}\right)\, \delta a_{\mu} + \mathcal{F}^{\mu\nu}\, \delta a_{\mu,\nu}\, ,
\end{align}
where $\cF^{\mu\nu}$ is a momentum canonically conjugated to $a_{\mu}$. Derivatives of the potential are organised in Faraday 2-form $f_{\mu\nu}$: 
\begin{align}
f=\dd a \Longrightarrow f_{\mu \nu} =a_{\nu,\mu } -a_{\mu, \nu}  \, .
\label{def: dwuforma faradaya}
\end{align}
Hence, the symplectic formula $\delta\Lag_{ed}$ takes the following form:
\begin{align}
\delta \mathcal{L}_{ed}\left(a_{\mu}, \, f_{\mu\nu} \right) =  \left(\partial_{\nu}\mathcal{F}^{\mu\nu}\right)\, \delta a_{\mu} -\frac 12\, \mathcal{F}^{\mu\nu}\, \delta f_{\mu\nu}\, .
\label{var lag em}
\end{align}
The above expression implies the skew-symmetry of the momentum $\cF^{\mu\nu}$. 

The definition of the Faraday 2-form~\eqref{def: dwuforma faradaya} geometrically guarantees the first pair of Maxwell equations:
\begin{align}
    \dd f= \dd^2 a =0\, ,
\end{align}
whereas the variational structure generates the second pair of Maxwell equations:
\begin{align}
    \frac{\partial \Lag_{ed}}{\partial a_{\mu}}  =   \partial_{\nu}\mathcal{F}^{\mu\nu} \, ,
    \label{maxeq2}
\end{align}
and the constitutive relation:
\begin{align}
    \frac{\partial \Lag_{ed}}{\partial f_{\mu\nu}}  = -\frac 12\, \cF^{\mu\nu}\, .
    \label{constrel}
\end{align}

\section{Vacuum electrodynamics}
The theory in the absence of any medium is described by a Lagrangian that depends only on the Faraday 2-form~\cite{Gravitation, ingarden}:
\begin{align}
    \Lag_{ed}= -\frac{\sqrt{|\det g|}}{4}\,f_{\alpha\beta}\,  f_{\mu\nu}\, g^{\alpha\mu}\, g^{\beta \nu} = -\frac{\sqrt{|\det g|}}{4}\,f_{\alpha\beta}\, f^{\alpha\beta}\, ,
    \label{Led}
\end{align}
where the metric serves as a fixed “\textit{background field}”. Then the constitutive relation~\eqref{constrel} implies:
\begin{align}
    \cF^{\mu\nu} = \sqrt{|\det g|}\, f^{\mu\nu}\, , 
    \label{constitutive}
\end{align}
whereas the second pair of Maxwell equations~\eqref{maxeq2} is given by the condition:
\begin{align}
\frac{\partial \mathcal{L}_{ed}}{\partial a_{\mu}}=0=   \partial_{\nu}\mathcal{F}^{\mu\nu} \,.
\end{align}
The symmetric stress-energy tensor density is defined as follows:
\begin{align}
    {\cal T}^{\mu\nu} := 2 \frac{\partial \Lag_{ed}}{\partial g_{\mu\nu}} = \sqrt{|\det g|}\,\left[f^{\mu \alpha}f^{\nu}_{\ \alpha} - \frac{1}{4}\, g^{\mu\nu}\, f_{\alpha\beta}\,f^{\alpha\beta}\right]\, .
    \label{calT ed}
\end{align}

\section{Electrodynamics with external sources}

The more general case includes the appearance of external sources which affect the electromagnetic field. Then the Lagrangian of such a system has the following form:
\begin{align}
    \Lag = \Lag_{ed} + \Lag_{\rm source} + \Lag_{ I}\, .
\end{align}
The simplest example of such  a system is charged dust, where the interaction term $\Lag_I$ is given by:
\begin{align}
   \Lag_{I} = \sqrt{|\det g|}\, j^{\mu}\, a_{\mu}\, , 
\end{align}
where $j^{\mu}$ contains the information about the matter. Then, the effective Lagrangian, which describes the dynamics of electromagnetic fields, is:
\begin{align}
    \Lag_{\rm eff} := \Lag_{ed} + \Lag_{I} = -\frac{\sqrt{|\det g|}}{4}\,f_{\alpha\beta}\, f^{\alpha\beta} + \sqrt{|\det g|}\, j^{\mu}\, a_{\mu}\, .
\end{align}
The constitutive relation~\eqref{constrel} stays the same as in the vacuum case:
\begin{align}
    \cF^{\mu\nu} = \sqrt{|\det g|}\, f^{\mu\nu}\, , 
\end{align}
but the second part of Maxwell equations~\eqref{maxeq2} is:
\begin{align}
      \partial_{\nu}\mathcal{F}^{\mu\nu} = \sqrt{|\det g|}\, j^{\mu}\, .
\end{align}
The right-hand side of the above equation describes the media as a \textit{source} of electromagnetic fields and is called a \textit{current density vector}.

\chapter{Proca theory}
\label{Proca th}

The theory proposed by A. Proca \cite{Proca}   describes the massive bosons with spin-1. Therefore, it was somehow an extension of the electrodynamics and Klein-Gordon scalar field. The particle is represented by the vector potential $b_{\mu}$, whose derivatives are combined in the closed 2-form $B_{\mu\nu}$:
\begin{align}
    B_{\mu\nu}:=b_{\nu,\mu} - b_{\mu\nu}\, ,
    \label{Proca field}
\end{align}
whereas the field equation is given by the ''Klein-Gordon''-like operator acting on the vector potential. To analyse such a theory, especially interactions, the Lagrangian formalism is necessary, thus:
\begin{align}
    \Lag_{P} = -\frac{\sqrt{|\det g|}}{4}\,\left( B_{\mu\nu}B^{\mu\nu} + 2\frac{m^2}{\hbar^2} \, b_{\mu}b^{\mu}\right)\, ,
    \label{LagP}
\end{align}
where $m$ denotes the mass of the boson, whereas $\hbar$ is the reduced Planck constant (or Dirac constant):
\begin{align}
    \hbar \approx  2.6\cdot 10^{-66}\, [\textbf{cm}^2]\, ,
    \label{hbar}
\end{align}
presented in the geometrical units -- for details see the red pages in \cite{Gravitation}. The variational formula is analogous to the electrodynamics one~\eqref{var lag em}:
\begin{align}
\delta \mathcal{L}_{P}\left(b_{\mu}, \, B_{\mu\nu} \right) = \left(\partial_{\nu}\mathcal{B}^{\mu\nu}\right)\, \delta b_{\mu} -\frac 12\, \mathcal{B}^{\mu\nu}\, \delta B_{\mu\nu}\, .
\label{var lag P}
\end{align}
 Whence, the field equations are the following:
 \begin{align}
     {\cal B}^{\mu\nu} &= -2 \frac{\partial \Lag_P}{\partial B_{\mu\nu}} = \sqrt{|\det g|}\, B^{\mu\nu}\, , \label{const Proca}\\
     \partial_{\nu}\mathcal{B}^{\mu\nu}&= \frac{\partial \Lag_P}{\partial b_{\mu}} = -\sqrt{|\det g|}\, \frac{m^2}{\hbar^2}\, b^{\mu}\, .
 \end{align}
 In \textbf{Lemma~\ref{lem div}} was presented proof that the partial divergence of a skew-symmetric tensor density is equal to the covariant divergence, so:
 \begin{align}
     \partial_{\nu}\mathcal{B}^{\mu\nu}&=\mnabla_{\nu}\mathcal{B}^{\mu\nu}=\sqrt{|\det g|}\, \mnabla_{\nu}B^{\mu\nu} =\sqrt{|\det g|}\, \mnabla_{\nu} \left( \mnabla^{\mu}b^{\nu} - \mnabla^{\nu}b^{\mu} \right) = \nonumber \\ 
     &=\sqrt{|\det g|}\left( \mnabla_{\nu}  \mnabla^{\mu}b^{\nu} - \mBox \,b^{\mu} \right) = -\sqrt{|\det g|}\, \frac{m^2}{\hbar^2}\, b^{\mu}\, .
 \end{align}
Using the \textbf{Lemma~\ref{lem comut}}, where the covariant derivatives commutation formula was presented, the following equality holds:
\begin{align}
     \mnabla_{\nu}  \mnabla^{\mu}b^{\nu}  = \mnabla^{\mu}\mnabla_{\nu}  b^{\nu} + b^{\sigma}\, \kolo{K}_{\sigma}^{\ \mu}\, .
\end{align}
Therefore, the equation for potential $b$ takes the following form:
\begin{align}
    \mBox \,b^{\mu} -  \mnabla^{\mu}\mnabla_{\nu}  b^{\nu} - b^{\sigma}\, \kolo{K}_{\sigma}^{\ \mu} - \frac{m^2}{\hbar^2}\, b^{\mu}=0\, .
    \label{eq P1}
\end{align}
To simplify it, the  covariant divergence has to be taken:
\begin{align}
    \mnabla_{\mu}\mBox \,b^{\mu} -  \mBox \mnabla_{\nu}  b^{\nu} - \mnabla_{\mu} \left(b^{\sigma}\, \kolo{K}_{\sigma}^{\ \mu}\right) - \frac{m^2}{\hbar^2}\, \mnabla_{\mu} b^{\mu}=0\, .
\end{align}
The first term was already calculated -- see formula~\eqref{com box}:
\begin{align}
    \mnabla_{\mu}\mBox \,b^{\mu} =   \mnabla^{\alpha}\left(b^{\beta} \kolo{K}_{\alpha\beta} \right)\, ,
\end{align}
therefore, the divergence $\mnabla_{\nu}  b^{\nu} $ satisfies the Klein-Gordon equation:
\begin{align}
     \mBox \mnabla_{\nu}  b^{\nu}+ \frac{m^2}{\hbar^2}\, \mnabla_{\mu} b^{\mu}=0\, .
\end{align}
Of course, assuming the Lorentz gauge:
\begin{align}
    \mnabla_{\nu}  b^{\nu} =0\, ,
\end{align}
is in coherence with the above equation and does not produce any contradictions. For this gauge, the Proca equation~\eqref{eq P1} takes the following form:
\begin{align}
    \mBox \,b^{\mu} =  \frac{m^2}{\hbar^2}\, b^{\mu}+ b^{\sigma}\, \kolo{K}_{\sigma}^{\ \mu} \, .
    \label{eq P2}
\end{align}
Obviously, for flat spacetimes, the above formula reduces to the “Klein-Gordon  equation'' for the vector field.

\ 

The symmetric stress-energy tensor density is defined as follows:
\begin{align}
    {\cal T}^{\mu\nu} := 2 \frac{\partial \Lag_{P}}{\partial g_{\mu\nu}} &= \sqrt{|\det g|}\,\left(B^{\mu \alpha}B^{\nu}_{\ \alpha} - \frac{1}{4}\, g^{\mu\nu}\, B_{\alpha\beta}\,B^{\alpha\beta}\right) + \nonumber \\
    & \quad +\sqrt{|\det g|}\, \frac{m^2}{\hbar^2}\, \left(b^{\mu}\,b^{\nu} - \frac 12\,g^{\mu\nu}\, b_{\alpha}b^{\alpha
    } \right)\, .
    \label{calT P}
\end{align}

 \chapter{Fierz-Lanczos theory}
 \label{FL th}
The Lanczos theory is used to describe the spin-2 particle, as a one-form of the electromagnetic potential describes the spin-1 particle. Moreover, the Lanczos field could be represented by a tensor that has identical properties to the Weyl tensor, which suggests the relation between the Lanczos potential and the connection (but only in the linearised case). All details and further references are presented in \cite{marian}.

\section{Lanczos potential}
 \label{lanczospot}
The mentioned procedure allows one to extract the Lanczos potential from an affine symmetric connection $\Gamma_{\kappa\lambda\mu}$ (not necessarily metric), where the first index $\kappa$ is lowered using the background metric tensor $g$. Originally, the construction of the Lanczos potential was based on the  \textit{linearised} symmetric connection (cf.~\cite{marian}). This \textit{linerisation} appeared as a perturbation of the metric structure:
\begin{align}
    g_{\mu\nu} \longmapsto g_{\mu\nu} + h_{\mu\nu}\, ,
\end{align}
where $h_{\mu\nu}$ is a small tensorial correction. However, the Lanczos field can be formulated for any perturbation of the metric connection, not necessarily related to the metric tensor  -- cf.~\eqref{decGamma}. 

The difference betwen the symmetric connection $\Gamma$ and the metric conection $\mGamma$, denoted as $N_{\kappa\lambda\mu}$~\eqref{decGamma},  has 40 independent components (due to the symmetry in the first two indices), whereas its totally symmetric part $N_{(\kappa\lambda\mu)}$ has 20 independent components. Thus, their difference:
\begin{align}
    \widetilde{N}_{\kappa\lambda\mu}:= {N}_{\kappa\lambda\mu} - {N}_{(\kappa\lambda\mu)}\, , 
    \label{difN}
\end{align}
also has 20 components. The next step consists of taking the skew-symmetric part with respect to the first two indices:
\begin{align}
    \widetilde{L}_{\kappa\lambda\mu} := \widetilde{N}_{[\kappa\lambda]\mu}\, ,
    \label{def tildeL}
\end{align}
followed by taking the “metric trace”:
\begin{align}
    \widetilde{L}_{\kappa} := \widetilde{L}_{\kappa\lambda\mu} \, g^{\lambda\mu}\, .
    \label{def L trace}
\end{align}
Finally, the Lanczos potential $L_{\kappa\lambda\mu}$ is defined by:
\begin{align}
    L_{\kappa\lambda\mu} := \widetilde{L}_{\kappa\lambda\mu} -\frac{1}{3}\, \left( \widetilde{L}_{\kappa}\, g_{\lambda\mu} - \widetilde{L}_{\lambda}\, g_{\kappa\mu} \right)\, .
    \label{def L pot}
\end{align}
Therefore, the Lanczos potential $L_{\kappa\lambda\mu}$ has 16 from 20 independent components, because the trace $\widetilde{L}_{\kappa}$ took 4 of them.  Interestingly, the skew-symmetrisation in formula~\eqref{def tildeL} does not change the number of independent components, but only reorganises them. It means that this relation could be inverted. Indeed, it holds that:
\begin{align}
    \widetilde{N}_{\kappa\lambda\mu} = \frac 34\, \widetilde{L}_{\kappa(\lambda\mu)}\, .
\end{align}
A similar relation connected Kijowski and Riemann tensors -- see formulae~\eqref{rel KR} and~\eqref{Riem od Kij}. 

\section{Relation between Lanczos potential and non-metricity tensor} 

In this dissertation, the above construction of the Lanczos potential is applied to the non-metricity tensor $N$~\eqref{rozklad tensora N tot}:
\begin{align}
N_{\kappa  \lambda \mu} = \widetilde{ A}_{\kappa \lambda \mu} -\frac{1}{18}\left( g_{\kappa \lambda}\, h_{\mu} + g_{\kappa \mu}\, h_{\lambda} - 5g_{\lambda\mu}\, h_{\kappa}\right) + \frac{2}{5}\,\left(g_{\kappa \lambda}\, A_{\mu} + g_{\kappa \mu}\, A_{\lambda} \right)\, .
\label{rozklad tensora N tot 1}
\end{align}
Its totally symmetric part is given by:
\begin{align}
    N_{(\kappa  \lambda \mu)} = \widetilde{A}_{(\kappa\lambda\mu)} + \frac{1}{6}\, g_{(\kappa\lambda}\, h_{\mu)}  + \frac{4}{5}\, g_{(\kappa\lambda}\, A_{\mu)} \, ,
\end{align}
whereas their difference~\eqref{difN} reads:
\begin{align}
    \widetilde{N}_{\kappa\lambda\mu} = \frac{1}{3} \left(2\widetilde{A}_{\kappa\lambda\mu} - \widetilde{A}_{\lambda\mu\kappa} - \widetilde{A}_{\mu\kappa\lambda} \right) + \frac{1}{3}\, g_{\lambda\mu}\, h_{\kappa} + \frac{2}{15} \left(g_{\kappa\lambda}\, A_{\mu} + g_{\kappa\mu}\, A_{\lambda} - 2g_{\lambda\mu}\, A_{\kappa} \right)\, . 
\end{align}
The next quantity, $\widetilde{L}_{\kappa\lambda\mu}$~\eqref{def tildeL}, is given by:
\begin{align}
    \widetilde{L}_{\kappa\lambda\mu} = \widetilde{A}_{[\kappa\lambda]\mu} - \frac{1}{3} \, g_{\mu[\kappa}\, h_{\lambda]} + \frac{4}{5}\, g_{\mu[\kappa}\, A_{\lambda]}\, ,
\end{align}
and its trace, $\widetilde{L}_{\kappa}$~\eqref{def L trace}, equals:
\begin{align}
    \widetilde{L}_{\kappa} = \frac{1}{2}\, h_{\kappa} - \frac{6}{5}\, A_{\kappa}\, .
\end{align}
Finally, the Lanczos potential~\eqref{def L pot} reduces to:
\begin{align}
    L_{\kappa\lambda\mu} = \widetilde{A}_{[\kappa\lambda]\mu}\, .
    \label{pot L ttA}
\end{align}
This means that the decomposition of the non-metricity tensor~\eqref{rozklad tensora N tot 1} can be written in the following form:
\begin{align}
N_{\kappa  \lambda \mu} = S_{\kappa \lambda \mu} + L_{\kappa \lambda \mu} - \frac{1}{18}\left( g_{\kappa \lambda}\, h_{\mu} + g_{\kappa \mu}\, h_{\lambda} - 5g_{\lambda\mu}\, h_{\kappa} \right) + \frac{2}{5}\,\left(g_{\kappa \lambda}\, A_{\mu} + g_{\kappa \mu}\, A_{\lambda} \right)\, ,
\label{rozklad tensora N tot 2}
\end{align}
where
\begin{align}
    S_{\kappa\lambda\mu} := \widetilde{A}_{(\kappa \lambda) \mu} \, ,
\end{align}
which has 16 independent components.

\ 

For the affine theory that does not depend on the traceless part of the Riemann tensor $W^{\kappa}_{\ \lambda\mu\nu}$, or equivalently, on the traceless part of the Kijowski tensor $U^{\kappa}_{\ \lambda\mu\nu}$, the associated Lanczos potential vanishes — see \textbf{Lemma~\ref{lem absence}}. However, when the whole curvature is present, then the associated Lanczos potential equals -- see  formula~\eqref{pot tA} in \textbf{Lemma~\ref{lem presence}}:
\begin{align}
     L_{\kappa\lambda\mu} &=   \widetilde{A}_{[\kappa\lambda]\mu}=\frac{16\pi}{\sqrt{|\det g|}}\, \mnabla_{\nu}\left[ \Omega_{[\kappa\lambda]\mu}^{\ \ \ \ \ \nu} + \frac13\, g_{\mu[\kappa}\, \cO_{\lambda]}^{\ \ \nu} \right] =\frac{16\pi}{\sqrt{|\det g|}}\, \mnabla_{\nu} \mathfrak{O}_{[\kappa\lambda]\mu}^{\ \ \ \ \ \nu}  \, ,  
 \end{align}
 whereas the last equality is obtained from \textbf{Lemma~\ref{lem dec Omega}}.

\section{Lanczos field}

The construction of the Lanczos field bases on the linearised Riemann tensor of the  corrections $N$ of the symmetric connection  $\Gamma$ (cf.~\eqref{decGamma} and~\eqref{metRimm}). Therefore, the ``linearised Riemann tensor'' is given by:

\begin{align}
     \mathfrak{R}_{\kappa   \lambda \mu \nu}:=\mnabla_{\mu} N_{\kappa  \lambda \nu } -\mnabla_{ \nu}N_{\kappa \lambda \mu}\, ,
\end{align}
where the $\kappa$ index is lowered by the background metric $g$ and $\mnabla$ is a covariant derivative associated with the background metric structure. The next step contains the following symmetrisation:
\begin{align}
    \mathfrak{r}_{\kappa    \lambda \mu \nu}:=\mathfrak{R}_{[\kappa    \lambda] \mu \nu} + \mathfrak{R}_{ [\mu \nu]\kappa    \lambda}\, .
\end{align}
Interestingly, an above symmetrisation could be simplified via the following lemma:
\begin{lemma}
\label{lem lanczpot}
For any tensor $T_{\alpha\beta \mu\nu} $ which satisfies
\begin{align}
        T_{\alpha\beta (\mu\nu)}&=0=T_{(\alpha\beta)\mu \nu}\, , & T_{\alpha[ \kappa\lambda\mu]}&=0\, ,
    \end{align}
    the following identity holds:
    \begin{align}
        T_{\alpha\beta \mu\nu} = T_{\mu\nu \alpha\beta }\, .
    \end{align}     
\end{lemma}
\begin{proof}
    \begin{align}
        T_{\alpha\beta \mu\nu}&= - T_{\alpha \mu\nu\beta} - T_{\alpha\nu\beta\mu} =T_{\mu\alpha \nu\beta} +T_{\nu\alpha \beta\mu} =-  T_{\mu \nu\beta\alpha} -T_{\mu \beta\alpha\nu}  -T_{\nu\beta\mu\alpha } -T_{\nu\mu\alpha\beta }=\nonumber\\
        &=2T_{\mu \nu\alpha\beta} +T_{\beta \mu\alpha \nu} +T_{\beta \nu \mu\alpha} = 2T_{\mu \nu\alpha\beta}  -T_{\beta \alpha\nu\mu} =2T_{\mu \nu\alpha\beta}  -T_{\alpha\beta \mu\nu}\, , 
    \end{align}
    what finishes the proof.
\end{proof}
It means that:
\begin{align}
    \mathfrak{r}_{\kappa    \lambda \mu \nu}= 2 \mathfrak{R}_{[\kappa    \lambda] \mu \nu}\,  .
    \label{frak r}
\end{align}
The Lanczos field $\mathfrak{L}_{\kappa    \lambda \mu \nu}$ is a totally traceless part of the above tensor, thus:
\begin{align}
    \mathfrak{L}_{\kappa\lambda\mu\nu} &:= \mathfrak{r}_{\kappa    \lambda \mu \nu} - \frac 12\, \left( \mathfrak{r}_{\alpha\mu}\, g_{\nu\beta} -\mathfrak{r}_{\alpha\nu}\,g_{\mu\beta} + g_{\alpha\mu}\, \mathfrak{r}_{\nu\beta} -g_{\alpha\nu}\,\mathfrak{r}_{\mu\beta} \right)+\nonumber\\
    & \quad +\frac { \mathfrak{r}}6\, \left(g_{\alpha\mu}\, g_{\beta\nu} - g_{\alpha\nu}\,g_{\beta\mu} \right)\, ,
    \label{def lanczos field}
\end{align}
where
\begin{align}
     \mathfrak{r}_{\mu\nu}&:= \mathfrak{r}_{\alpha\mu\beta\nu}\, g^{\alpha\beta}\, ,&  \mathfrak{r}&:=  \mathfrak{r}_{\alpha\beta}\, g^{\alpha\beta}\, .
\end{align}

\section{Relation between Lanczos field and algebraically traceless Riemann tensor}

\sectionmark{Lanczos field and traceless Riemann tensor}
Since the Lanczos potential \(L_{\mu\nu\kappa}\), related to the non-metricity tensor \(N\), was equal to the totally traceless part \( \widetilde{A}_{[\mu\nu]\kappa} \) \eqref{pot L ttA}, the above procedure can be used for the non-metricity part of the linearised traceless tensor \(W\)~\eqref{rozklad pelny W}, denoted by \(C\) (see~\eqref{def tensor C}):
\begin{align}
    C_{\kappa\lambda\mu\nu}
    &:= \mathrm{linear}\!\bigl(W^{\alpha}{}_{\lambda\mu\nu} - \kolo{W}^{\alpha}{}_{\lambda\mu\nu}\bigr)\, g_{\kappa\alpha} \\
    &= \mnabla_{\mu}A_{\kappa\nu\lambda} - \mnabla_{\nu}A_{\kappa\mu\lambda}
      + \tfrac{1}{3}\!\left( g_{\kappa\nu}\,\mnabla_{\sigma}A^{\sigma}{}_{\mu\lambda}
      - g_{\kappa\mu}\,\mnabla_{\sigma}A^{\sigma}{}_{\nu\lambda}\right)\,.
\end{align}
Hence
\begin{align}
    \mathfrak{r}_{\kappa\lambda\mu\nu} &= C_{[\kappa\lambda]\mu\nu} + C_{[\mu\nu]\kappa\lambda}\,, \\
    \mathfrak{r}_{\mu\nu} &= -C_{(\mu|\alpha\beta|\nu)}\, g^{\alpha\beta} = -C_{(\mu\nu)}\,, \\
    \mathfrak{r} &= -C_{(\mu\nu)}\, g^{\mu\nu} = 0\,.
\end{align}
Finally,
\begin{align}
    \mathfrak{L}_{\kappa\lambda\mu\nu}
    &= C_{[\kappa\lambda]\mu\nu} + C_{[\mu\nu]\kappa\lambda}
       + \tfrac{1}{2}\!\left( C_{(\alpha\mu)}\, g_{\nu\beta} - C_{(\alpha\nu)}\,g_{\mu\beta}
       + g_{\alpha\mu}\, C_{(\nu\beta)} - g_{\alpha\nu}\,C_{(\mu\beta)} \right).
\end{align}

An interesting (and nontrivial) inverse problem is how to decompose the linearised algebraically traceless Riemann tensor \(C\) into the Lanczos field and the remaining part. The decomposition for the totally traceless part and traces is provided in the following lemma (cf.\ \textbf{Lemma~\ref{lemm W dec}}):

\begin{lemma}
If the tensor \(C_{\kappa\lambda\mu\nu}\) satisfies
\begin{align}
    C_{\kappa[\lambda\mu\nu]}&=0\,, & C_{\kappa\lambda(\mu\nu)}&=0\,, 
    & C_{\kappa\lambda\mu\nu}\, g^{\kappa\lambda}&=0\,, & C_{\kappa\lambda\mu\nu}\, g^{\kappa\mu}&=0\,,
\end{align}
then it admits the decomposition
\begin{align}
    C_{\kappa\lambda\mu\nu}
    &= \widetilde{C}_{\kappa\lambda\mu\nu}
      -\tfrac{1}{6}\, g_{\kappa\lambda}\, C_{[\mu\nu]}
      + \tfrac{1}{8}\!\left(g_{\kappa\nu}\,C_{(\lambda\mu)} - g_{\kappa\mu}\, C_{(\lambda\nu)}\right)
      + \tfrac{1}{12}\!\left(g_{\kappa\nu}\,C_{[\lambda\mu]} - g_{\kappa\mu}\, C_{[\lambda\nu]}\right)\nonumber\\
    &\quad +\tfrac{3}{8}\!\left(C_{(\kappa\nu)}\,g_{\lambda\mu} - C_{(\kappa\mu)}\, g_{\lambda\nu}\right)
      + \tfrac{5}{12}\!\left(C_{[\kappa\nu]}\,g_{\lambda\mu} - C_{[\kappa\mu]}\, g_{\lambda\nu}\right),
\end{align}
where \(\widetilde{C}_{\kappa\lambda\mu\nu}\) is the totally traceless part and
\begin{align}
    C_{\kappa\nu}:=C_{\kappa\lambda\mu\nu}\, g^{\lambda\mu}\, .
\end{align}
\end{lemma}

\begin{proof}
The proof is a straightforward verification of the stated identities and symmetries.
\end{proof}

The last step is the extraction of the Lanczos field \(\mathfrak{L}\) from \(\widetilde{C}\). Since \(\widetilde{C}\) has no definite symmetry in its first two indices, it can be split into symmetric and skew-symmetric parts:
\begin{align}
    \widetilde{C}_{\kappa\lambda\mu\nu}
    &= \widetilde{C}_{[\kappa\lambda]\mu\nu} + \widetilde{C}_{(\kappa\lambda)\mu\nu} \nonumber\\
    &= \tfrac{1}{2}\!\left(\widetilde{C}_{[\kappa\lambda]\mu\nu} + \widetilde{C}_{[\mu\nu]\kappa\lambda}\right)
       + \tfrac{1}{2}\!\left(\widetilde{C}_{[\kappa\lambda]\mu\nu} - \widetilde{C}_{[\mu\nu]\kappa\lambda}\right)
       + \widetilde{C}_{(\kappa\lambda)\mu\nu} \\
    &= \tfrac{1}{2}\,\mathfrak{L}_{\kappa\lambda\mu\nu} + \tfrac{1}{2}\,\mathfrak{M}_{\kappa\lambda\mu\nu}
       + \widetilde{C}_{(\kappa\lambda)\mu\nu},
\end{align}
where
\begin{align}
\mathfrak{L}_{\kappa\lambda\mu\nu} &=\widetilde{C}_{[\kappa\lambda]\mu\nu} - \widetilde{C}_{[\mu\nu]\kappa\lambda}\, ,\, \\
    \mathfrak{M}_{\kappa\lambda\mu\nu} &:= \widetilde{C}_{[\kappa\lambda]\mu\nu} - \widetilde{C}_{[\mu\nu]\kappa\lambda}.
\end{align}

The traceless Riemann tensor \(C_{\kappa\lambda\mu\nu}\) has \(64\) independent components, and its trace \(C_{\kappa\lambda}\) has \(15\), hence \(\widetilde{C}_{\kappa\lambda\mu\nu}\) has \(49\) independent components. The Lanczos field \(\mathfrak{L}\) carries \(10\) degrees of freedom, so the remaining parts account for \(39\) independent parameters. These objects satisfy the algebraic properties
\begin{align}
    \mathfrak{L}_{[\kappa\lambda\mu\nu]}&=0\,, & \mathfrak{M}_{(\kappa\lambda\mu\nu)}&=0\,.
\end{align}

\newpage

    \bibliographystyle{acm}     
     \addcontentsline{toc}{chapter}{Bibliography}
    \bibliography{bibliografia}    

@article{Zeldovich1,
  author    = {Ya. Zel'dovich},
  title     = {Cosmological constant and elementary particles},
  journal   = {Soviet Physics JETP},
  year      = {1967},
  volume    = {6},
  pages     = {316}
}

@article{Zeldovich2,
  author    = {Ya. Zel'dovich},
  title     = {The Cosmological constant and the theory of elementary particles},
  journal   = {Soviet Physics Uspekhi},
  year      = {1968},
  volume    = {11},
  number    = {3},
  pages     = {381}
}

@article{borowiec,
     AUTHOR = {Borowiec, Andrzej and Godłowski, Włodzimierz and Szydłowski,
              Marek},
     TITLE = {Dark matter and dark energy as effects of modified gravity},
   JOURNAL = {Int. J. Geom. Methods Mod. Phys.},
  FJOURNAL = {International Journal of Geometric Methods in Modern Physics},
    VOLUME = {4},
      YEAR = {2007},
    NUMBER = {1},
     PAGES = {183}
}

@article{horava,
  author    = {P. Ho{\v{r}}ava},
  title     = {Quantum gravity at a {L}ifshitz point},
  JOURNAL = {Phys. Rev. D},
    VOLUME = {79},
      YEAR = {2009},
    NUMBER = {8},
     PAGES = {084008}
}

@book{eddington,
  author    = {A. S. Eddington},
  title     = {The {M}athematical {T}heory of {R}elativity},
  publisher = {Cambridge University Press},
  year      = {1923}
}

@article{cosmo,
  author    = {D. M. Ghilencea},
  title     = {Non-metric geometry as the origin of mass in gauge theories of scale invariance},
  journal   = {Eur. Phys. J. C},
  year      = {2023},
  volume    = {83},
  pages     = {176}
}

@article{vollick,
  author  = {D. N. Vollick},
  title   = {Palatini approach to {B}orn-{I}nfeld-{E}instein theory and a geometric description of electrodynamics},
  journal = {Phys. Rev. D},
  volume  = {69},
  pages   = {064030},
  year    = {2004}
}

@article{banados,
  author  = {M. Ba{\~n}ados and P. G. Ferreira},
  title   = {Eddington’s theory of gravity and its progeny},
  journal = {Phys. Rev. Lett.},
  volume  = {105},
  pages   = {011101},
  year    = {2010},
  note    = {Erratum: Phys. Rev. Lett. 113, 119901 (2014)}
}

@article{Extended,
  author  = {Salvatore Capozziello and Mariafelicia De Laurentis},
  title   = {Extended {T}heories of {G}ravity},
  journal = {Phys. Rep.},
  volume  = {509},
  pages   = {167},
  year    = {2011}
}

@article{fRgravity,
  author  = {S. Capozziello and R. Cianci and C. Stornaiolo and S. Vignolo},
  title   = {f({R}) gravity with torsion: {A} geometric approach within the {J}-bundles framework},
  journal = {Int. J. Geom. Meth. Mod. Phys.},
  volume  = {5},
  pages   = {765},
  year    = {2008}
}

@article{Hor-Lif,
  author  = {T. P. Sotiriou},
  title   = {Hořava-{L}ifshitz gravity: a status report},
  journal = {J. Phys.  Conf. Ser.},
  volume  = {283},
  pages   = {012034},
  year    = {2011}
}

@article{Szczyrba,
  author  = {W. Szczyrba},
  title   = {A symplectic structure on the set of {E}instein metrics},
  journal = {Commun. Math. Phys.},
  volume  = {51},
  pages   = {163},
  year    = {1976}
}

@book{Mercury,
  author    = {U. J. {Le Verrier}},
  title     = {Théorie du mouvement de {M}ercure},
  publisher =  {Annales de l'Observatoire Impérial de Paris},
  year      = {1859}
}

@book{newton,
  author    = {Isaac Newton},
  title     = {Philosophiae {N}aturalis {P}rincipia {M}athematica},
  publisher = {London},
  year      = {1687}
}

@article{clifford,
  author  = {W. Clifford},
  title   = {On the {S}pace-{T}heory of {M}atter},
  journal = {Proc. Cambridge Phil. Soc.},
  volume  = {2},
  year    = {1876}
}

@article{ein1,
  author  = {A. Einstein},
  title   = {Zur allgemeinen {R}elativitätstheorie},
  journal = {Sitzungsberichte der Königlich Preußischen Akademie der Wissenschaften (Berlin)},
  pages   = {778},
  year    = {1915}
}

@article{hil,
  author  = {D. Hilbert},
  title   = {Die {G}rundlagen der {P}hysik},
  journal = {Nachrichten von der Königlichen Gesellschaft der Wissenschaften zu Göttingen, Mathematisch-Physikalische Klasse},
  pages   = {395},
  year    = {1915}
}

@article{ein2,
  author  = {A. Einstein},
  title   = {Die {G}rundlage der allgemeinen {R}elativitätstheorie},
  journal = {Annalen der Physik},
  volume  = {354},
  number  = {7},
  pages   = {769},
  year    = {1916}
}

@article{weyl,
  author  = {H. Weyl},
  title   = {Gravitation und {E}lektrizität},
  journal = {Sitzungsberichte der Königlich Preußischen Akademie der Wissenschaften zu Berlin},
  pages   = {465},
  year    = {1918}
}

@article{palatini,
  author  = {A. Palatini},
  title   = {Deduzione invariantiva delle equazioni gravitazionali dal principio di {H}amilton},
  journal = {Rendiconti del Circolo Matematico di Palermo},
  volume  = {43},
  pages   = {203},
  year    = {1919}
}

@article{buchdahl,
  author  = {H. A. Buchdahl},
  title   = {Non-linear {L}agrangians and cosmological theory},
  journal = {Mon. Not. Roy. Astron. Soc.},
  volume  = {150},
  pages   = {1},
  year    = {1970}
}

@article{lovelock,
  author  = {D. Lovelock},
  title   = {The {E}instein tensor and its generalizations},
  journal = {J. Math. Phys.},
  volume  = {12},
  number  = {3},
  pages   = {498},
  year    = {1971}
}

@article{horndeski,
  author  = {G. W. Horndeski},
  title   = {Second-order scalar-tensor field equations in a four-dimensional space},
  journal = {Internat. J. Theoret. Phys.},
  volume  = {10},
  number  = {6},
  pages   = {363},
  year    = {1974}
}

@article{harada1,
  author  = {J. Harada},
  title   = {Dark energy in conformal {K}illing gravity},
  JOURNAL = {Phys. Rev. D},
    VOLUME = {108},
      YEAR = {2023},
    NUMBER = {10},
     PAGES = {104037}
}

@article{harada2,
  author  = {J. Harada},
  title   = {Emergence of the {C}otton tensor for describing gravity},
  JOURNAL = {Phys. Rev. D}, 
    VOLUME = {103},
      YEAR = {2021},
    NUMBER = {11},
     PAGES = {L121502},
}

@article{lanczos,
  author  = {C. Lanczos},
  TITLE = {Lagrangian multiplier and {R}iemannian spaces},
   JOURNAL = {Rev. Mod. Phys.},
    VOLUME = {21},
      YEAR = {1949},
     PAGES = {497},
}

@book{lic,
  author    = {B. Bąk},
  title     = {A contribution to unification of gravity and electromagnetism},
  publisher = {Bachelor's thesis (written in Polish)},
  year      = {2018},
}

@book{mag,
  author    = {B. Bąk},
  title     = {The variational principles in the general relativity theory},
  publisher = {Master's thesis (written in Polish)},
  year      = {2020},
}

@article{nonmetricity,
  author    = {B. Bąk and J. Kijowski},
  title     = {How the non-metricity of the connection arises naturally in the classical theory of gravity},
  journal   = {Journal of Mathematical Physics},
  volume    = {65},
  pages     = {092501},
  year      = {2024}
}

@article{APP,
  author    = {B. Bąk},
  title     = {Variational formulations of {G}eneral {R}elativity},
  journal   = {Acta Physica Polonica B, Proceedings Supplement},
  volume    = {15},
  pages     = {A1.1},
  year      = {2022}
}

@article{kobayashi,
  author    = {S. Kobayashi},
  title     = {Theory of connections},
  journal   = {Annali di Matematica Pura ed Applicata},
  volume    = {43},
  pages     = {119},
  year      = {1957}
}

@book{Gravitation,
  author    = {C. W. Misner and K. S. Thorne and J. A. Wheeler},
  title     = {Gravitation},
  publisher = {W. H. Freeman and Company},
  address   = {San Francisco},
  year      = {1973}
}

@book{Weinberg,
  author    = {S. Weinberg},
  title     = {Gravitation and {C}osmology: {P}rinciples and {A}pplications of the {G}eneral {T}heory of {R}elativity},
  publisher = {John Wiley \& Sons},
  year      = {1972}
}

@book{ingarden,
  author    = {R. S. Ingarden and A. Jamiołkowski},
  title     = {Classical {E}lectrodynamics},
  publisher = {PWN Scientific Publishers},
  address   = {Warsaw},
  year      = {1980}
}

@article{born-infeld,
  author    = {M. Born and L. Infeld},
  title     = {Foundations of the {N}ew {F}ield {T}heory},
  journal   = {Proceedings of the Royal Society A},
  volume    = {144},
  pages     = {425},
  year      = {1934}
}

@book{CJK,
  author    = {P. T. Chruściel and J. Jezierski and J. Kijowski},
  title     = {Hamiltonian {F}ield {T}heory in the {R}adiating {R}egime},
  series    = {Lecture Notes in Physics Monographs},
  volume    = {70},
  publisher = {Springer},
  year      = {2001}
}

@article{Kij-Moreno2015,
  author    = {J. Kijowski and G. Moreno},
  title     = {Symplectic structures related to higher order variational problems},
  journal   = {Int. J. Geom. Meth. Mod. Phys.},
  volume    = {12},
 pages     = {1550084},
  year      = {2015}
}

@article{Kij-Fer,
  author    = {M. Ferraris and J. Kijowski},
  title     = {General Relativity is a gauge type theory},
  journal   = {Lett. Math. Phys.},
  volume    = {5},
    pages    = {127},
  year      = {1981}
}

@book{Geometria,
  author    = {J. Kijowski},
  title     = {Geometria różniczkowa jako narzędzie nauk przyrodniczych},
  publisher = {Monografie CSZ},
  address   = {Warszawa},
  year      = {2015}
}

@article{Proca,
  author  = {A. Proca},
  title   = {Sur la théorie ondulatoire des électrons positifs et négatifs},
  journal = {Journal de Physique et le Radium},
  volume  = {7},
  year    = {1936}
}

@book{Fock,
  author    = {V. A. Fock},
  title     = {The {T}heory of {S}pace, {T}ime and {G}ravitation},
  publisher = {Pergamon Press},
  year      = {1964}
}

@article{pieszy,
  author  = {J. Kijowski},
  title   = {A simple derivation of canonical structure and quasi-local {H}amiltonians in general relativity},
  journal = {Gen. Rel. Grav.},
  volume  = {29},
  pages   = {307},
  year    = {1997}
}

@article{newvariationalprinciple,
  author  = {J. Kijowski},
  title   = {On a new variational principle in general relativity and the energy of gravitational field},
  journal = {Gen. Rel. Grav.},
  volume  = {9},
  pages   = {857},
  year    = {1978}
}

@article{universality,
  author  = {J. Kijowski},
  title   = {Universality of the {E}instein theory of gravitation},
  journal = {Int. J. Geom. Meth. Mod. Phys.},
  volume  = {13},
  number  = {8},
  year    = {2016}
}

@article{senger,
  author  = {J. Kijowski and K. Senger},
  title   = {Covariant jets of a connection and higher order curvature tensors},
  journal = {J. Geom. Phys.},
  volume  = {163},
  pages   = {104092},
  year    = {2021}
}

@book{Tulcz,
  author    = {J. Kijowski and W. M. Tulczyjew},
  title     = {A {S}ymplectic {F}ramework for {F}ield {T}heories},
  publisher = {Springer},
  series    = {Lecture Notes in Physics},
  volume    = {107},
  year      = {1979}
}

@article{marian,
  author  = {J. Jezierski and J. Kijowski and M. Wiatr},
  title   = {Localizing energy in {F}ierz-{L}anczos theory},
  journal = {Phys. Rev. D},
  volume  = {102},
  pages   = {024015},
  year    = {2020}
}

@article{kij2024,
  author  = {J. Kijowski},
  title   = {Gravity on a {L}arge {S}cale — {D}oes {I}t {N}ecessarily {L}ook like {I}t {D}oes on a {S}mall {S}cale?},
  journal = {Astronomy},
  volume  = {3},
  pages   = {29},
  year    = {2024}
}

@article{olmo,
  author  = {V. I. Afonso and C. Bejarano and R. Ferraro and G. J. Olmo},
  title   = {Determinantal {B}orn-{I}nfeld coupling of gravity and electromagnetism},
  journal = {Phys. Rev. D},
  volume  = {105},
  pages   = {084067},
  year    = {2022}
}

@article{indie1,
  author  = {S. Jana and S. Kar},
  title   = {Born-{I}nfeld gravity coupled to {B}orn-{I}nfeld electrodynamics},
  journal = {Phys. Rev. D},
  volume  = {92},
  pages   = {084004},
  year    = {2015}
}

@article{string,
  author  = {A. A. Tseytlin},
  title   = {Born-{I}nfeld action, supersymmetry and string theory},
  journal = {arXiv:hep-th/9908105},
  year    = {1999}
}

@book{kerner,
  author    = {R. Kerner},
  title     = {Our {C}elestial {C}lockwork: {F}rom {A}ncient {O}rigins to {M}odern {A}stronomy of the {S}olar {S}ystem},
  publisher = {World Scientific Publishing Co},
  year      = {2021},
}

@book{AKW,
  author    = {A. K. Wróblewski},
  title     = {Historia fizyki},
  publisher = {Wydawnictwo Naukowe PWN},
  address   = {Warszawa},
  year      = {2007}
}

@article{rogatko2024,
  author  = {M. Rogatko},
  title   = {Dark photon–dark energy stationary axisymmetric black holes},
  journal = {Phys. Rev. D},
  volume  = {109},
  pages   = {104030},
  year    = {2024}
}

@article{holdom,
  author  = {B. Holdom},
  title   = {Two \( U(1) \)'s and \(\epsilon\) charge shifts},
  journal = {Phys. Lett. B},
  volume  = {166},
  pages   = {196},
  year    = {1986}
}

@book{thedarkphoton,
  author       = {M. Fabbrichesi and E. Gabrielli and G. Lanfranchi},
  title        = {The {P}hysics of the {D}ark {P}hoton},
    publisher = {Springer Cham},
  year      = {2020}
}
\end{document}